\documentclass[preprint]{elsarticle}

\journal{Journal of Computational Physics}

\def\d{\delta} 
 
\def\e{{\epsilon}}

\def\norm#1{\|#1\|} 
\def\wb{{\bar w}}

\usepackage{algorithm}
\usepackage{algorithmic}
\usepackage{amssymb}
\usepackage{amsmath}
\usepackage{amsthm}
\usepackage{caption}
\usepackage{comment}
\usepackage{float}
\usepackage[latin1]{inputenc}
\usepackage[pdfborder={0 0 0}, colorlinks=true, linkcolor=blue, citecolor=red]{hyperref}\hypersetup{pdfborder=0 0 0}
\usepackage{mathrsfs}
\usepackage{multicol}
\usepackage{multirow}
\usepackage{subfigure}
\usepackage{wrapfig}
\usepackage{ulem}
\usepackage{soul,xargs}
\usepackage[pdftex,dvipsnames]{xcolor}

\usepackage[colorinlistoftodos,prependcaption,textsize=tiny]{todonotes}
\newcommandx{\question}[2][1=]{\todo[linecolor=red,backgroundcolor=red!20,bordercolor=red,#1]{#2}}
\newcommandx{\change}[2][1=]{\todo[linecolor=blue,backgroundcolor=blue!25,bordercolor=blue,#1]{#2}}
\newcommandx{\info}[2][1=]{\todo[linecolor=OliveGreen,backgroundcolor=OliveGreen!25,bordercolor=OliveGreen,#1]{#2}}
\newcommandx{\improve}[2][1=]{\todo[linecolor=Plum,backgroundcolor=Plum!25,bordercolor=Plum,#1]{#2}}
\newcommandx{\thiswillnotshow}[2][1=]{\todo[disable,#1]{#2}}
\newcommandx{\answer}[2][1=]{\todo[linecolor=blue,backgroundcolor=White!25,bordercolor=Plum,#1]{#2}}

\newtheorem{proposition}{Proposition}
\newtheorem{remark}{Remark}
\graphicspath{{./PaperFigures/}}

\begin{document}

% MY MACROs
\newcommand{\TODO}[1]{ \fbox{\parbox{3in}{\bf TODO: #1}}}

\newcommand{\grbf}[1] {\mbox{\boldmath${#1}$\unboldmath}}
\newcommand{\gbf}[1] {\mathbf{#1}}

\newcommand{\beq} {\begin{equation}}
\newcommand{\eeq} {\end{equation}}
\newcommand{\bdm} {\begin{displaymath}}
\newcommand{\edm} {\end{displaymath}}
\newcommand{\bit}{\begin{itemize}}
\newcommand{\eit}{\end{itemize}}
\newcommand{\bde}{\begin{description}}
\newcommand{\ede}{\end{description}}
\newcommand{\bce}{\begin{center}}
\newcommand{\ece}{\end{center}}
\newcommand{\ben} {\begin{enumerate}}
\newcommand{\een} {\end{enumerate}}
\newcommand{\bea} {\begin{eqnarray}}
\newcommand{\eea} {\end{eqnarray}}
\newcommand{\barr} {\begin{array}}
\newcommand{\earr} {\end{array}}
\newcommand{\bean} {\begin{eqnarray*}}
\newcommand{\eean} {\end{eqnarray*}}
\newcommand{\edoc} {\end{document}}
\newcommand{\On}{\hspace{0.2cm} \mbox{on} \hspace{0.2cm}}
\newcommand{\In}{\hspace{0.2cm} \mbox{in} \hspace{0.2cm}}
\newcommand{\Minimize}{\hspace{0.2cm} \mbox{minimize} \hspace{0.2cm}}
\newcommand{\Maximize}{\hspace{0.2cm} \mbox{maximize} \hspace{0.2cm}}
\newcommand{\SubjectTo}{\hspace{0.2cm} \mbox{subject} \hspace{0.15cm}
            \mbox{to} \hspace{0.2cm}}

\def\etal{{\it et al.~}}

\newcommand{\point}[1] {\ensuremath{\boldsymbol{#1}}}
\newcommand{\vvec}[1] {\ensuremath{\boldsymbol{#1}}}
\renewcommand{\vec}[1] {\ensuremath{\boldsymbol{#1}}}
\newcommand{\cvec}[1] {\ensuremath{\boldsymbol{{#1}}}}
\newcommand{\ten}[1] {\ensuremath{\boldsymbol{#1}}}
\newcommand{\tenfour}[1] {\ensuremath{\boldsymbol{\mathsf{#1}}}}
\newcommand{\comp}[1]{\begin{pmatrix}{ #1 }\end{pmatrix}}
\newcommand{\vv}{\ensuremath{\vec{v}}}
\newcommand{\vr}{\left(\ensuremath{\vec{r}}\right)}
\newcommand{\att}{\left(t\right)}
\newcommand{\tvv}{\ensuremath{\tilde{\vec{v}}}}
\newcommand{\ww}{\ensuremath{\vec{w}}}
\newcommand{\GG}{\ensuremath{\ten{G}}}
\newcommand{\FF}{\ensuremath{\boldsymbol{\mathfrak{F}}}}
\newcommand{\tEE}{\ensuremath{\tilde{\vec{E}}}}
\newcommand{\nn}{\ensuremath{\vec{n}}}
\renewcommand{\SS}{\ensuremath{\ten{S}}}
\newcommand{\tSS}{\ensuremath{\tilde{\ten{S}}}}
\newcommand\tK{{\bf K}}
\newcommand{\cp}{\ensuremath{{c_p}}}
\newcommand{\cs}{\ensuremath{{c_s}}}
\newcommand{\tcs}{\ensuremath{{\tilde{c}_s}}}
\newcommand{\Vp}{\ensuremath{{V^e_p}}}
\newcommand{\Vs}{\ensuremath{{V^e_s}}}
\newcommand{\Ss}{\ensuremath{{S^e_s}}}
\newcommand{\nxnx}[1] {\ensuremath{\nn\times\left(\nn\times{#1}\right)}}

\newcommand{\dd}[2] {\ensuremath{\frac{\partial {#1}}{\partial {#2}}}}
\newcommand{\DD}[2] {\ensuremath{\frac{d {#1}}{d {#2}}}}
\newcommand{\Grad} {\ensuremath{\nabla}}
\newcommand{\Gradv}[1]{\ensuremath{\nabla_{#1}}}
\newcommand{\Div} {\ensuremath{\nabla\cdot}}
\newcommand{\Divv}[1] {\ensuremath{\nabla_{#1}\cdot}}
\newcommand{\Curl} {\ensuremath{\nabla\times}}
\newcommand{\Cten}{\ensuremath{\tenfour{C}}}

%    Notation for an expression evaluated at a particular condition. The
%    optional argument can be used to override automatic sizing of the
%    right vert bar, e.g. \eval[\biggr]{...}_{...}
\newcommand{\eval}[2][\right]{\relax
  \ifx#1\right\relax \left.\fi#2#1\rvert}

\providecommand{\abs}[1]{\left\lvert#1\right\rvert}
\providecommand{\norm}[1]{\left\Vert#1\right\Vert}

\newcommand{\Nel} {\ensuremath{{N_T}}}
\newcommand{\Ned} {\ensuremath{{N_\text{ed}}}}
\newcommand{\De} {\ensuremath{{\mathsf{D}^e}}}
\newcommand{\Dep} {\ensuremath{{\mathsf{D}^{e'}}}}
\newcommand{\Dhat} {\ensuremath{\hat{\mathsf{D}}}}
\newcommand{\Dehat} {\ensuremath{{\hat{\mathsf{D}}^e}}}
\newcommand{\half} {\ensuremath{\frac{1}{2}}}
\newcommand{\IN} {\ensuremath{{\mathcal{I}_N}}}
\newcommand{\INe} {\ensuremath{{\mathcal{I}_{N_e}}}}

\newcommand{\mc}[1]{\mathcal{#1}}
\newcommand{\mcC}[1]{\mathcal{#1}}
\newcommand{\mb}[1]{\mathbf{#1}}
\newcommand{\mbb}[1]{\mathbb{#1}}
\newcommand{\mi}[1]{\mathit{#1}}
\newcommand{\mf}[1]{\mathfrak{#1}}
\newcommand{\Oi}[1]{\int_{\mc{D}} #1 \, d\vec{x}}
\newcommand{\TOi}[1]{\int_{0}^T \int_{\mc{D}} #1 \, d\vec{x}dt}
\newcommand{\Ti}[1]{\int_{0}^T #1 \, dt}
\newcommand{\TDei}[1]{\int_{0}^T\int_{\De} #1 \, d\vec{x}dt}
\newcommand{\TDeid}[1]{\int_{0}^T\int_{\De,N} #1 \, d\vec{x}dt}
\newcommand{\Dei}[1]{\int_{\De} #1 \, d\vec{x}}
\newcommand{\DeI}[1]{\int_{D} #1 \, d\vec{x}}
\newcommand{\DeiM}[1]{\int_{\Dhat} #1 \, d\vec{r}}
\newcommand{\DeMd}[1]{\int_{\De,N_e} #1 \, d\vec{x}}
\newcommand{\DeiMd}[1]{\int_{\Dhat,N_e} #1 \, d\vec{r}}
\newcommand{\DeMdN}[1]{\int_{D,N} #1 \, d\vec{x}}
\newcommand{\DeiMdN}[1]{\int_{\Dhat,N} #1 \, d\vec{r}}
\newcommand{\DeiMdNC}[2]{\int_{\Dhat,N_{#2}} #1 \, d\vec{r}}
\newcommand{\pDei}[1]{\int_{\partial \De} #1 \, d\vec{x}}
\newcommand{\pDeiM}[1]{\int_{\partial \Dhat} #1 \, d\vec{r}}
\newcommand{\pDeiMd}[1]{\int_{\partial \Dhat,N_e} #1 \, d\vec{r}}
\newcommand{\pDeiMdN}[1]{\int_{\partial \Dhat,N} #1 \, d\vec{r}}
\newcommand{\pDeiMdNC}[2]{\int_{\partial \Dhat,N_{#2}} #1 \, d\vec{r}}
\newcommand{\pMortar}[2]{\int_{\mc{M}_{#2}} #1 \, d\vec{r}}
\newcommand{\pMortard}[2]{\int_{\mc{M}^{#2},N_{#2}} #1 \, d\vec{r}}
\newcommand{\pMortardd}[3]{\int_{\mc{M}^{#2},N_{#3}} #1 \, d\vec{r}}
\newcommand{\pMortardh}[3]{\int_{\mc{M}_{#2},N_{#3}} #1 \, d\vec{r}}
\newcommand{\pMortardhn}[4]{\int_{{#2}_{#3},N_{#4}} #1 \, d\vec{r}}
\newcommand{\TpDei}[1]{\int_{0}^T\int_{\partial \De} #1 \, d\vec{x}}
\newcommand{\TpDeid}[1]{\int_{0}^T\int_{\partial \De,N} #1 \, d\vec{x}}
\newcommand{\Deid}[1]{\int_{\De,N_e} #1 \, d\vec{x}}
\newcommand{\DeidN}[1]{\int_{\De,N} #1 \, d\vec{x}}
\newcommand{\pDeidN}[1]{\int_{\partial \De,N} #1 \, d\vec{x}}
\newcommand{\pDeid}[1]{\int_{\partial \De,N_e} #1 \, d\vec{x}}
\newcommand{\nor}[1]{\left\| #1 \right\|}
\newcommand{\snor}[1]{\left| #1 \right|}
\newcommand{\LRp}[1]{\left( #1 \right)}
\newcommand{\LRs}[1]{\left[ #1 \right]}
\newcommand{\LRa}[1]{\left< #1 \right>}
\newcommand{\LRc}[1]{\left\{ #1 \right\}}
\newcommand{\pp}[2]{\frac{\partial #1}{\partial #2}}
\newcommand{\ppcp}[1]{\frac{\partial #1}{\partial \vec{c_p}}}
\newcommand{\ppcpj}[1]{\frac{\partial #1}{\partial c_p}}
\newcommand{\ppcpmf}[1]{\frac{\partial #1}{\partial \vec{\mf{c_p}}}}
\newcommand{\ppm}[1]{\frac{\partial #1}{\partial \vec{m}}}
\newcommand{\ppho}[3]{\frac{\partial^{#3} #1}{{\partial #2}^{#3}}}
\newcommand{\ppmi}[1]{\frac{\partial #1}{\partial m_i}}
\newcommand{\ppmj}[1]{\frac{\partial #1}{\partial m_j}}
\newcommand{\ppcpm}[1]{\frac{\partial #1}{\partial \vec{c_p}^-}}
\newcommand{\ppcpp}[1]{\frac{\partial #1}{\partial \vec{c_p}^+}}
\newcommand{\ppcs}[1]{\frac{\partial #1}{\partial \vec{c_s}}}
\newcommand{\ppcsm}[1]{\frac{\partial #1}{\partial \vec{c_s}^-}}
\newcommand{\ppcsp}[1]{\frac{\partial #1}{\partial \vec{c_s}^+}}
\newcommand{\td}[2]{\frac{d #1}{d #2}}
\newcommand{\Rvert}[1]{\left. #1 \right|}
\newcommand{\bs}[1]{\boldsymbol{#1}}
\newcommand{\ipDed}[3]{\LRp{ #1, #2}_{\De,N}^{#3}}

\newcommand{\bigO}{\mathcal{O}}

\newcommand{\iOm}[1]{\int_{\Omega} #1 \, d\Omega}
\newcommand{\iGb}[1]{\int_{\Gamma_{b}} #1 \, ds}
\newcommand{\iGs}[1]{\int_{\Gamma_{s}} #1 \, ds}
\newcommand{\iGsy}[1]{\int_{\Gamma_{s}} #1 \, ds(\mb{y})}
\newcommand{\RE}{\mbox{{I}}\hspace{-0.5ex}\mbox{{R}}}
\newcommand{\iC}[1]{\int_{0}^{2\pi} #1 \, d\theta}
\newcommand{\iGr}[1]{\int_{\Gamma_{\infty}} #1 \, ds}

\newcommand{\jump}[1] {\ensuremath{[\![{#1}]\!]}}
\newcommand{\diff}[1] {\ensuremath{\LRs{#1}}}
\newcommand{\average}[1] {\ensuremath{\LRc{\!\!\LRc{#1}\!\!}}}

\newcounter{tablecell}

\newcommand*{\numbercell}{%
  \stepcounter{tablecell}%
  \thetablecell. % **
}

\newtheorem{propo}{Proposition}
\newtheorem{coro}{Corollary}
\newtheorem{theo}{Theorem}
\newtheorem{lem}{Lemma}

\newtheorem{defi}{definition}

\newtheorem{rema}{remark}

\newcommand{\figlab}[1]{\label{fig:#1}}
\newcommand{\eqnlab}[1]{\label{eq:#1}}
\newcommand{\theolab}[1]{\label{theo:#1}}
\newcommand{\corolab}[1]{\label{coro:#1}}
\newcommand{\propolab}[1]{\label{propo:#1}}
\newcommand{\lemlab}[1]{\label{lem:#1}}
\newcommand{\defilab}[1]{\label{defi:#1}}
\newcommand{\remalab}[1]{\label{rema:#1}}
\newcommand{\tablab}[1]{\label{tab:#1}}

\newcommand{\figref}[1]{\ref{fig:#1}}
\newcommand{\theoref}[1]{\ref{theo:#1}}
\newcommand{\defiref}[1]{\ref{defi:#1}}
\newcommand{\remaref}[1]{\ref{rema:#1}}
\newcommand{\cororef}[1]{\ref{coro:#1}}
\newcommand{\proporef}[1]{\ref{propo:#1}}
\newcommand{\lemref}[1]{\ref{lem:#1}}
\newcommand{\eqnref}[1]{\eqref{eq:#1}}
\newcommand{\alglab}[1]{\label{alg:#1}}
\newcommand{\algref}[1]{\ref{alg:#1}}
\newcommand{\seclab}[1]{\label{sect:#1}}
\newcommand{\secref}[1]{\ref{sect:#1}}
\newcommand{\tabref}[1]{\ref{tab:#1}}

\renewcommand{\u}{u}
\newcommand{\uk}{\u^k}
\newcommand{\upk}{\u^{+k}}
\newcommand{\ukp}{\u^{k+1}}
\newcommand{\umkp}{\u^{-k+1}}
\newcommand{\umk}{\u^{-k}}
\renewcommand{\v}{v}
\newcommand{\vk}{{\v}^k}
\newcommand{\vkp}{{\v}^{k+1}}
\newcommand{\un}{\u_h}
\renewcommand{\e}{e}
\newcommand{\w}{w}
\renewcommand{\wb}{\mb{\w}}
\newcommand{\J}{J}
\newcommand{\Jb}{\mb{J}}
\newcommand{\q}{q}
\renewcommand{\r}{r}
\newcommand{\qh}{\hat{\q}}
\newcommand{\qs}{\q^*}
\newcommand{\qht}{\tilde{\qh}}
\newcommand{\qbar}{\overline{\q}}
\newcommand{\qb}{{\mb{\q}}}
\newcommand{\sig}{{\sigma}}
\newcommand{\thetah}{{\hat{\theta}}}
\newcommand{\thetas}{{\theta^*}}
\newcommand{\thetahn}{{\thetah_h}}
\newcommand{\sign}[1]{\text{ sgn}\!\LRp{#1}}
\newcommand{\sigh}{\hat{\sig}}
\newcommand{\sigs}{\sig^*}
\newcommand{\sigk}{\sig^{k}}
\newcommand{\sigpk}{\sig^{+k}}
\newcommand{\sigkp}{\sig^{k+1}}
\newcommand{\sigmkp}{\sig^{-k+1}}
\newcommand{\sigmk}{\sig^{-k}}
\newcommand{\sigb}{{\bs{\sig}}}
\newcommand{\sigbk}{\sigb^k}
\newcommand{\sigbkp}{\sigb^{k+1}}
\newcommand{\sigbn}{\sigb_h}
\newcommand{\sigbs}{\sigb^*}
\newcommand{\taub}{{\bs{\tau}}}
\newcommand{\taustar}{\tau_*}
\newcommand{\kappab}{{\bs{\kappa}}}
\newcommand{\gamb}{{\bs{\gamma}}}
\newcommand{\ut}{\tilde{\u}}
\newcommand{\sigbt}{\tilde{\sigb}}
\newcommand{\vt}{\tilde{\v}}
\newcommand{\taubt}{\tilde{\taub}}
\newcommand{\ub}{\mb{\u}}
\newcommand{\ubl}{\mb{\u}_{\lambda}}
\newcommand{\ubf}{\mb{\u}_{f}}
\newcommand{\ubp}{\mb{\u}^+}
\newcommand{\ubm}{\mb{\u}^-}
\newcommand{\vb}{\mb{\v}}
\newcommand{\xb}{\mb{x}}
\newcommand{\um}{\u^-}
\newcommand{\up}{\u^+}
\newcommand{\ubk}{\ub^k}
\newcommand{\ubkp}{\ub^{k+1}}

\newcommand{\m}{{m}}

\newcommand{\ud}{{\dot{\u}}}
\newcommand{\vd}{{\dot{\v}}}
\newcommand{\sigbd}{{\dot{\sigb}}}
\newcommand{\taubd}{\dot{\taub}}

\renewcommand{\d}{{d}}
\newcommand{\R}{{\mathbb{R}}}
\newcommand{\kap}{\kappa}
\newcommand{\F}{\mb{F}}
\newcommand{\f}{f}
\newcommand{\g}{g}
\newcommand{\fb}{\mb{\f}}
\newcommand{\gb}{\mb{\g}}
\newcommand{\gD}{g_D}
\newcommand{\pOmega}{\partial \Omega}
\newcommand{\K}{T}
\newcommand{\Kp}{\K^+}
\newcommand{\Km}{\K^-}
\newcommand{\pK}{{\partial \K}}
\newcommand{\pKin}{{\pK^{\text{in}}}}
\newcommand{\pKout}{{\pK^{\text{out}}}}
\newcommand{\Kj}{{\K_j}}
\newcommand{\Gh}{{\mc{E}_h}}
\newcommand{\Gho}{{\mc{E}_h^o}}
\newcommand{\Ghb}{{\mc{E}_h^{\partial}}}
\newcommand{\Poly}{{\mc{P}}}
\newcommand{\D}{{D}}
\newcommand{\p}{{p}}
\newcommand{\px}{\partial_x}
\newcommand{\ppx}{\partial^2_x}
\newcommand{\pres}{{q}}
\newcommand{\presl}{{q}_{\lambda}}
\newcommand{\presf}{{q}_{f}}
\newcommand{\Ts}{{\mb{T}}}
\newcommand{\Vb}{{\mb{V}}}
\newcommand{\V}{{V}}
\newcommand{\Lam}{{\Lambda}}
\newcommand{\Lamb}{\bs{\Lambda}}
\newcommand{\mub}{\bs{\mu}}
\newcommand{\lambdab}{\bs{\lambda}}
\newcommand{\Thn}{\mathcal{T}_h}
\newcommand{\Th}{{\Ts_h}}
\newcommand{\Vh}{{\V_h}}
\newcommand{\Vbh}{{\Vb_h}}
\newcommand{\Lamh}{{\Lam_h}}
\newcommand{\Lambh}{{\Lamb_h}}
\newcommand{\Lamho}{{\overline{\Lam}_h}}
\newcommand{\Lambht}{{\Lamb_h^t}}
\newcommand{\Ltwo}{{L^2\LRp{\Omega}}}
\newcommand{\Ltwoe}{{L^2\LRp{\e}}}
\newcommand{\Ltw}{{L^2\LRp{\K}}}
\newcommand{\VhK}{{\V_h\LRp{\K}}}
\newcommand{\ThK}{{\Ts_h\LRp{\K}}}
\newcommand{\VbhK}{{\Vb_h\LRp{\K}}}
\newcommand{\Lamhe}{{\Lam_h\LRp{\e}}}
\newcommand{\Lambhe}{{\Lamb_h\LRp{\e}}}
\renewcommand{\L}{\mc{L}}
\renewcommand{\D}{\mc{D}}
\newcommand{\T}{\mc{T}}
\newcommand{\Tt}{\tilde{\T}}
\newcommand{\A}{\mb{A}}
\newcommand{\B}{\mb{B}}
\renewcommand{\D}{\mb{D}}
\newcommand{\Am}{\A^-}
\newcommand{\Ap}{\A^+}
\newcommand{\Aa}{\snor{\A}}
\newcommand{\Rb}{\mb{R}}
\newcommand{\C}{\mb{C}}
\renewcommand{\S}{\mb{S}}
\newcommand{\Sm}{\S^{-}}
\newcommand{\Sp}{\S^{+}}
\newcommand{\Spm}{\S^{\pm}}

\newcommand{\Lte}{{L^2\LRp{\Gh}}}
\newcommand{\n}{\mb{n}}
\newcommand{\nm}{\n^-}
\newcommand{\np}{\n^+}
\newcommand{\uh}{{\hat{\u}}}
\newcommand{\uhk}{{\uh}^k}
\newcommand{\uhkp}{{\uh}^{k+1}}
\newcommand{\us}{{\u^*}}
\newcommand{\uhn}{{\uh_h}}
\newcommand{\ubh}{{\hat{\ub}}}
\newcommand{\ubs}{{\ub^*}}
\newcommand{\ubht}{{\tilde{\ubh}}}
\newcommand{\uht}{{\tilde{\uh}}}
\newcommand{\uhtn}{{\uht_h}}
\newcommand{\uhd}{{\dot{\uh}}}
\newcommand{\ubhk}{\ubh^k}
\newcommand{\ubhkp}{\ubh^{k+1}}
\newcommand{\sigbh}{{\hat{\sigb}}}
\newcommand{\Fh}{{\hat{\F}}}
\newcommand{\fh}{{\bar{\f}}}
\newcommand{\Fs}{\F^*}
\renewcommand{\c}{{c}}

\newcommand{\zb}{\mb{z}}
\renewcommand{\sb}{\mb{s}}
\newcommand{\rb}{\mb{r}}
\newcommand{\betab}{\bs{\beta}}
\newcommand{\veps}{{\varepsilon}}
\newcommand{\vepsb}{\bs{\veps}}
\newcommand{\Hb}{\mb{H}}
\newcommand{\Hbt}{\Hb^t}
\newcommand\E{\mathcal{E}}
\newcommand\Eh{\mathcal{E}_h}
\newcommand\Ehi{\mathcal{E}_{h,i}}
\newcommand\Ehiint{\mathcal{E}_{h,i}^\mathrm{int}}
\newcommand\EhG{\mathcal{E}_{h}^{\Gamma}}
\newcommand\EhiG{\mathcal{E}_{h,i}^{\Gamma}}
\newcommand\EhijG{\mathcal{E}_{h,i,j}^{\Gamma}}
\newcommand\EhjiG{\mathcal{E}_{h,j,i}^{\Gamma}}
\newcommand{\Eb}{\mb{E}}
\newcommand{\Ebb}{\mbb{E}}
\newcommand{\Ebt}{\Eb^t}
\newcommand{\hb}{\mb{h}}
\newcommand{\eb}{\mb{e}}
\newcommand{\Hbh}{\hat{\mb{H}}}
\newcommand{\Ebh}{\hat{\mb{E}}}
\newcommand{\Hbs}{{\mb{H}^*}}
\newcommand{\Ebs}{{\mb{E}^*}}
\newcommand{\Ebht}{\Ebh^t}
\newcommand{\Hbht}{\Hbh^t}
\newcommand{\Lb}{\mb{L}}
\newcommand{\Gb}{\mb{G}}
\newcommand{\Lbh}{\hat{\Lb}}
\newcommand{\Lbs}{\Lb^*}
\newcommand{\Ib}{\mb{I}}
\newcommand{\I}{I}
\newcommand{\Q}{Q}
\newcommand{\Sigb}{\bs{\Sigma}}
\newcommand{\Piu}{\Pi_\u}
\newcommand{\Pisigb}{\Pi_\sigb}
\newcommand{\Z}{Z}
\newcommand{\Y}{Y}
\newcommand{\rhou}{\rho\u}
\newcommand{\rhov}{\rho\v}
\newcommand{\rhoe}{\rho_e}
\newcommand{\rhoa}{\tilde\rho}
\newcommand{\rhoaa}{\hat\rho}
\renewcommand{\P}{P}
\newcommand{\h}{H}

% color
\newcommand{\mygreen}[1]{{\color[rgb]{0,.65,0} #1}}
\newcommand{\mywhite}[1]{{\color[rgb]{1.0,1.0,1.0} #1}}
\newcommand{\myred}[1]{{\color[rgb]{0.65,0.0,0.0} #1}}
\newcommand{\myblue}[1]{{\color[rgb]{0.0,0.0,1.0} #1}}
\newcommand{\mygray}[1]{{\color[rgb]{.15,.35,.60} #1}}
\newcommand{\mythemec}[1]{{\color[RGB]{51,108,121} #1}}
\newcommand{\mydarkred}[1]{{\color[rgb]{.6,.1,.1} #1}}
\newcommand{\mydarkblue}[1]{{\color[rgb]{.1,.1,.9} #1}}
\newcommand{\mygreenback}[1]{{\color[rgb]{.75,1,.75} #1}}
\newcommand{\myorange}[1]{{\color[rgb]{.76,.39,.13} #1}}
\newcommand{\mygrass}[1]{{\color[rgb]{.19,.64,.13} #1}}
\newcommand{\mysierp}[1]{{\color[RGB]{209,28,209} #1}}

\newcommand{\tanbui}[2]{\textcolor{blue}{#1} \textcolor{red}{#2}}

\newcommand{\vel}{\bs{\vartheta}}
\newcommand{\vels}{\vel^*}
\newcommand{\vele}{\vel^e}
\newcommand{\velh}{\hat{\vel}}
\newcommand{\velk}{\vel^k}
\newcommand{\velkp}{\vel^{k+1}}
\newcommand{\vhk}{\hat\v^{k}}

\newcommand{\phis}{\phi^*}
\newcommand{\phie}{\phi^e}
\newcommand{\phih}{\hat{\phi}}
\newcommand{\phihk}{\phih^k}
\newcommand{\phik}{\phi^k}
\newcommand{\phikp}{\phi^{k+1}}

\begin{frontmatter}

\title{Sparse Grids based Adaptive Noise Reduction strategy for Particle-In-Cell schemes}

    \author[srk]{Sriramkrishnan Muralikrishnan}
    \author[Antoine]{Antoine J. Cerfon}
    \author[srk]{Matthias Frey}
    \author[Lee]{Lee F. Ricketson}
    \author[srk]{Andreas Adelmann}
    \address[srk]{Paul Scherrer Institut, 5232 Villigen, Switzerland.}
    \address[Antoine]{Courant Institute of Mathematical Sciences, New York University, New York NY 10012, USA.}
    \address[Lee]{Lawrence Livermore National Laboratory, Livermore, USA.}

\begin{abstract}
    We propose a sparse grids based adaptive noise reduction strategy for electrostatic particle-in-cell (PIC) simulations. Our approach is based on the key 
    idea of relying on sparse grids instead of a regular grid in order to increase the number of particles per cell for the same total number of particles, as first introduced
    in \cite{ricketson2016sparse}. Adopting a new filtering perspective for this idea, we construct the algorithm so that it can be easily integrated into high performance large-scale PIC code bases. Unlike the physical and Fourier domain filters typically used in PIC codes, our approach automatically adapts to mesh size, number of particles per cell, smoothness of the density profile and the initial sampling technique. Thanks to the truncated combination technique \cite{leentvaar2008pricing,benk2012hybrid,benk2012variants}, we can reduce the larger grid-based error of the standard sparse grids approach for non-aligned and non-smooth
    functions. We propose a heuristic based on formal error analysis
    for selecting the optimal truncation parameter at each time step, and develop a natural framework to minimize the total error in sparse PIC simulations. We demonstrate its efficiency and performance by means of two test cases: the diocotron
    instability in two dimensions, and the three-dimensional electron dynamics in a Penning trap. Our run time performance studies indicate that our new scheme can 
    provide significant speedup and memory reduction as compared to regular PIC for achieving comparable accuracy in the charge density deposition.
\end{abstract}


\end{frontmatter}

\section{Introduction}
    Particle-in-cell (PIC) schemes have been a popular and effective method for the simulation of kinetic plasmas for a long period of time \cite{hockney1988computer,birdsall2004plasma,dawson1983particle}. Compared to continuum kinetic codes,
PIC schemes effectively reduce the dimension from six to three for three-dimensional simulations whereas compared to pure particle codes with direct summation it reduces the computation of self-consistent forces from $\mc{O}(N_p^2)$ to $\mc{O}(N_p + N_c)$ where $N_p$ is the total 
number of particles and $N_c \ll N_p$ is the number of mesh points. Even though the fast multipole method \cite{greengard1987fast} reduces the complexity of pure particle schemes to $\mc{O}(N_p)$, such an approach has other limitations such as the need for overly restrictive small time stepsizes. The other attractive features of PIC schemes include simplicity, ease of parallelization and robustness for wide variety of physical scenarios \cite{ricketson2016sparse}.

The main drawback of PIC schemes as compared to deterministic continuum kinetic schemes is the numerical error associated with particle noise, which decreases slowly as one increases the number of particles. Specifically, the noise in PIC schemes decreases as $1/\sqrt{P_c}$ where $P_c = N_p/N_c$ is the number of particles per cell. High fidelity large-scale 3D PIC simulations thus often require at least $\mc{O}(10^9)$ grid points and 
$\mc{O}(10^{12})$ particles to get the desired accuracy level \cite{nakashima2017large}.
These simulations require hours to complete even on large-scale state-of-the-art supercomputers available today. Thus, noise reduction approaches are of great interest to the PIC community to improve accuracy and also to speed up the computations and reduce memory 
requirements.

    There have been several efforts in this area over the past years and a brief overview is given in section \secref{noise_review}. Some of the strategies
    such as the $\delta f$ technique \cite{denton1995deltaf,aydemir1994unified,sydora1999low} are applicable for certain classes of plasma physics problems
    and give great computational savings whereas for others they are not obvious to apply. Filtering is a common noise reduction technique which finds applications in many production-level PIC codes such as TRISTAN-MP \cite{spitkovsky2005simulations,buneman1993computer}, ORB5 \cite{jolliet2007global}, IMPACT-T \cite{terzic2007particle} and Warp-X \cite{vay2018warp}, to name a few. One of the primary reasons for it is its simplicity and ease of 
    implementation in these frameworks. The stencil width and number of passes in case of digital filters and the cut-off wavenumber in case of Fourier domain
    filters is typically selected based on experience and knowledge about the physical problem at hand. Thus these could result in scenarios where either
    too much signal is smoothed or the high frequency noise is not removed sufficiently. Even if we managed to choose the parameters in the filter 
    so that they are optimal for a particular mesh size, number of particles per cell, time instant/smoothness of the function and the initial sampling technique they may no longer be optimal once we change any of the above and require tuning once again.

    Our objective in this work is to develop a noise reduction strategy, or filtering scheme, that automatically adapts itself to the aforementioned parameters. Just like other filtering techniques we also want it to be easily integrated into existing production-level PIC codes. Towards that goal, we adopt a key idea introduced in the recent work \cite{ricketson2016sparse} which combined sparse
    grids with the PIC scheme. In that article, the authors showed that owing to the large cell sizes involved in sparse grids compared to regular grids, the PIC
    scheme combined with sparse grids has many more particles per cell than its regular counterpart. This led to significant noise reduction and enormous
    speedups for certain classes of problems which have smooth or axis-aligned density profiles.

    Let us give a brief overview of the present work. By re-visiting the idea introduced in \cite{ricketson2016sparse} from a filtering 
    perspective, we first construct a sparse grids based noise reduction strategy for electrostatic 
    PIC simulations. Compared to existing filtering approaches, this sparse grids based approach is superior for functions which are 
    smooth or aligned with an axis. In simple terms this can be understood as follows: with any filtering technique the reduction in noise
    comes with a price, which is an increase in the grid-based error. In case of the sparse grids filtering, the noise reduction is maximal due to the least number of grid points involved in the approximation which translates to maximum particles per cell in the context of PIC. At the same time the increase in grid-based error for 
    smooth or axis-aligned functions is minimal. However, the same cannot be said for other functions in general and in those cases the
    increase in grid-based error associated with sparse grids is high. In order to tackle that issue, we use the so-called truncated combination 
    technique \cite{leentvaar2008pricing,benk2012hybrid,benk2012variants}, which reduces the large grid-based error of standard sparse 
    grids technique for non-aligned and non-smooth functions. This is because the truncated combination technique uses a different choice of coarse grids with finer mesh sizes than those used in the standard sparse grid combination. The truncation parameter involved in the combination technique is crucial for
    minimizing the sum of grid-based error and particle noise. Hence, we propose a heuristic
    based on formal error analysis to calculate the optimal truncation parameter on the fly which minimizes the total error.  

    This paper is organized as follows.\ Section \secref{pic} introduces the PIC method in the context of electrostatic Vlasov-Poisson equations.\ Section \secref{noise_review} briefly reviews the existing noise reduction strategies in PIC and provides motivation and objectives for this article.\ Section \secref{sparse_noise_reduction} explains in detail the components and algorithm for the sparse grids based noise 
    reduction strategy.\ Numerical results for the 2D diocotron test case and 3D penning trap are presented in 
    section \secref{numerical_results} and section \secref{conclusions} presents the conclusions and future work.

\section{Particle-in-cell method}
\seclab{pic}
    In this work, we consider the non-relativistic electrostatic Vlasov-Poisson system with a fixed magnetic field, and introduce the PIC method in that setting. The electrons are immersed in a uniform, immobile, neutralizing background ion population and the system is given by
    \begin{equation}
        \eqnlab{Vlasov}
        \frac{\partial f}{\partial t} + \vb \cdot \nabla_{{\bf x}} f + \frac{q_e}{m_e}\LRp{\Eb + \vb \times \B_{ext}} \cdot \nabla_{\vb} f = 0,
    \end{equation}
where $\Eb = \Eb_{sc} + \Eb_{ext},$ and the self-consistent fields due to space charge are given by 
\[
    \Eb_{sc} = -\nabla \phi, \quad -\Delta \phi = \rho = \rho_e - \rho_i.
\]
In the above equation $f({\bf x},\vb,t)$ is the electron phase-space distribution, $q_e$ and $m_e$ are the electron charge and mass respectively. The total
electron charge in the system is given by $Q_e=q_e\int\int f d{\bf x}d\vb$, the electron charge density by $\rho_e({\bf x}) = q_e\int f d\vb$ and the
constant ion density by $\rho_i = \frac{Q_e}{\int d{\bf x}}$. Throughout this paper we use bold letters for vectors and non-bold ones for scalars.

The particle-in-cell method discretizes the phase space distribution $f({\bf x},\vb,t)$ in a Lagrangian way by means of macro-particles (hereafter referred to as particles for simplicity). At time $t=0$, the
distribution $f$ is sampled to get the particles and after that a typical computational cycle in PIC consists of the following steps:

\begin{enumerate}
    \item \label{step1_pic} Assign a shape function e.g., cloud-in-cell \cite{birdsall2004plasma} to each particle $p$ and deposit the electron charge onto an underlying mesh.
    \item \label{step2_pic} Use a grid-based Poisson solver to compute $\phi$ by solving $-\Delta \phi = \rho$ and differentiate $\phi$ to 
        get the electric field $\Eb=-\nabla \phi$ on the mesh.
    \item Interpolate $\Eb$ from the grid points to particle locations ${\bf x}_p$ using the same interpolation function as used for the 
        charge deposition in step 1.
    \item By means of a time integrator advance the particle positions and velocities using
        \begin{align*}
            \frac{d\vb_p}{dt} &= \frac{q_e}{m_e}\LRp{\Eb + \vb \times \B_{ext}}|_{{\bf x}={\bf x}_p}, \\
            \frac{d{\bf x}_p}{dt} &= \vb_p.
        \end{align*}
\end{enumerate}
    
     The sources of different errors in the PIC simulations and their order of accuracies for typical choices are 
    as follows. For simplicity if we consider a uniform mesh with spacing $h$ in all the directions then for the shape functions used in typical PIC schemes (B-splines), the grid-based error scales as $\mc{O}(h^2)$. This is a result of approximating $\delta$ functions in the configuration space by shape functions of compact support. The Poisson equation is typically solved by means of FFT solvers or by multigrid
    methods. In case of multigrid solvers the equation is discretized by second-order finite difference or finite element schemes. The field
    solves together with the interpolation (typically linear) accounts for an additional $\mc{O}(h^2)$. The particle noise is the result of approximating the expected value of the shape function by an arithmetic mean over a finite number of discrete particles. It scales as 
    $(N_p h^d)^{-1/2}$ \cite{ricketson2016sparse}, where $d$ is the spatial dimension of the problem. The initial distribution is 
    sampled using one of the standard sampling techniques such as the naive Monte-Carlo strategy \cite{aydemir1994unified}, importance 
    sampling \cite{aydemir1994unified} or by means of the quiet start \cite{wang2011particle,wang2012adaptive,myers20174th}. The choice of 
    initial sampling plays an important role in determining the constant associated with the particle noise.\ Finally, for time integration usual choices are the second-order leap frog scheme \cite{birdsall2004plasma} and Runge-Kutta schemes of order 2 and higher. If we consider the leap frog scheme then the error in the time discretization scales as $\mc{O}(\Delta t^2)$. The mesh size $h$, time step $\Delta t$ and the number of particles $N_p$ in most PIC simulations are such that the dominant error comes from the 
   particle noise. Hence, high fidelity simulations 
    typically require a large number of particles to minimize it.  The high noise associated with PIC simulations motivated researchers to develop several noise
    reduction strategies which we will discuss next.

\section{Noise reduction strategies in PIC}
\seclab{noise_review}

    Noise reduction can be achieved in several ways in the context of PIC simulations and they can be categorized as: (i) variance reduction techniques such as the $\delta f$ method \cite{denton1995deltaf,aydemir1994unified,sydora1999low} and quiet start \cite{sydora1999low}; (ii) phase space remapping \cite{wang2011particle,wang2012adaptive,myers20174th}; (iii) filtering in physical domain
    \cite{birdsall2004plasma,verboncoeur2005particle,spitkovsky2005simulations,buneman1993computer,vay2011numerical},
    Fourier domain \cite{birdsall2004plasma,jolliet2007global} and wavelet domain \cite{gassama2007wavelet,terzic2007particle,terzic2011new}.\ The list is
    not exhaustive and there are many other contributions in this area. In addition, recently a noise reduction strategy using kernel density estimation algorithm has been proposed in \cite{wu2018reducing}, where the authors adaptively select the shape functions in PIC which minimize the sum of bias squared and variance of the error in the density. Also, in \cite{ricketson2016sparse} sparse grid techniques are used to achieve noise reduction in PIC, as we will explain in detail in the next section since our work is based on this strategy. In
    this section, we will focus more on the filtering strategies due to their relevance in the context
    of our work.

    The goal of filtering in PIC simulations is to smooth high frequency oscillations usually associated with noise. Filtering can be done in
    any field quantity, although the most common one in electrostatic PIC is the charge density \cite{verboncoeur2005particle} as it is the origin of noise  and the potential and electric field are smoother because of the integration involved in solving Poisson's equation.
    In case of filtering in the physical domain, one typically selects a filter of certain stencil width e.g., binomial filter, and does few passes on the field quantity. On the other hand for filters in the Fourier domain a maximum wavenumber is specified by the user 
    and the filter eliminates all the wavenumbers higher than the specified cut-off wavenumber \cite{birdsall2004plasma}. In almost all the filtering strategies, the number of passes/stencil width in the 
    physical domain or the cut-off wavenumber in the Fourier domain has to be chosen
    a priori such that the total error, which is the sum of grid-based error (bias) and particle noise (variance) is minimized. However, in practice there are not many 
    constructive strategies available to pick these parameters and in many cases the values are chosen based on rule of thumb and previous 
    experiences \cite{shalaby2017sharp}.\ If we managed to choose these parameters so that they are optimal for a particular time instant, mesh, 
    number of particles per cell and sampling technique, they are unlikely to remain so as the simulation evolves. Indeed, due to non-linear space charge effects fine scale structures appear in the density and this 
     changes the smoothness of the profile continuously with time. Hence an ideal filter should be adaptive with respect to 
    all aforementioned parameters to minimize the total error.

    Our objective is to develop a noise reduction strategy for PIC simulations which automatically adapts with 
    respect to time, mesh size, number of particles per cell and the initial sampling technique. Since our ultimate goal is large-scale 3D
    high fidelity plasma simulations we also want the approach to be scalable in high performance computing 
    architectures. Finally, it would be desirable if the strategy integrates well with the existing high performance large-scale PIC code bases such as OPAL \cite{adelmann2019opal} so that the user community of these codes can benefit from it. Towards these goals, we propose a sparse grids based adaptive
    noise reduction strategy in the following section.

\section{Sparse grids based noise reduction}
\seclab{sparse_noise_reduction}

    The sparse grid combination technique was first introduced in \cite{griebel1990combination} as a way to approximate smooth functions
    on rectangular grids efficiently, by a specific linear combination of their approximations on different coarse grids. If we consider linear 
    interpolation as an example, then for a regular grid of mesh size $h$ we need $\mc{O}(h^{-d})$ grid points to get an accuracy of  
    $\mc{O}(h^2)$. The sparse grid combination technique on the other hand uses only $\mc{O}(h^{-1}|log(h)|^{(d-1)})$ total number 
    of points to get an accuracy of $\mc{O}(h^2|log(h)|^{(d-1)})$ for smooth functions, which is only slightly deteriorated 
    compared to the regular grids. Thus we can clearly see the advantages of sparse grids in high dimensions and it is mainly used to
    tackle the curse of dimensionality in many applications \cite{bungartz2004sparse}. The key idea is the cancellations that happen between the error expansions in the different 
    coarse grids, which are called component grids in the sparse grids terminology. Also, the scalar values that multiply each component grid involved in the combination are called the combination coefficients. In Figure \figref{truncated_comb} an illustration is shown, where we can see the different component grids and their combination coefficients involved in approximating a $2^8\times2^8$ regular grid. The literature on sparse grid combination technique and sparse grids in general is vast and the readers can refer to 
    \cite{bungartz2004sparse,griebel1990combination,pfluger2010spatially,heene2018massively,griebel1995efficient} and 
    the references therein for more details.

    In \cite{ricketson2016sparse}, the authors combined sparse grid combination technique with the 
    PIC method and showed noise reduction and order of magnitude speedups for a certain class of problems.
    One of the key ideas from \cite{ricketson2016sparse} is that for the same total
    number of particles $N_p$, sparse grids have many more particles per cell compared to regular
    grids and hence significantly less noise. However, this comes with a limitation which is explained as follows. 
    While for smooth functions (Figure \figref{time0_diocotron_tau1_density}) or functions aligned with an axis (Figure \figref{time0_penning_tau1_density}) the increase in grid-based 
    error for using sparse grids is very small in comparison to regular grids, for non-smooth or non-aligned functions (Figures \figref{time875_diocotron_tau1_density}, \figref{time200_penning_tau1_density} and \figref{time300_penning_tau1_density}) the increase is much higher \cite{trefethen2017cubature}. Hence, the sparse PIC proposed
    in \cite{ricketson2016sparse} showed large speedups and memory savings for smooth functions, whereas for ones with fine scale structures that are not aligned with the grid the regular PIC itself was more advantageous. The sparse PIC for these cases seemed to provide benefits only in the limit of asymptotically fine meshes at least for two-dimensional problems where the weaker dependence of the computational cost on dimension is not as noticeable.

    We base our work on the key idea of increased particles per cell in sparse grids shown in
    \cite{ricketson2016sparse} compared to regular grids for the same total number of particles
    $N_p$. However, in this paper we provide the following novelties and improvements compared to
    \cite{ricketson2016sparse}. 

    Firstly, the sparse PIC scheme in \cite{ricketson2016sparse} performs all the operations e.g., charge deposition and Poisson solve on the sparse grids and does not introduce regular grids
    at all except for visualization purposes or post-processing. By contrast in the
    current approach we use sparse grids only for noise reduction in the charge density while
    all the operations such as charge deposition and Poisson solve happen in the regular grid
    just like in usual PIC codes. 

    The reason for pursuing this approach is that we 
    want our strategy to fit well within the existing large-scale high performance regular PIC code bases in a light weight and easy to integrate manner. While
    the parallelization strategy outlined in \cite{cerfon2019sparse} for sparse PIC is attractive
    in its own sense, it is also very different from the parallelization strategy adopted in 
    many of the existing large-scale PIC code bases. In that sense, since the strategy
    proposed in this paper is an add-on for regular PIC it can be more easily integrated into existing
    frameworks such as OPAL \cite{adelmann2019opal}, Warp-X \cite{vay2018warp}, TRISTAN-MP \cite{spitkovsky2005simulations,buneman1993computer} within their parallelization framework without disturbing other parts.
    
    One may be concerned that since both regular grids and sparse grids are involved in our 
    current approach the significant benefits mentioned in \cite{ricketson2016sparse} may not
    be realized. However, in the regime where particle operations dominate we still get significant speedups compared to the regular PIC for similar accuracy 
    as demonstrated with the numerical results in section \secref{numerical_results}. Apart from
    these software reasons, our approach here also opens the door to understanding the noise reduction
    achieved with sparse grids from a filtering perspective as we show in the next section. This has the potential to extend sparse grids based noise reduction/filtering in many more directions for future research.
   
   Secondly, we tackle the high grid-based error of sparse grids for non-aligned or non-smooth functions by means of the so-called truncated combination technique \cite{leentvaar2008pricing,benk2012hybrid,benk2012variants}. This technique was proposed as a modification to the standard
   sparse grid combination technique to tackle the convergence issues in certain types of PDEs
   in financial applications caused by the presence of extremely anisotropic grids in the standard
   sparse grid technique. Here, in the context of PIC it provides a natural way to 
   minimize the sum of grid-based error and particle noise within the framework of sparse grids based noise reduction without much modification to the 
   standard sparse grid combination technique. 

   Finally, we propose a heuristic based on a generalization of the formal error analysis in \cite{ricketson2016sparse} to adaptively select the optimal truncation parameter in the combination technique depending on the smoothness of the density profile, mesh size, number of particles per cell and the initial sampling technique. In what follows we explain in detail these three contributions and lay out the algorithm for the noise reduction strategy.  
   
   \subsection{A filtering perspective}

   Let us consider a domain of size $[0,L]^d$ where $d$ is the dimension (typically $d=2$ or $3$), and for simplicity a regular grid of mesh size
   $h=\frac{L}{2^n}$ in all the directions. In our noise reduction strategy, after step \ref{step1_pic} in 
   the PIC algorithm shown in section \secref{pic} we perform a sparse grids projection of the
   charge density as follows
    \begin{equation}
       \eqnlab{sparse_filter}
        \varrho_e = G \rhoa_e = \LRp{\sum_{i=1}^{nc} c_i P_i R_i}\rhoa_e. 
    \end{equation}
Here, $\rhoa_e$ and $\varrho_e$ are the 
the charge densities on the regular grid before and after the sparse grids transformation. $R_i$ and $P_i$ are the transfer operators which transfer the density from the regular grid to the $i$th 
component grid in the sparse grid combination technique and vice versa, respectively. We also denote the transfer operators simply as $R$ and $P$ in places where the subscript $i$ is not needed. $c_i$ is 
the combination coefficient for the $i$th component grid which is a scalar value and $nc$ is the
number of component grids involved in the combination technique. 

    One requirement for the transfer operators $P_i$ and $R_i$ is to ensure global charge conservation.
%For this purpose, we use cloud-in-cell deposition to construct the restriction operator $R_i$ which transfers the density from the regular
%grid to the $i$th component grid. Thus it is given by 
In our approach, we use the cloud-in-cell or linear interpolation operators and they are given by  
\begin{equation}
\eqnlab{restriction}
    R_i({\bf x} - \tilde{\bf x}) = \frac{h^d}{V_i}\prod_{m=1}^d W_m\LRp{x_m - \tilde x_m}, \quad P_i({\bf x} - \tilde{\bf x}) = \prod_{m=1}^d W_m\LRp{x_m - \tilde x_m} 
\end{equation}
where $W_m(\zeta) = \max\LRc{0,1-\frac{|\zeta|}{h_m}}$. Here, ${\bf x}$ and $\tilde{\bf x}$ are the locations of the grid points in the $i$th component grid and regular
grid respectively and $h_m$ is the mesh size of the $i$th component grid along the $m$th coordinate axis. $V_i$ is the volume of each cell in the 
$i$th component grid.
% This is multiplied by the ratio of volume of a cell in regular grid $h^d$ to the volume of a cell in the $i$th component grid to conserve the total charge. 
% For the prolongation
%operator $P_i$, we use the same cloud-in-cell or linear interpolation (without the factor $h^d/V_i$) and interpolate from the $i$th component grid to the regular grid.

    Let us consider the standard sparse grid combination technique in \cite{griebel1990combination}, then the sparse grids projection or interpolation in equation
    \eqnref{sparse_filter} essentially removes high frequency components which are 
    coupled between the axes. This is because the sparse grid combination corresponding
    to a regular grid of mesh size $h$ does not have the fine resolution $h$ in all the 
    directions. In this sense the sparse grid combination acts as a multi-dimensional low pass filter and keeps only
    certain wavenumbers resolved by a regular grid of mesh size $h$. This is the filtering point of view for the noise
    reduction obtained from the sparse grids. It can also be understood from a Monte Carlo point of view as shown in \cite{ricketson2016sparse} by means of 
    increased particles per cell in the sparse grids compared to the regular grid.
    However, in the sparse PIC presented in \cite{ricketson2016sparse} the particles deposit directly onto the component grids
    unlike the strategy pursued here. These two approaches are related as stated in the following proposition and hence the noise reduction obtained with the sparse grids can be understood from a Monte Carlo point of view or from a filtering perspective. In later sections we use either viewpoint to 
    explain the noise reduction with sparse grids depending on the context.

\begin{proposition}
    \propolab{MC_filtering_equivalence}
For node-centered grids and linear interpolation shape functions the direct charge density deposition
    onto the component grids in the sparse PIC approach \cite{ricketson2016sparse} is equivalent to first depositing the charge density onto the regular grid
    and then transferring it to the component grids by means of the operator $R$ in equation \eqnref{restriction}\footnote{Let us refer this as two-step approach for simplicity.}. 
    In case of the cell-centered grids an exact equivalence between the two approaches does not hold.\ There the two-step approach can be viewed as direct charge deposition onto the component grids with
    a different shape function than the standard hat function which is also second-order accurate in space. 
\end{proposition}
\begin{proof}
    The proof is given in appendix A.
\end{proof}

%\begin{remark}
%    \remalab{cell_center_remark}
%%In case of the cell-centered grids the coarse grid points are not also the grid points in the fine grid. %Thus $W_c$ will have a discontinuity in the
%%    derivative for some of the intervals $[x_j,x_{j+1}]$ depending on the
%%    ratio $h_c/h_f$. 
%%Hence, the equivalence between the two approaches does not hold. and there will be additional grid-based errors in the current
%%    approach compared to the direct deposition in \cite{ricketson2016sparse}.
%    In case of the cell-centered grids an exact equivalence between the two approaches does not hold.\ There the approach of first depositing the charge density onto the regular grid
%    and then transferring it to the component grids by means of the operator $R$ can be viewed as direct charge deposition onto the component grids with
%    a different shape function than the standard hat function. We verified that this new shape function satisfies the sufficient conditions given in \cite{filbet2003numerical} for shape functions. Also, empirically we 
%    observe it to be of second order and the difference in error levels between the two approaches for cell-centered grids is very small. Since a theoretical analysis of it is
%    more involved we defer it for future work.
%\end{remark}

    The advantage of the Monte Carlo point of view is that we can estimate the grid-based error and particle noise with explicit dependence on the number of particles and mesh size as we show in the section \secref{error_analysis}. From a pure filtering perspective this may be very difficult or not possible. 

    Now, we are interested in knowing how much grid-based error and particle noise are increased and decreased by the sparse grids filter respectively. To answer this, we observe that for interpolation the sparse grid combination technique is equivalent to
    the sparse grids based on hierarchical bases \cite{pfluger2010spatially}. The latter
    is identified based on an optimization process \cite{bungartz2004sparse} which guarantees for
    smooth functions, the least number of degrees of freedom for maximal accuracy of $\mc{O}\LRp{|log(h)|^{d-1}h^2}$ based on $L^2$ or $L^\infty$ norm. Thanks to this, in the context of PIC, the sparse grids transformation in equation \eqnref{sparse_filter} gives maximal
    noise reduction (because of the least number of grid points and hence maximum particles per cell) and at the same time the increase in grid-based error is minimal for smooth functions. Thus compared to other filters, the one based on the standard sparse grid combination technique is optimal in the sense of minimizing the total error
    for functions which are either smooth or aligned with an axis.

    \subsection{Truncated combination technique to handle non-aligned and non-smooth functions}

    The optimality mentioned in the previous section for sparse grids filtering is no longer 
    applicable in case of non-smooth functions or functions which are not aligned with either of the axes. 
    Here the grid-based error is significantly large compared to the regular grid because
    of large mixed derivatives \cite{trefethen2017cubature}. This is why in \cite{ricketson2016sparse}, the authors reported poor performance 
    of sparse PIC for the diocotron instability test case as it falls into the non-aligned category when simulated with a Cartesian grid. There are a few ways to tackle this problem as mentioned in \cite{ricketson2016sparse,cerfon2019sparse}, namely optimized coordinate systems which evolve with the 
    charge density or by means of spatially adaptive sparse grids. These strategies, which are perhaps more elegant from a mathematical point of view 
    and more efficient, however require significant changes 
    in the existing regular PIC code bases. Also no detailed, robust algorithm is known at the moment. Here, we pursue another direction using the truncated 
    combination technique \cite{leentvaar2008pricing,benk2012hybrid,benk2012variants} which is much simpler and can be easily implemented in existing 
    codes.

\begin{figure}[h!b!t!]
    \vspace{-15mm}
    \begin{center}
    \includegraphics[trim=9cm 3cm 3cm 3cm,clip=true,width=0.8\textwidth]{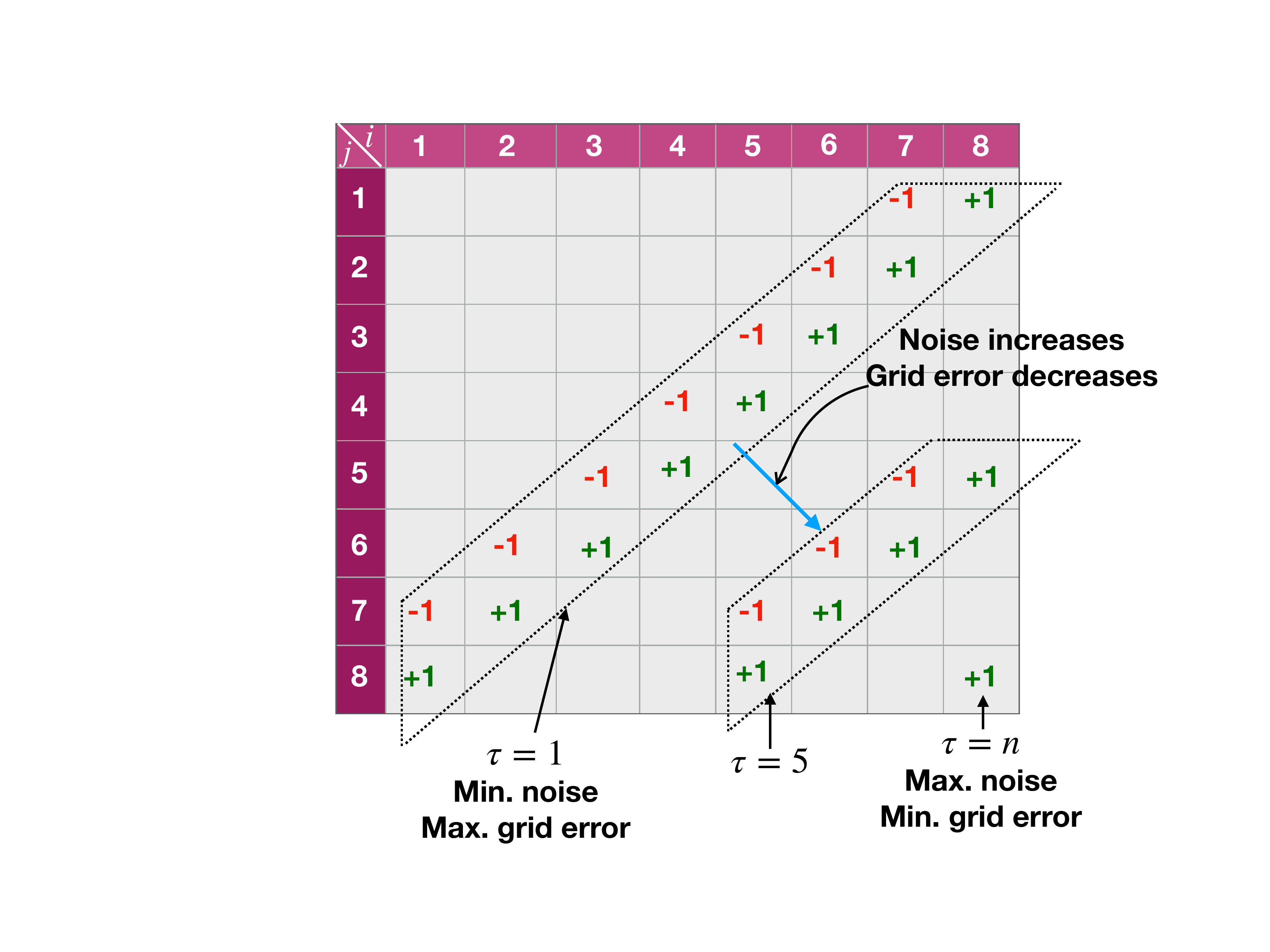}
    \caption{Schematic explaining the sparse grid combination technique and how the truncated combination can be used to minimize the 
    total error. Here, $\tau=1$ corresponds to the standard sparse grid combination technique 
        and $\tau=n$ corresponds to the regular grid. The indices $i$ and $j$ on the row and column headers indicate the mesh sizes of the component grids involved in the combination technique such that the $(i,j)$th component grid has mesh sizes $h_i=\frac{L}{2^i}$ and $h_j=\frac{L}{2^j}$, where $L$ is the length of the domain in each direction. The $+1$ and $-1$ are the combination coefficients $c_i$ in equation \eqnref{sparse_filter} corresponding to the component grids.}
    \figlab{truncated_comb}
    \end{center}
    \vspace{-5mm}
\end{figure}

    In Figure \figref{truncated_comb}, we show for a 2D problem with a regular mesh of size $2^8\times2^8$ the different combination strategies. The truncated combination technique \cite{leentvaar2008pricing,benk2012hybrid,benk2012variants} introduces a truncation parameter $\tau$\footnote{In this work we consider the same truncation parameter $\tau$ in all the directions.},
which is a positive integer that determines the component grids involved in the combination. If we consider a $2^n\times2^n$ regular grid, then the value of $\tau=1$ corresponds to 
the standard combination technique in \cite{griebel1990combination} and $\tau=n$ corresponds to 
the regular grid. By increasing $\tau$, fewer component grids are used in the combination technique, but each with finer 
mesh sizes than the previous $\tau$. This alleviates the issue of non-aligned and non-smooth functions by controlling the error term associated with the mixed fourth 
derivatives. Thus, the truncated combination technique provides a unified framework to 
transition from standard sparse grids to regular grid in terms of approximation capability by
increasing $\tau$.

Let us consider a PIC simulation with $N_p$ total number of particles and $2^n\times2^n$ regular grid with mesh size $h=\frac{L}{2^n}$. Now the regular grid with $\tau=n$ will have the least grid-based error and 
maximal noise because it has the mesh size $h$ in all the directions. The 
standard sparse grids technique with $\tau=1$ on the other extreme has maximal grid-based error and minimal noise as it has the mesh size $h$ in directions aligned with $x$ or $y$
axis but not in others. As we increase $\tau$ from $1$ to $n$ as shown in Figure \figref{truncated_comb}, we decrease the grid-based error
because of the inclusion of finer mesh sizes in the component grids but at the same time increase the particle noise due to decreased particles per cell or inclusion of higher wavenumbers in the filtering process \eqnref{sparse_filter}. Thus depending on the
smoothness and the orientation of the function there is an optimal $\tau$
at which the total error, which is the sum of grid-based error and particle noise, is minimized.
In the following we will present a formal error analysis and propose a heuristic approach to estimate the optimal $\tau$.

\subsection{Formal error analysis}
\seclab{error_analysis}

In \cite{ricketson2016sparse}, a formal error analysis is shown for sparse PIC quantifying the grid-based error and particle noise. Since our codes are
based on cell-centered grids (as it is the default choice in many plasma PIC codes \cite{adelmann2019opal,vay2018warp}), as mentioned in Proposition \proporef{MC_filtering_equivalence} the direct charge 
deposition in \cite{ricketson2016sparse} and the current approach are not exactly equivalent because of the differences in the shape
functions. Nevertheless, the order of accuracy is same for both the approaches and they differ only in constants.  
%empirically we observed same orders of convergence and 
%very small difference in the error levels between the two approaches. A rigorous error analysis of the current 
%approach for cell-centered grids is more involved and we leave it for future work. 
Hence, we will mostly
follow the steps in \cite{ricketson2016sparse} and generalize it to include the truncated combination technique.% To that extent, we would like to mention that
%the estimates derived here are not based on a rigorous error analysis and hence can be improved. A rigorous error analysis of the current approach
%is left as a part of future work and will be reported elsewhere. 

As shown in \cite{ricketson2016sparse} and appendix B, approximating $\rho_e$ in PIC simulations consists of two parts namely grid-based error and particle noise. In what follows we will quantify these two components and get an estimate of the total error.

%As explained earlier 
%our current 
%approach is slightly different from the one in \cite{ricketson2016sparse}. The particles deposit directly onto the component grids in 
%\cite{ricketson2016sparse}, whereas in the present strategy the particles first deposit onto the regular grid and from the regular 
%grid we transfer the density to the component grids using the restriction operator $R_i$ in equation \eqnref{sparse_filter}. 
%%Because of the additional interpolation involved in the current approach the grid-based error is increased, however only by a constant, as the grid-based error for sparse grids $\mc{O}\LRp{|log(h)|^{d-1}h^2}$ is more than that for 
%the regular grids which is $\mc{O}(h^2)$. For particle noise, this additional step 
%reduces the error as we are projecting the density which is already filtered through the regular grid onto the sparse grid space. Nevertheless, 
%through our numerical experiments in the context of interpolation we found that the difference in error levels between the two approaches is very 
%small and since the goal here is to develop a heuristic rather than a rigorous error analysis we use the estimates in \cite{ricketson2016sparse} and
%generalize it to include the truncated combination technique.

\subsubsection{Grid-based error}
 Let us remind that as per our notations, $\rho_e$ is the exact electron charge density given by 
 \[
     \rho_e({\bf x}) = q_e\int \f({\bf x},\vb) d\vb = \int\int \f({\bf \xi},\vb) \delta({\bf x}-{\bf \xi}) d{\bf \xi}d\vb,
 \]
 and $\varrho_e$ is the density on the regular grid after sparse grids transformation in equation \eqnref{sparse_filter}. We will denote the grid
 error component of the total error as $||\rho_e-\varrho_e||_{grid}$, where for simplicity we have denoted the $L^\infty$ norm $||.||_{L^\infty}$ by $||.||$ (equivalently, we can also use the
$L^2$-norm). In our approach the grid-based error comes from the approximation of delta-functions in the configuration space by shape functions of compact support as well as from the transfer operators $R$ and $P$. 
%The particle noise is introduced by means of approximating integrals in the 
% mean estimation by summation over finite number of particles. 
%Let us denote by $\varrho_e$ the charge density after this approximation and hence $||\overline\rho_e-\varrho_e||$ represents the particle noise. By means of the triangle inequality we can then write the total error 
%$||\rho_e-\varrho_e||$ as 
%\begin{equation}
%||\rho_e-\varrho_e|| \leq ||\rho_e-\overline\rho_e|| + ||\overline\rho_e-\varrho_e||.
%\end{equation}
%We want to minimize the total error and towards that we will first provide the estimate for the grid-based error and then come to the particle noise.

Towards quantifying the grid-based error, for simplicity, let us consider a 2D PIC simulation in a periodic domain $[0,L]^2$ and 
a regular mesh of size $2^n \times 2^n$. Let the mesh size of the regular grid be $h_n=\frac{L}{2^n}$ and the mesh sizes of the component grids be $h_i=\frac{L}{2^i}$ and $h_j=\frac{L}{2^j}$ for the $(i,j)$th component 
grid in Figure \figref{truncated_comb}. In our approach, we use the cloud-in-cell or linear interpolation operators for all the grid transfer operations. Hence, similar to \cite{ricketson2016sparse,griebel1990combination,leentvaar2008pricing}, we assume an error expansion of the form $C_1(h_i)h_i^2 + C_2(h_j)h_j^2 + D_1(h_i,h_j)h_i^2h_j^2$ where $C_1,C_2,D_1$ are appropriate coefficient functions with a uniform upper bound. The summation over the component grids in equation \eqnref{sparse_filter}
leads to pair-wise cancellations both in the standard sparse grid combination technique as well as in the truncated combination 
technique as shown in Figure \figref{truncated_comb}. After multiplying with
the combination coefficients and summing across all the component grids
we get
\begin{align}
    (\rho_e-\varrho_e)_{grid} &= C_1(h_n)h_n^2 + C_2(h_n)h_n^2 \nonumber \\ &+\frac{4h_n^2L^2}{2^{2\tau}}\LRs{\frac{1}{4}\sum_{\substack{i+j=n+\tau \\ i,j\geq \tau}}D_1(h_i,h_j) - \sum_{\substack{i+j=n+\tau-1 \\ i,j\geq \tau}}D_1(h_i,h_j)},
\end{align}
where we used the fact that $h_ih_j = \frac{h_nL}{2^{\tau}}$ when 
$i+j=n+\tau$ and $h_ih_j = \frac{h_nL}{2^{(\tau-1)}}$ when $i+j=n+\tau-1$. Taking the norm on both sides of the above equation and noting that there are
$n-(\tau-1)$ component grids with $i+j=n+\tau$ and $(n-1)-(\tau-1)$ component grids with $i+j=n+\tau-1$ we get
\begin{align}
    ||\rho_e-\varrho_e||_{grid} &\leq \kappa_1 h_n^2 + \kappa_2 h_n^2 + \frac{4\beta_1 h_n^2L^2}{2^{2\tau}}\LRs{\frac{n-(\tau-1)}{4}+\LRc{(n-1)-(\tau-1)}} \nonumber \\
    &\leq h_n^2 \LRp{\kappa_1  + \kappa_2  + \beta_1 L^2 2^{-2\tau}\LRs{5(n-\tau)+1}}. \eqnlab{2D_grid_error}
\end{align}
Here, $\kappa_1,\kappa_2$ and $\beta_1$ are the constants corresponding to the upper bounds such that $||C_1(h_n)||\leq\kappa_1$, $||C_2(h_n)||\leq\kappa_2$ and $||D_1(h_i,h_j)||\leq\beta_1$, $\forall h_i,h_j$. The same expression for the error is 
also obtained in \cite{leentvaar2008pricing} for the truncated combination
in 2D. Similarly one can derive the estimates in 3D and the grid-based error in that case is given by
\begin{align}
    ||\rho_e-\varrho_e||_{grid} &\leq h_n^2 \LRp{\kappa_1  + \kappa_2  + \kappa_3 + (\beta_1 + \beta_2 + \beta_3) L^2 2^{-2\tau}\LRs{5(n-\tau)+1} \right.\nonumber \\&+ \left. \gamma L^4 2^{-(4\tau+1)}\LRc{25(n-\tau)^2-5(n-\tau)+2}}, \eqnlab{3D_grid_error}
\end{align}
where the upper bounds for the coefficient functions in 3D are such that $||C_d(h_n)||\leq\kappa_d$, $||D_d(h_i,h_j)||\leq\beta_d$ and $||F(h_i,h_j,h_k)||\leq\gamma$ for $d=1,2,3$ and $\forall h_i,h_j,h_k$. By plugging in 
$\tau=1$ and $\tau=n$ in \eqnref{2D_grid_error} and \eqnref{3D_grid_error}
we recover the estimates for the standard sparse grid combination in \cite{griebel1990combination} and for regular 
grids respectively.

\subsubsection{Particle noise}

Now, we will consider the estimates for the particle noise in the total 
error. The particle noise is the result of approximating the expected value of the shape function by an arithmetic mean over a finite number of 
discrete particles. % In contrast to \cite{ricketson2016sparse}, where the particles
%deposit directly onto the component grids in our current approach they 
%deposit onto the regular grid and then by means of $R$ we transfer it to
%the component grids. The presence of regular grid in the current 
%approach acts as an additional layer of filter and reduces the error in the
%noise dominated regime compared to the approach in \cite{ricketson2016sparse}. This is observed empirically by means of 
%numerical experiments we performed in the context of interpolation. Nevertheless, the difference in error levels observed is very small and as 
%of the moment since we do not know how to exactly quantify the particle 
%noise in the current approach we follow the analysis in \cite{ricketson2016sparse}. 
As per the error analysis in \cite{ricketson2016sparse}, in 2D 
the particle noise in each component grid is 
$\mc{O}\LRp{1/\sqrt{N_ph_ih_j}}$ and as stated in the grid error 
estimates we have $n-(\tau-1)$ component grids each with $h_ih_j = \frac{h_nL}{2^{\tau}}$ and $(n-1)-(\tau-1)$ component 
grids with $h_ih_j = \frac{h_nL}{2^{(\tau-1)}}$. Thus we can write an estimate for the particle noise as
\begin{align}
    ||\rho_e-\varrho_e||_{noise} &= \mc{O}\LRp{\sigma\LRs{\frac{n-(\tau-1)}{\sqrt{\frac{N_ph_nL}{2^\tau}}}+\frac{(n-1)-(\tau-1)}{\sqrt{\frac{N_ph_nL}{2^{(\tau-1)}}}}}} \nonumber \\
                        &= \mc{O}\LRp{\sigma\LRc{\frac{2^{0.5(\tau-1)}\LRs{(n-\tau)(1+\sqrt{2})+\sqrt{2}}}{\sqrt{N_p h_n L}}}}, \eqnlab{2D_particle_noise}
\end{align}
where $\sigma$ is the particle noise constant. Following the same procedure, the noise estimate in 3D is given by 
\begin{equation}
    ||\rho_e-\varrho_e||_{noise} = \mc{O}\LRp{\sigma \LRc{\frac{2^{(\tau-2)}\LRs{(3+\sqrt{2})(n-\tau)^2 + (5+\sqrt{2})(n-\tau) + 4}}{\sqrt{N_p h_n L^2}}}}.
    \eqnlab{3D_particle_noise}
\end{equation}
Again, by plugging in $\tau=1$ and $\tau=n$ in equations \eqnref{2D_particle_noise}, \eqnref{3D_particle_noise} we recover the estimates shown in \cite{ricketson2016sparse} for standard sparse grids
technique and regular grids respectively. With the grid and particle error
estimates in hand we will show how these can be used in practice to adaptively select the optimal $\tau$.

\subsubsection{A heuristic approach for estimating coefficients of the error}
\seclab{heuristic}
   In order to use the grid and particle error estimates derived in the previous section we need to have an estimate of the coefficients. To that extent, we note that a rigorous derivation of coefficients for the current approach in case of cell-centered grids depends on the ratio of the mesh sizes of the component grids to the regular grid and is more involved. Instead in this section we approximate the grid and particle coefficients based on heuristic arguments and empirical
   observations and intend to improve these choices in the future iterations of our algorithm.
   % In this section we provide some empirical choices for them based on proportionality
%   arguments. To that extent, we note that there are additional constants involved in both these upper bounds which depend on the details
%   of the numerical scheme. We do not consider them in the present study and will include in the future iterations of our algorithm. 
Let us first consider the grid-based error before we discuss the particle noise. As explained in \cite{ricketson2016sparse,cerfon2019sparse} and equations \eqnref{2D_grid} and \eqnref{3D_grid} in appendix B,
   the coefficient functions in the grid error estimates are proportional to the derivatives of the charge density $\rho_e$ such that
\begin{align*}
    &C_1 \propto \frac{\partial^2\rho_e}{\partial x^2},C_2 \propto \frac{\partial^2\rho_e}{\partial y^2},C_3\propto \frac{\partial^2\rho_e}{\partial z^2},
D_1\propto \frac{\partial^4\rho_e}{\partial x^2\partial y^2} \\
    &D_2\propto \frac{\partial^4\rho_e}{\partial y^2\partial z^2},D_3\propto \frac{\partial^4\rho_e}{\partial z^2\partial x^2}, F\propto \frac{\partial^6\rho_e}{\partial x^2\partial y^2\partial z^2}.
\end{align*}
    
    In PIC we only have an approximation of $\rho_e$ on the regular grid, which we call $\rhoa_e$ as defined in equation \eqnref{rhoaa1}, and this also contains the particle noise. In order to have 
    a realistic approximation of the derivatives of the charge density from noisy regular PIC data $\rhoa_e$ we perform a denoising by thresholding in 
    the Fourier domain. Specifically, we first take the Fourier transform of the density on the regular grid $\rhoaa_e = \mc{F}\LRp{\rhoa_e}$ and perform a
    hard thresholding such that  
    \begin{equation}
        \rhoaa_e =
          \begin{cases}
              \rhoaa_e &\quad |\rhoaa_e|\geq \epsilon,\\
              0  &\quad |\rhoaa_e| < \epsilon,
          \end{cases}
       \eqnlab{denoising}     
    \end{equation}
where $\epsilon$ is the threshold for denoising and $|\rhoaa_e|$ denotes the magnitude of the Fourier transform $\rhoaa_e$. This type of denoising is common
in signal processing as well as wavelet denoising \cite{donoho1995adapting} techniques. 

The threshold parameter $\epsilon$ is a function of the number of particles per cell $P_c$, the initial sampling method and also the distribution $f$. It determines how much noise and signal
is removed by the denoising process. Too low a value will not remove much noise and too high a value may remove a significant portion of the signal along with the noise. However, in contrast to denoising techniques in signal 
processing where after applying this threshold one performs an inverse 
transform to get the signal in the physical domain, we emphasize the fact that for our scheme we only use it for
selecting the truncation parameter $\tau$ (which performs the final filtering). Hence the 
threshold $\epsilon$ does not need to be optimal, and we only need to ensure that we do not pick up excessive noise.% Also we would like to note that, if we directly use hard thresholding like this and inverse transform to the
%physical domain then we get non-physical ripples in the density profile which is not desired. 

At present, we use an ad-hoc strategy to select the value of $\epsilon$ as a certain percentage of the maximum 
value of $|\rhoaa_e|$, namely $\epsilon = \alpha \max\LRp{|\rhoaa_e|}$, where $\alpha$ 
denotes the percentage. To determine $\alpha$ in our algorithm, for a certain number of particles per cell $(P_c)_{ref}$ (e.g., 5) we run the PIC simulation for a few different values of $\alpha$ and pick the minimum value necessary for denoising. To reduce the run time we use a coarse mesh necessary for the problem in these simulations. Once we pick the value of $\alpha$
for a reference number of particles per cell $(P_c)_{ref}$, we run simulations
with other values of $P_c$ by multiplying $\alpha$ by $\sqrt{(P_c)_{ref}/P_c}$, as we know the noise
in PIC methods scales as $1/\sqrt{P_c}$. In the numerical results in section 
\secref{numerical_results}, we performed the experiments with few different $\alpha$ to examine
robustness. It is observed that within a range of $\alpha$ (which is problem specific) the results not change much. In our future work we will develop a more systematic way to pick the threshold from the density data, based on techniques similar to the ones used in wavelet denoising \cite{donoho1995adapting}. Machine learning techniques can also be used for this purpose and this is another direction we will pursue.

After denoising the charge density, we compute the derivatives in the 
Fourier domain and perform inverse transforms. Next, in order to fix the constants in front of these derivatives in appendix B we derive the
grid-based error for regular PIC schemes. Since, each component grid in the
sparse grid combination technique is a regular grid with mesh sizes $h_i$, $h_j$ and $h_k$, equations \eqnref{2D_grid} and \eqnref{3D_grid} can be used for determining the constants involved in the upper bounds. % However, as mentioned in remark \remaref{cell_center_remark} for cell-centered grids 
%the shape functions for the current approach are not same as the standard hat functions. Here we come up with approximations for the constants in the 
%current approach.
To that extent, we note that by means of the grid transfer operations with
$R$ and $P$ we incur two times the grid-based error of similar magnitude given in equations \eqnref{2D_grid} and \eqnref{3D_grid}. Moreover, the 
charge density $\rhoa_e$ in the regular grid adds another $1/12$ in 
front of the second derivative terms. Summing all these contributions we get an estimate for the coefficients in equations \eqnref{2D_grid_error} and \eqnref{3D_grid_error} as  
\begin{align}
    &\kappa_1 = \frac{1}{4}\left\Vert\frac{\partial^2\bar\rho_e}{\partial x^2}\right\Vert,\kappa_2 = \frac{1}{4}\left\Vert\frac{\partial^2\bar\rho_e}{\partial y^2}\right\Vert,\kappa_3= \frac{1}{4}\left\Vert\frac{\partial^2\bar\rho_e}{\partial z^2}\right\Vert,
    \beta_1 = \frac{1}{72} \left\Vert\frac{\partial^4\bar\rho_e}{\partial x^2\partial y^2}\right\Vert \nonumber \\
    &\beta_2 = \frac{1}{72}\left\Vert\frac{\partial^4\bar\rho_e}{\partial y^2\partial z^2}\right\Vert,\beta_3 =\frac{1}{72}\left\Vert\frac{\partial^4\bar\rho_e}{\partial z^2\partial x^2}\right\Vert, \gamma = \frac{1}{864}\left\Vert\frac{\partial^6\bar\rho_e}{\partial x^2\partial y^2\partial z^2}\right\Vert,
    \eqnlab{grid_constants}
\end{align}
where $\bar\rho_e$ is the denoised charge density defined in equation \eqnref{rhoa1}.

Finally, following the particle noise estimates in equations \eqnref{2D_noise} and \eqnref{3D_noise} as well as \cite{ricketson2016sparse,terzic2007particle}, for our algorithm we take 
\begin{equation}
\eqnlab{particle_constant}
\sigma = \sqrt{\LRp{2/3}^{d}\left\Vert Q_e\rhoa_e\right\Vert}
\end{equation}
in equations \eqnref{2D_particle_noise} and \eqnref{3D_particle_noise}, where $d$ is the dimension and $\rhoa_e$ is the charge density on the regular grid before denoising
as defined in equation \eqnref{rhoaa1}. Here, we use the density $\rhoa_e$ instead of the denoised density $\bar\rho_e$ as it helps in adjusting the  
particle constant with respect to different sampling techniques.% We again note that this is only an approximation as our current approach in case of the 
%cell-centered grids are not standard hat functions. Nevertheless, our choice is based on empirical observations where the difference in error levels between
%those two shape functions is very small.

Through numerical experiments we also found another choice for the coefficients in the grid-based error and particle noise as  
\begin{align}
    &\kappa_1 = \left\Vert k_x^2\rhoaa_e\right\Vert,\kappa_2 =\left\Vert k_y^2\rhoaa_e\right\Vert,\kappa_3 =\left\Vert k_z^2\rhoaa_e\right\Vert,\beta_1 =\left\Vert k_x^2k_y^2\rhoaa_e\right\Vert \nonumber \\
    &\beta_2 =\left\Vert k_y^2k_z^2\rhoaa_e\right\Vert, \beta_3 =\left\Vert k_x^2k_z^2\rhoaa_e\right\Vert, \gamma =\left\Vert k_x^2k_y^2k_z^2\rhoaa_e\right\Vert, \sigma = \sqrt{\left\Vert Q_e\rhoa_e\right\Vert}
    \eqnlab{empirical_constants}    
\end{align}
where $k_x,k_y$ and $k_z$ are the wavenumbers in $x,y$ and $z$ respectively. We do not present detailed results, but for the numerical experiments in section \secref{numerical_results} as well
as for other synthetic examples in the context of interpolation we found this choice yields similar optimal $\tau$ values as that of the constants in 
equations \eqnref{grid_constants} and \eqnref{particle_constant}. It has an added advantage that we do not need to take inverse transform of the 
derivatives, which is three in 2D and seven in 3D. Thus it may be of interest from a practical point of view, and for the numerical experiments in section
\secref{numerical_results} we observed up to 7 times speedup in the $\tau$ estimation part with this choice compared to the ones in equations \eqnref{grid_constants} and \eqnref{particle_constant}.

In Algorithm \ref{al:tauEstimator} we consolidate the steps in the optimal $\tau$ estimator algorithm. For the range of $\tau$, we 
    consider $[1,n-3]$ for 2D and $[1,n-2]$ for 3D where $2^{n}$ is the number of points in the regular grid in each dimension. This is because in 2D, the total
    number of degrees of freedom in $\tau=n-2$ sparse grid is same as that of the regular grid and $\tau=n-1$ has more degrees of freedom than the regular
    grid. Thus we include up to $n-3$ in 2D to have a higher value of $P_c$ than for the regular 
    grid and hence less particle noise. In 3D on the other hand the permissible range of $\tau$ is $[1,n-2]$ and even with $\tau=n-2$
    we have fewer degrees of freedom than for the regular grid.
\subsection{Implementation in a HPC PIC code base.}

    Once the optimal $\tau$ is obtained from Algorithm \ref{al:tauEstimator} we need to perform sparse grids noise reduction. In Algorithm \ref{al:transferToSparse} we present a matrix-free implementation of the sparse grids filtering in equation \eqnref{sparse_filter}. This 
    implementation would be suitable for large-scale high performance PIC code bases like OPAL. In these codes the density in the 
    regular grid is domain decomposed between different processors and in
    Algorithm \ref{al:transferToSparse} each processor holds the entire
    component grid in the combination technique. For moderate values of $\tau$, each component grid
    has very few degrees of freedom compared to the regular grid and this is not
    very expensive in terms of memory. However for high $\tau$, the component grids involved in the combination has a considerable number of degrees of freedom (especially in 3D) and hence both memory as well as the MPI\_Allreduce step in Algorithm \ref{al:transferToSparse} 
    could present a bottleneck. In our future work we will also split up the 
    component grids between processors which would require a more 
    complicated parallelization strategy as shown in \cite{strazdins2015highly}.

    If the parallelization of the code base uses MPI for inter-node parallelism and OpenMP for intra-node parallelism then the for loop over component grids in Algorithm \ref{al:transferToSparse} can be done in parallel with OpenMP. 
    Algorithms \ref{al:tauEstimator} and \ref{al:transferToSparse} would go
    in between steps \ref{step1_pic} and \ref{step2_pic} in the regular PIC procedure outlined in section \secref{pic}. Ingredients such as FFT which are required for the tauEstimator algorithm are already available in many of the large-scale PIC
    code bases and hence these two algorithms can be incorporated inside them very easily without any modification to the other parts. 
    
\begin{algorithm}
  \begin{algorithmic}[1]
      \STATE Compute Fourier transform of the charge density $\rhoaa_e = \mc{F}\LRp{\rhoa_e}$.
      \STATE Perform denoising by hard thresholding according to equation \eqnref{denoising}.
      \STATE Compute the constants for the grid-based error with \eqnref{grid_constants} and the particle error constant \eqnref{particle_constant}.
      \FOR {$\tau=1$ \TO $n-3$ for 2D and $n-2$ for 3D} 
       \STATE Evaluate grid-based error and particle noise using equations \eqnref{2D_grid_error},\eqnref{2D_particle_noise} for 2D and \eqnref{3D_grid_error},\eqnref{3D_particle_noise} for 3D. 
    \ENDFOR
      \STATE Select the $\tau$ with minimum total error.
    \end{algorithmic}
  \caption{tauEstimator: An algorithm for estimating optimal $\tau$.}
  \label{al:tauEstimator}
    %\vspace{-5mm}
\end{algorithm}
\begin{algorithm}
  \begin{algorithmic}[1]
      \FOR{$i=1$ \TO $nc$}
      \STATE Each processor deposits their regular grid partition of $\rhoa_e$ to the
      $i$th component grid using the transfer operator $R_i$ in equation \eqnref{restriction}.
      %\STATE Multiply by $h^d/V_i$, where $V_i$ is the volume of each cell in the $i$th component grid. 
      \STATE MPI\_Allreduce to add contributions from all processors on the $i$th component grid.
      \STATE Each processor interpolates from the $i$th component grid to their regular grid partition of $\rhoa_e$ using transfer operator $P_i$ in equation \eqnref{restriction}.
      \STATE Multiply by combination coefficient and accumulate.
    \ENDFOR
    \end{algorithmic}
  \caption{transferToSparse: An algorithm for sparse grids based noise reduction with a given $\tau$.}
  \label{al:transferToSparse}
\end{algorithm}

\begin{remark}
    In general the charge density $\varrho_e$ after sparse grids transformation is not guaranteed to be positive everywhere. This is not unique to our approach and also happens in other noise reduction strategies such as high-order shape functions \cite{myers20174th}, compensating filters \cite{birdsall2004plasma} and wavelet-based density estimation \cite{del2010wavelet}. In our numerical
    results in section \secref{numerical_results} we do not observe any
    problems caused by this. However, we could adopt the density redistribution procedure used in \cite{myers20174th} to make the charge
    density positive everywhere after the sparse grids transformation. This will be studied in future versions of the algorithm. Also, as shown
    in \cite{shalaby2017sharp}, the filtering procedures used in explicit PIC simulations improve energy conservation but at the loss of 
    momentum conservation. In our future study we will investigate in detail the impact of the noise reduction strategy on energy and momentum 
    conservation and report the results.
\end{remark}

\section{Numerical results}
\seclab{numerical_results}
    In this section we will test the performance of the adaptive noise reduction strategy on two benchmark problems in plasma physics and
    beam dynamics; namely two-dimensional diocotron instability, and three-dimensional electron dynamics in a Penning trap with a neutralizing ion background. These test cases produce fine-scale structures during the nonlinear evolution and thus can be used to evaluate the ability of the adaptive $\tau$ method to capture them
    while still reducing noise. Also, they are very relevant to the large-scale accelerator simulations which we intend to 
    perform in our future works. 

    In all the simulations we consider a periodic box $\Omega=[0,L]^d$, where $d$ is the dimension and $L$ is the length in each dimension.
    The charge to mass ratio $q_e/m_e$ in all our simulations is $-1$. In measuring the error in field quantities we use the relative discrete $L^2$-norm
    also known as the normalized root mean squared error given by
    \begin{equation}
        \eqnlab{error_def}
        \mathcal{E}(\boldsymbol \psi) = \sqrt{\frac{\sum_{i=1}^{N_{points}}\LRp{\boldsymbol \psi({\bf x}_i) - \boldsymbol \psi_{ref}({\bf x}_i)}^2}{\sum_{i=1}^{N_{points}}(\boldsymbol \psi_{ref}({\bf x}_i))^2}},  
    \end{equation}
    where $\boldsymbol \psi$ is any field quantity, $\boldsymbol \psi_{ref}$ is the reference field which is
    obtained from an ensemble average of high-resolution regular PIC simulations and ${\bf x}_i$ are the locations of points in the domain at which we measure the error. This error is for a particular time instant and we measure the error
    at few instants in the whole simulation. In both numerical examples, we calculate the error for regular PIC, adaptive $\tau$ PIC and fixed $\tau$ PIC with the range of $\tau$ taken to be the same as the one used in the tauEstimator Algorithm \ref{al:tauEstimator}. By means of these error curves we can see how well the adaptive $\tau$ algorithm 
    performs in terms of picking the optimal $\tau$ and also how the errors compare to that of the regular PIC results with different number of particles per cell $P_c$. We always define the number of particles per cell $P_c$ based on the regular grid. It is given by 
    \[
        P_c = \frac{N_p}{N_c} = \frac{N_p}{2^{nd}}.
    \]
        
        For the time integration we use the leap frog method and for the Poisson equation we use the second order cell-centered finite difference
        method as in \cite{martin2000cell,frey2019architecture} with single level and without
        any spatial adaptivity. For solving the linear system arising 
        from the discretized Poisson equation we use the smoothed aggregation algebraic multigrid (SAAMG) from the second generation Trilinos
        MueLu library \cite{berger2019muelu}. The stopping tolerance for 
        the iterative solver is set as $10^{-10}$ multiplied by the infinity norm of the right hand side. More details on the solver can be found in 
        \cite{frey2019architecture}. The code is written on top of a C++ miniapp based on the particle accelerator library 
        OPAL \cite{adelmann2019opal} and box structured adaptive 
        mesh refinement library AMReX \cite{AMReX_JOSS}. Even though 
        FFT solver would be the most accurate and fastest option \cite{gholami2016fft} in this context, the reason for the 
        above choice of field solver is in our future work we want to extend the current approach to include adaptive mesh refinement. Also, the conclusions of the present
        study will not be much affected by this choice and will be applicable 
        for FFT solver too.

    All the computations are performed on the Merlin6 HPC cluster at the Paul Scherrer Institut, the details of which are as follows. Each Merlin6 
    node consists of 2 sockets  and each socket in turn has Intel Xeon Gold 6152 processor with $22$ cores at 2.1-3.7GHz. There are 2 threads in each core,
    however in all the present computations we only use single thread. Each node contains 384 GB DDR4 memory in total.

\subsection{2D diocotron instability}
    As a first example, we consider the 2D diocotron instability test case as already described in \cite{ricketson2016sparse}. In this test case, we have electrons with a hollow density profile immersed in a neutralizing immobile and uniform ion background and confined by a uniform external axial magnetic field. The magnetic field is strong enough that the electron dynamics is dominated by advection in the self-consistent $\Eb_{sc} \times \B_{ext}$ velocity field \cite{driscoll1990experiments,fine1991,cerfon2013,cerfon2016}. The initial electron density profile is not monotonic in the radial direction, which translates to an $\Eb_{sc} \times \B_{ext}$ shear flow which is unstable to what is known as the Kelvin-Helmholtz shear layer instability \cite{driscoll1990experiments,DrazinReid2004,cerfon2016} in fluid dynamics, and the diocotron instability in beam and plasma physics \cite{aydemir1994unified,davidson2001physics,driscoll1990experiments}. This instability deforms the initially axisymmetric electron density distribution, leading, in the nonlinear phase, to the formation of a discrete number of vortices, and eventually breakup \cite{cerfon2016,davidson2001physics}. This test case
    has importance both from a fundamental physics point of view \cite{aydemir1994unified,davidson2001physics,driscoll1990experiments} as well as in practical applications such as beam collimation \cite{jo2018control}.

    The parameters for this test case are as follows and are very similar to the ones in \cite{ricketson2016sparse}. We apply a uniform 
    external magnetic field $\B_{ext} = \LRc{0,0,5}$ along the $z-$axis in a domain of length $L=22$. The external electric field $\Eb_{ext}={\bf 0}$ for this 
    problem. The initial distribution is given by
    \begin{align}
        \nonumber
        f(t=0) &= \frac{C}{2\pi} e^{-|\vb|^2/2} \exp\LRc{-\frac{(r-L/4)^2}{2(0.03L)^2}},\\
       \eqnlab{dist_diocotron}
        r &= \sqrt{(x-L/2)^2 + (y- L/2)^2},
    \end{align}
    and the constant $C$ is chosen such that the total electron charge $Q_e=-400$. We sample the phase space using Gaussian
    distribution in the velocity variables with mean $0$ and standard deviation $1$. For the configuration space, we use a uniform
    distribution for $\theta$ in $[0,2\pi]$, and for $r$ a Gaussian distribution with mean $L/4$ and standard deviation $0.03L$. From $(r,\theta)$
    we do the polar to Cartesian transformation to get $(x,y)$ for the particles. We will refer this sampling as Gaussian sampling to differentiate it from the uniform sampling which will be presented later in this section.

        For denoising in equation \eqnref{denoising}, we take $\epsilon=\alpha \sqrt{(P_c)_{ref}/P_c} \max(|\rhoaa_e|)$ as explained in section 
        \secref{heuristic}, where $(P_c)_{ref} = 5$ and $\alpha=0.01$. This means that with $5$ particles
    per cell, charge densities with Fourier amplitude less than 1 percent of the maximum amplitude will be set to 0 and for other $P_c$ the threshold
    will be scaled accordingly. The time step of the time integrator is chosen as $\Delta t=0.02$
    and the simulation is run till final time $T=17.5$.

\begin{figure}[h!t!b!]
\subfigure[time=0]{
\includegraphics[trim=23cm 6cm 23cm 6cm,clip=true,width=0.3\columnwidth]{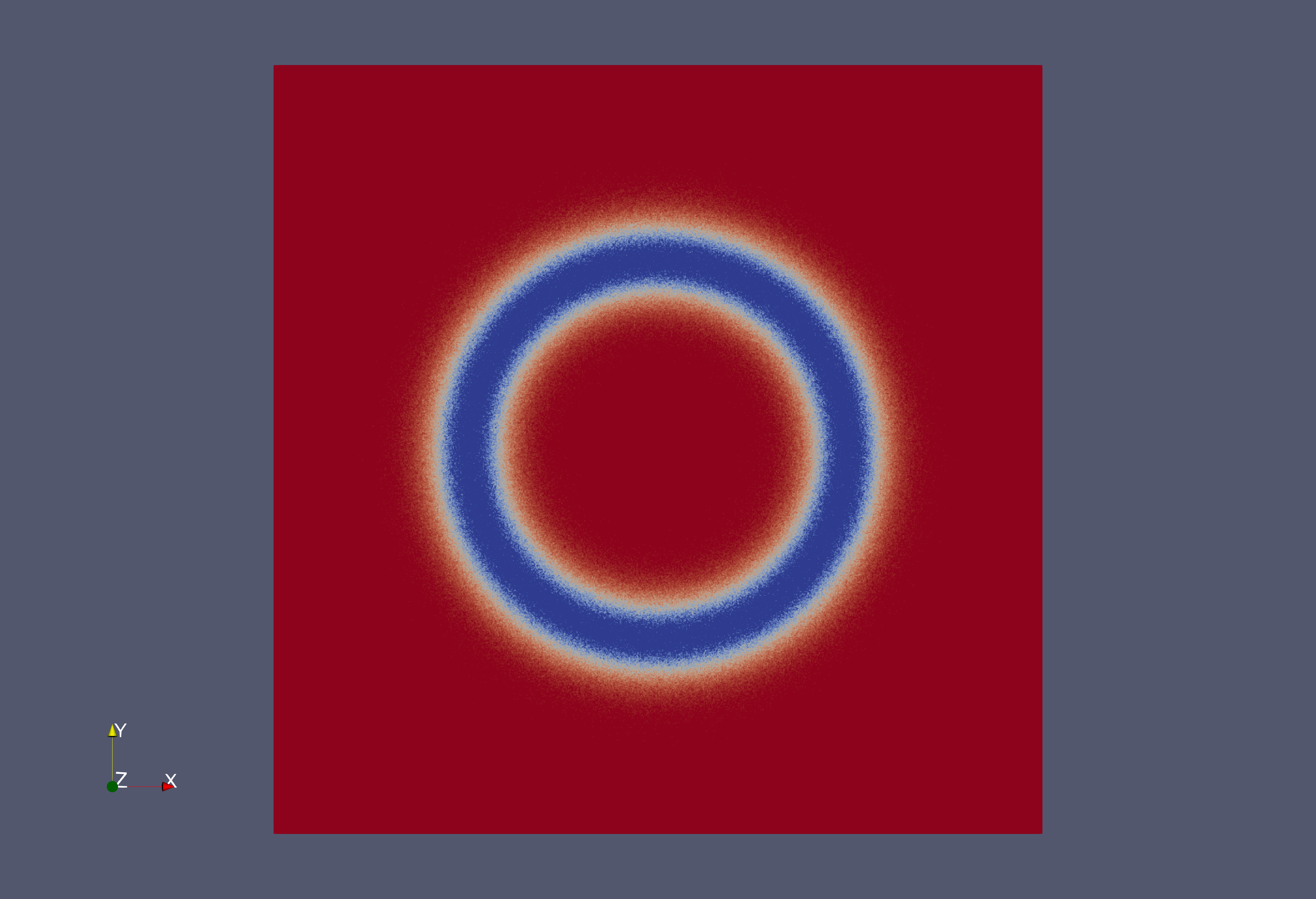}
}
\subfigure[time=10]{
\includegraphics[trim=23cm 6cm 23cm 6cm,clip=true,width=0.3\columnwidth]{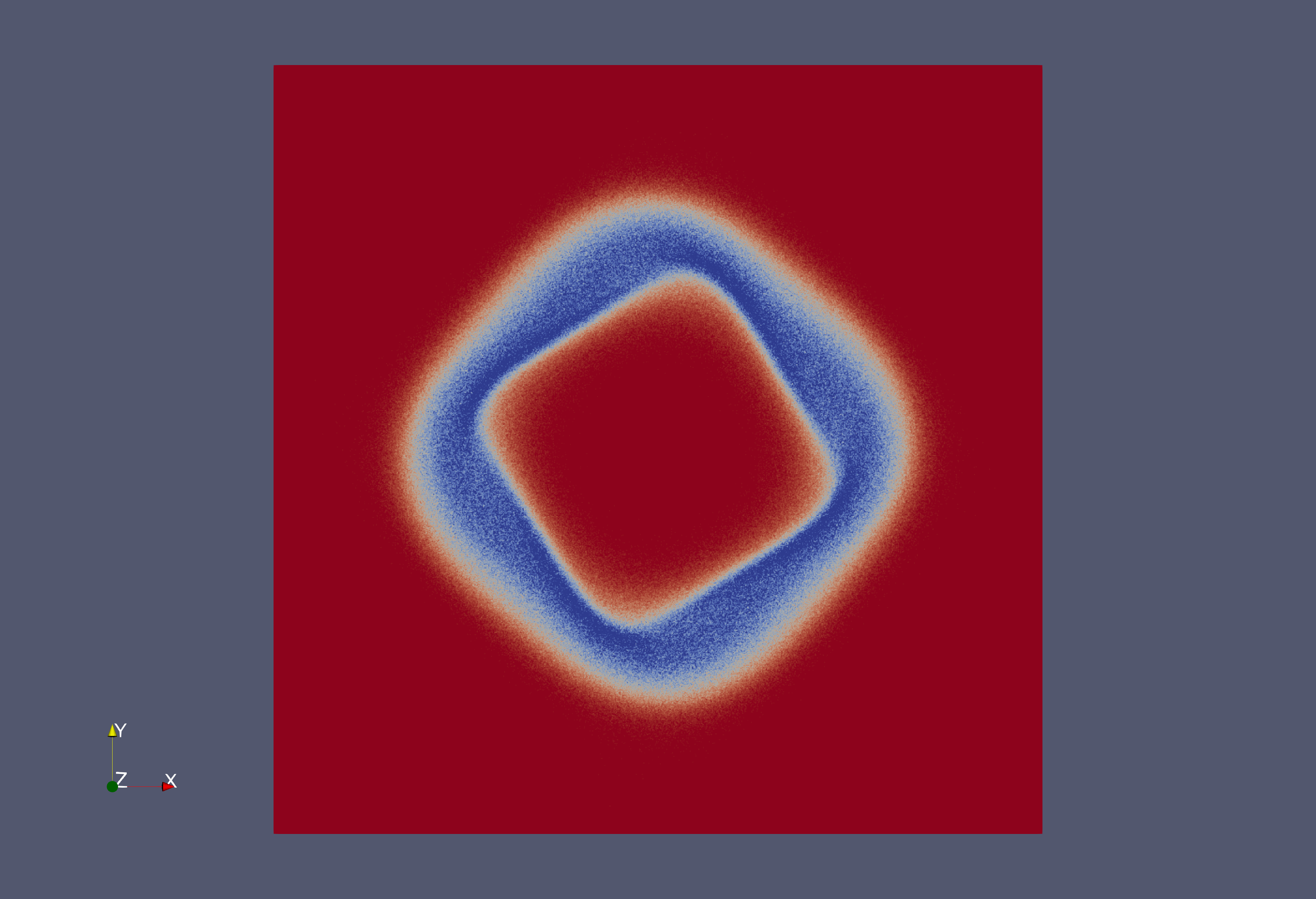}
}
\subfigure[time=17.5]{
\includegraphics[trim=23cm 6cm 23cm 6cm,clip=true,width=0.3\columnwidth]{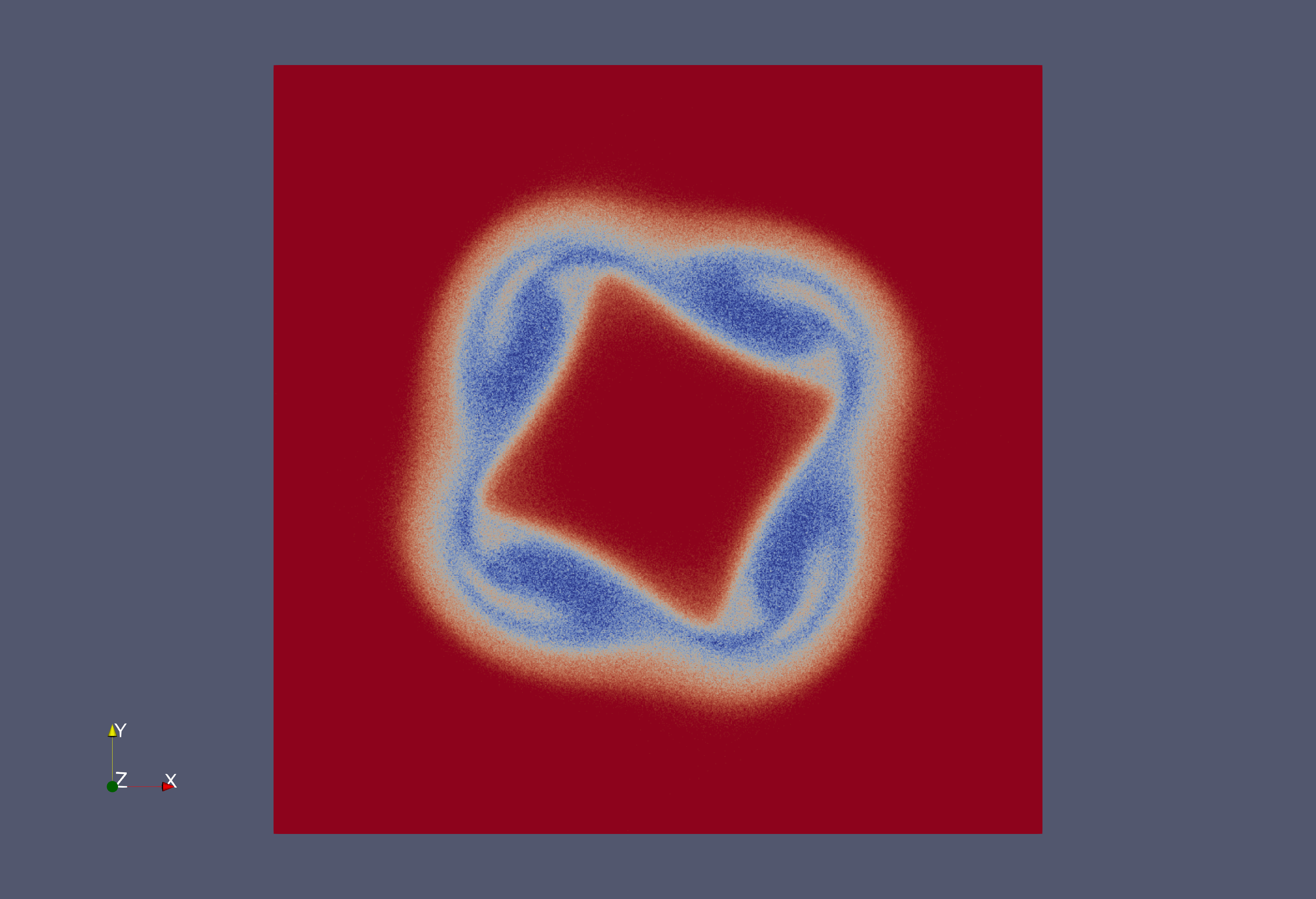}
}
\subfigure[time=0]{
\includegraphics[trim=23cm 6cm 23cm 6cm,clip=true,width=0.3\columnwidth]{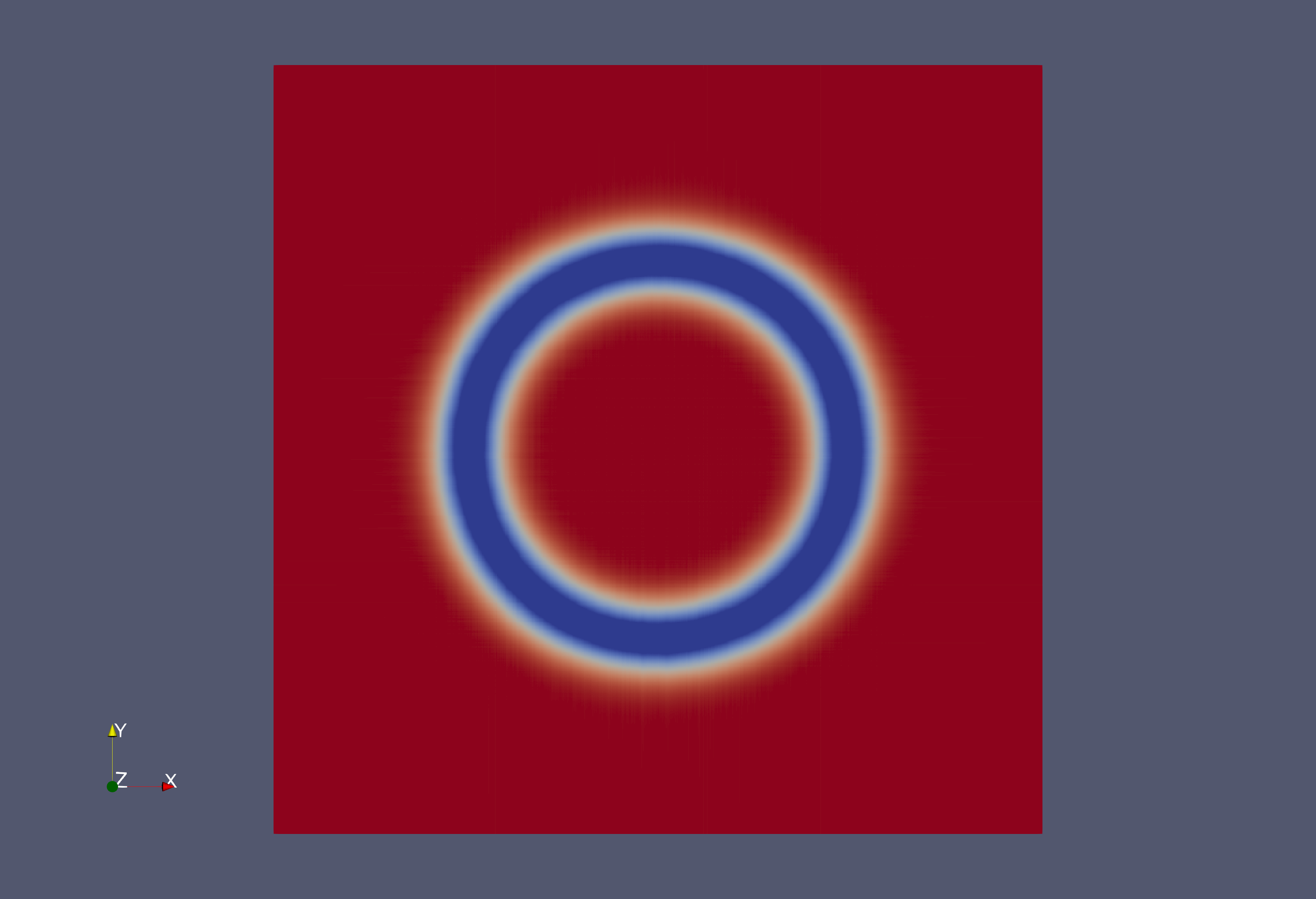}
    \figlab{time0_diocotron_tau1_density}
}
\subfigure[time=10]{
\includegraphics[trim=23cm 6cm 23cm 6cm,clip=true,width=0.3\columnwidth]{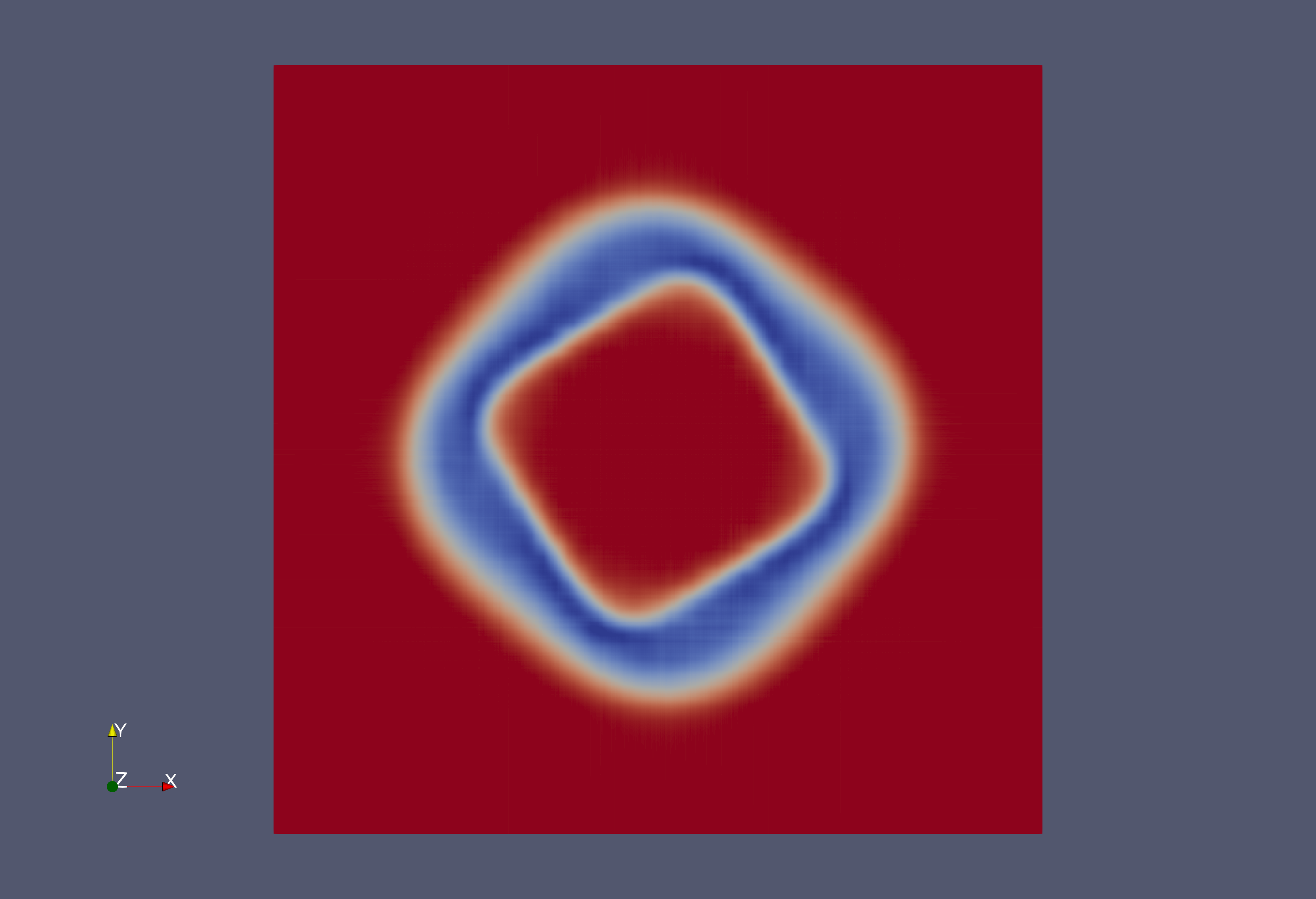}
}
\subfigure[time=17.5]{
\includegraphics[trim=23cm 6cm 23cm 6cm,clip=true,width=0.3\columnwidth]{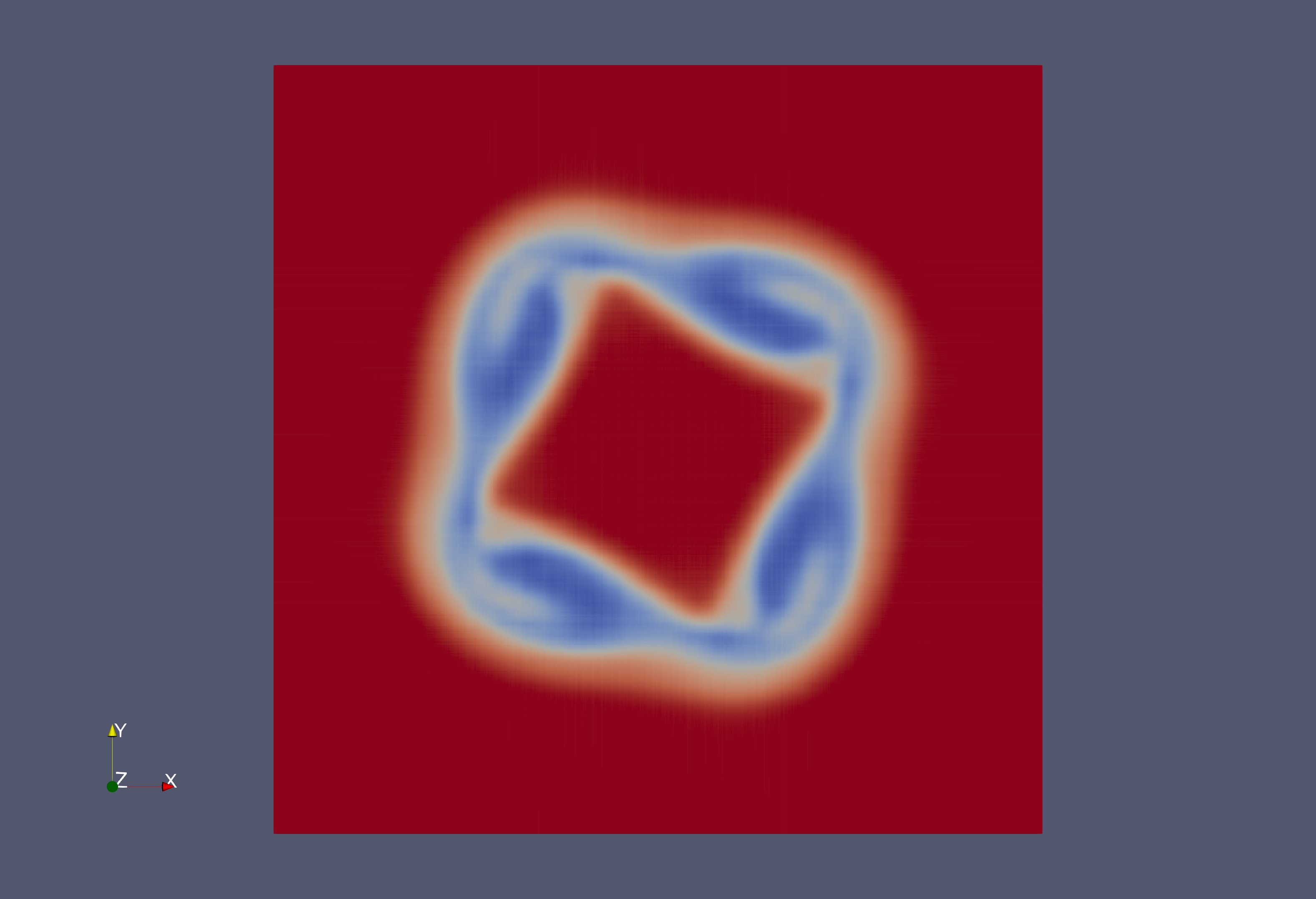}
    \figlab{time875_diocotron_tau1_density}
}
\subfigure[time=0]{
\includegraphics[trim=23cm 6cm 23cm 6cm,clip=true,width=0.3\columnwidth]{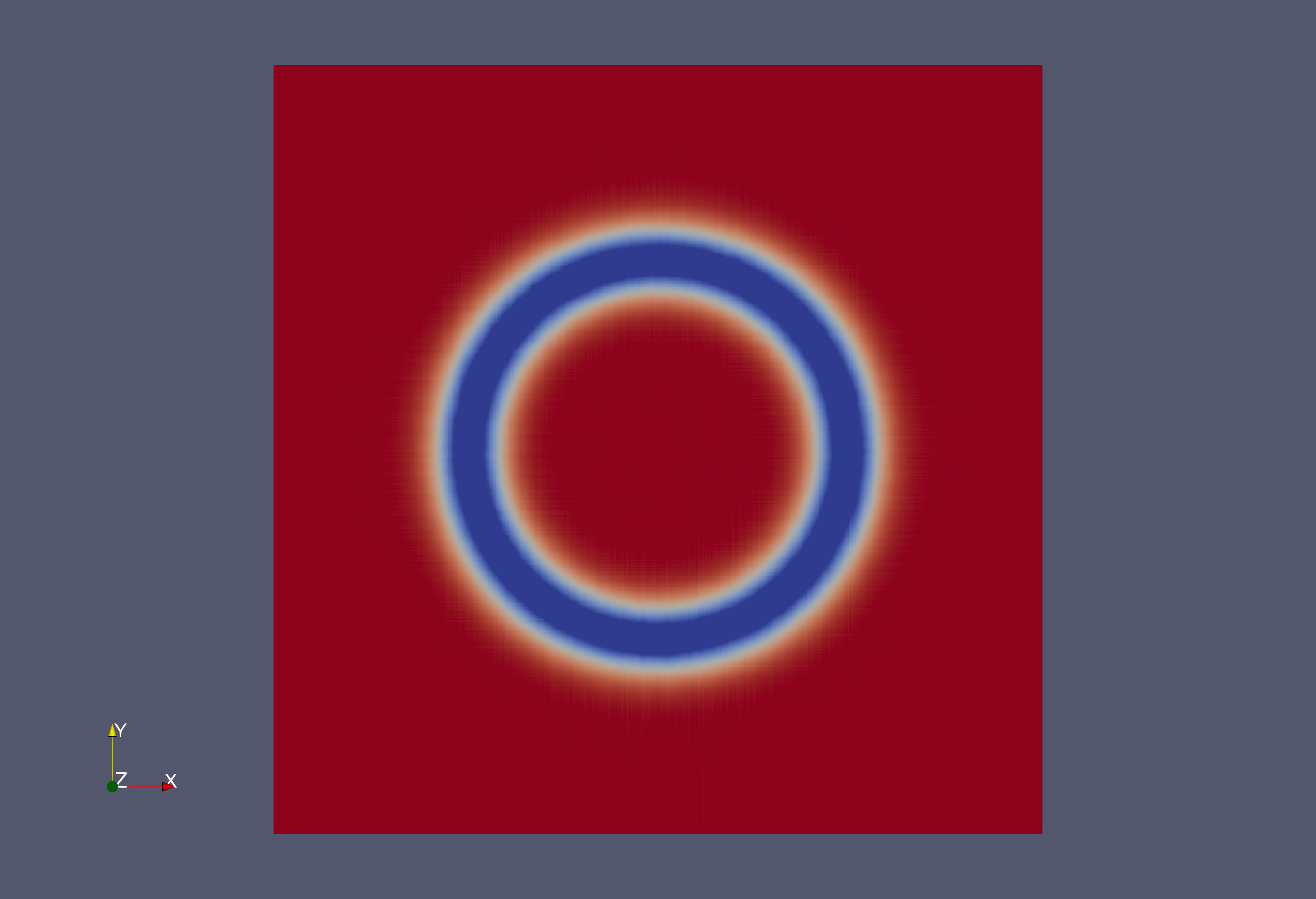}
}
\subfigure[time=10]{
\includegraphics[trim=23cm 6cm 23cm 6cm,clip=true,width=0.3\columnwidth]{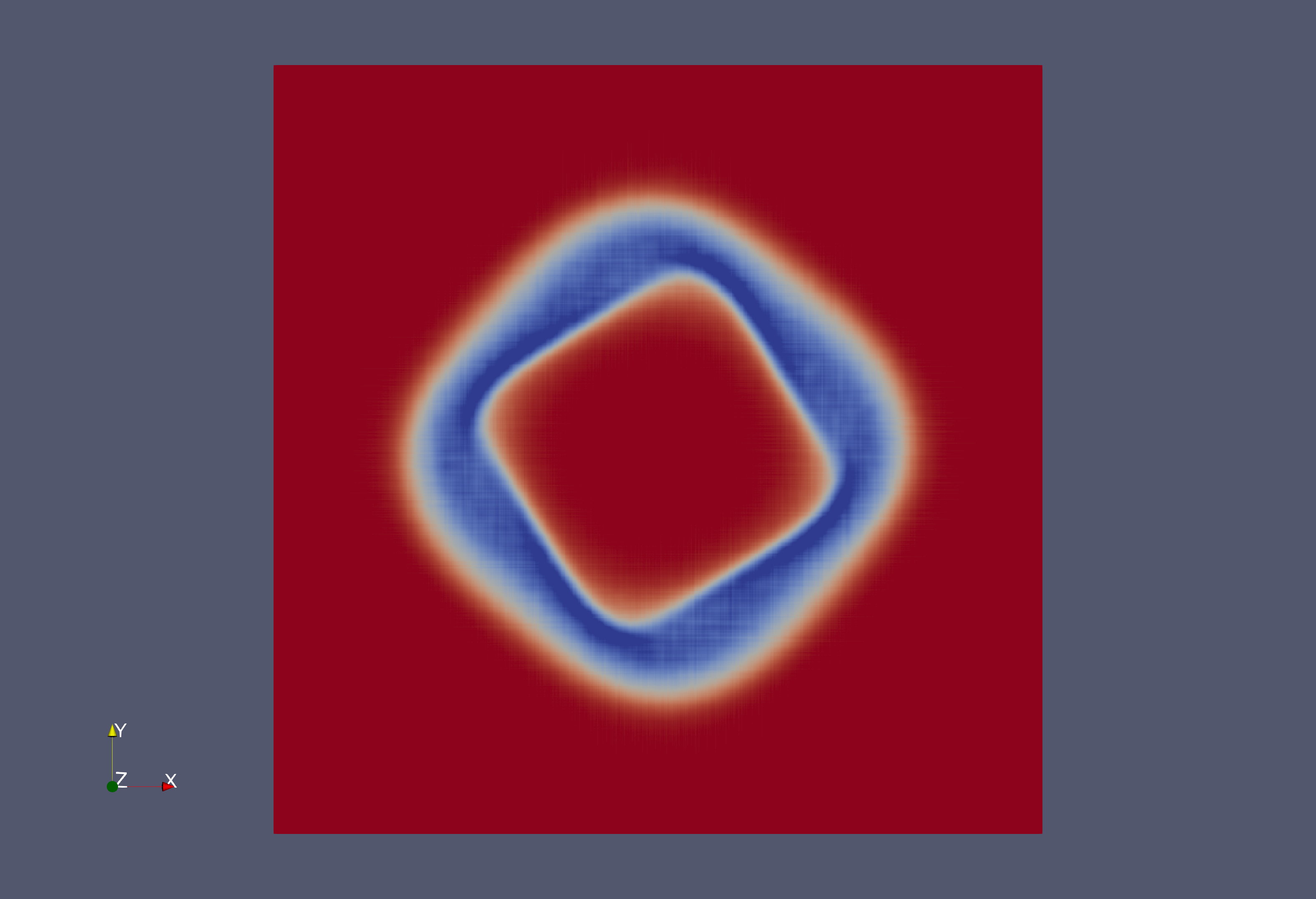}
}
\subfigure[time=17.5]{
\includegraphics[trim=23cm 6cm 23cm 6cm,clip=true,width=0.3\columnwidth]{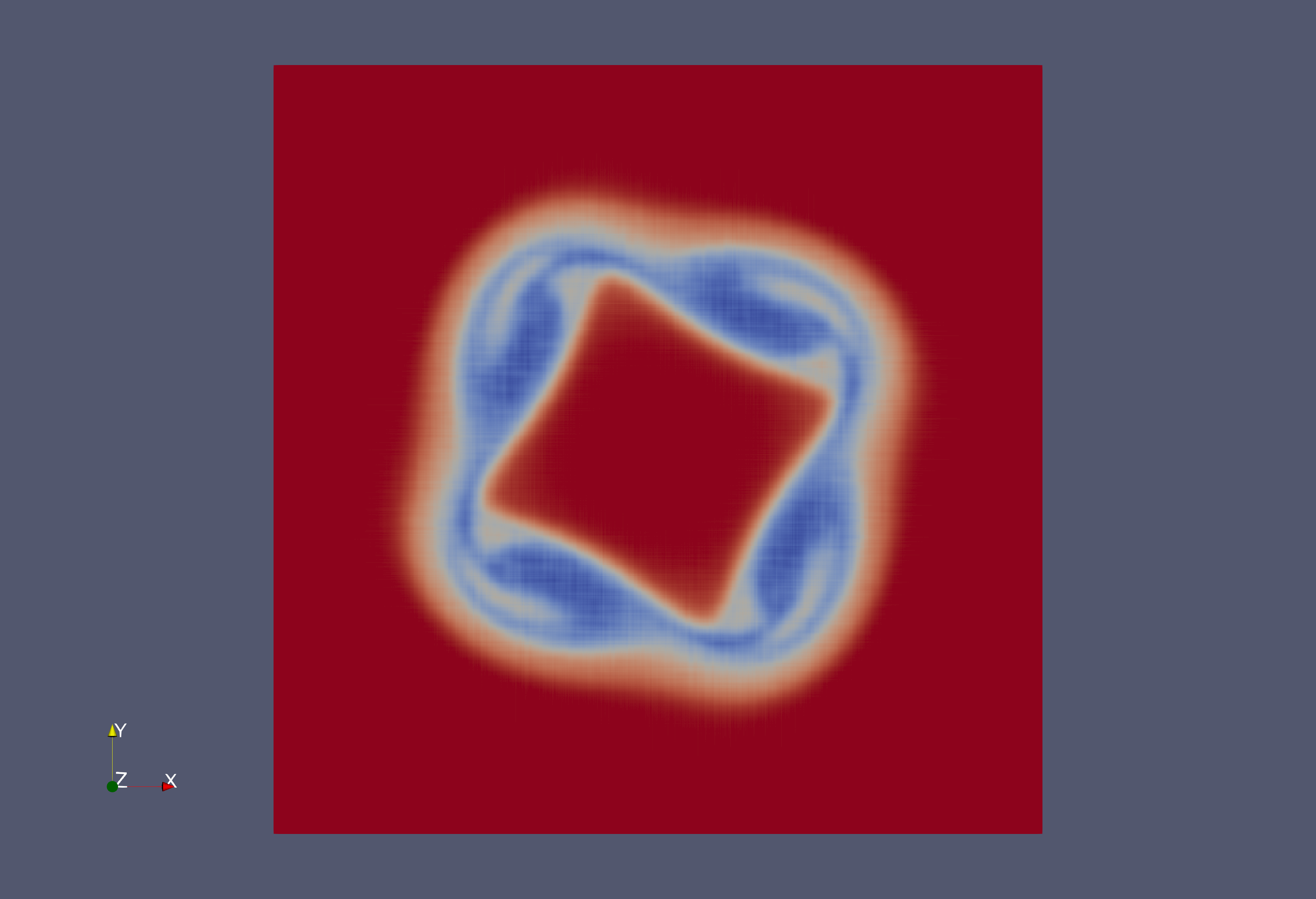}
}
\subfigure[time=0]{
\includegraphics[trim=23cm 6cm 23cm 6cm,clip=true,width=0.3\columnwidth]{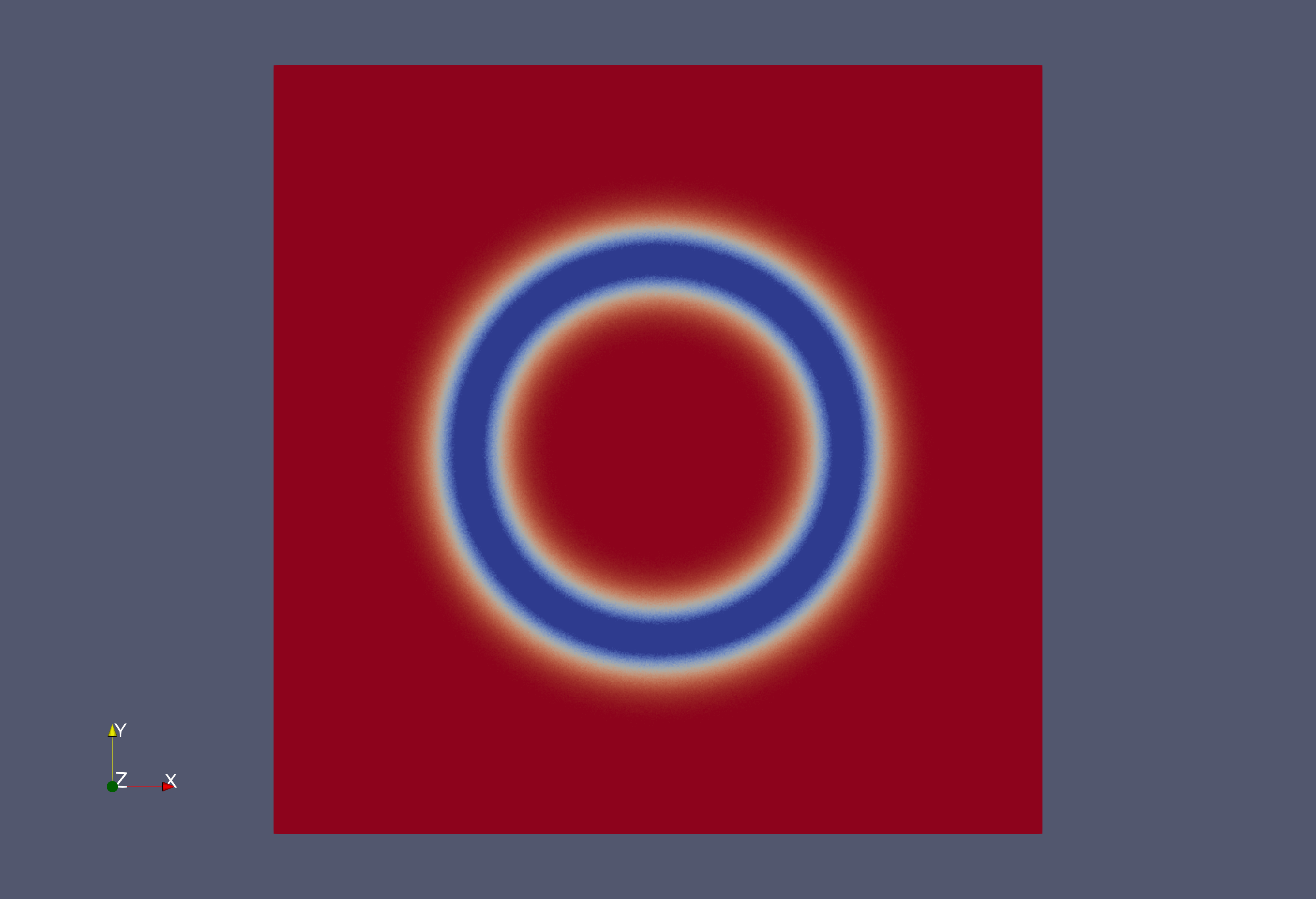}
}
\subfigure[time=10]{
\includegraphics[trim=23cm 6cm 23cm 6cm,clip=true,width=0.3\columnwidth]{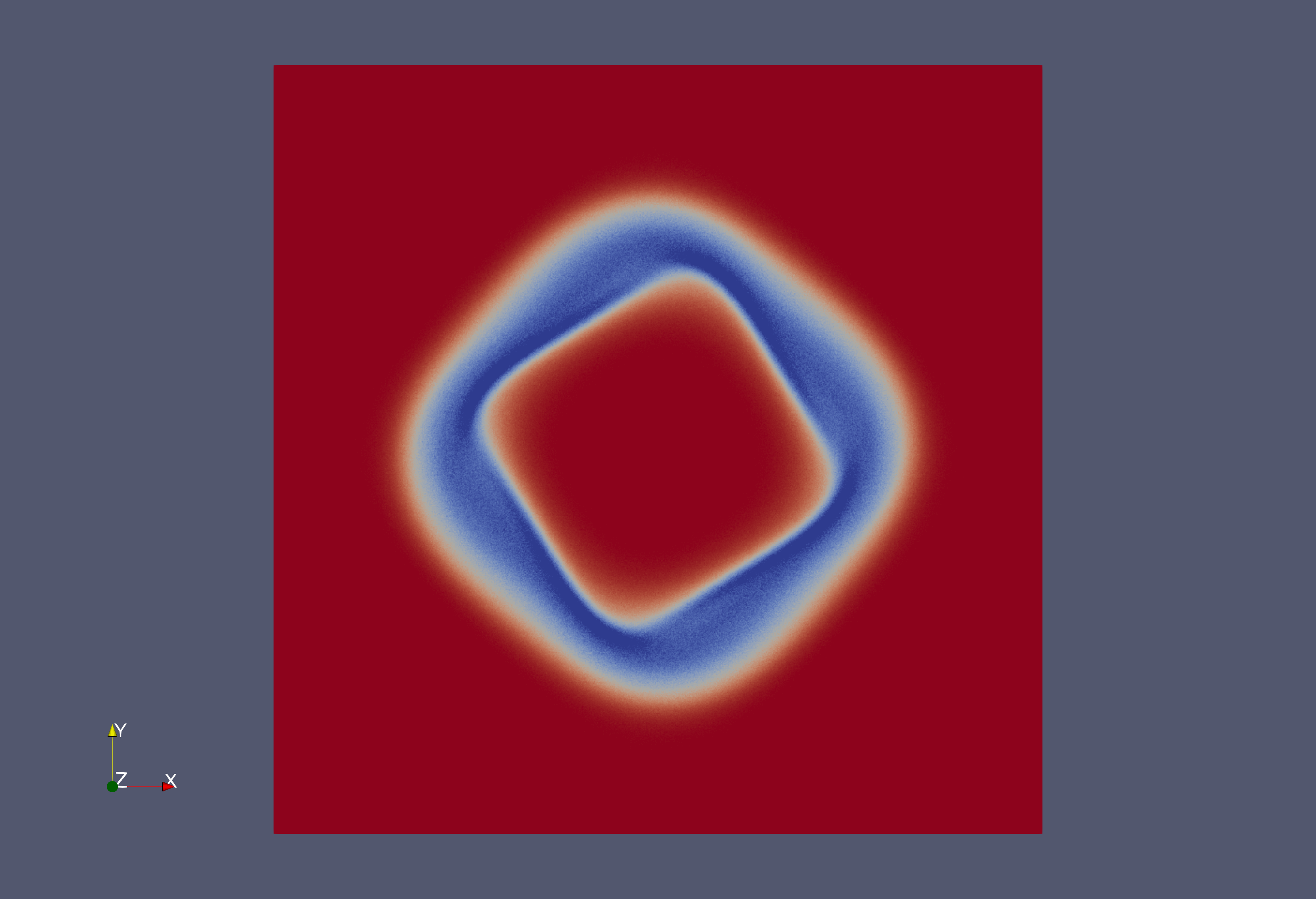}
}
\subfigure[time=17.5]{    
\includegraphics[trim=23cm 6cm 23cm 6cm,clip=true,width=0.3\columnwidth]{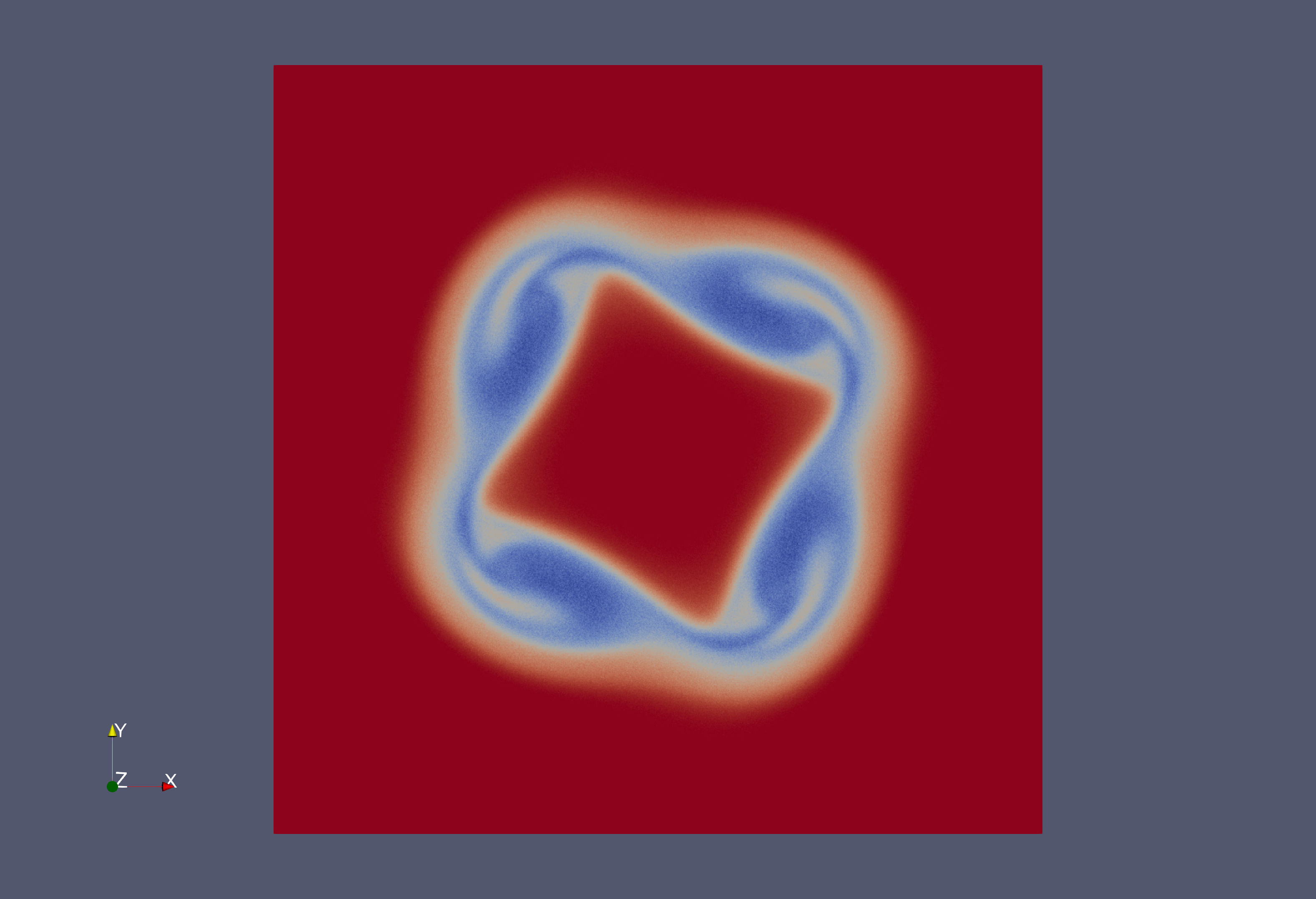}
}
    \caption{2D diocotron instability. Gaussian sampling: Evolution of the electron charge density with time for regular PIC, $P_c=5$ (first row); $\tau=1$, $P_c=5$ (second row); adaptive $\tau$, $P_c=5$ (third row); and regular PIC, $P_c=80$ (fourth row). The mesh considered here is $1024^2$.}
\figlab{diocotron_density}
\end{figure}

    Figure \figref{diocotron_density} shows the evolution of the electron charge density with time for regular, $\tau=1$ and adaptive $\tau$ PIC for a $1024^2$ mesh. For the first three rows $P_c=5$ and for the last row $P_c=80$. From the first and 
    second rows we can see that while the regular PIC results are dominated by noise, $\tau=1$ results are dominated by grid error due to the smearing
    of fine scale structures. This is also noted in \cite{ricketson2016sparse} in their sparse PIC studies. In contrast, the adaptive $\tau$ results in the third row strikes a 
    balance between the grid-based error and noise and are in close agreement (in visual norm) with the regular PIC results with high $P_c$ in the
    fourth row.

\begin{figure}[h!t!b!]
\subfigure[$256^2$, $P_c=5$]{
\includegraphics[width=0.5\columnwidth]{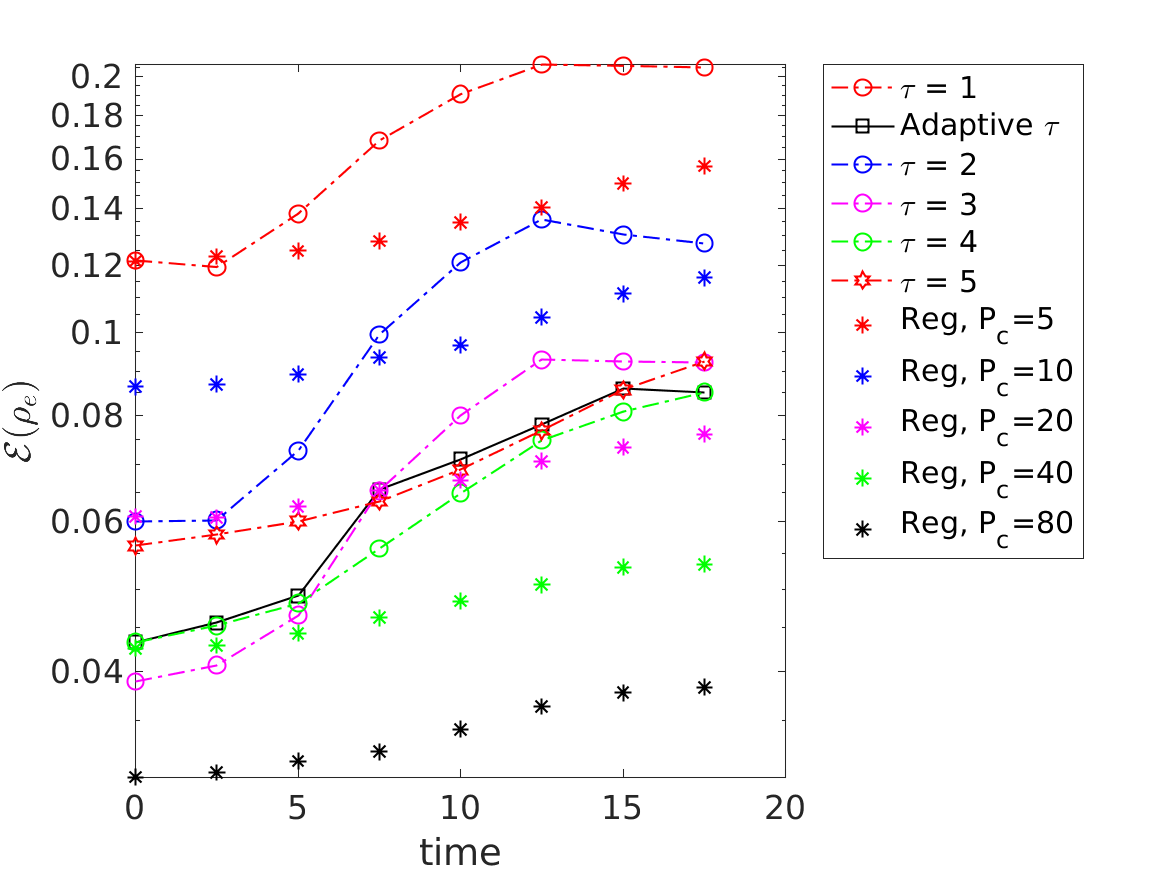}
}
\subfigure[$256^2$, $P_c=5$]{
\includegraphics[width=0.5\columnwidth]{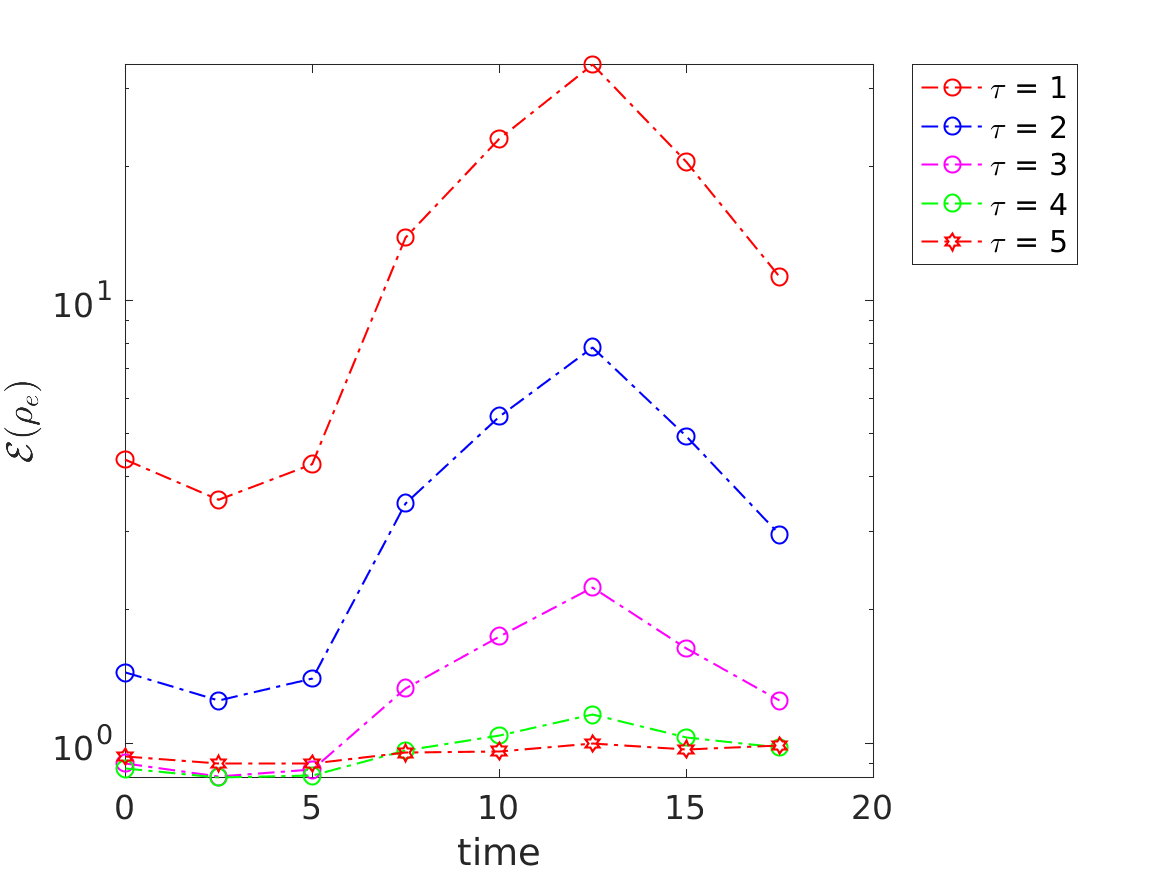}
}
\subfigure[$512^2$, $P_c=5$]{
\includegraphics[width=0.5\columnwidth]{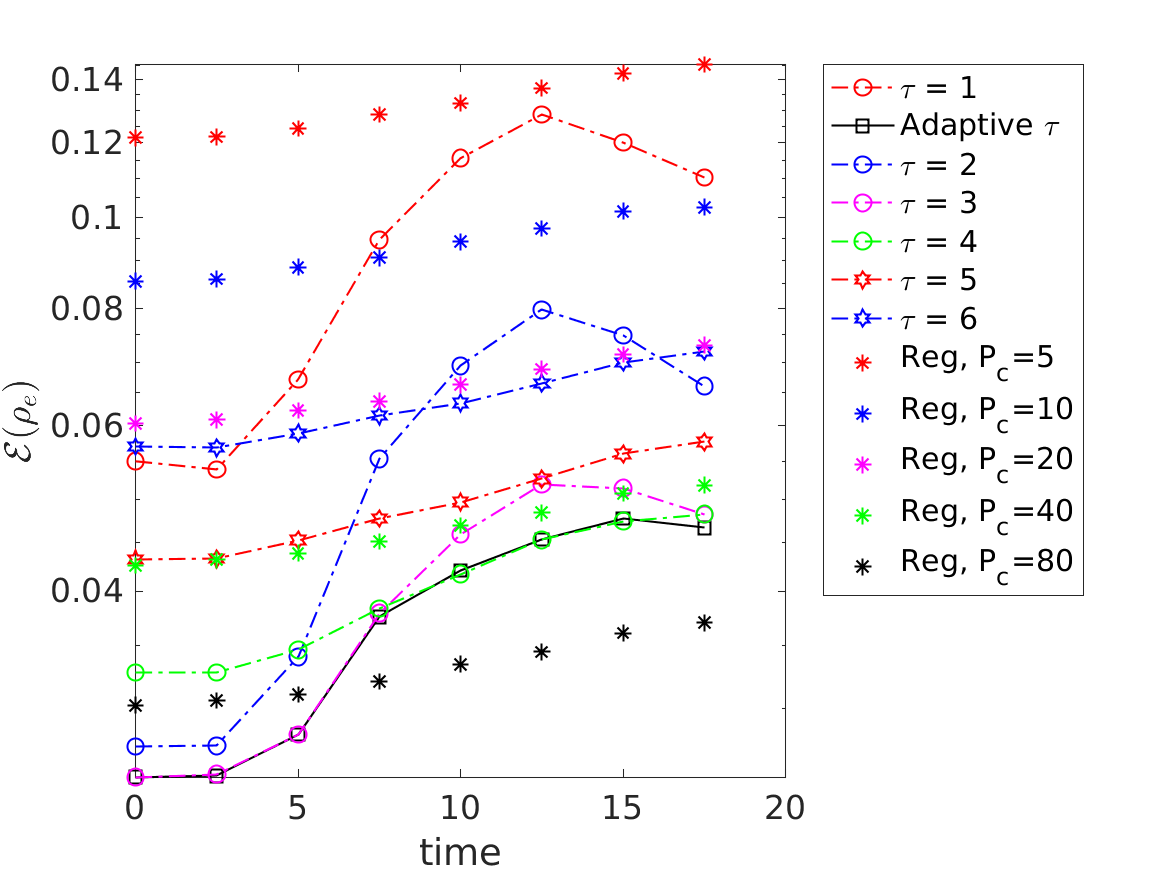}
}
\subfigure[$512^2$, $P_c=5$]{
\includegraphics[width=0.5\columnwidth]{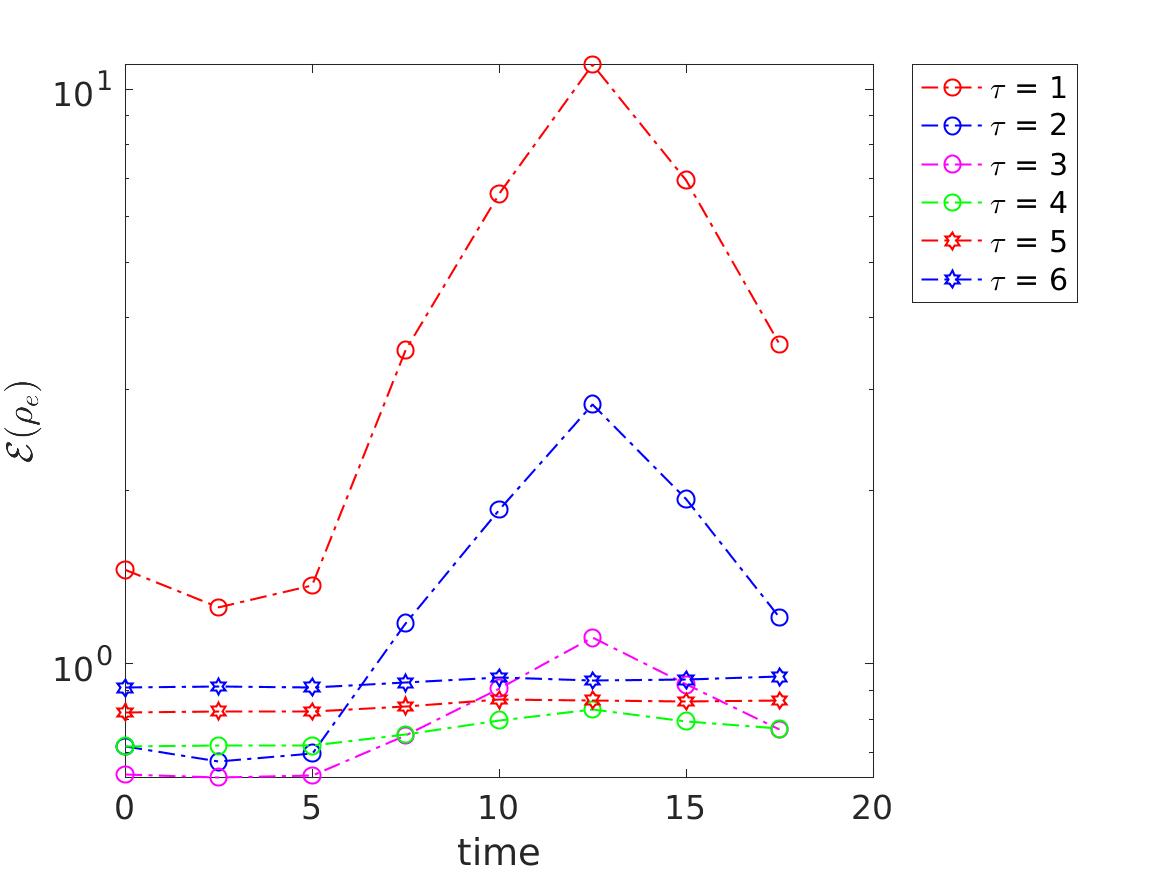}
    \figlab{512_pc5_gauss}
}
\subfigure[$1024^2$, $P_c=5$]{
\includegraphics[width=0.5\columnwidth]{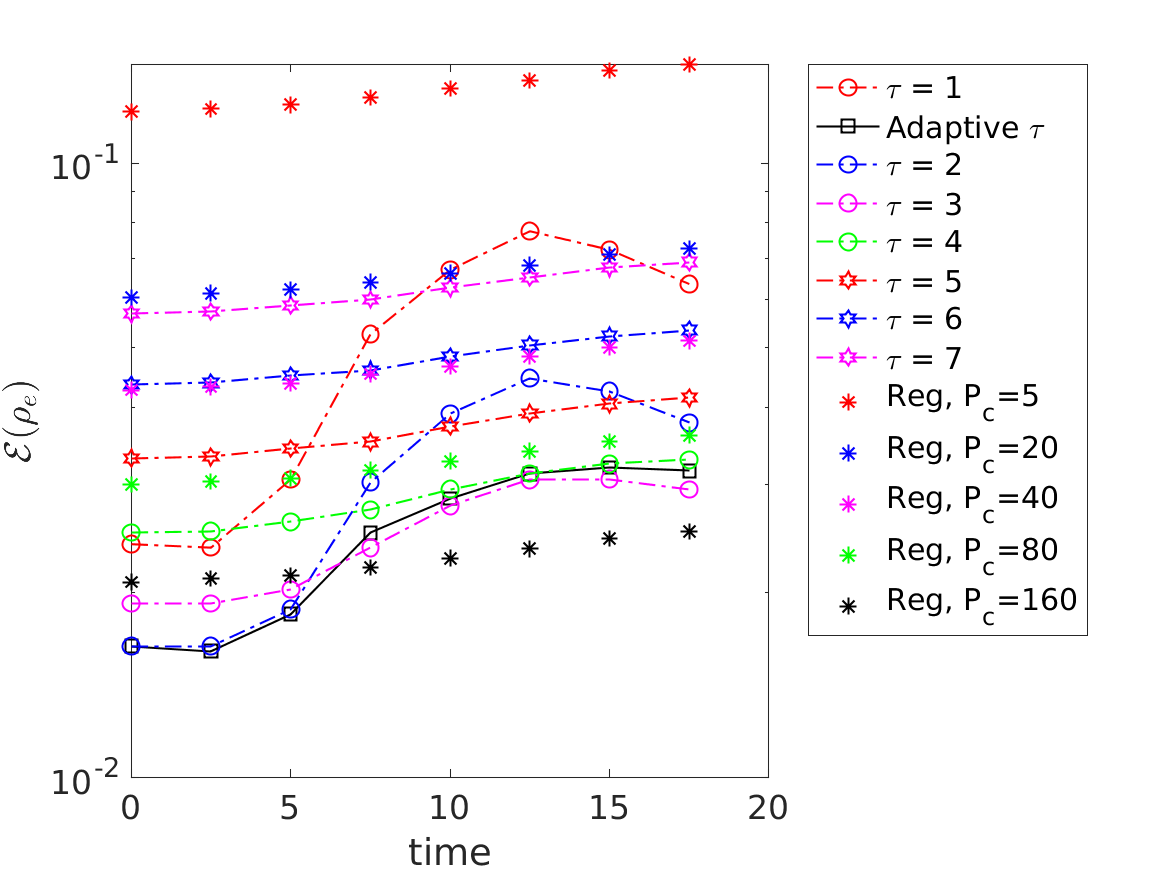}
}
\subfigure[$1024^2$, $P_c=5$]{
\includegraphics[width=0.5\columnwidth]{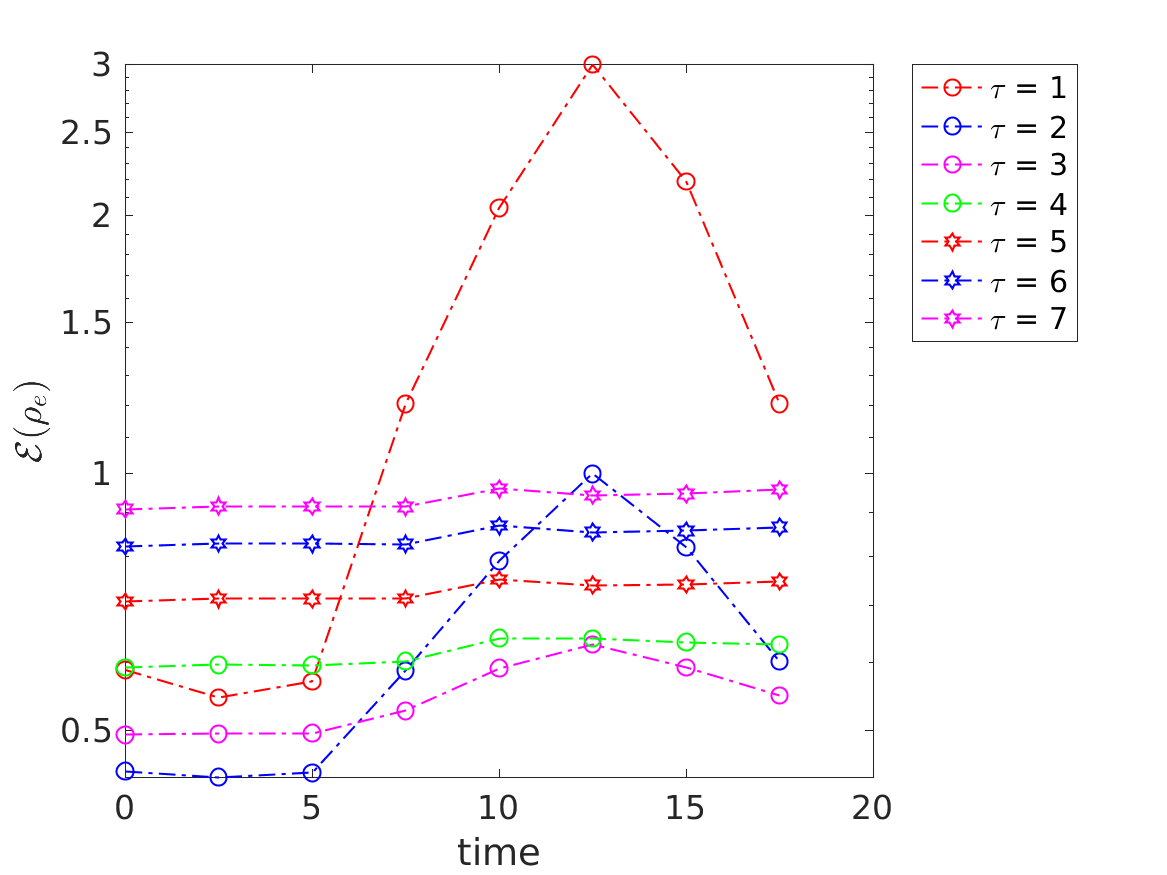}
}
    \caption{2D diocotron instability. Gaussian sampling: Electron charge density error comparison between regular (Reg), fixed $\tau$ and adaptive $\tau$ PIC. The left column is the actual error calculated using equation \eqnref{error_def} and the right column is the estimations from the $\tau$ estimator based on which the optimal $\tau$ is selected. The fixed as well as adaptive $\tau$ has the number of particles per cell $P_c=5$.}
\figlab{diocotron_pc5_gauss}
\end{figure}

\begin{figure}[h!t!b!]
\subfigure[$256^2$, $P_c=10$]{
\includegraphics[width=0.5\columnwidth]{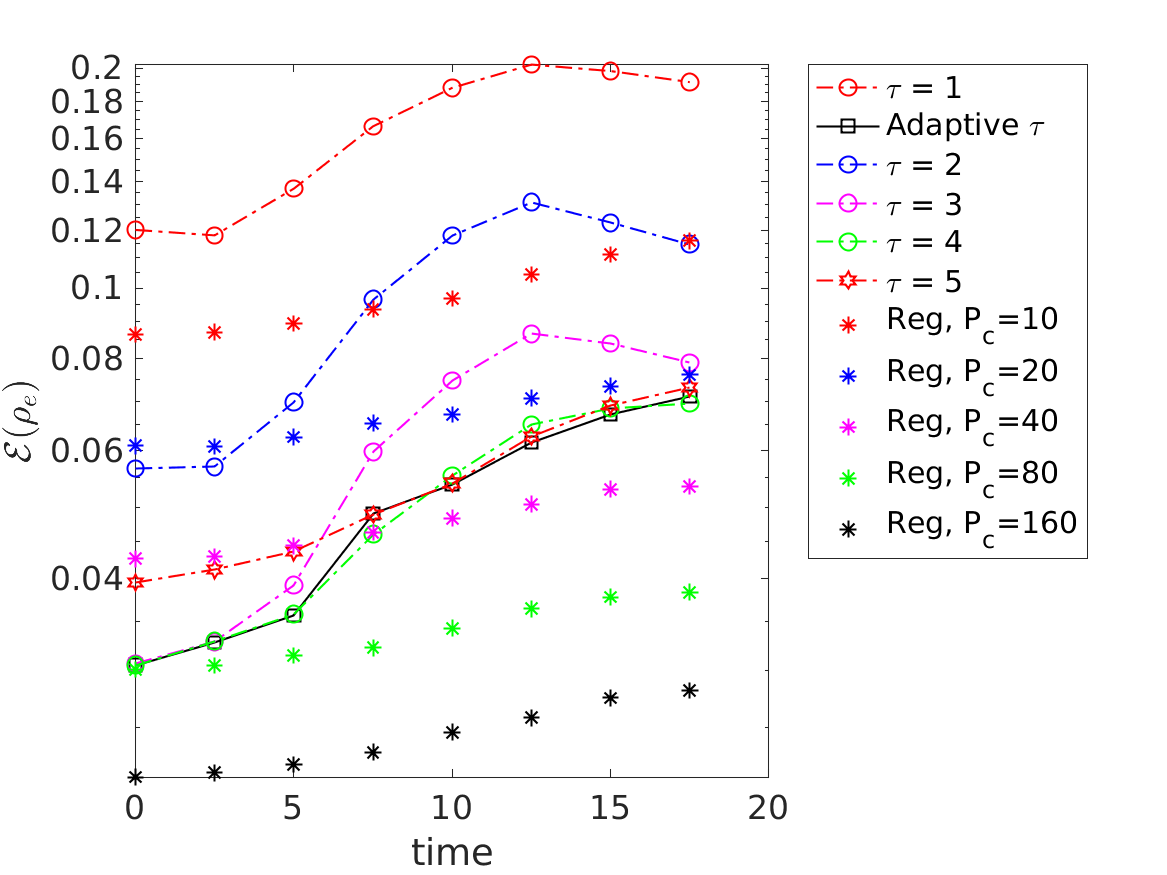}
}
\subfigure[$256^2$, $P_c=10$]{
\includegraphics[width=0.5\columnwidth]{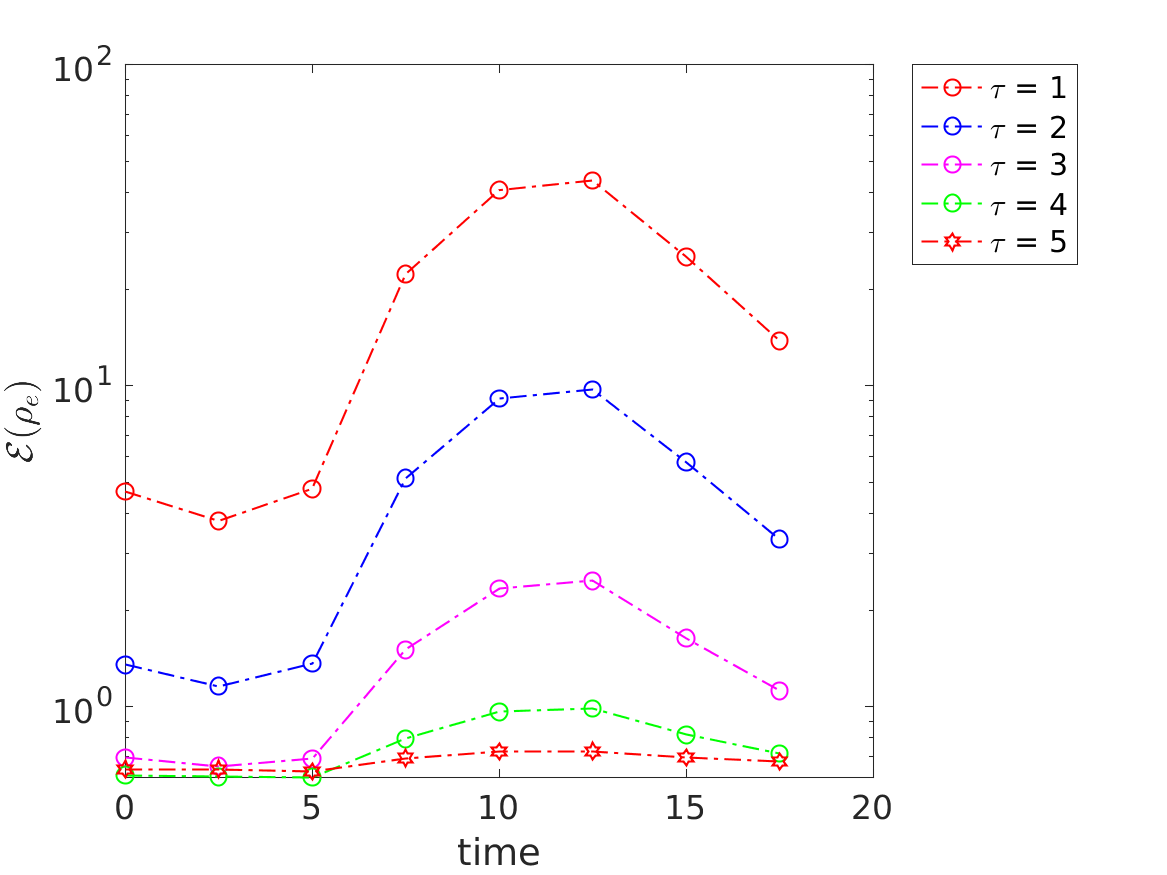}
}
\subfigure[$512^2$, $P_c=10$]{
\includegraphics[width=0.5\columnwidth]{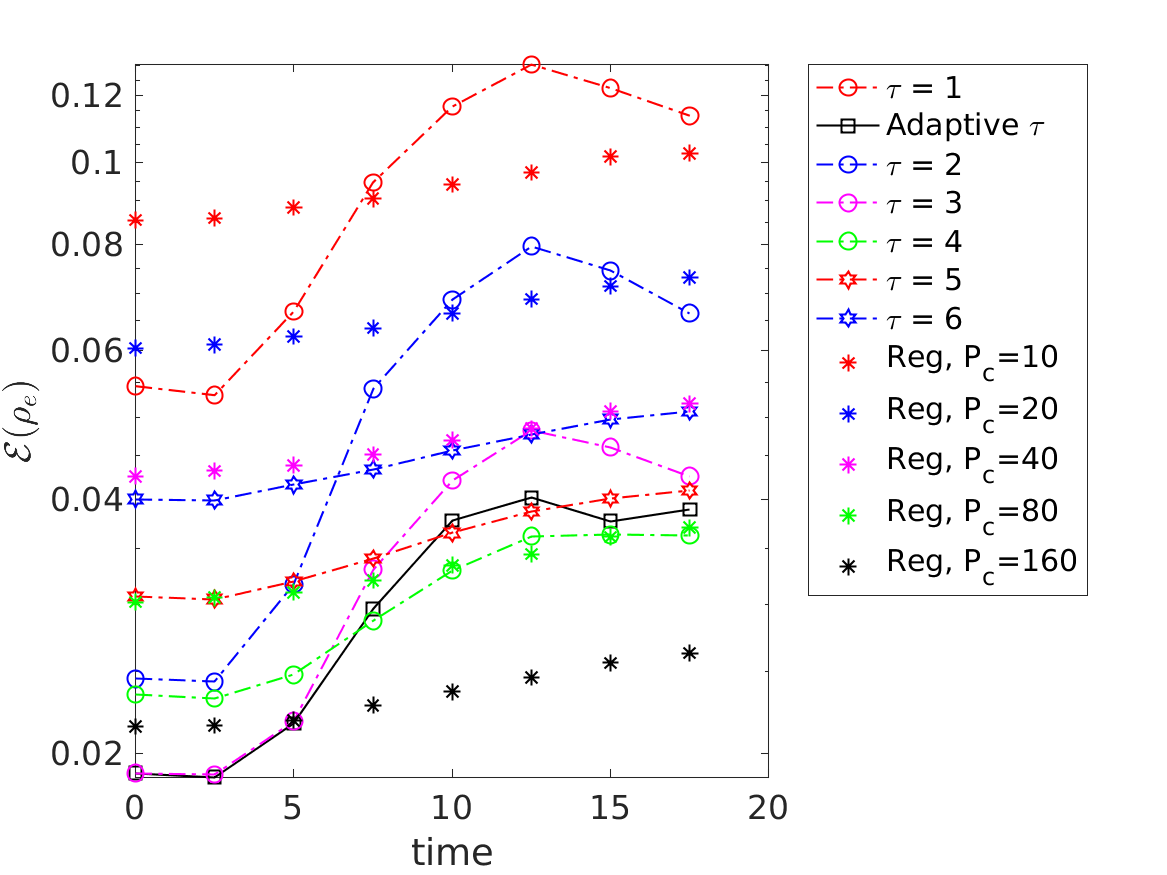}
}
\subfigure[$512^2$, $P_c=10$]{
\includegraphics[width=0.5\columnwidth]{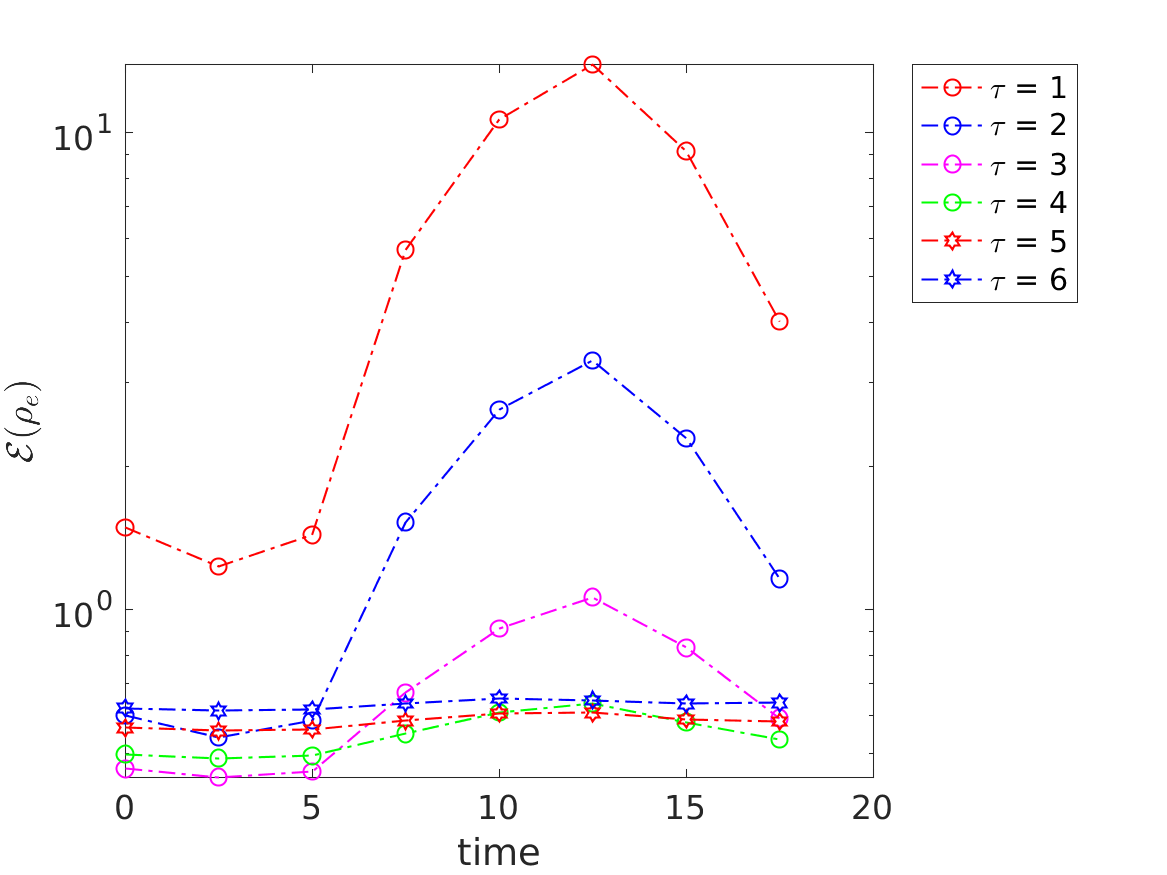}
    \figlab{512_pc10_gauss}
}
\subfigure[$1024^2$, $P_c=10$]{
\includegraphics[width=0.5\columnwidth]{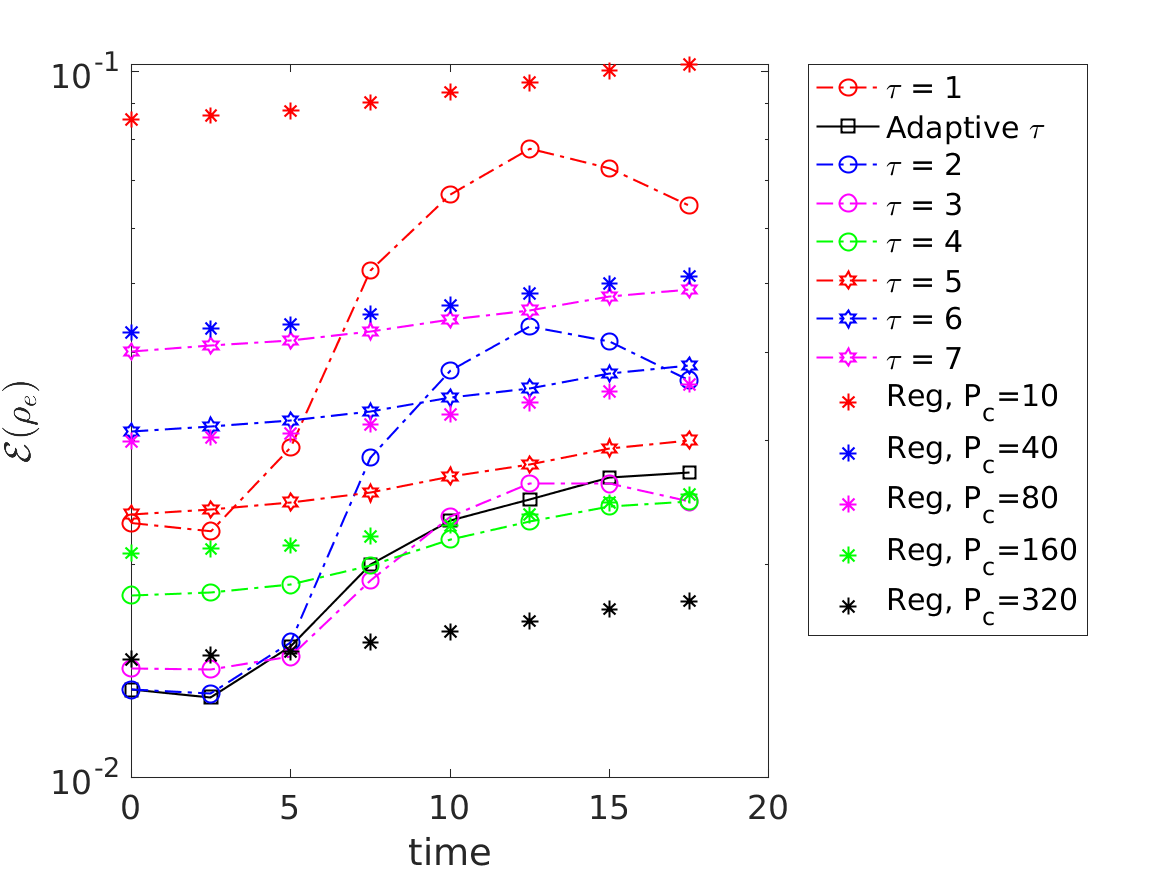}
}
\subfigure[$1024^2$, $P_c=10$]{
\includegraphics[width=0.5\columnwidth]{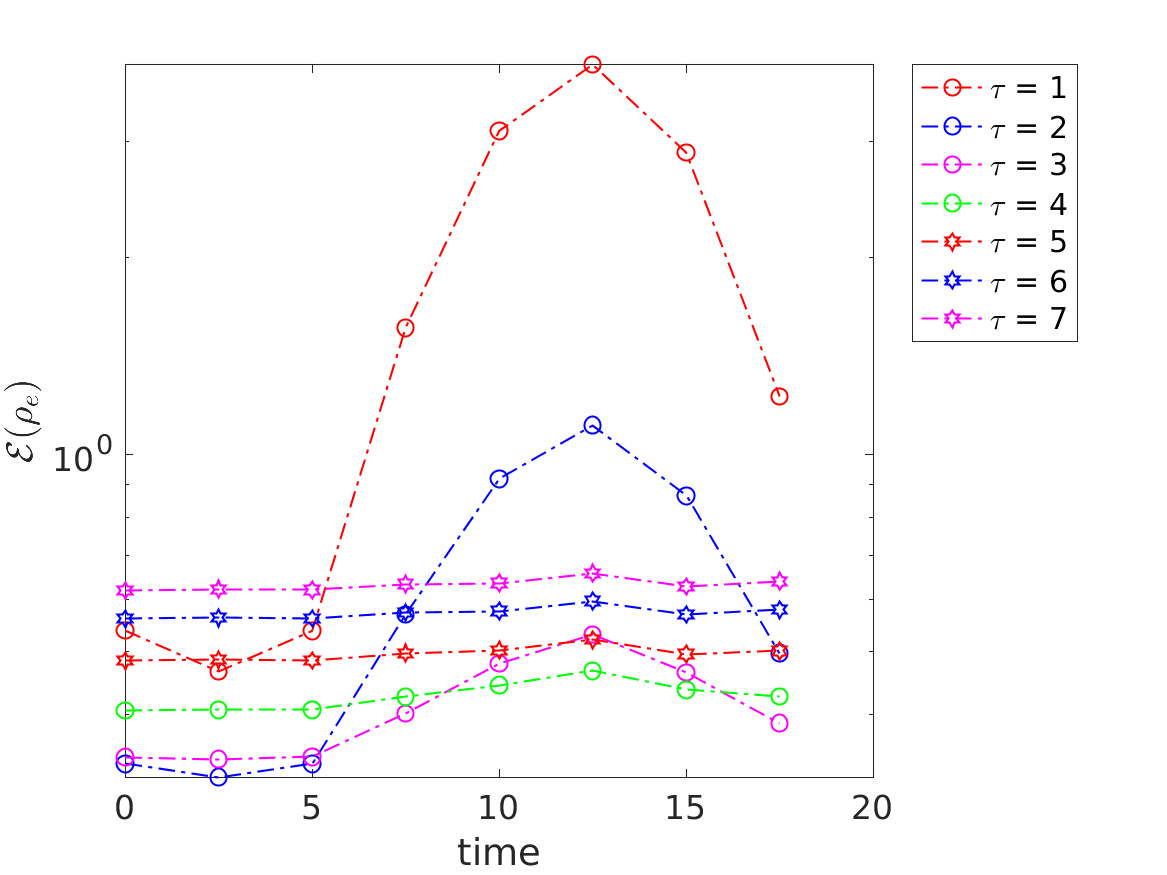}
}
    \caption{2D diocotron instability. Gaussian sampling: Electron charge density error comparison between regular (Reg), fixed $\tau$ and adaptive $\tau$ PIC. The left column is the actual error calculated using equation \eqnref{error_def} and the right column is the estimations from the $\tau$ estimator based on which the optimal $\tau$ is selected. The fixed as well as adaptive $\tau$ has the number of particles per cell $P_c=10$.}
\figlab{diocotron_pc10_gauss}
\end{figure}

\begin{figure}[h!t!b!]
\subfigure[$256^2$, $P_c=20$]{
\includegraphics[width=0.5\columnwidth]{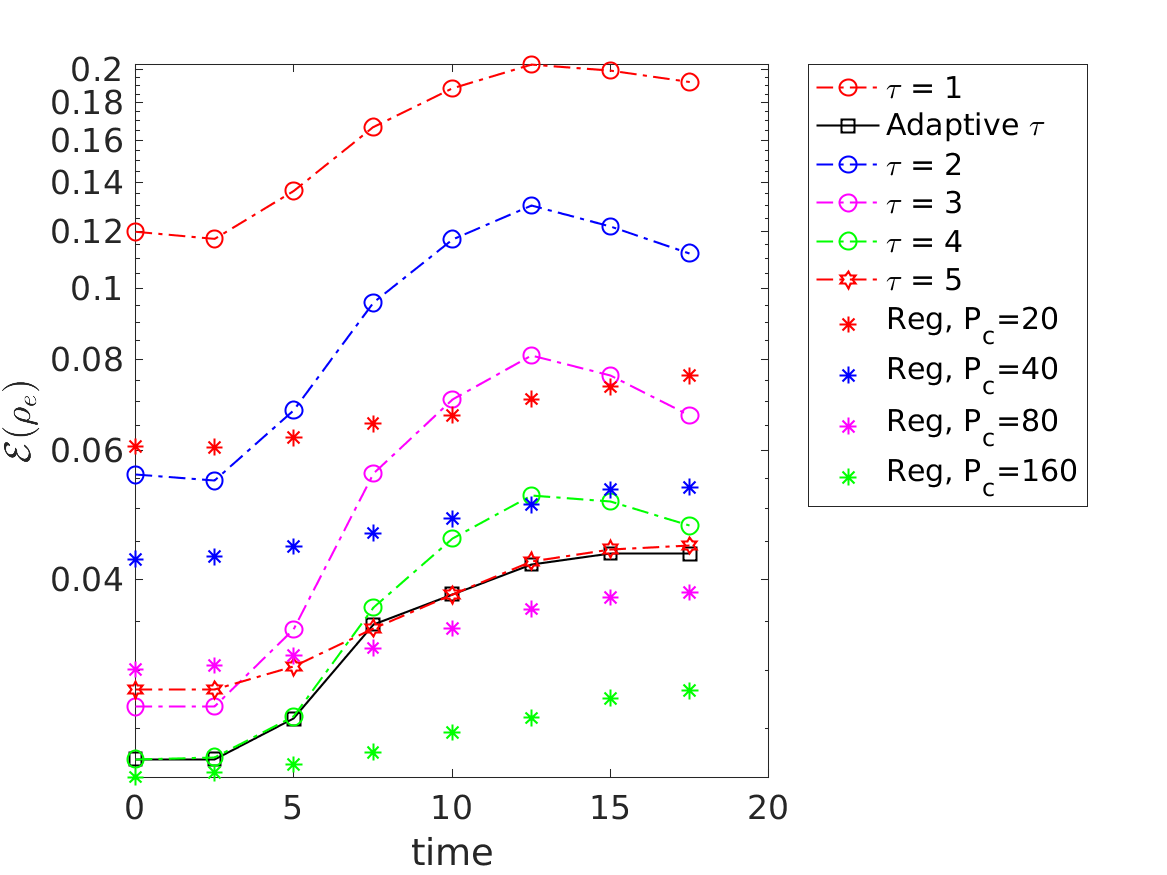}
}
\subfigure[$256^2$, $P_c=20$]{
\includegraphics[width=0.5\columnwidth]{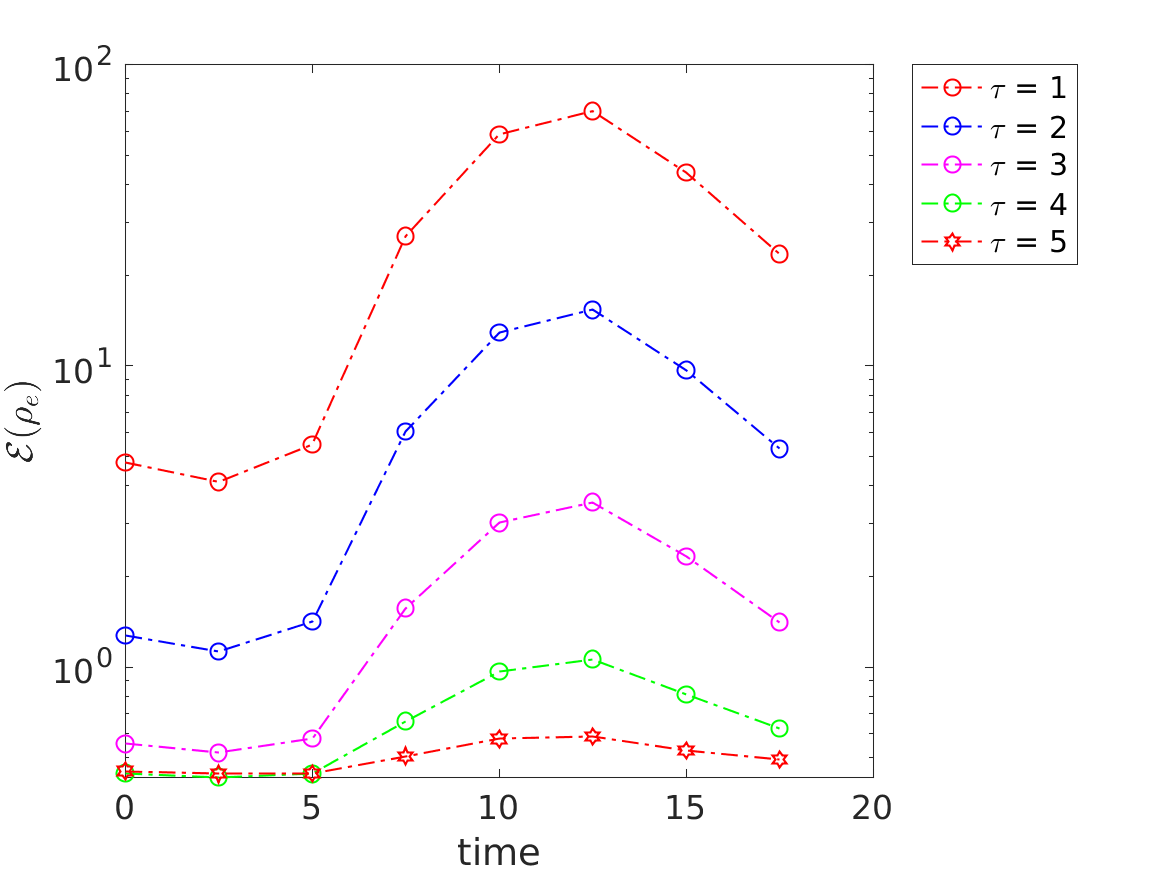}
}
\subfigure[$512^2$, $P_c=20$]{
\includegraphics[width=0.5\columnwidth]{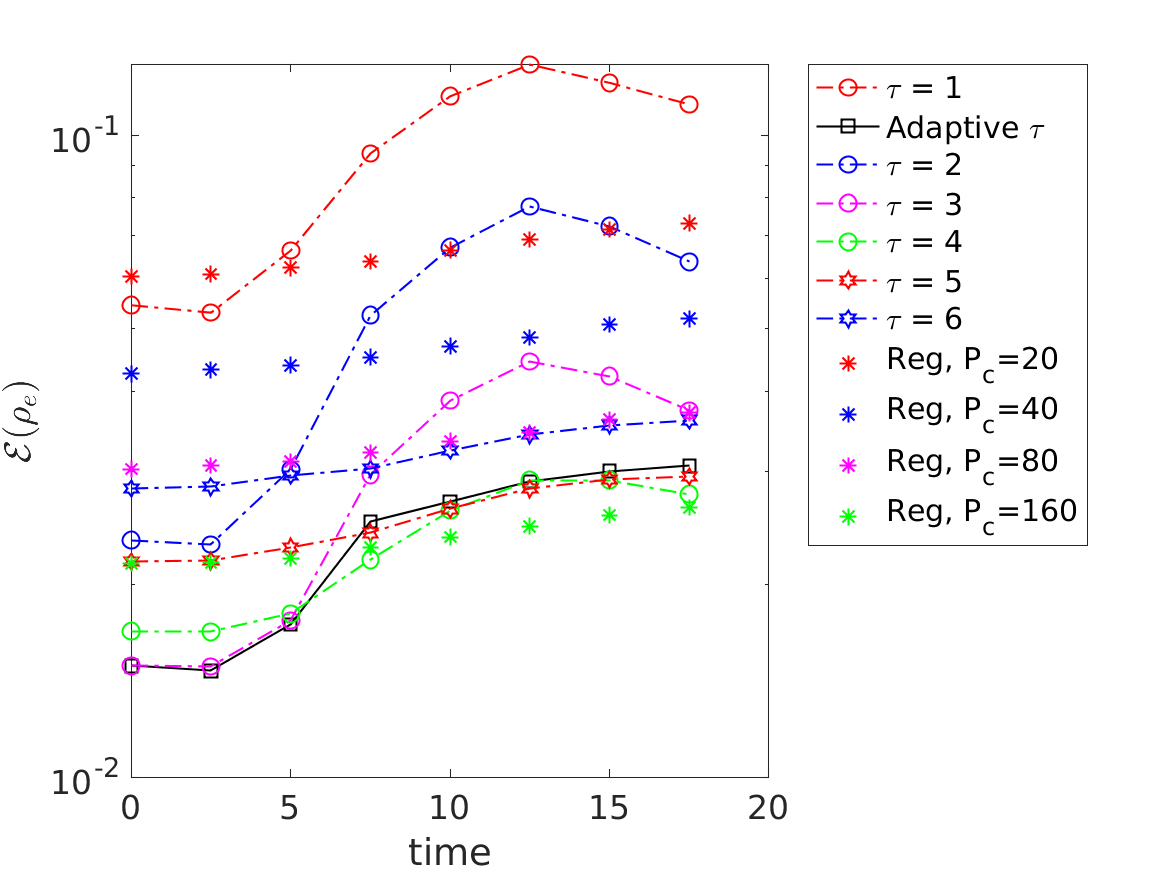}
}
\subfigure[$512^2$, $P_c=20$]{
\includegraphics[width=0.5\columnwidth]{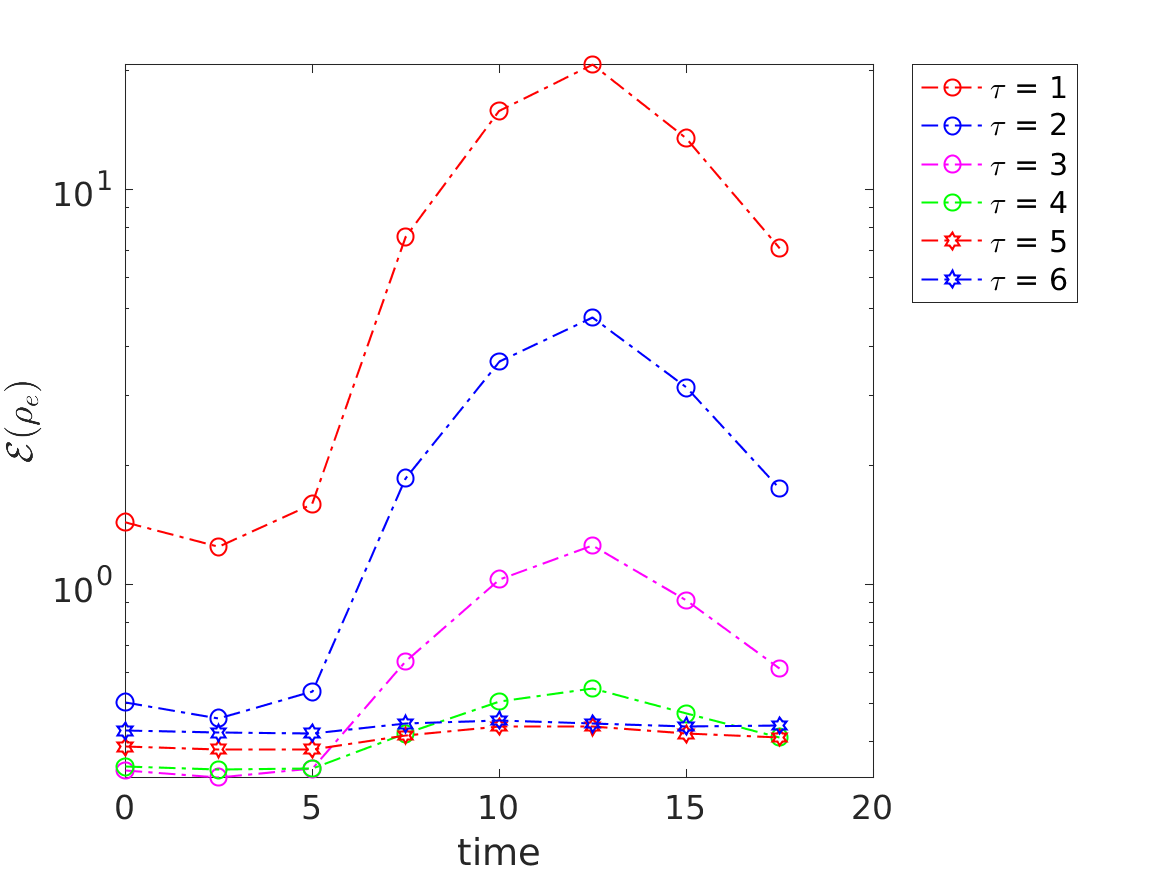}
}
\subfigure[$1024^2$, $P_c=20$]{
\includegraphics[width=0.5\columnwidth]{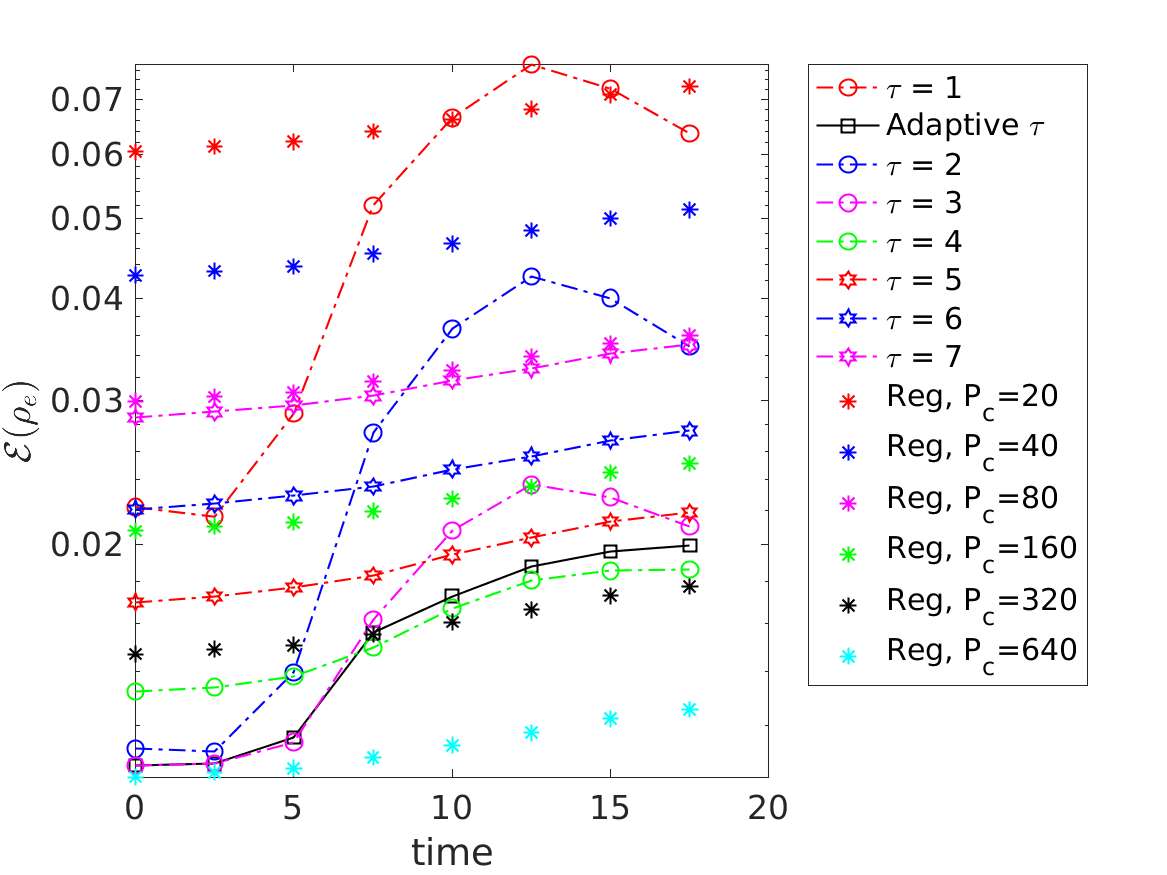}
    \figlab{1024_pc20_gauss}
}
\subfigure[$1024^2$, $P_c=20$]{
\includegraphics[width=0.5\columnwidth]{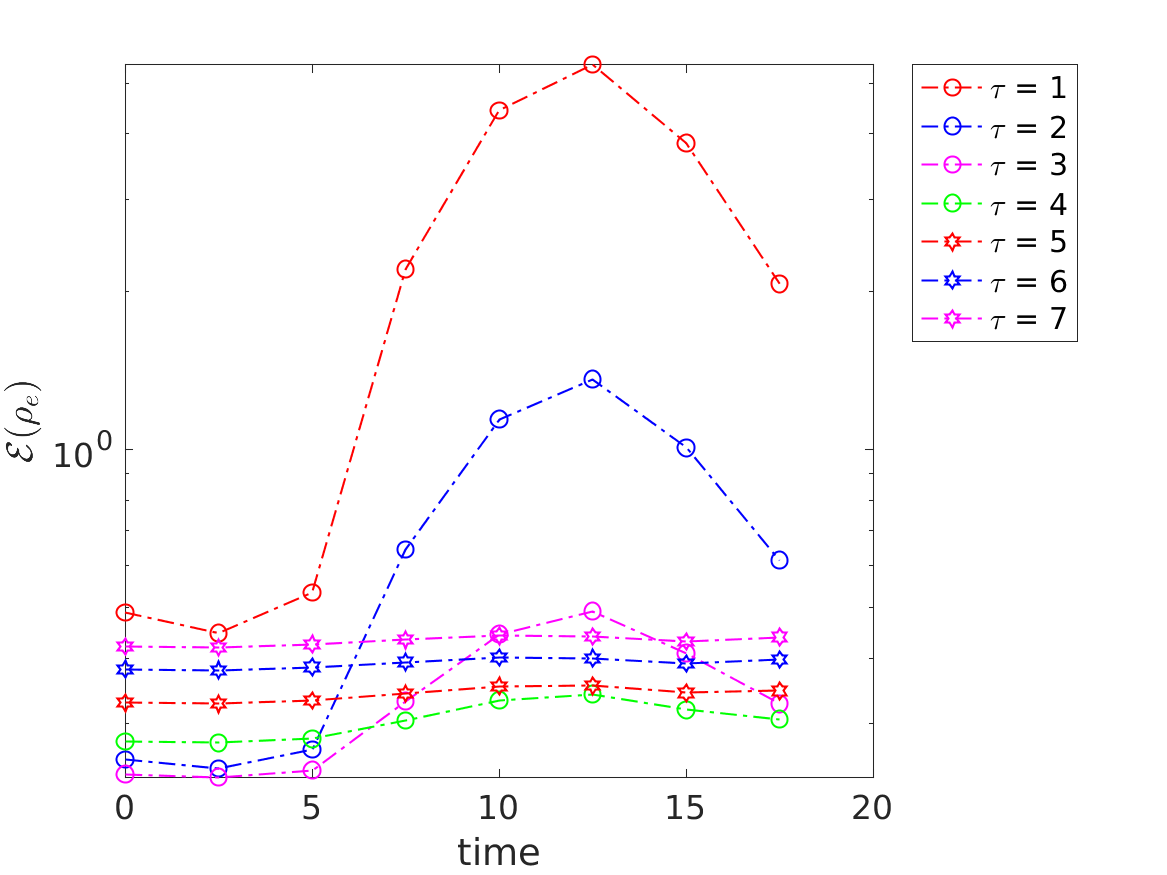}
}
    \caption{2D diocotron instability. Gaussian sampling: Electron charge density error comparison between regular (Reg), fixed $\tau$ and adaptive $\tau$ PIC. The left column is the actual error calculated using equation \eqnref{error_def} and the right column is the estimations from the $\tau$ estimator based on which the optimal $\tau$ is selected. The fixed as well as adaptive $\tau$ has the number of particles per cell $P_c=20$. The errors for regular PIC with $P_c=320$ and $640$ are
    calculated from that of $P_c=160$ based on the theoretical particle error scaling $1/\sqrt{Pc}$. This is based on the observation that the errors for the regular PIC are in the noise dominated regime.}
\figlab{diocotron_pc20_gauss}
\end{figure}
    In order to make a quantitative comparison, in the left columns of Figures \figref{diocotron_pc5_gauss}-\figref{diocotron_pc20_gauss}, the error in $\rho_e$ calculated using \eqnref{error_def} at 8 time instants is shown for three different meshes
    $256^2,512^2,1024^2$ and number of particles per cell $P_c=5,10,20$. For regular PIC we also carried out simulations at higher $P_c$, namely 
    $40,80,160$ in order to compare the accuracy level with adaptive $\tau$ results. The reference in equation \eqnref{error_def} is computed using the average
    of 8 independent regular PIC simulations each with a $1024^2$ mesh and $P_c=320$. In equation \eqnref{error_def}, the $N_{points}$ are taken as the cell-centered points
    in the mesh under consideration and the reference $\rho_e$ is interpolated to these points for calculating error. In Figure \figref{1024_pc20_gauss}, for calculating the error with regular PIC at
    $P_c=320,640$ we divided the error for $P_c=160$ by $\sqrt{2}$ and $\sqrt{4}$ respectively as we observed the errors are already in the noise dominated
    regime and follow the scaling $1/\sqrt{P_c}$. On the right columns of Figures \figref{diocotron_pc5_gauss}-\figref{diocotron_pc20_gauss} are the estimations
    of the error for different $\tau$ values from the $\tau$ estimator divided by the root mean squared value of the reference $\rho_e$. It is based on these
    curves that the optimal $\tau$, i.e. the one with minimum error, is selected at each time step during the simulation.

    From the left columns of Figures \figref{diocotron_pc5_gauss}-\figref{diocotron_pc20_gauss}, we can see that in general the adaptive $\tau$ performs well
    in terms of picking one of the $\tau$ values with the lowest error (if not the optimal $\tau$ at all time instants). The shapes of the error curves for 
    individual $\tau$ values are also similar for the estimated and actual ones. It demonstrates the ability of our estimator to predict correct error
    dynamics for different $\tau$ cases. While we do not have to worry about the magnitude of the errors in the estimator, the ordering of the error curves 
    between different $\tau$ values is of importance as it decides the optimal $\tau$, and we want it to be close to the actual scenario on the left 
    columns. To that extent, we make an observation that in the time interval $t\in[7.5,17.5]$ the difference in the magnitude of errors between different $\tau$ values in the estimator differs more from the actual scenario than in the time interval $t\in[0,7.5)$. More specifically, for low $\tau$ values ($\tau=1,2,3$) the estimator predicts 
    a significantly higher error compared to the other $\tau$ values in that regime. 

    One of the reasons for this behavior is for low $\tau$ cases e.g., $\tau=1,2$ and $3$, the number of component grids in the combination technique is higher than that for the high $\tau$ cases. Since we use the triangle inequality to bound the errors, both the grid and particle errors
    tend to be more over-estimated for the low $\tau$ cases than those for the high $\tau$ ones. Another reason is, in the estimates for the grid error we use the derivatives based on the regular grid. While this is a sharper upper bound for high $\tau$, the derivatives seen in reality by the low $\tau$
    cases for functions with fine scale structures will be smaller because of the larger mesh sizes. Indeed, fine scale structures form in the time interval  
    $t\in[7.5,17.5]$ and hence grid error dominated for the simulations with sparse grids noise reduction. 

        In spite of these differences, in all the cases even with the predicted sub-optimal $\tau$ the error values of the adaptive $\tau$ PIC is 
    significantly lower than that of the regular PIC with same $P_c$. If we use some problem specific information, then it may be possible to reduce the over-estimations in the grid and particle errors by introducing a correction factor for different $\tau$ values.

\begin{figure}[h!t!b!]
\subfigure[$256^2$, $P_c=5$]{
\includegraphics[width=0.3\columnwidth]{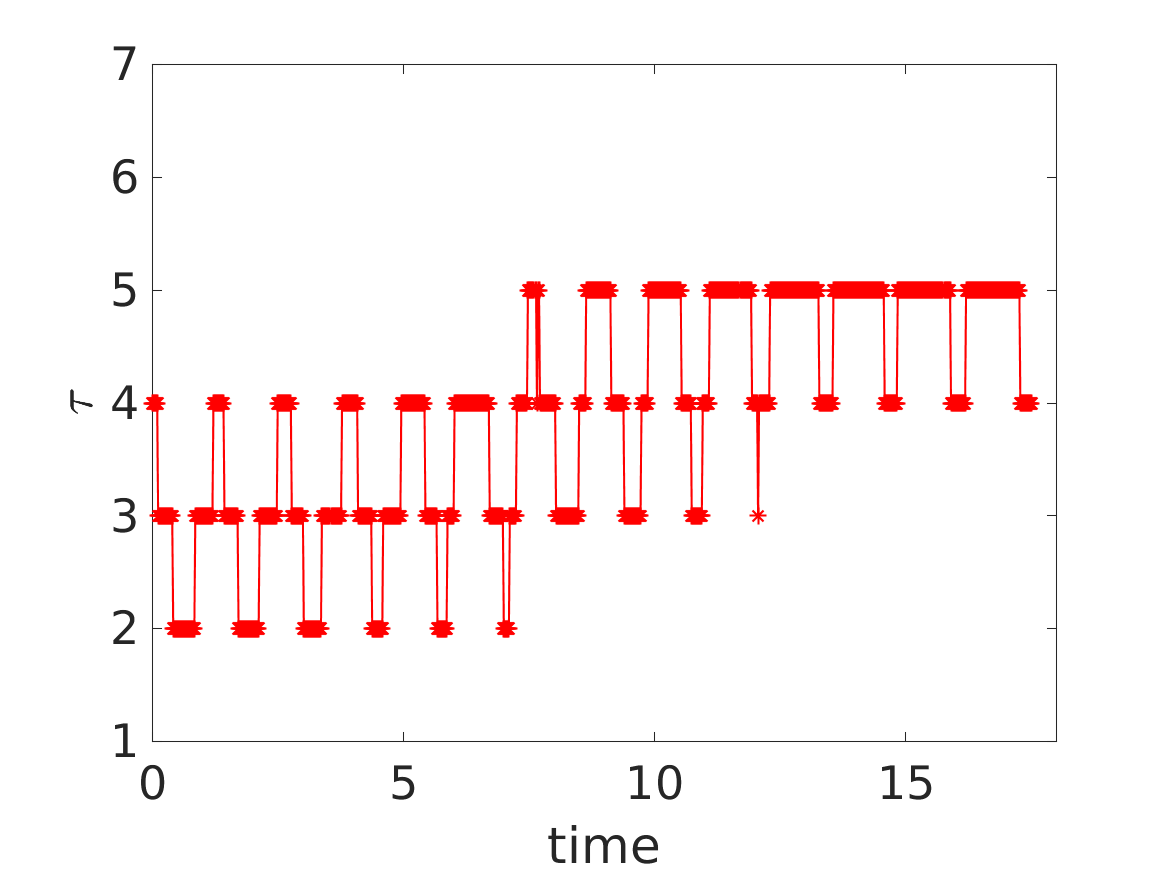}
}
\subfigure[$512^2$, $P_c=5$]{
\includegraphics[width=0.3\columnwidth]{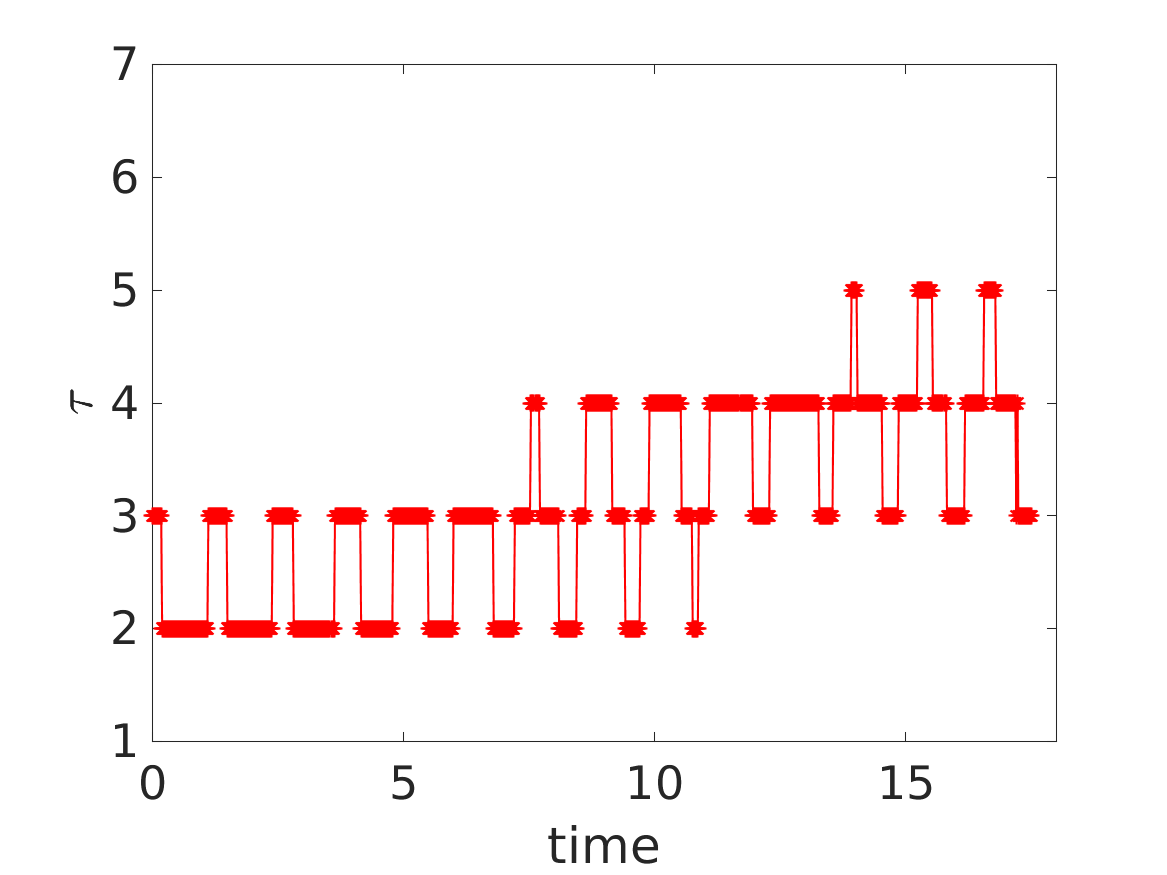}
}
\subfigure[$1024^2$, $P_c=5$]{
\includegraphics[width=0.3\columnwidth]{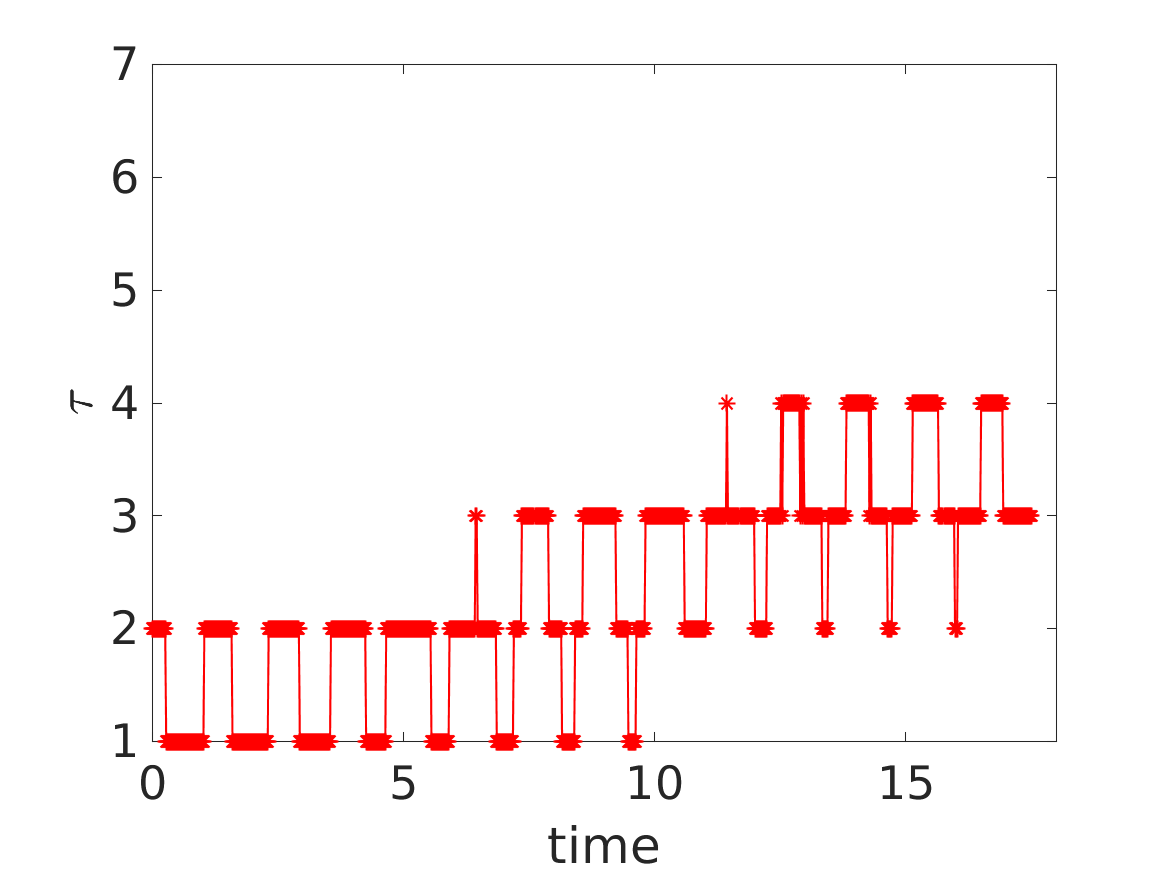}
}
\subfigure[$256^2$, $P_c=10$]{
\includegraphics[width=0.3\columnwidth]{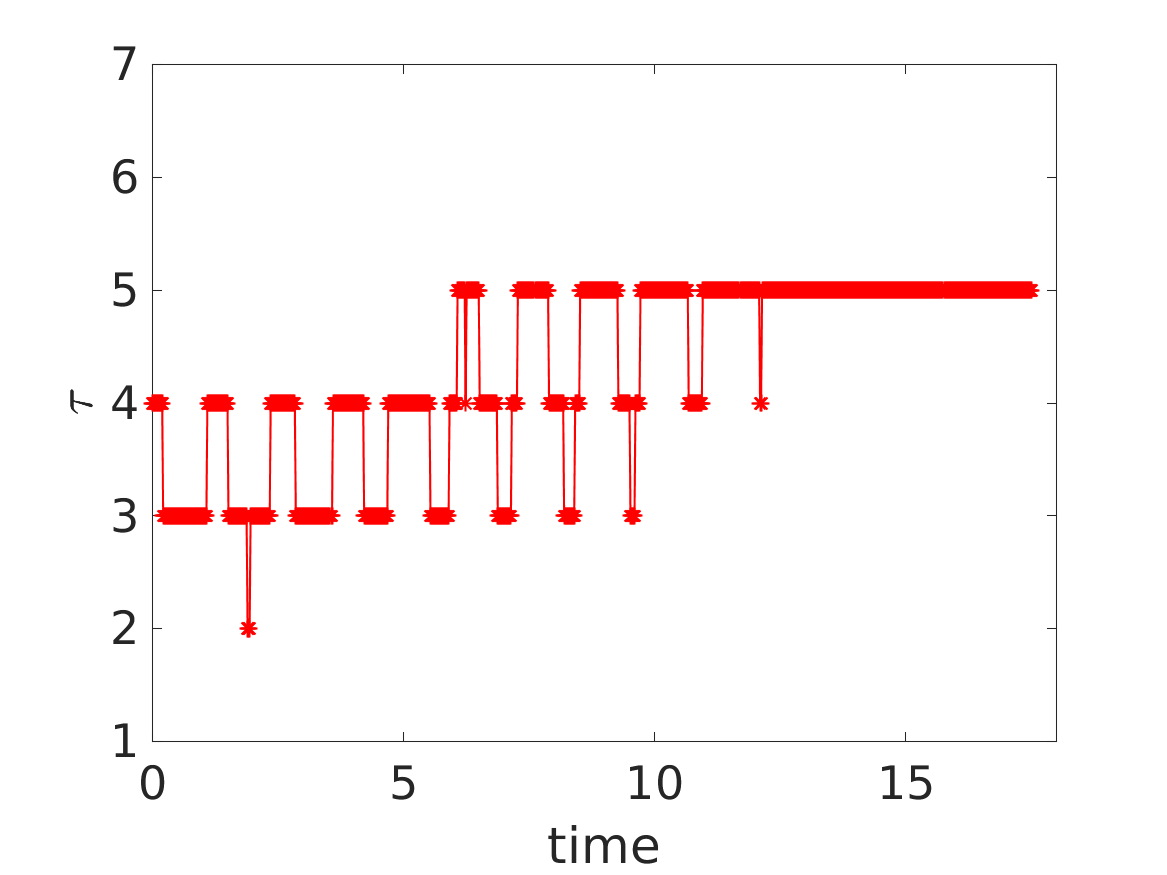}
}
\subfigure[$512^2$, $P_c=10$]{
\includegraphics[width=0.3\columnwidth]{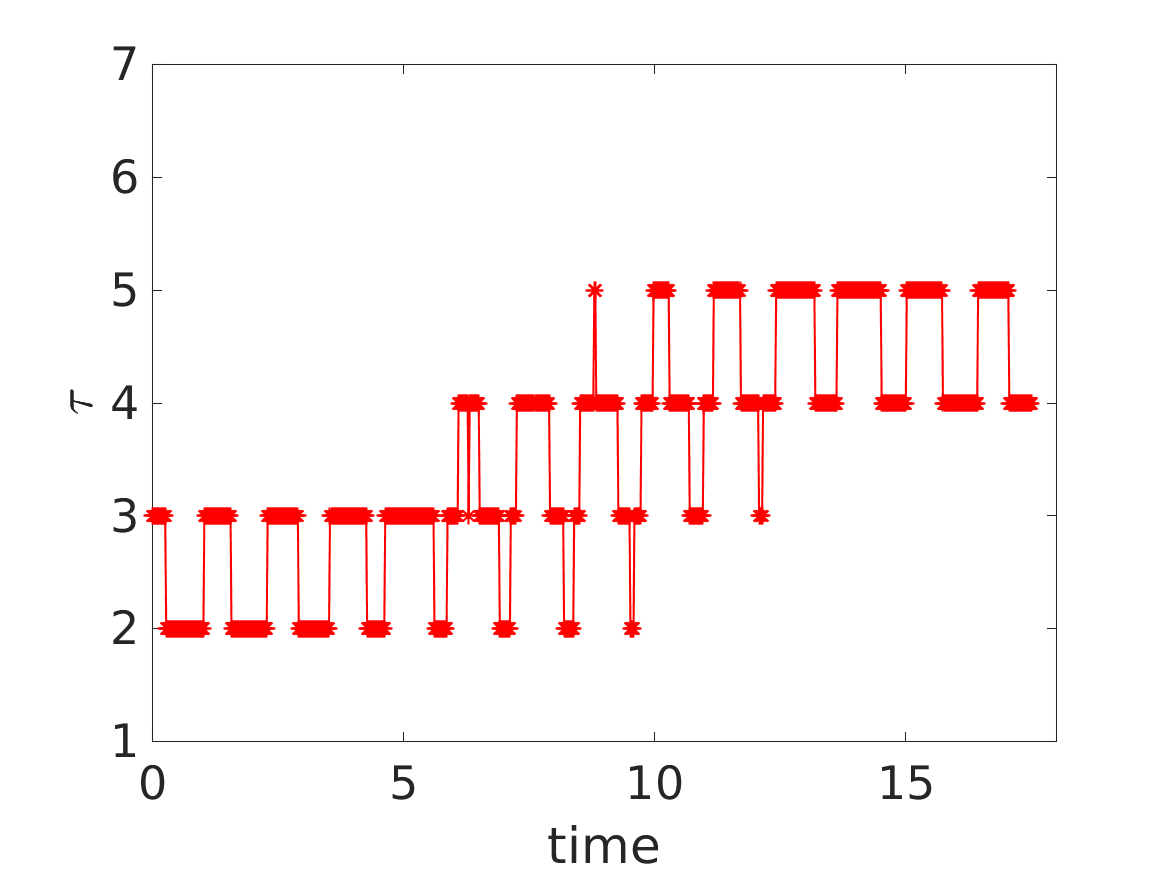}
}
\subfigure[$1024^2$, $P_c=10$]{
\includegraphics[width=0.3\columnwidth]{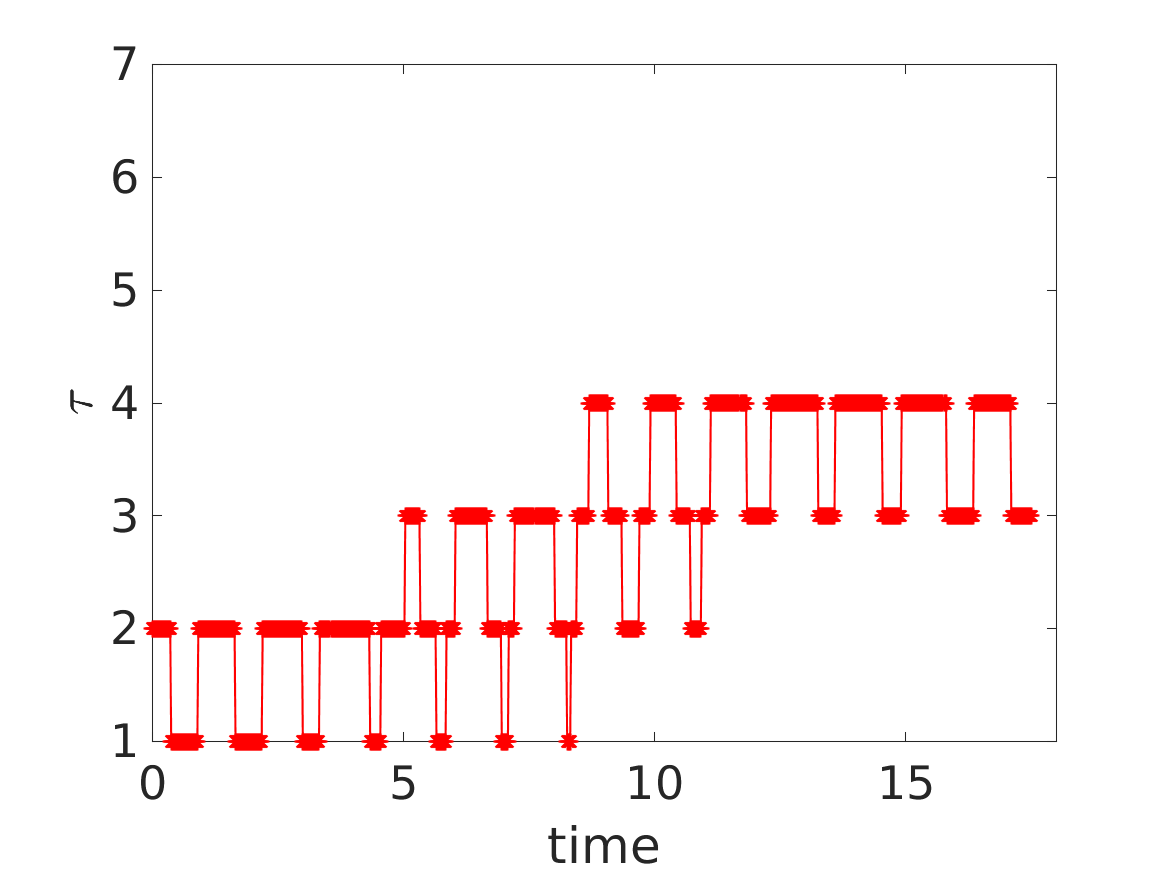}
}
\subfigure[$256^2$, $P_c=20$]{
\includegraphics[width=0.3\columnwidth]{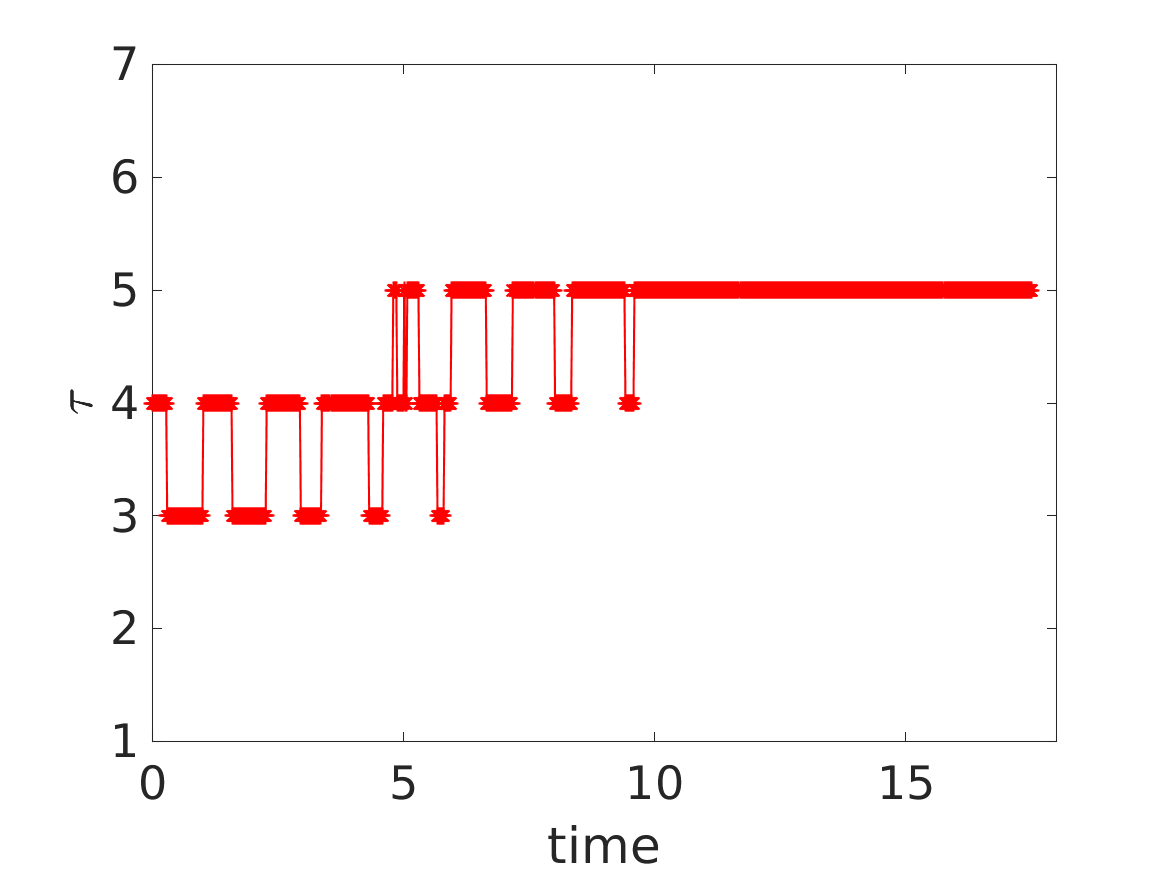}
}
\subfigure[$512^2$, $P_c=20$]{
\includegraphics[width=0.3\columnwidth]{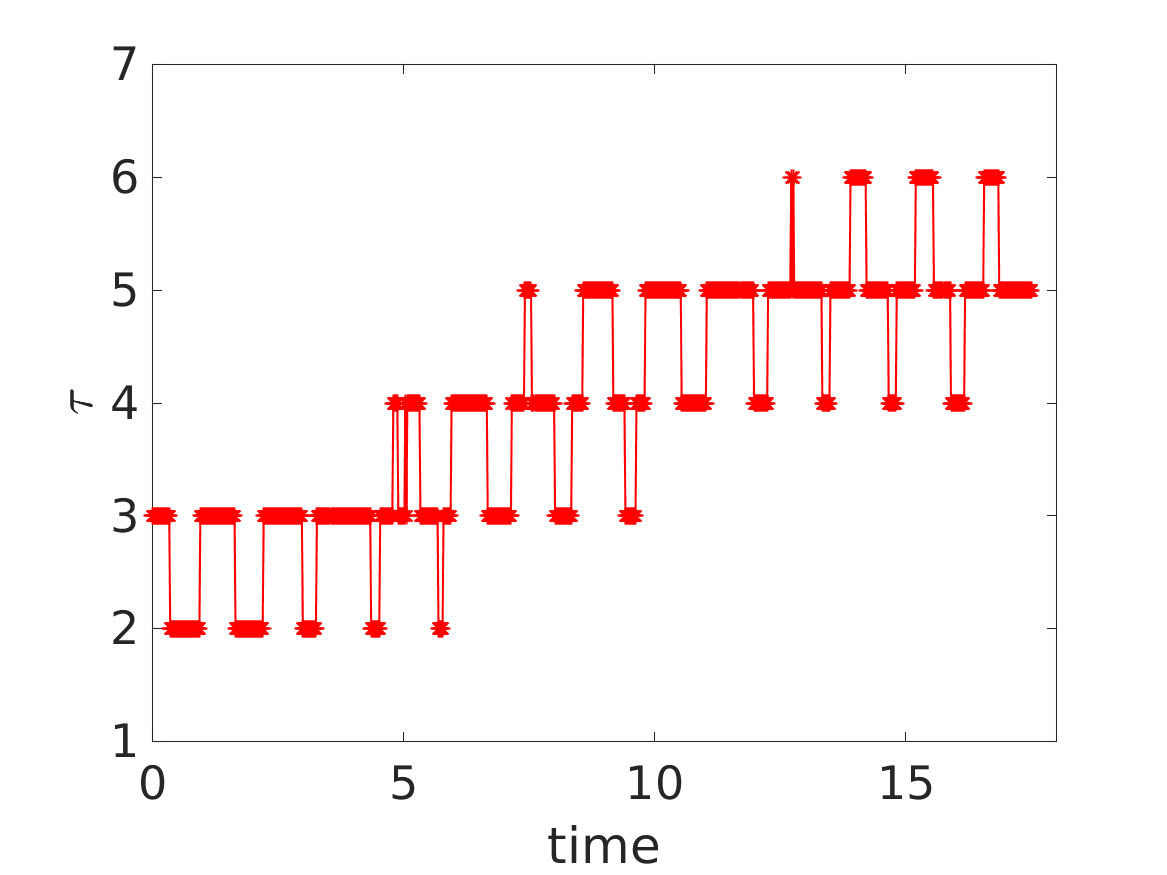}
}
\subfigure[$1024^2$, $P_c=20$]{
\includegraphics[width=0.3\columnwidth]{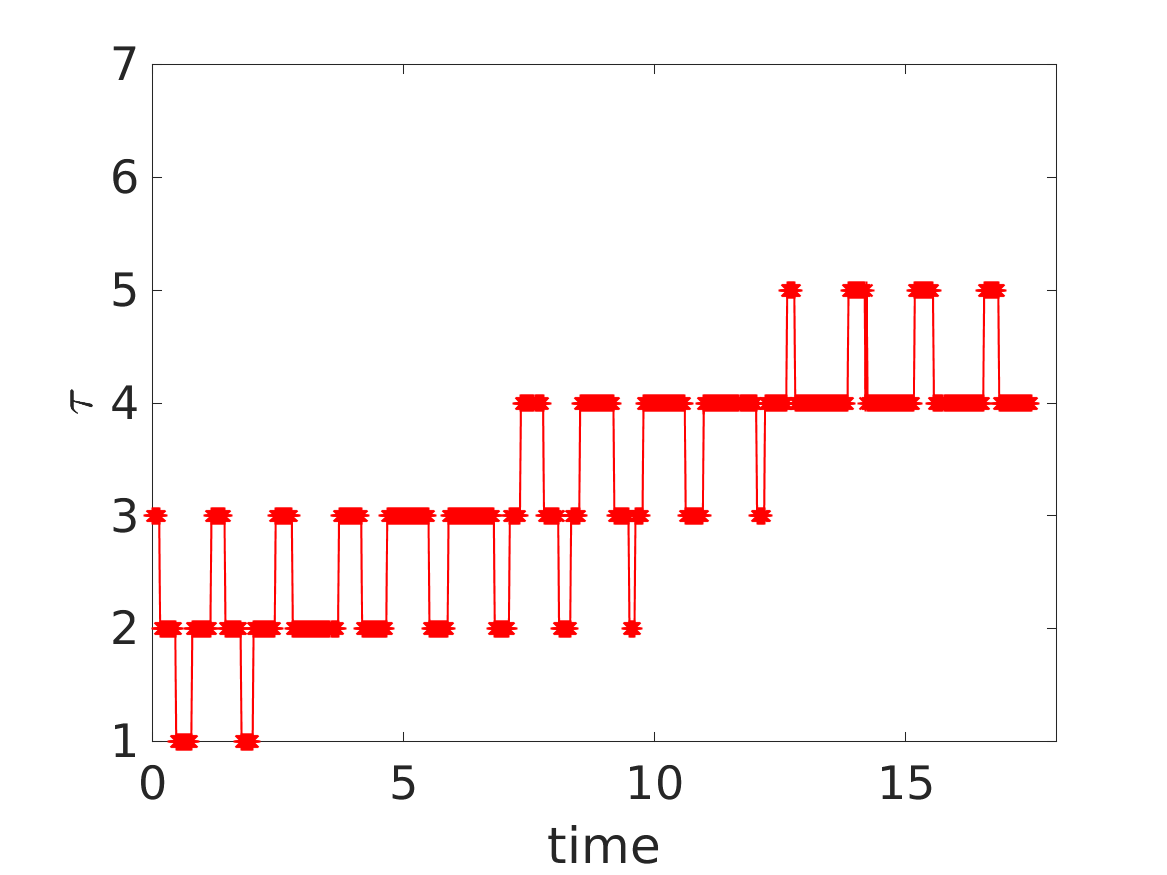}
}
    \caption{2D diocotron instability. Gaussian sampling: Time history of $\tau$ for different mesh sizes and 
    number of particles per cell $P_c$.}
\figlab{diocotron_tau_history_gauss}
\end{figure}

        In Figure \figref{diocotron_tau_history_gauss}, the time history of $\tau$ is shown for the meshes and $P_c$ considered in Figures 
        \figref{diocotron_pc5_gauss}-\figref{diocotron_pc20_gauss}. Here we can see that for the same $P_c$, when we decrease the mesh size, i.e., going from left to right in Figure \figref{diocotron_tau_history_gauss}, the $\tau$ values decrease. This is because we are moving from the grid error dominated regime to the particle error dominated regime. On the other hand,
        for the same mesh size and increasing $P_c$, i.e., moving from top to bottom in Figure \figref{diocotron_tau_history_gauss}, the $\tau$ values increase
        as we are moving from the particle error dominated regime to the grid error dominated regime. Also, for a particular mesh size and given $P_c$ the later time instants
        have higher $\tau$ compared to the earlier ones. This is due to the formation of fine scale structures in the problem and resolving them require
        a higher $\tau$. 

\begin{figure}[h!t!b!]
\subfigure[$256^2$, $P_c=5$]{
\includegraphics[width=0.3\columnwidth]{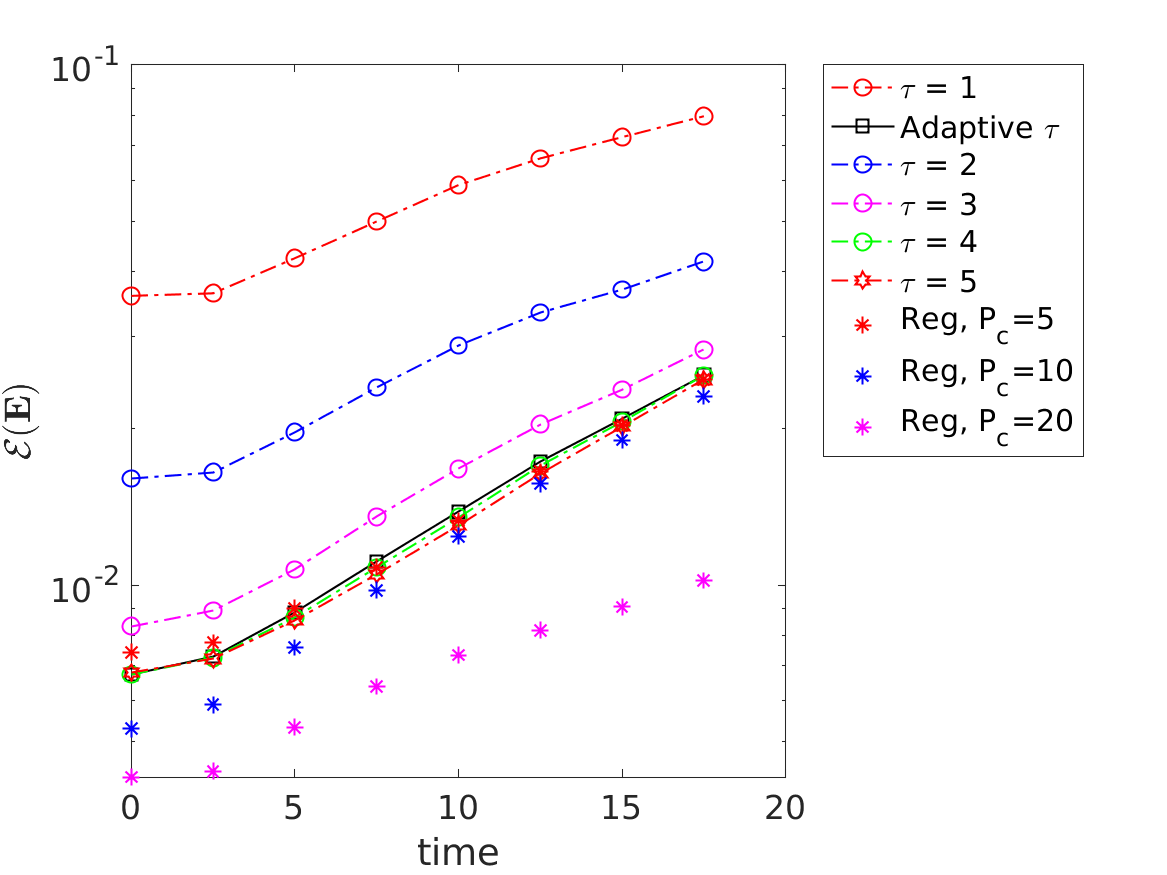}
}
\subfigure[$512^2$, $P_c=5$]{
\includegraphics[width=0.3\columnwidth]{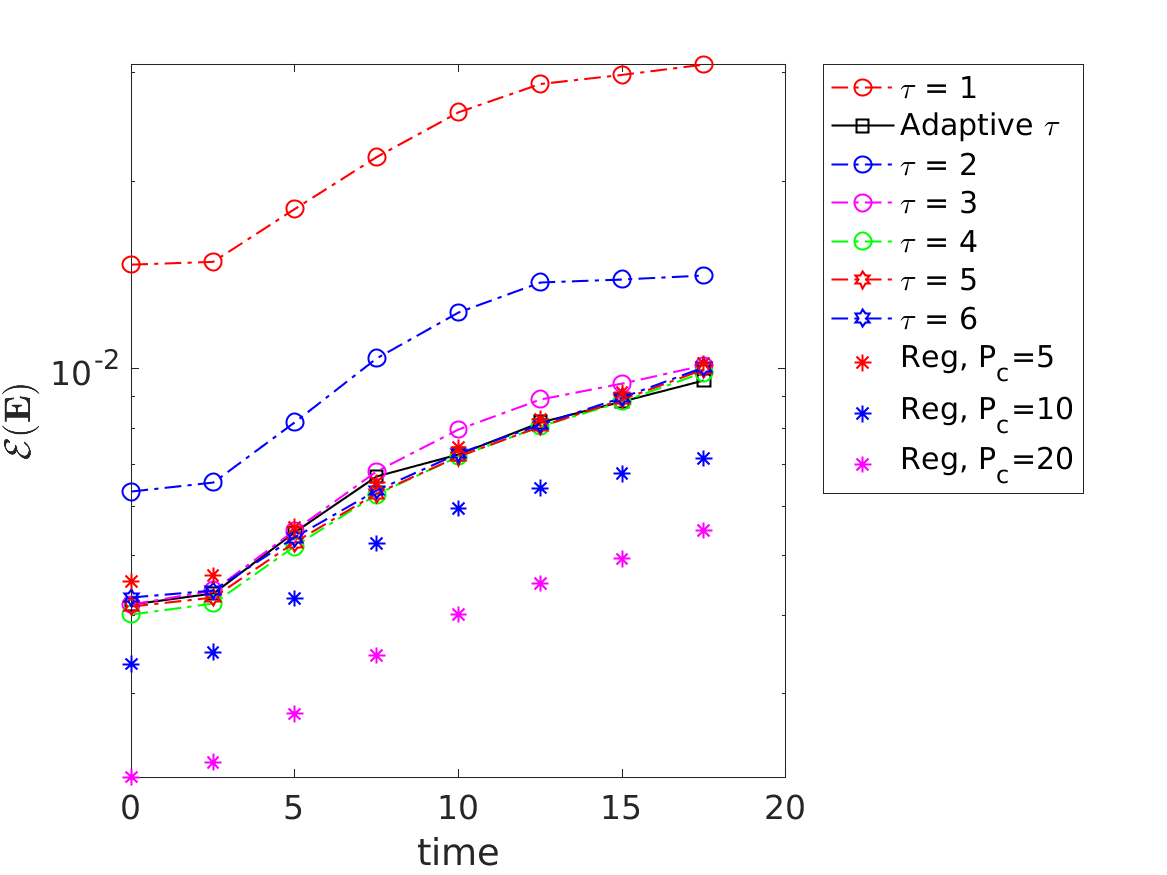}
}
\subfigure[$1024^2$, $P_c=5$]{
\includegraphics[width=0.3\columnwidth]{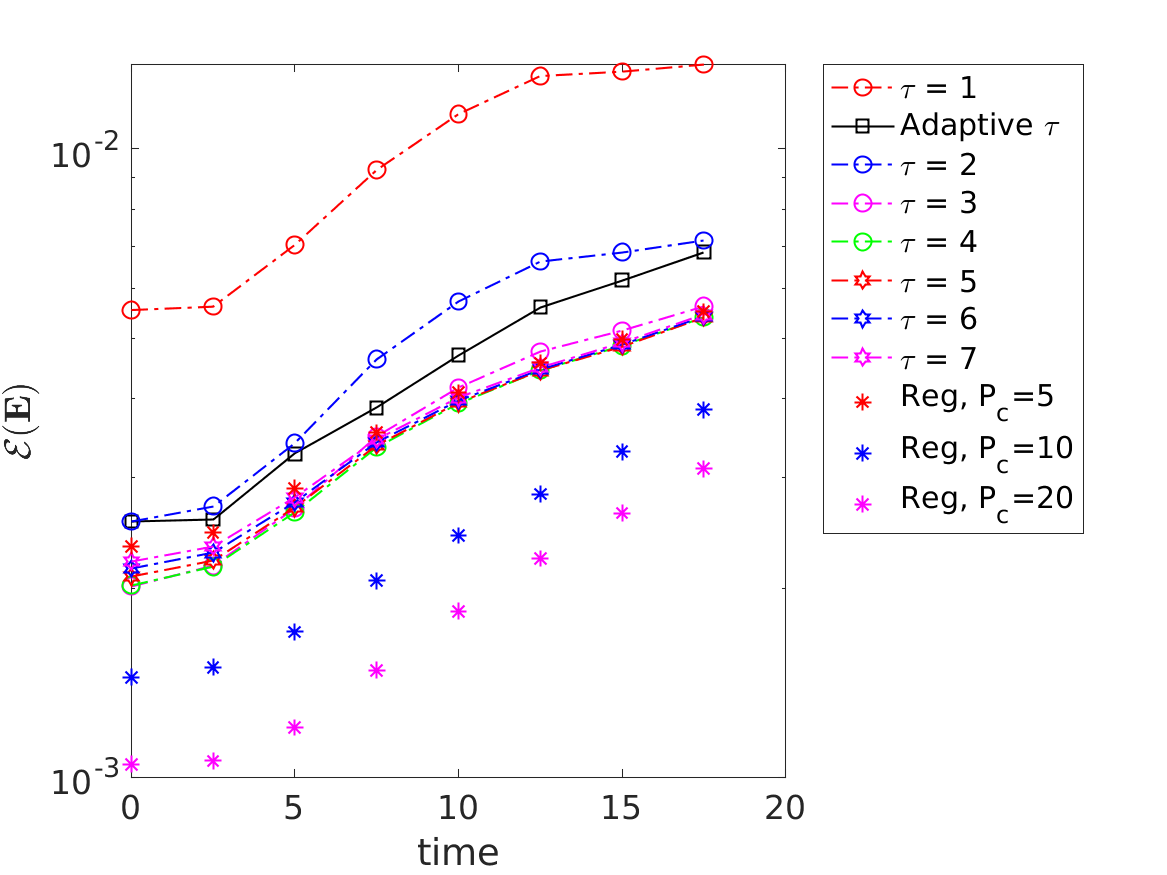}
}
\subfigure[$256^2$, $P_c=10$]{
\includegraphics[width=0.3\columnwidth]{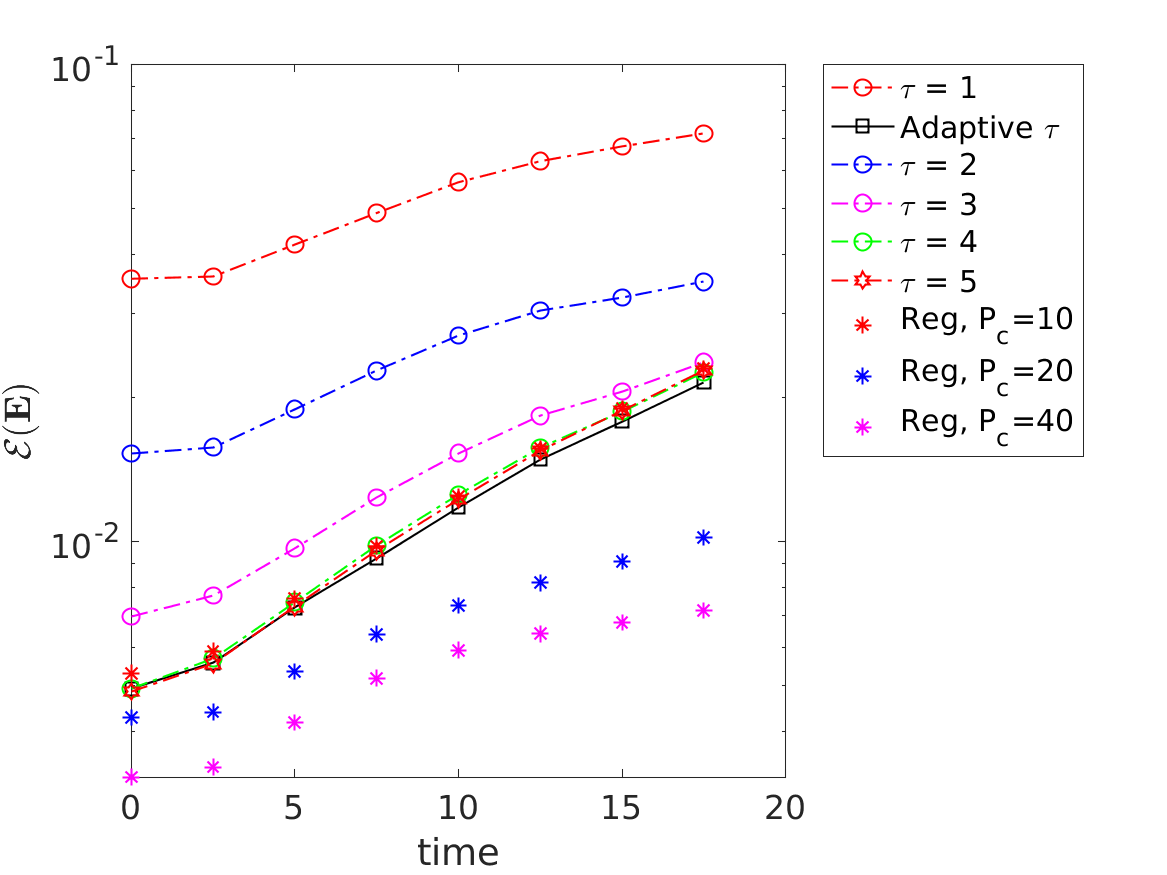}
}
\subfigure[$512^2$, $P_c=10$]{
\includegraphics[width=0.3\columnwidth]{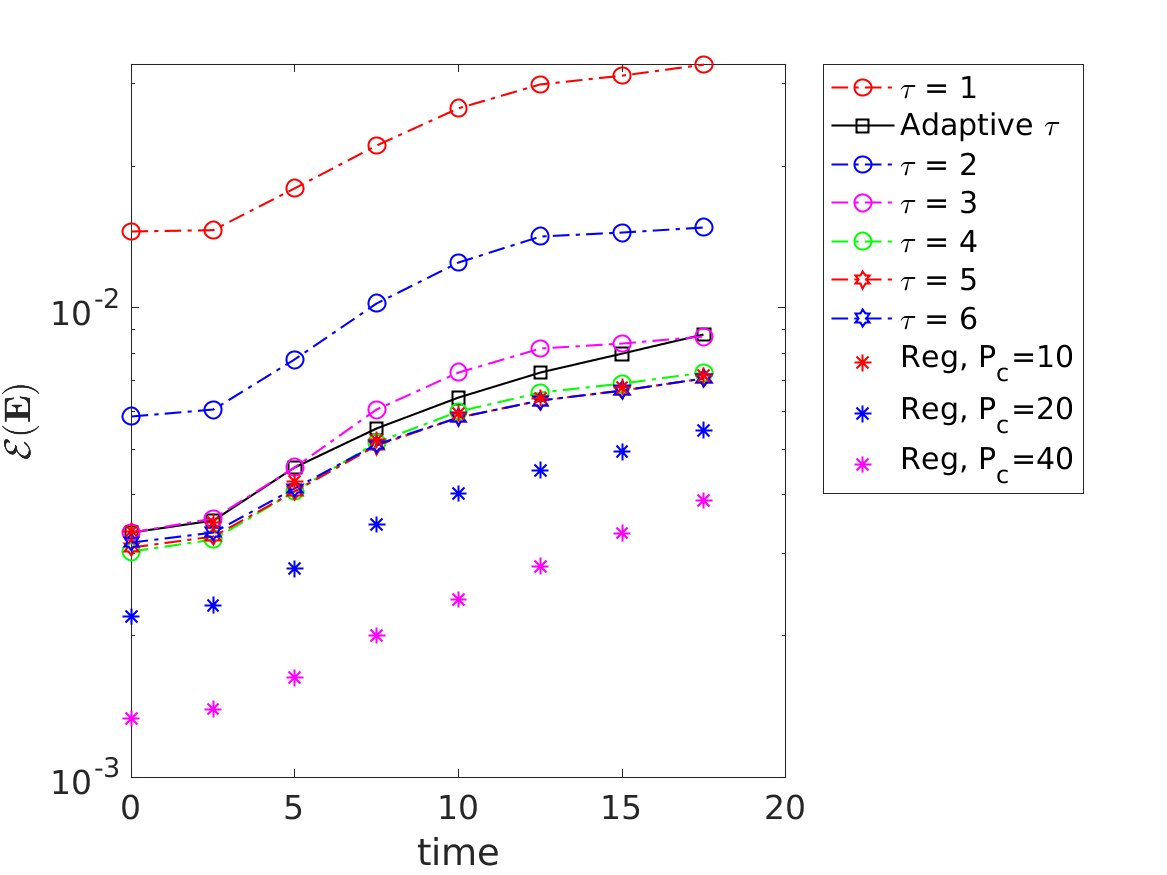}
}
\subfigure[$1024^2$, $P_c=10$]{
\includegraphics[width=0.3\columnwidth]{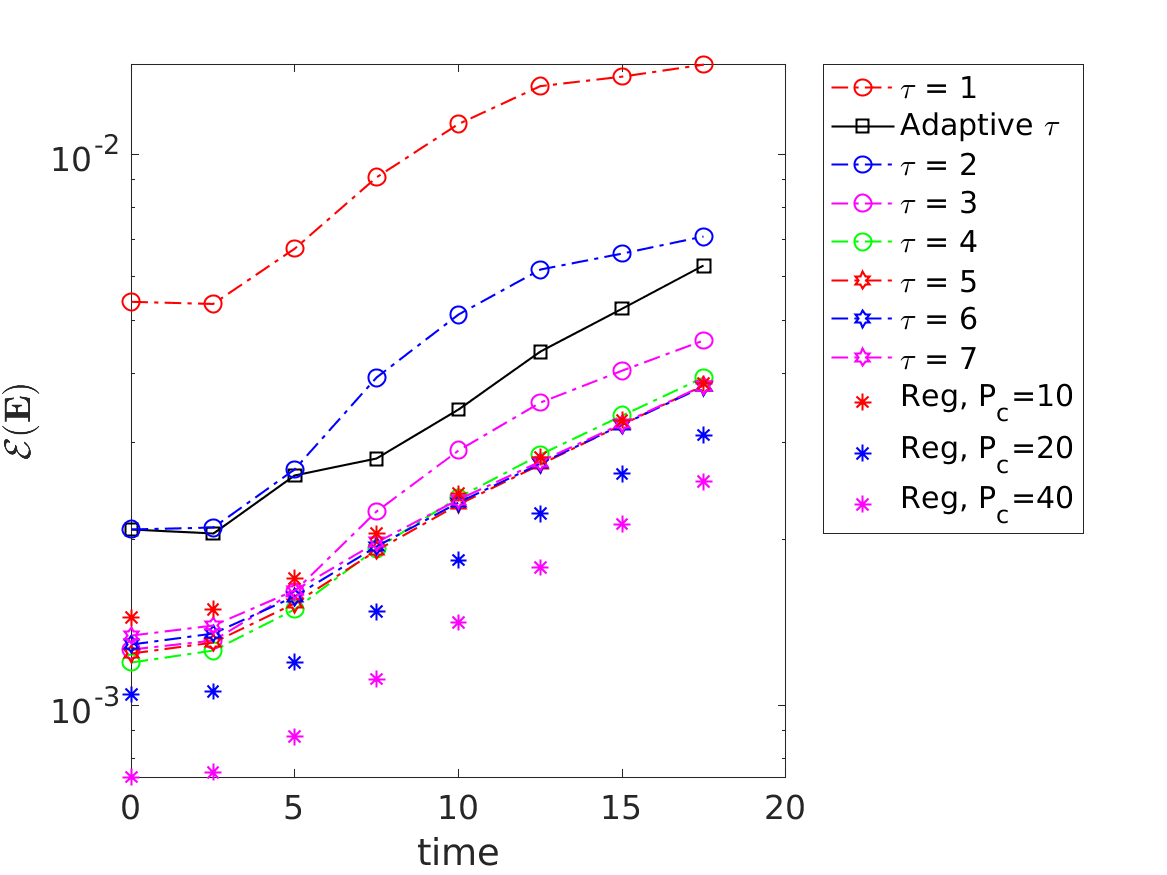}
}
\subfigure[$256^2$, $P_c=20$]{
\includegraphics[width=0.3\columnwidth]{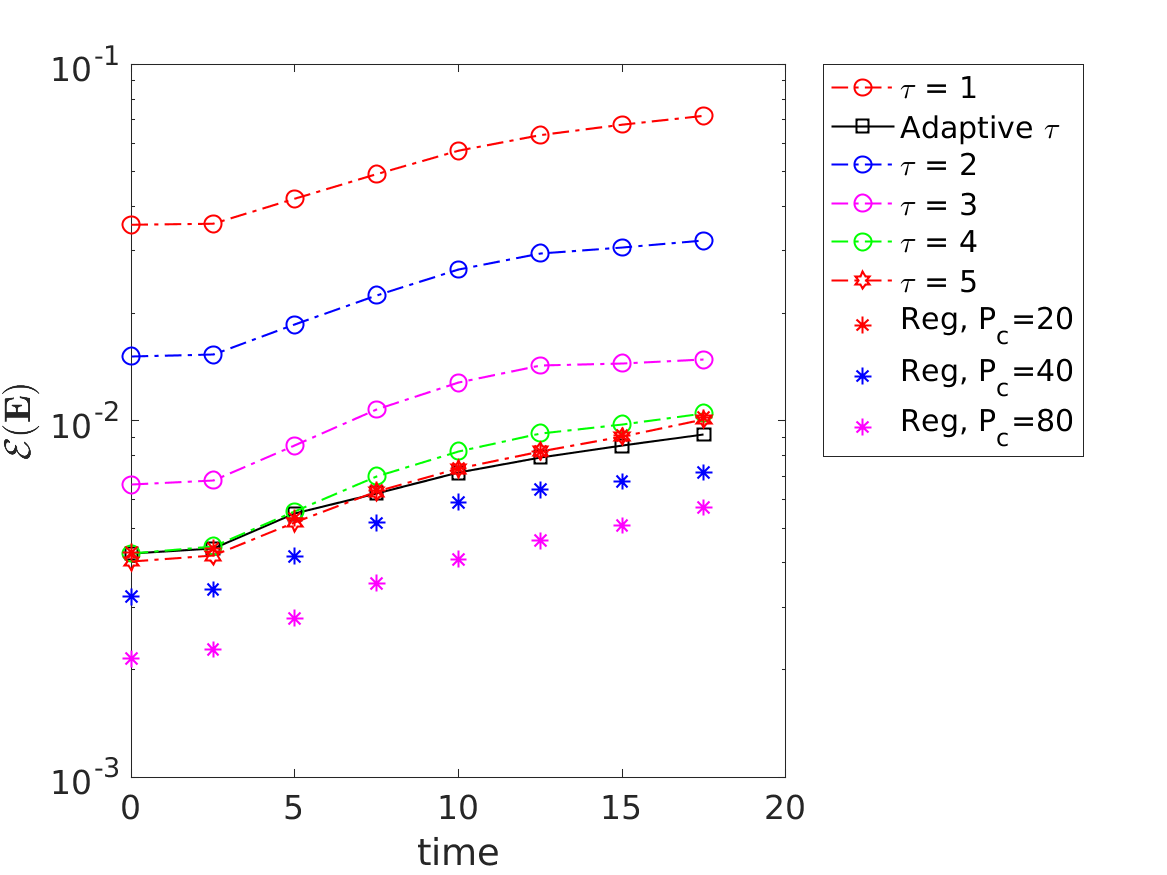}
}
\subfigure[$512^2$, $P_c=20$]{
\includegraphics[width=0.3\columnwidth]{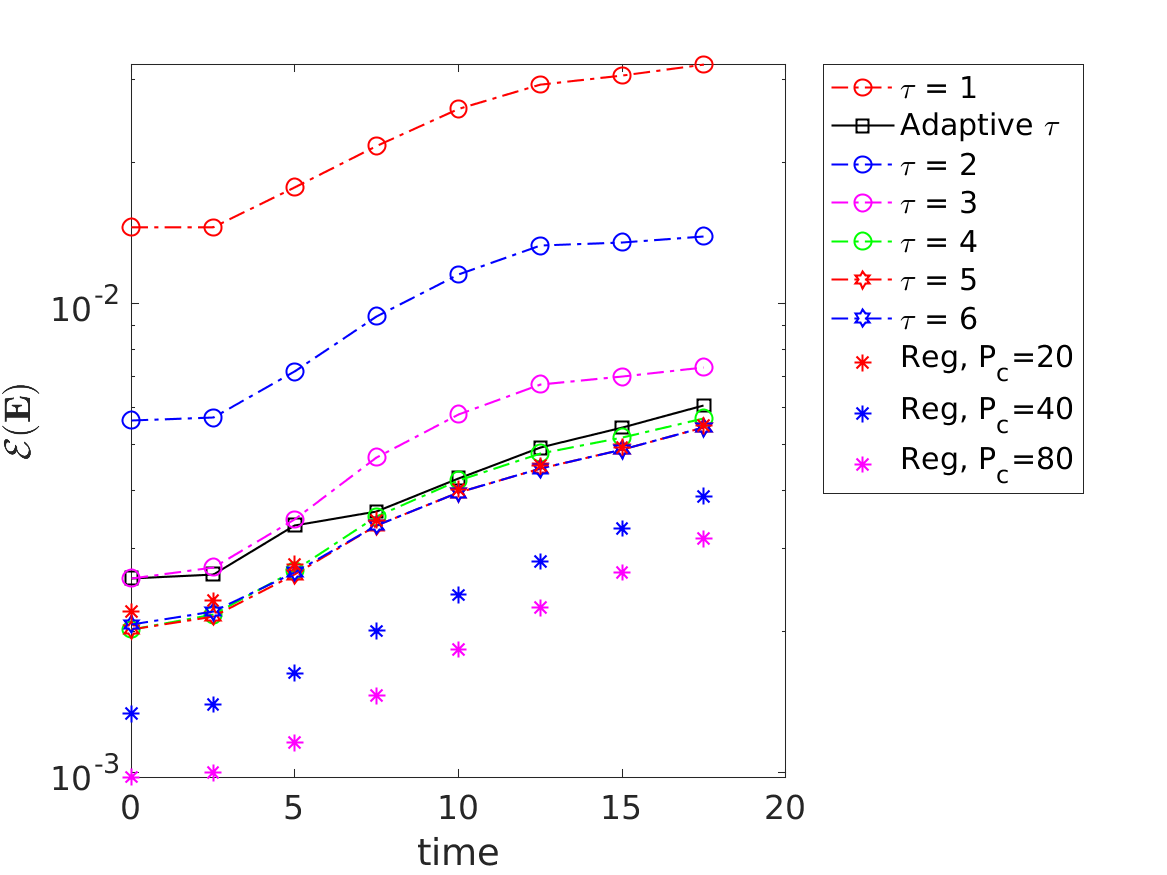}
}
\subfigure[$1024^2$, $P_c=20$]{
\includegraphics[width=0.3\columnwidth]{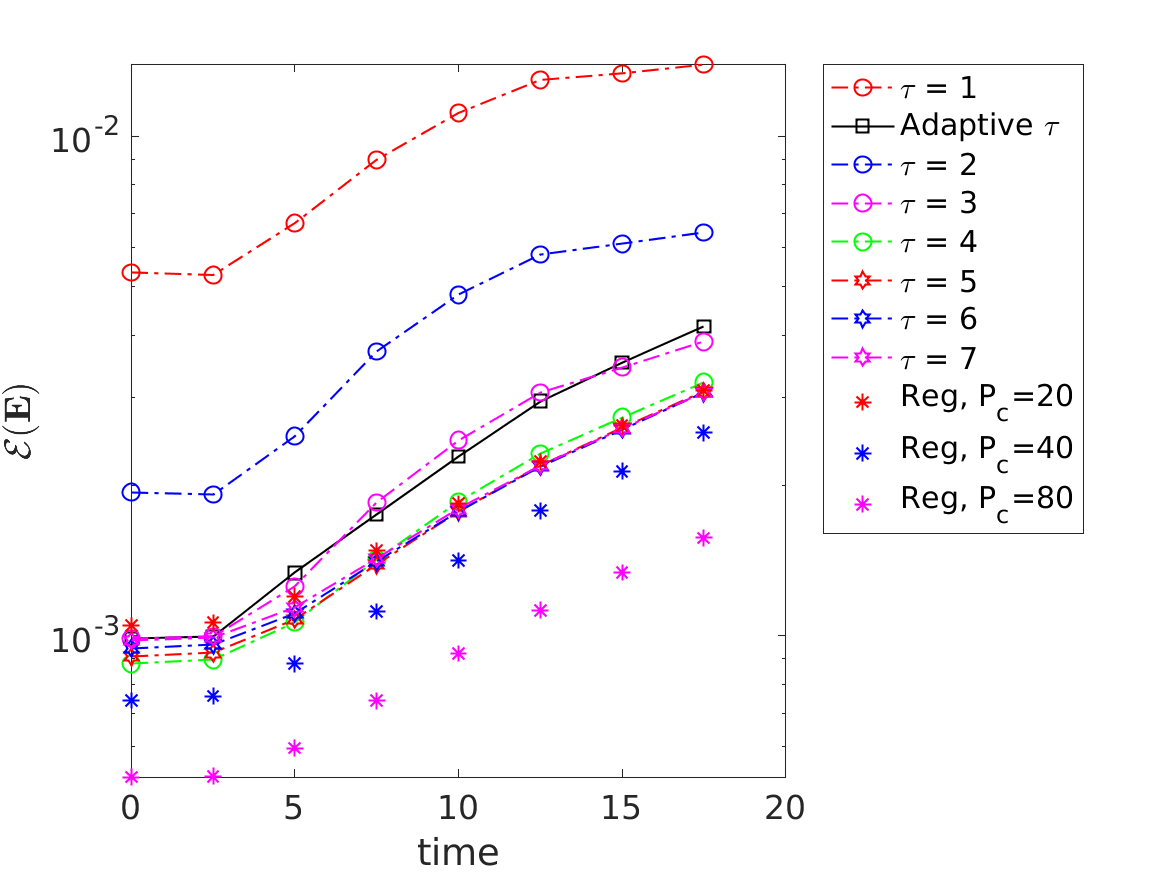}
}
    \caption{2D diocotron instability. Gaussian sampling: Electric field error comparison between regular (Reg), fixed $\tau$ and adaptive $\tau$ PIC.}
\figlab{diocotron_electric_field}
\end{figure}

In Figure \figref{diocotron_electric_field}, the error in the electric field $\Eb$ calculated using equation \eqnref{error_def} is shown 
for the meshes and $P_c$ considered. We can see that the adaptive $\tau$ errors at the best are similar to the regular PIC and in 
some cases it is higher than regular PIC error for the same $P_c$. We also notice that none of the fixed $\tau$ error levels are better than the regular PIC errors. 
The reason for this is as follows: the electric field is obtained by integrating the charge density, and integration is a smoothing operation which 
reduces the particle noise. Since in our adaptive $\tau$ noise reduction algorithm we increase the grid-based error to reduce the particle noise and minimize the total error in the density, this can result in either similar or even an increase in the electric field error as compared to the regular PIC if the 
integration itself is sufficient enough to reduce the noise. High-order shape functions are a promising option to address this limitation as they may reduce the particle 
noise without increasing the grid-based error. We will investigate the combination of high-order shape functions with our algorithm in future work. %Moreover, the electric field errors in Figure \figref{diocotron_electric_field} 

\begin{figure}[h!t!b!]
\subfigure[$256^2$, $P_c=5$]{
\includegraphics[width=0.3\columnwidth]{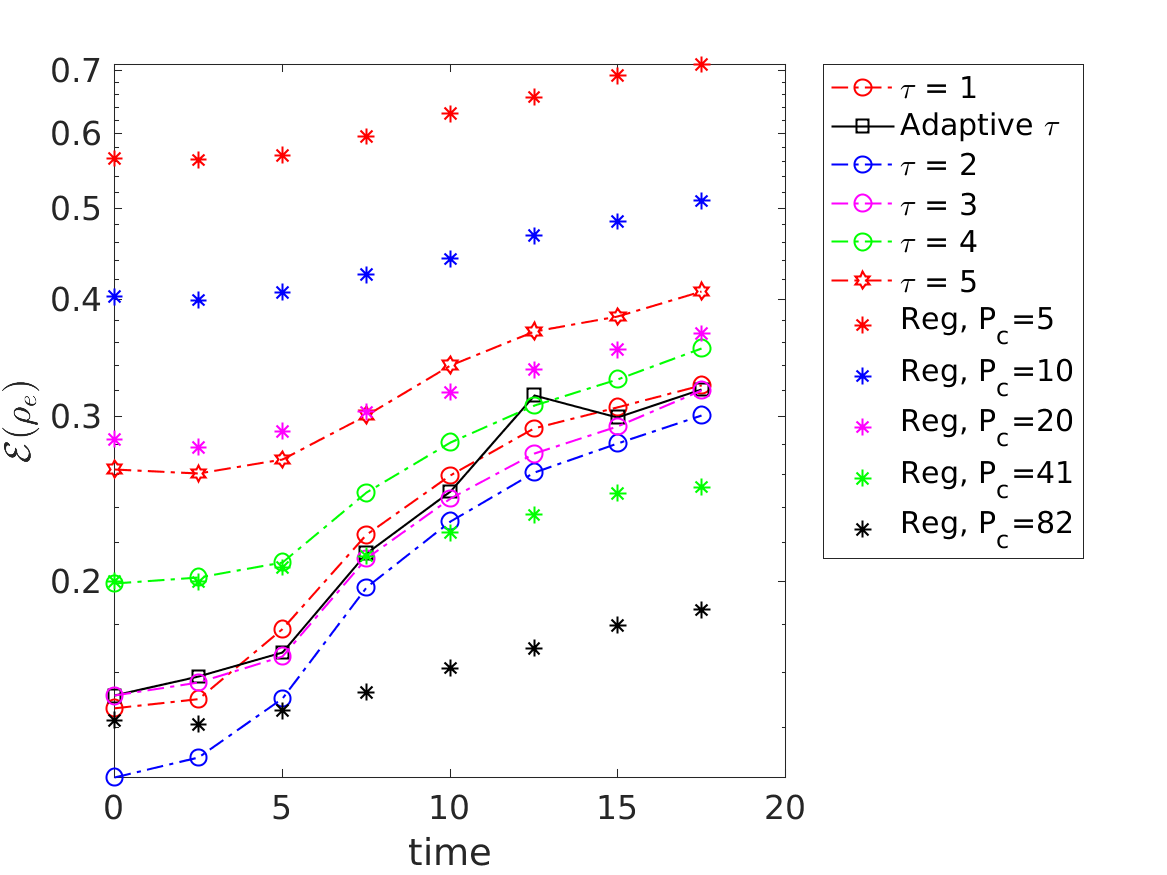}
}
\subfigure[$512^2$, $P_c=5$]{
\includegraphics[width=0.3\columnwidth]{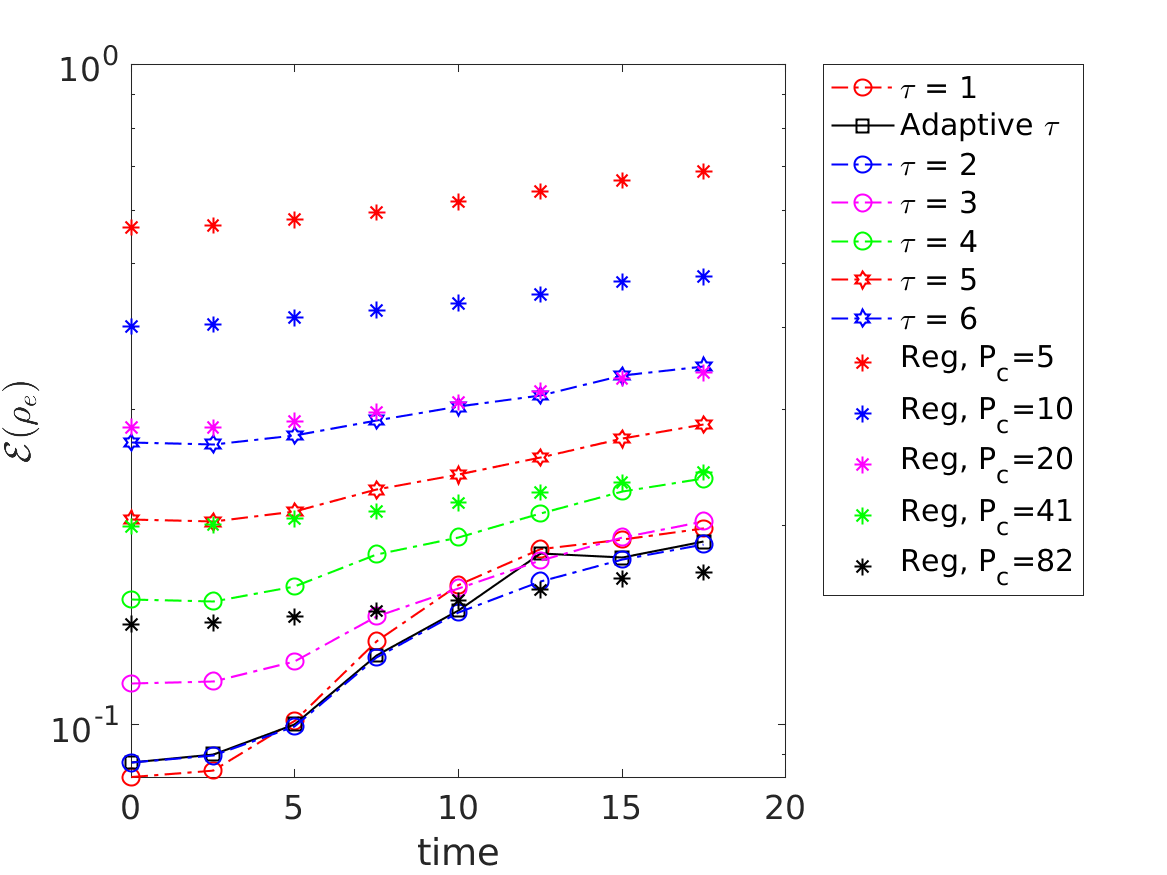}
}
\subfigure[$1024^2$, $P_c=5$]{
\includegraphics[width=0.3\columnwidth]{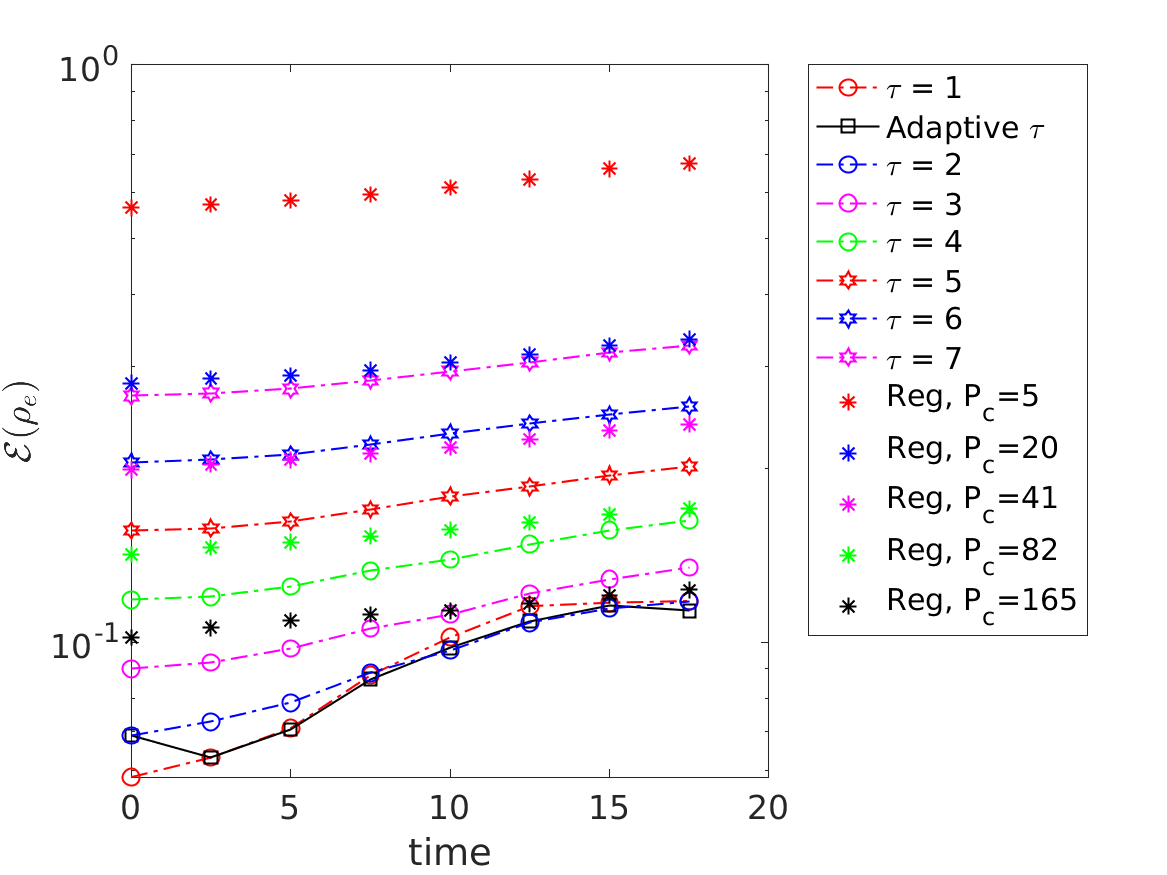}
}
\subfigure[$256^2$, $P_c=10$]{
\includegraphics[width=0.3\columnwidth]{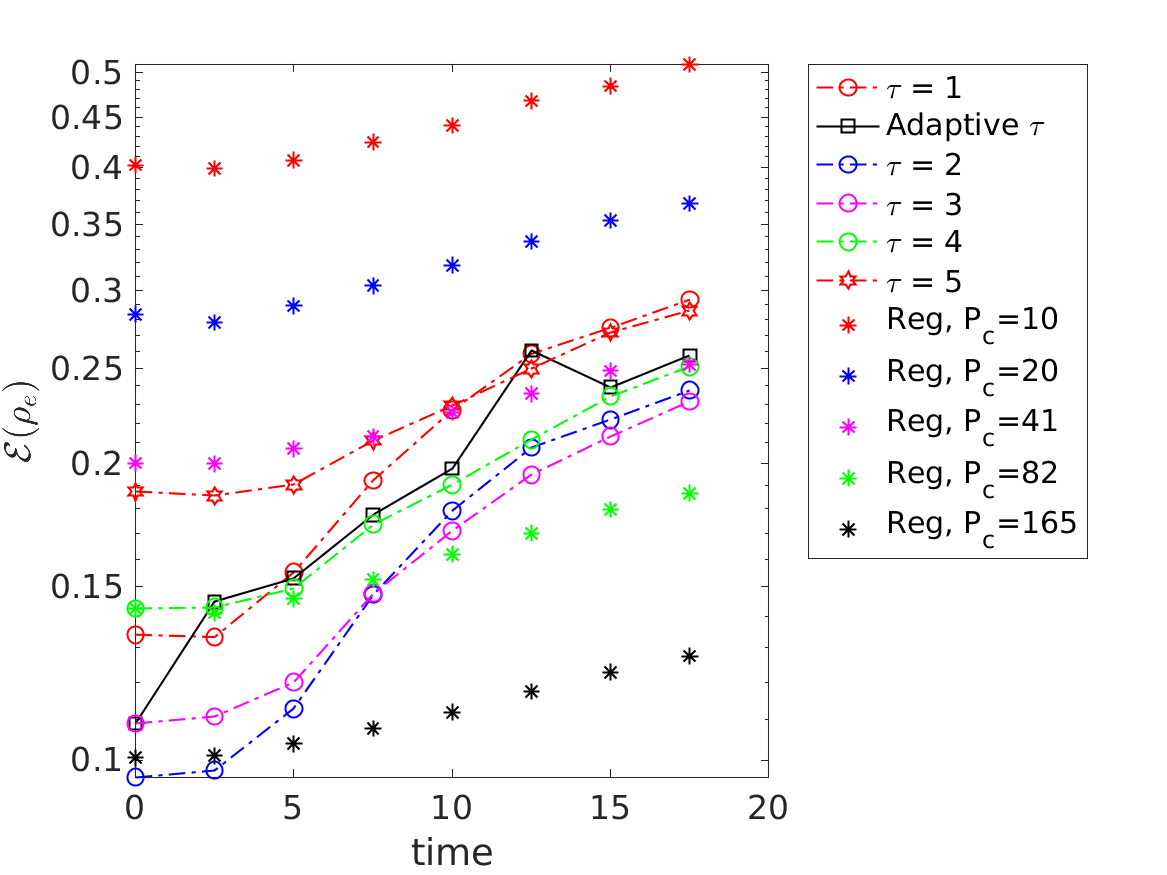}
}
\subfigure[$512^2$, $P_c=10$]{
\includegraphics[width=0.3\columnwidth]{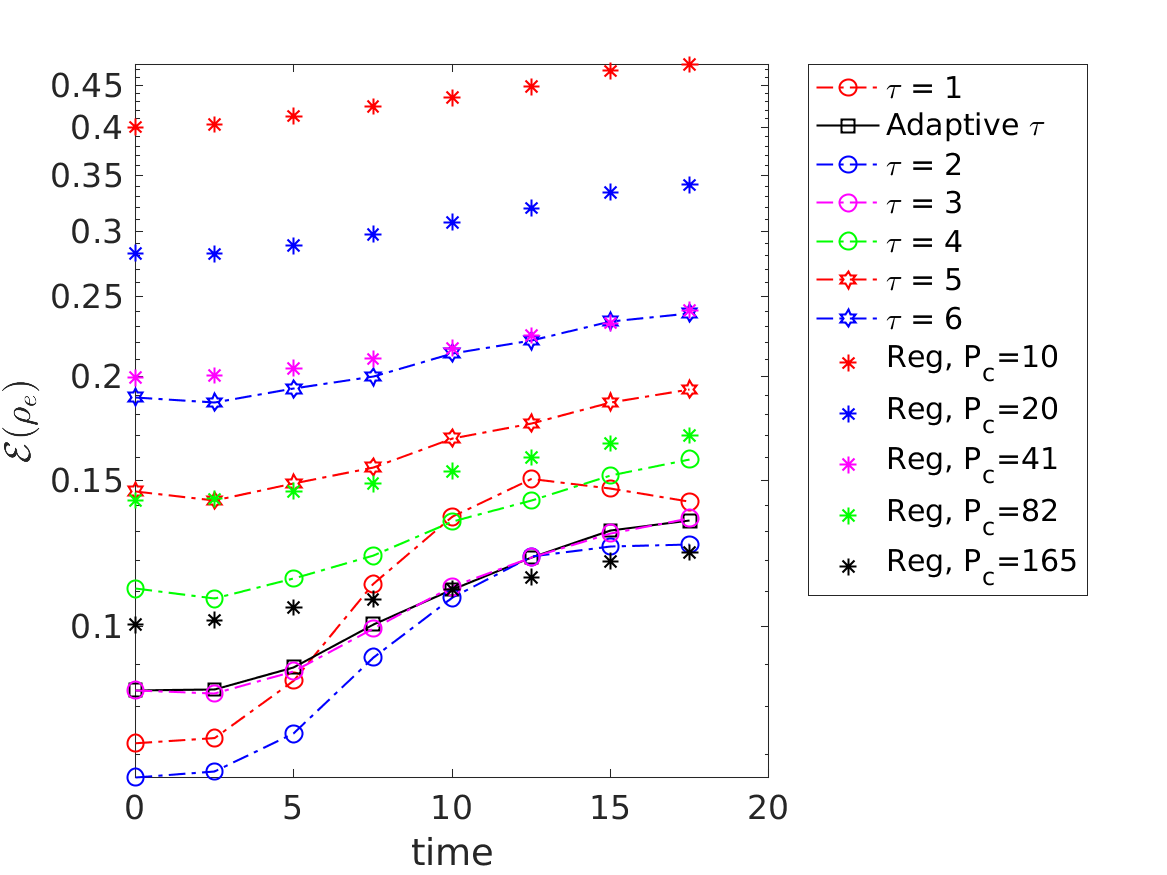}
}
\subfigure[$1024^2$, $P_c=10$]{
\includegraphics[width=0.3\columnwidth]{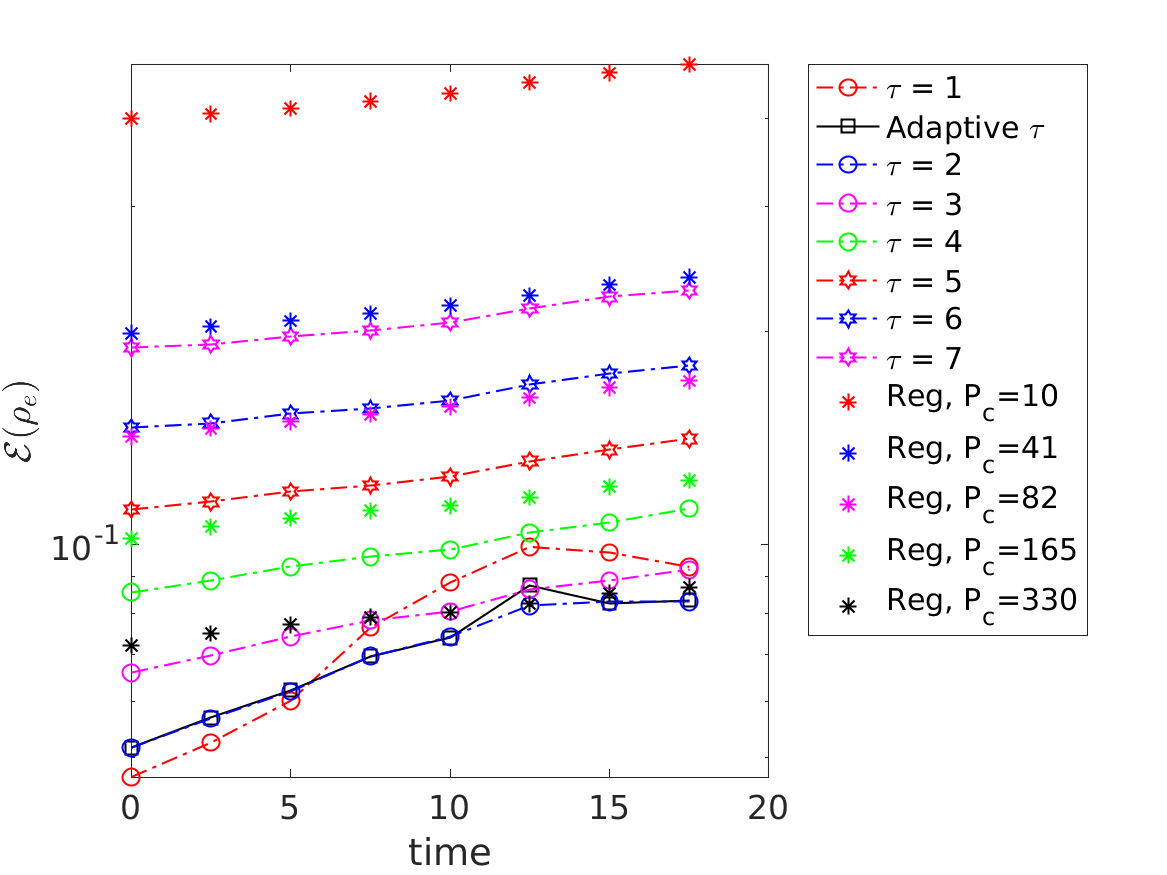}
}
\subfigure[$256^2$, $P_c=20$]{
\includegraphics[width=0.3\columnwidth]{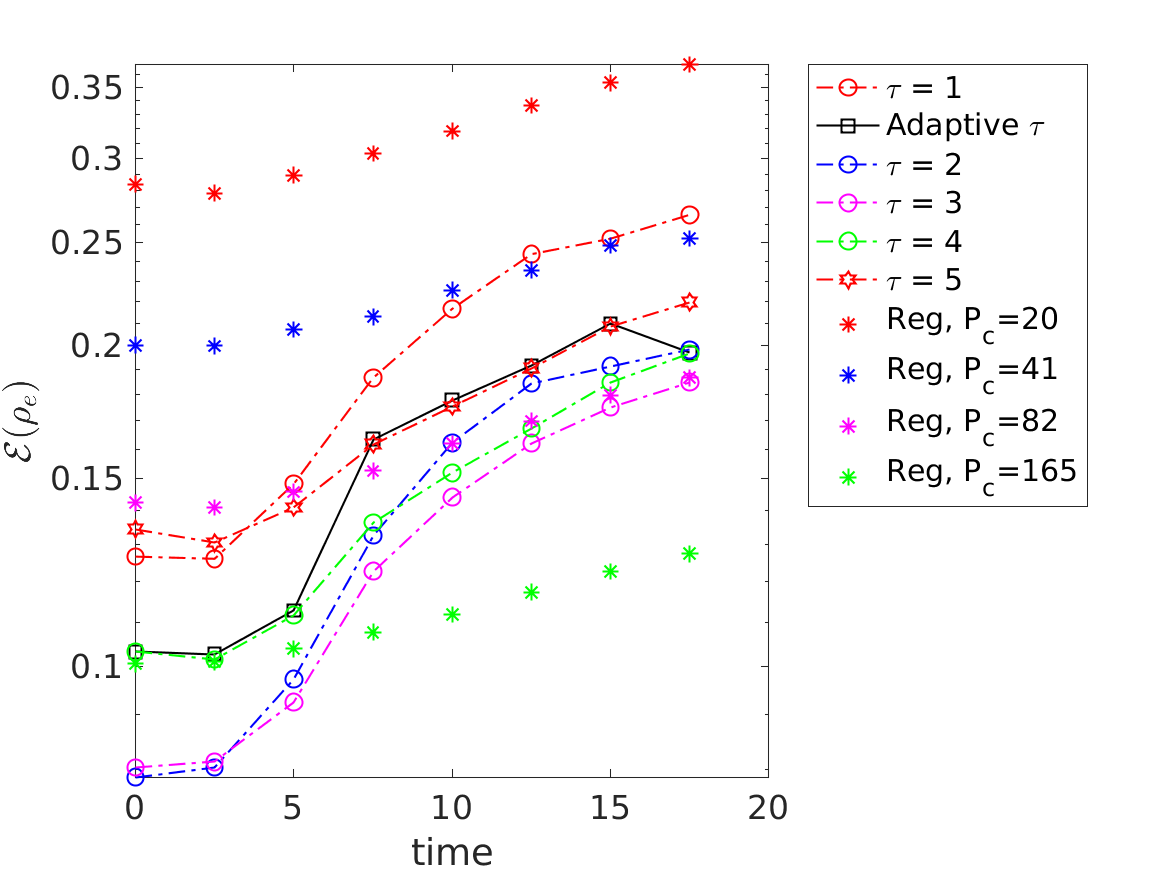}
}
\subfigure[$512^2$, $P_c=20$]{
\includegraphics[width=0.3\columnwidth]{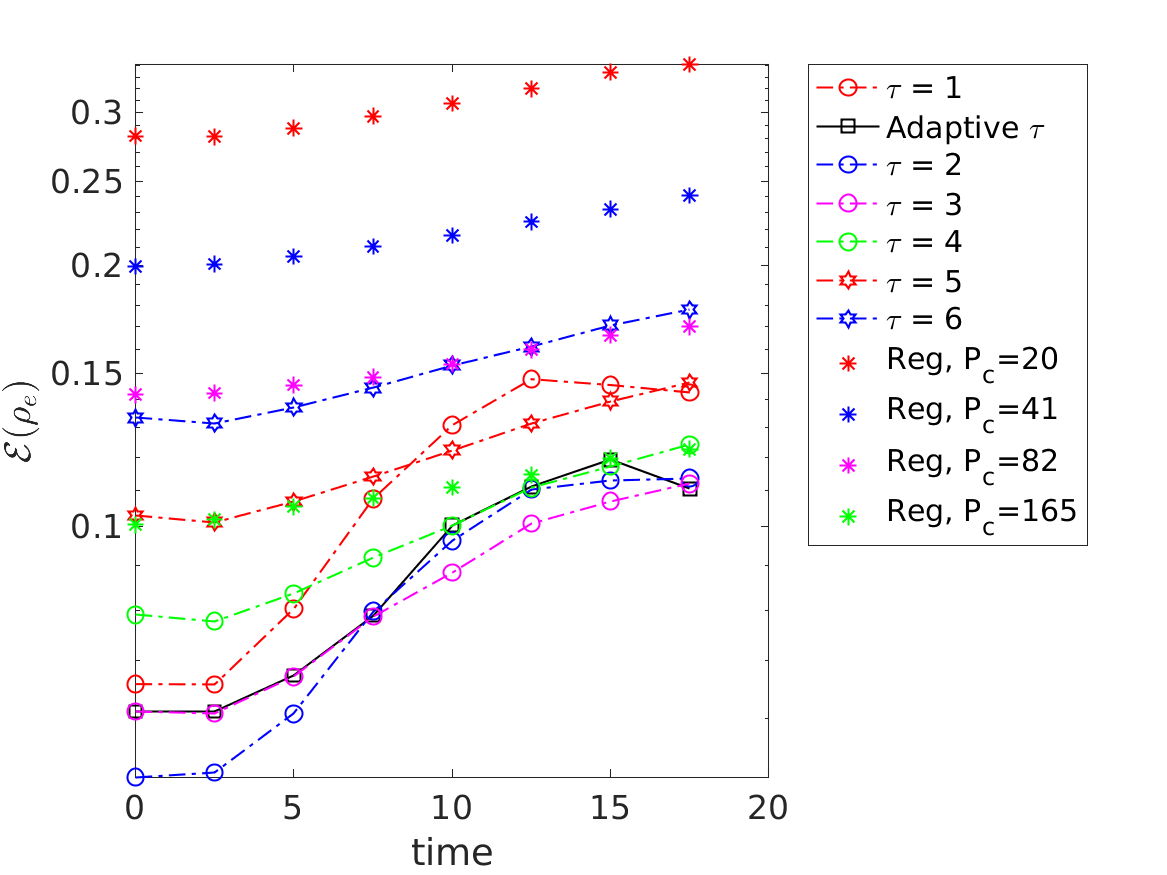}
}
\subfigure[$1024^2$, $P_c=20$]{
\includegraphics[width=0.3\columnwidth]{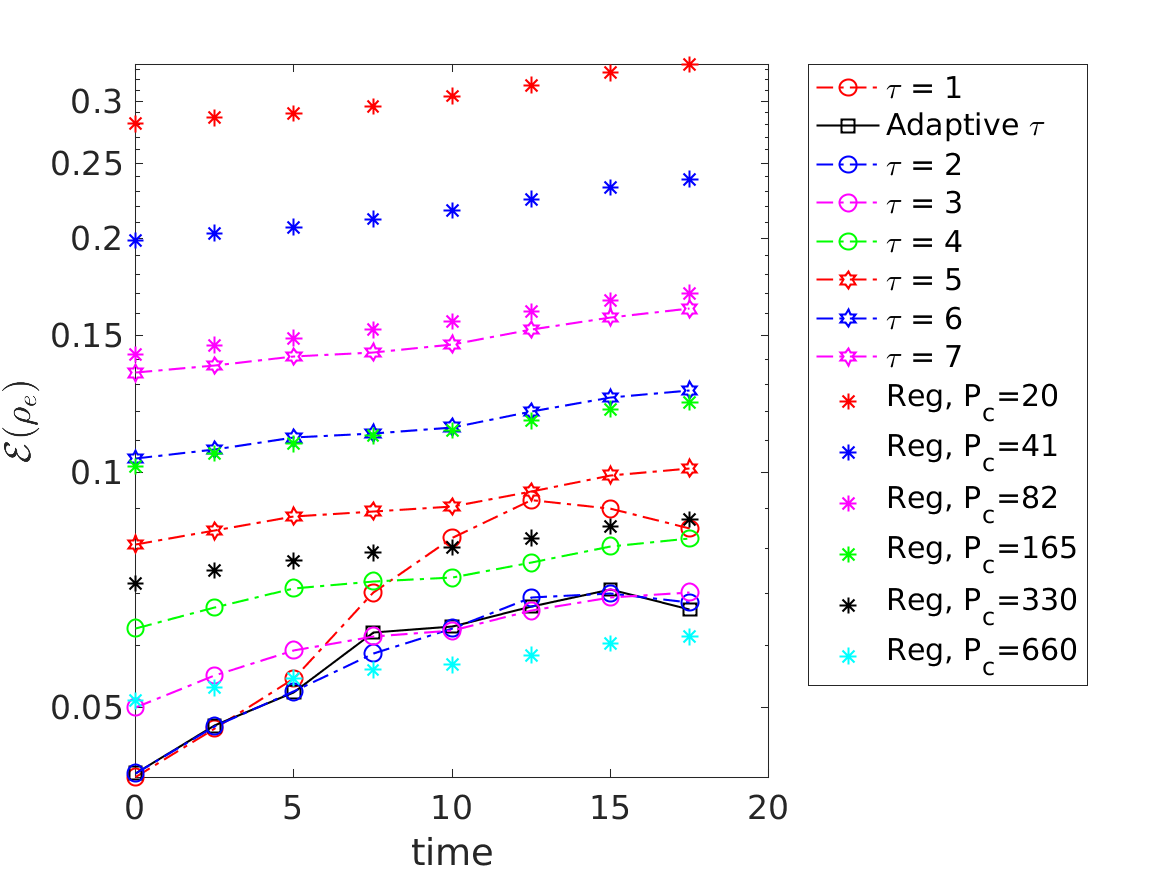}
}
    \caption{2D diocotron instability. Uniform sampling: Electron charge density error comparison between regular (Reg), fixed $\tau$ and adaptive $\tau$ PIC.  The errors for regular PIC with $P_c=330$ and $660$ are
    calculated from that of $P_c=165$ based on the theoretical particle error scaling $1/\sqrt{Pc}$. This is based on the observation that the errors for the regular PIC are in the noise dominated regime.}
\figlab{diocotron_density_uniform}
\end{figure}

\begin{figure}[h!t!b!]
\subfigure[$256^2$, $P_c=5$]{
\includegraphics[width=0.3\columnwidth]{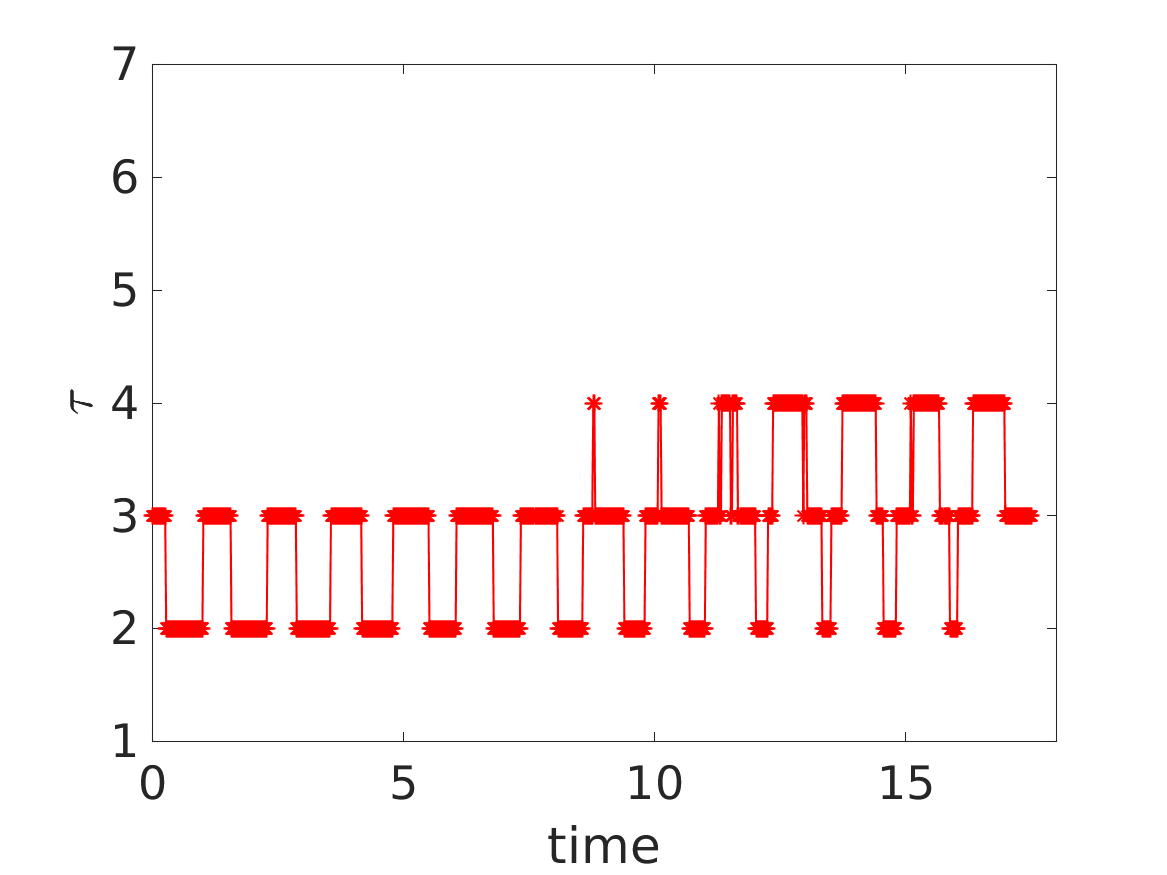}
}
\subfigure[$512^2$, $P_c=5$]{
\includegraphics[width=0.3\columnwidth]{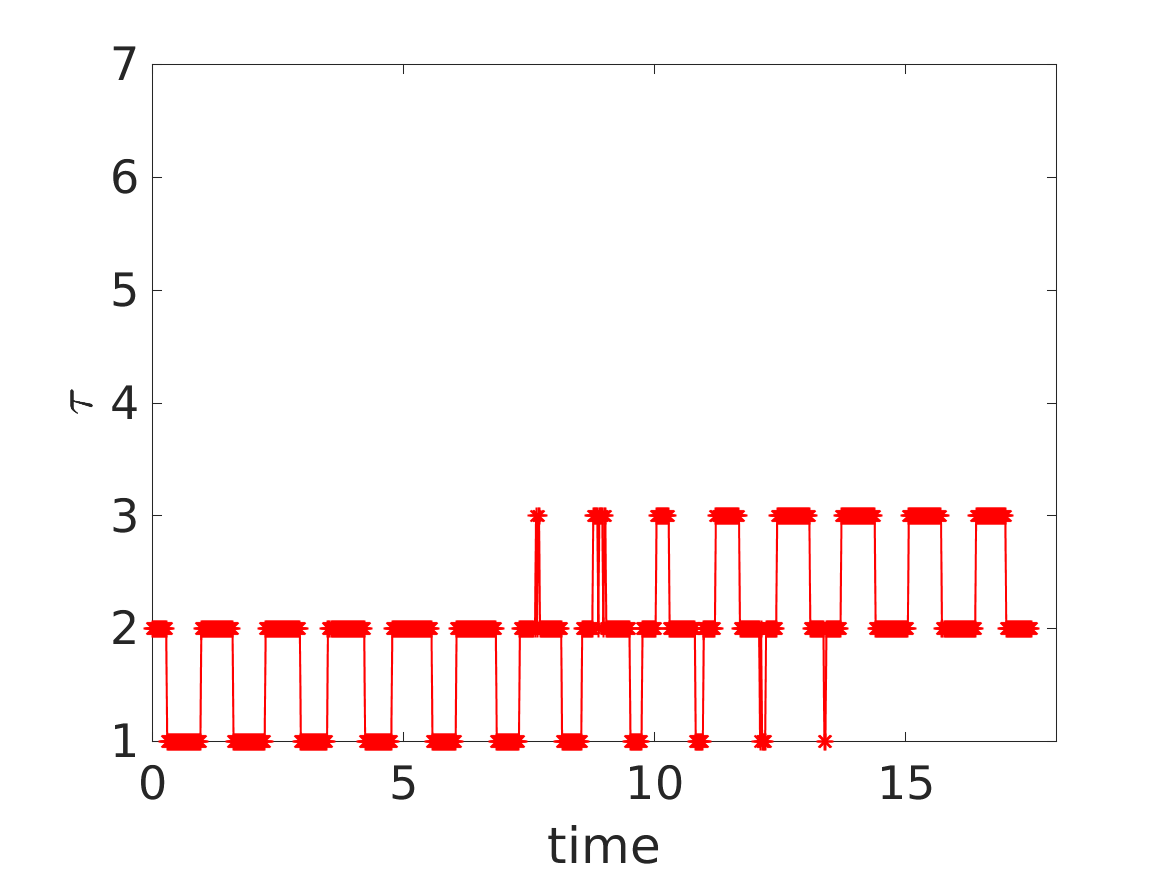}
}
\subfigure[$1024^2$, $P_c=5$]{
\includegraphics[width=0.3\columnwidth]{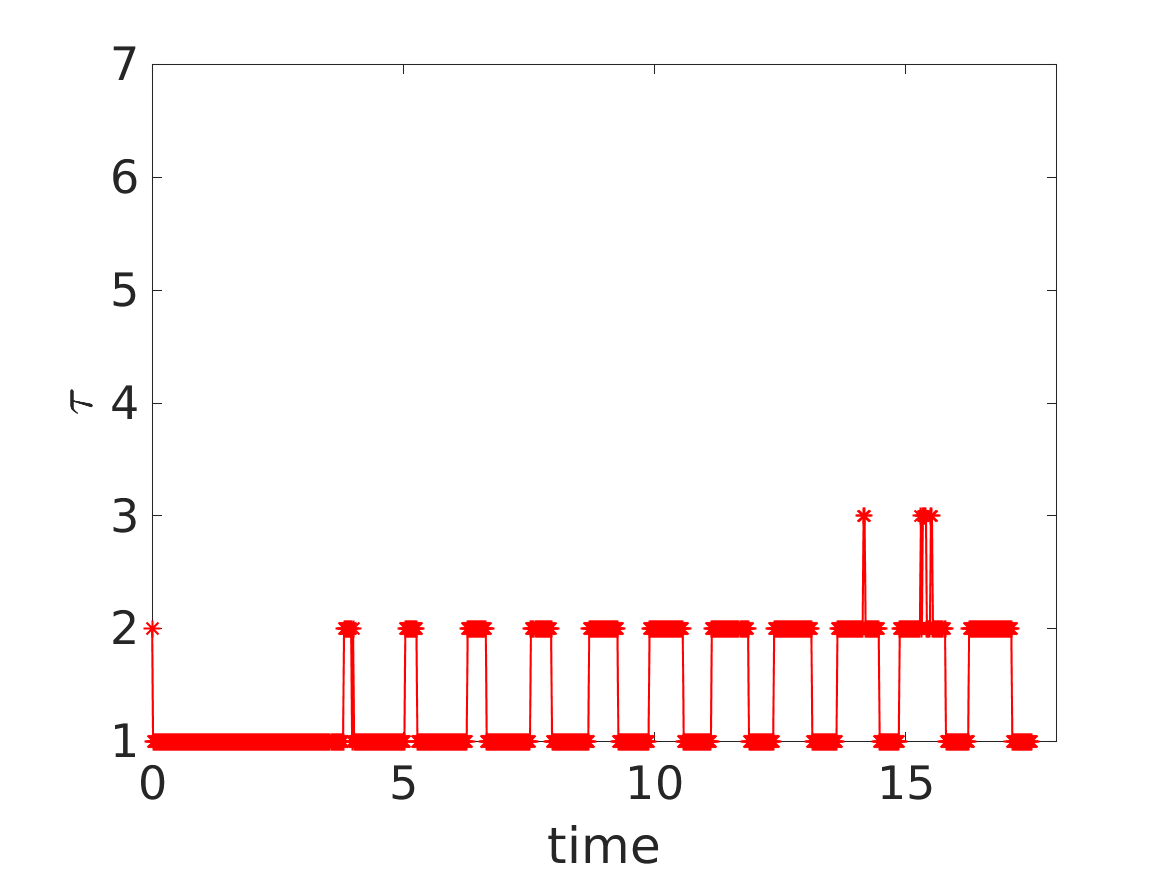}
}
\subfigure[$256^2$, $P_c=10$]{
\includegraphics[width=0.3\columnwidth]{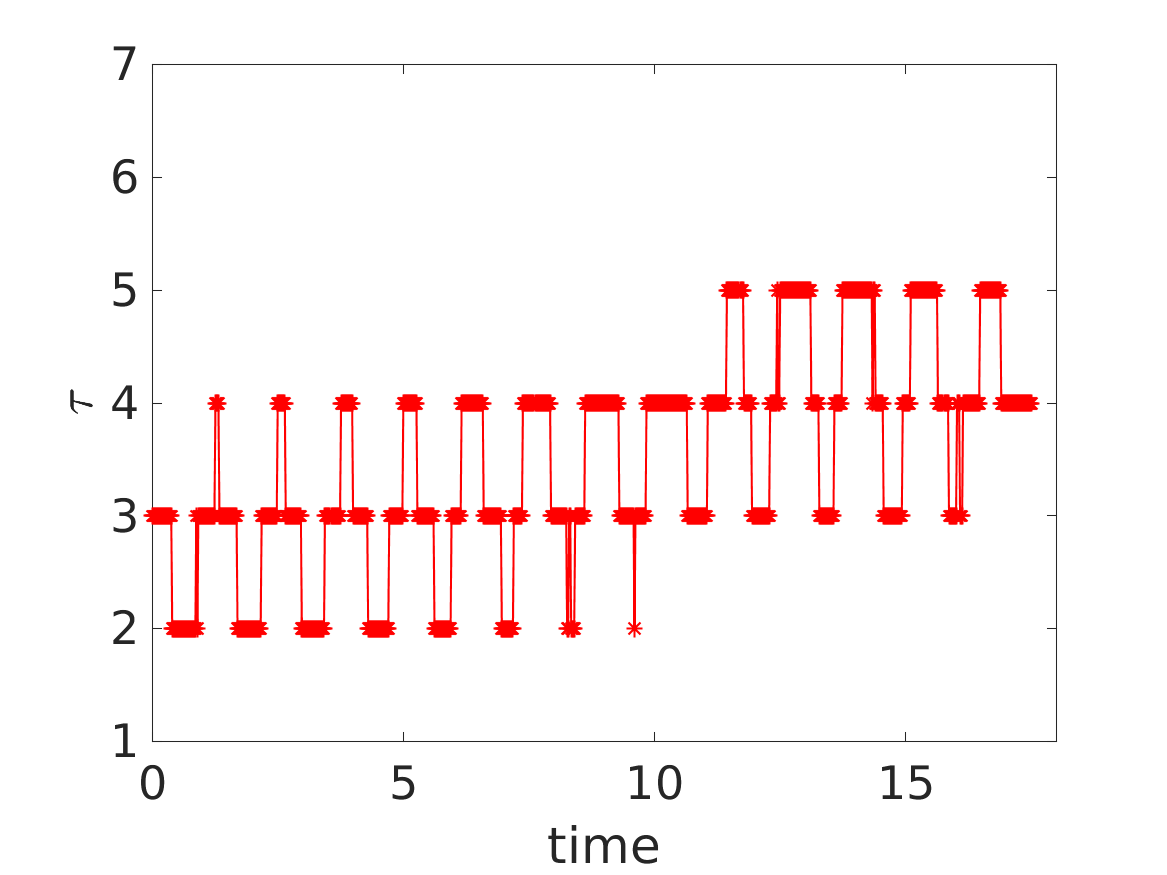}
}
\subfigure[$512^2$, $P_c=10$]{
\includegraphics[width=0.3\columnwidth]{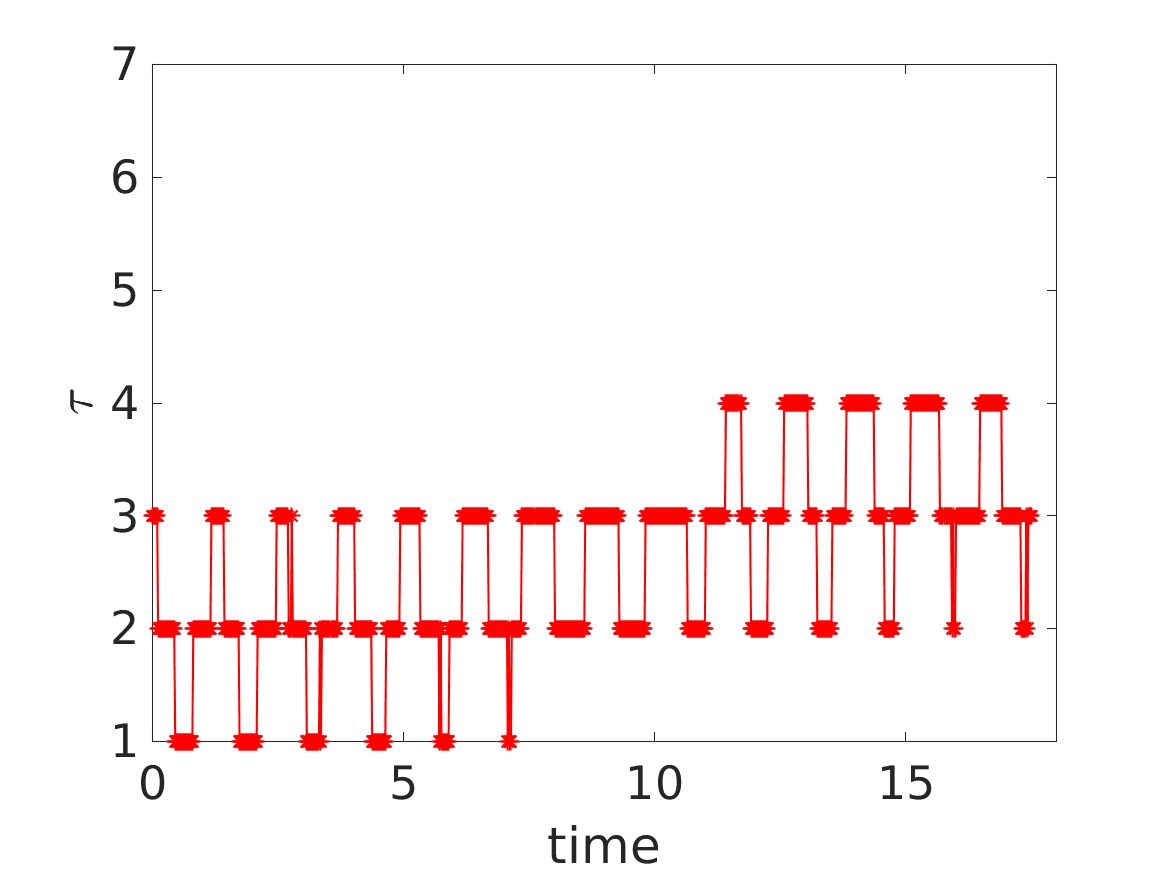}
}
\subfigure[$1024^2$, $P_c=10$]{
\includegraphics[width=0.3\columnwidth]{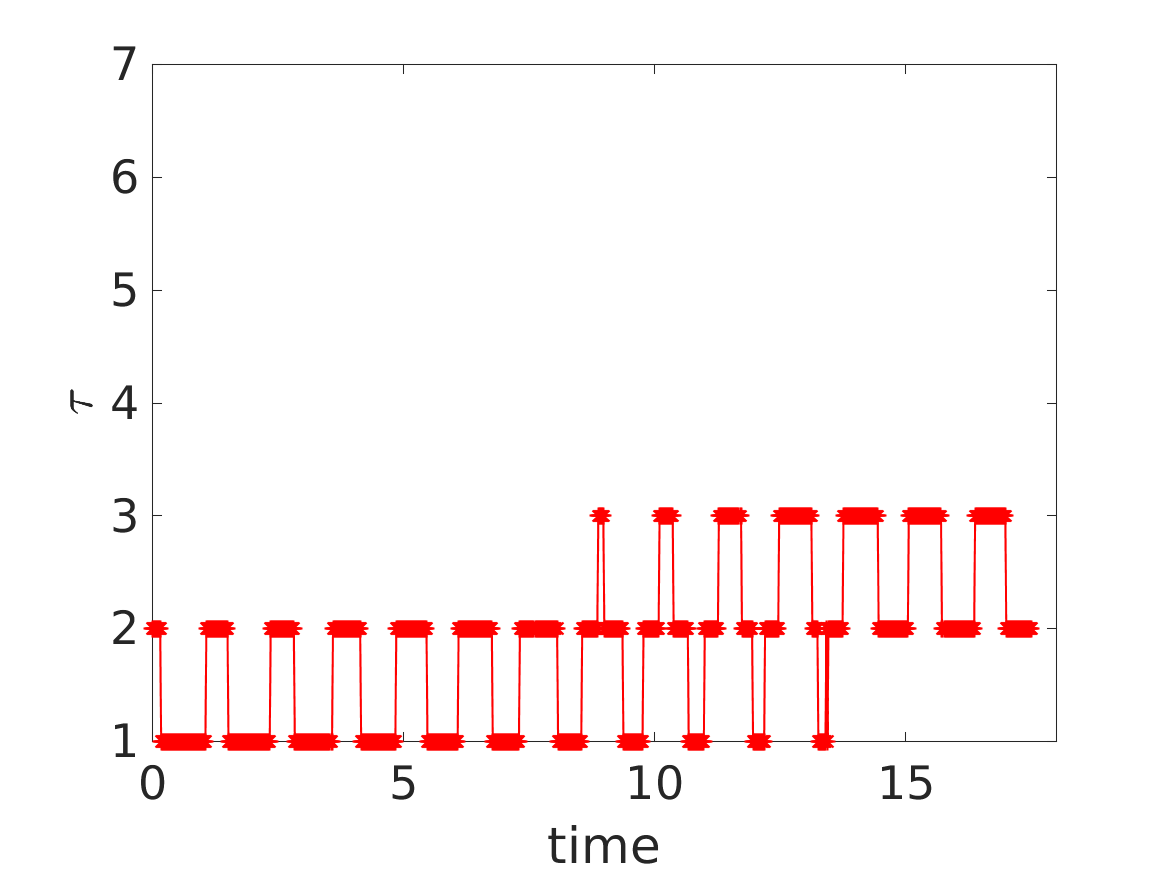}
}
\subfigure[$256^2$, $P_c=20$]{
\includegraphics[width=0.3\columnwidth]{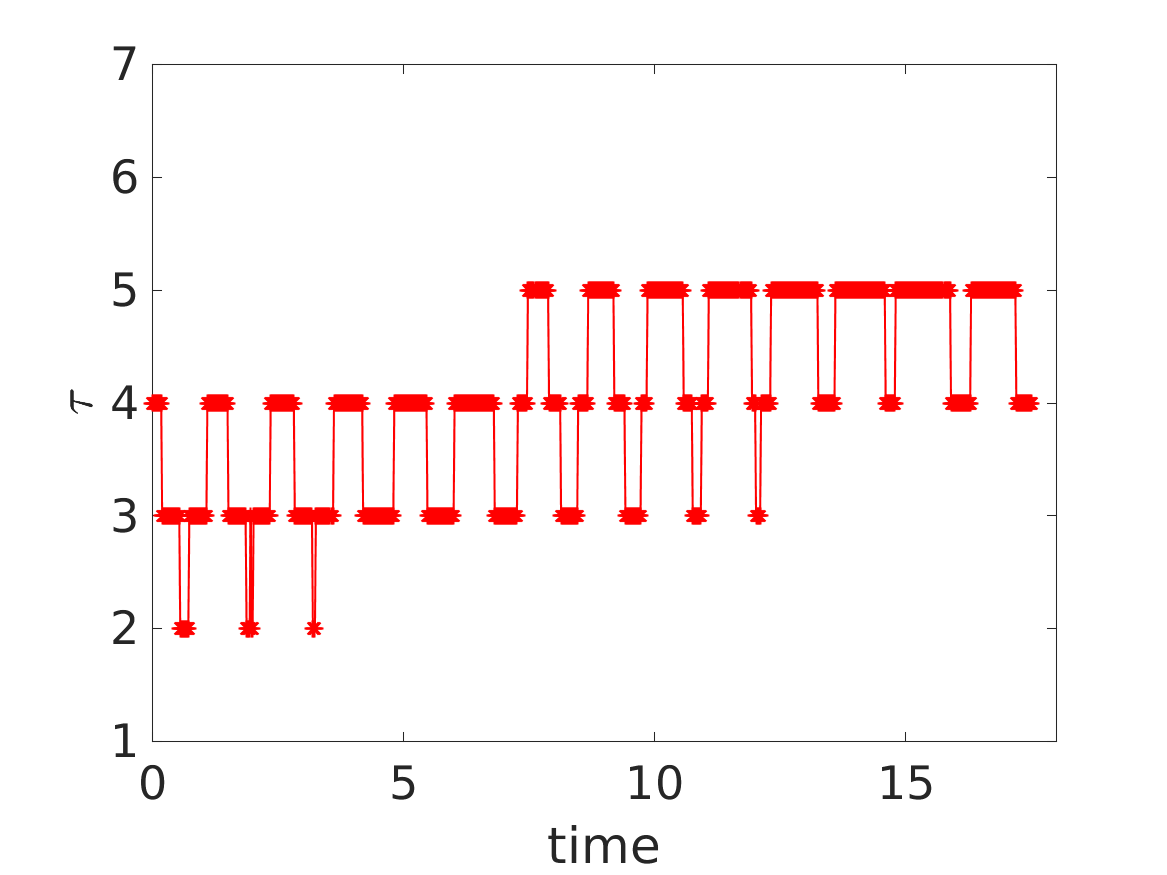}
}
\subfigure[$512^2$, $P_c=20$]{
\includegraphics[width=0.3\columnwidth]{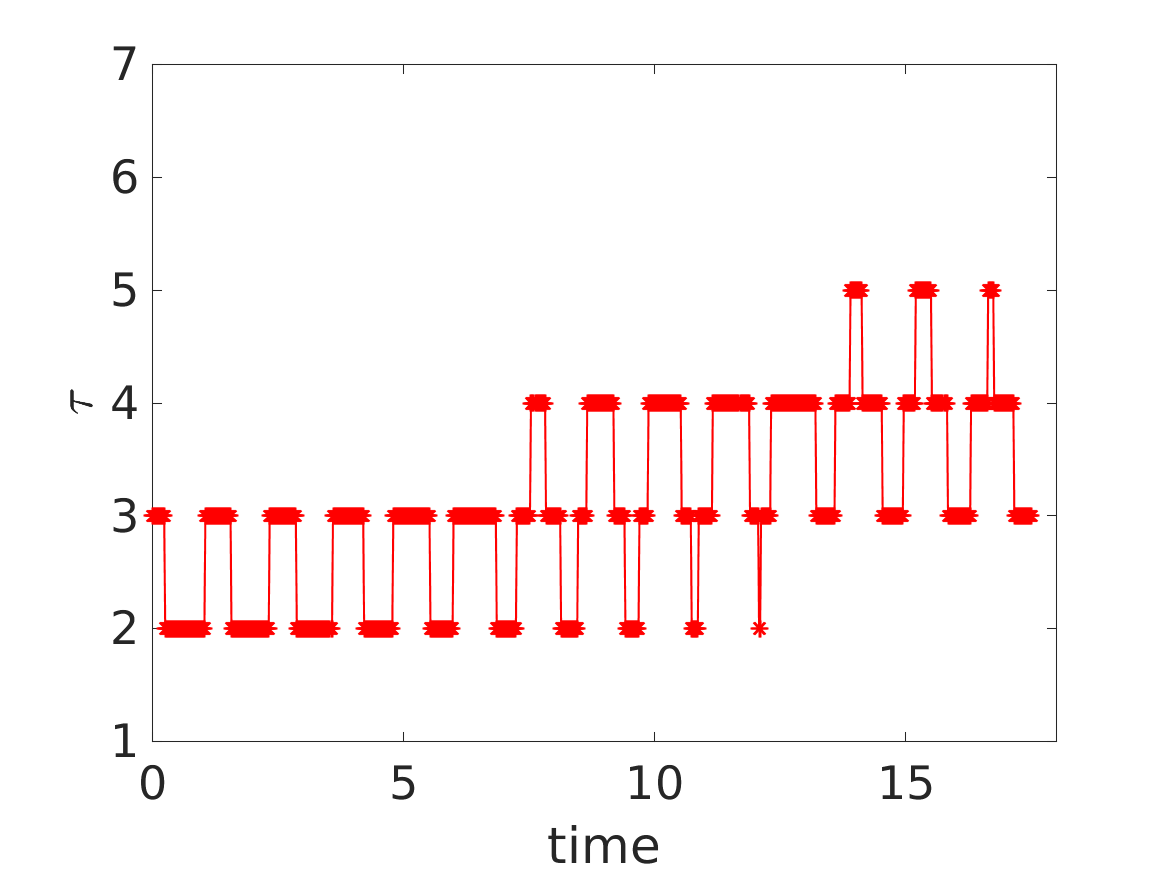}
}
\subfigure[$1024^2$, $P_c=20$]{
\includegraphics[width=0.3\columnwidth]{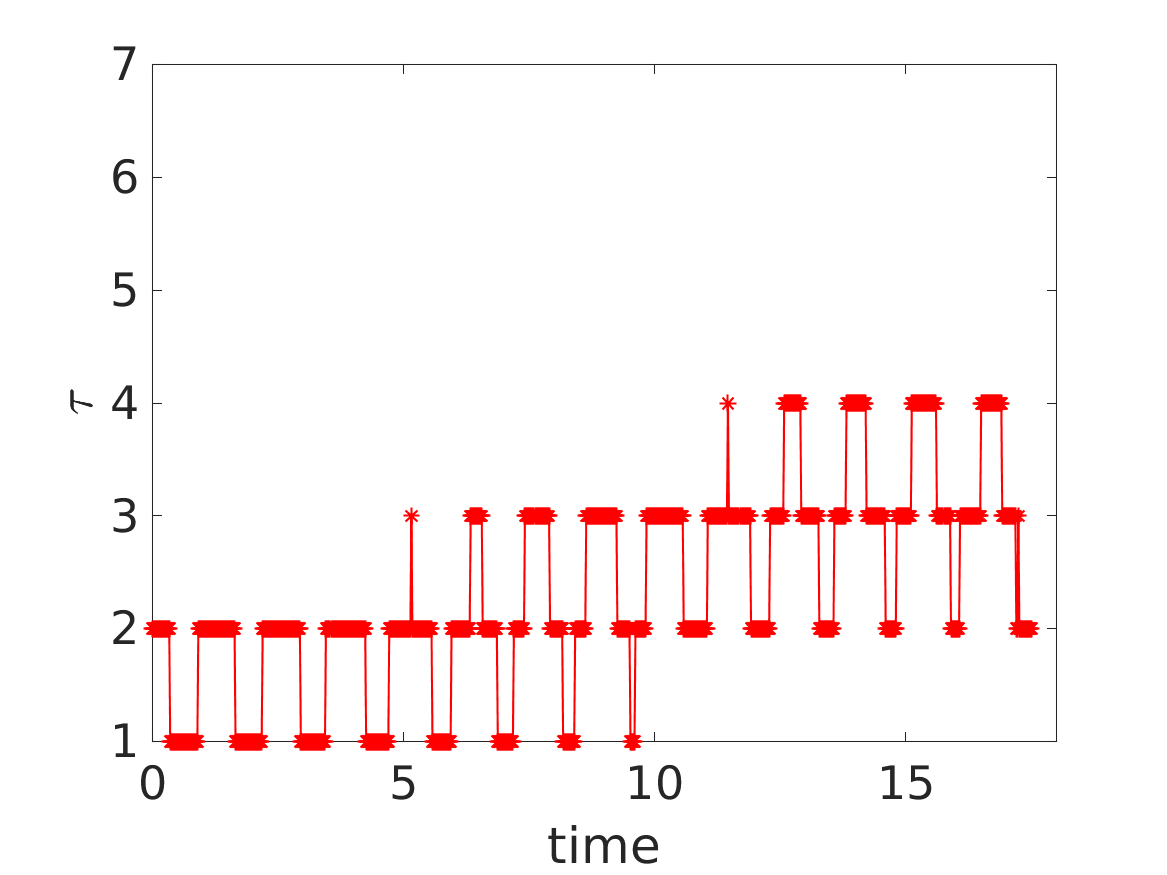}
}
    \caption{2D diocotron instability. Uniform sampling: Time history of $\tau$ for different mesh sizes and number of particles per cell $P_c$.}
\figlab{diocotron_tau_history_uniform}
\end{figure}

Having studied the adaptivity of the algorithm with respect to mesh size, $P_c$ and time (smoothness of the function) we will now
consider a different initial sampling technique and evaluate the performance. To that end, we sample the $f$ in equation 
\eqnref{dist_diocotron} with a uniform distribution in all the variables. The range for the velocity variables is chosen as $[-6,6]$ 
while for the configuration space it is $[0,L]$. Note that unlike the Gaussian sampling described earlier, with this sampling each particle 
will have a different constant charge $q_e$ \cite{aydemir1994unified} to match the distribution. But the charge to mass ratio is the same for all the 
particles. Similar to \cite{wang2011particle}, we ignore particles 
with weights less than $1.0\times10^{-9}$. This naive Monte-Carlo sampling strategy, as we will see in the results for this example, has higher noise levels than the Gaussian sampling. However it can be useful in scenarios where we do not know of an importance sampling technique to sample the distribution 
at hand.
  
Owing to higher noise levels in the sampling strategy, we needed $(P_c)_{ref}=5$ and $\alpha=0.03$ for the calculation of the denoising threshold. 
%Even though we performed experiments with all the mesh sizes and $P_c$ considered for the Gaussian sampling, here we will show the results only for $1024^2$ mesh to keep the length short. 
Figure \figref{diocotron_density_uniform} shows that in general, except for the coarsest mesh size $256^2$, the adaptive $\tau$ algorithm performs well in this case also in terms of picking a nearly optimal $\tau$ for most of the cases. The time history of $\tau$ in Figure \figref{diocotron_tau_history_uniform} shows that as expected, the optimal $\tau$ values for this
case are lower than that for the Gaussian sampling in Figure \figref{diocotron_tau_history_gauss} due to higher noise levels.% Similar performance benefits are also observed for other mesh sizes and $P_c$.

\begin{table}[h!b!t!]%[0.5\textwidth]
\centering
\begin{tabular}{|r|c|c|c||c|c|c|}
\hline
    & \multicolumn{3}{c||}{Gaussian sampling} & \multicolumn{3}{c|}{Uniform sampling}\\
\cline{2-7}
    \!\!\! Mesh \!\!\!\! & \multicolumn{3}{c||}{\!\!\scriptsize  $P_c$ \!\!} & \multicolumn{3}{c|}{\!\!\scriptsize  $P_c$ \!\!} \\
\hline
    & \scriptsize 5 & \scriptsize 10 & \scriptsize 20 & \scriptsize 5 & \scriptsize 10 & \scriptsize 20 \\
\cline{2-7}
    $256^2$ & 43.9 & 47.7 & 55.5 & 42.5 & 44.7 & 50.3 \\
    $512^2$ & 92.2 & 110 & 146 & 85 & 96 & 119 \\
    $1024^2$ & 435 & 501.7 & 654 & 394.7 & 438.5 & 544 \\
\hline
\end{tabular}
\caption{\label{tab:diocotron_runtime_adaptive}2D diocotron instability. Adaptive $\tau$ PIC: Total run time in seconds on 64 cores for different mesh sizes and number of particles per cell for the Gaussian and uniform sampling techniques.}
\end{table}

\begin{table}[h!b!t!]%[0.5\textwidth]
\centering
    \begin{tabular}{|r|c|c|c|c|c|c|}
\hline
\cline{2-7}
        \!\!\! Mesh \!\!\!\! &  \multicolumn{3}{c|}{\!\!\scriptsize  Regular PIC \!\!} & \multicolumn{3}{c|}{\!\!\scriptsize Reg/adaptive $\tau$\!\!}\\
\hline
\cline{2-7}
        $256^2$ & 51 ($20$) & 51 ($20$)  & 104.7 ($80$) & 1.2 & 1.1 & 1.9 \\
        $512^2$ & 201.7 ($40$) & 364.3 ($80$) & 709 ($160$) & 2.2 & 3.3 & 4.8 \\
        $1024^2$ & 1543.8 ($80$) & 2911.1 ($160$) &  5857 ($320$) & 3.5 & 5.8 & 8.9\\
\hline
\end{tabular}
    \caption{\label{tab:diocotron_runtime_gauss}2D diocotron instability. Gaussian sampling: Columns $2-4$ are the  total run time in seconds taken 
    by the regular PIC on 64 cores for different mesh sizes and number of particles per cell (within parentheses) to reach a comparable accuracy (based on visual norm from the left columns of Figures \figref{diocotron_pc5_gauss}-\figref{diocotron_pc20_gauss}) in charge density 
    that of the adaptive $\tau$ results in Table \ref{tab:diocotron_runtime_adaptive} at time $T=17.5$. Columns $5-7$ are the ratio of time 
    taken by regular PIC to the values in columns $2-4$ of Table \ref{tab:diocotron_runtime_adaptive} for adaptive $\tau$ PIC.}
\end{table}

\begin{table}[h!b!t!]%[0.5\textwidth]
\centering
    \begin{tabular}{|r|c|c|c|c|c|c|}
\hline
\cline{2-7}
    \!\!\! Mesh \!\!\!\! &  \multicolumn{3}{c|}{\!\!\scriptsize  Regular PIC \!\!} & \multicolumn{3}{c|}{\!\!\scriptsize  Reg/adaptive $\tau$ \!\!}\\
\hline
\cline{2-7}
        $256^2$ & 45.7 ($20$) & 56.1 ($41$)  & 80.1 ($82$) & 1.1 & 1.2 & 1.6 \\
        $512^2$ & 249.3 ($82$) & 462.9 ($165$) & 462.9 ($165$) & 2.9 & 4.8 & 3.9 \\
        $1024^2$ & 1926 ($165$) & 3665.4 ($330$) &  6964.3 ($660$) & 4.9 & 8.3 & 12.8\\
\hline
\end{tabular}
    \caption{\label{tab:diocotron_runtime_uniform}2D diocotron instability. Uniform sampling: Columns $2-4$ are the total run time in seconds taken by the regular PIC on 64 cores for different mesh sizes and number of particles per cell (within parentheses) to reach a comparable accuracy (based on visual norm from Figure \figref{diocotron_density_uniform}) in charge density that of the adaptive $\tau$ results in Table \ref{tab:diocotron_runtime_adaptive} at time $T=17.5$. Columns $5-7$ are the ratio of time 
    taken by regular PIC to the values in columns $5-7$ of Table \ref{tab:diocotron_runtime_adaptive} for adaptive $\tau$ PIC. The timing reported for $P_c=660$ and hence the speedup corresponding to it is a theoretically extrapolated one obtained by multiplying the timing for $P_c=330$ with a factor $1.9$ which is the ratio between timings for $P_c=330$ to $P_c=165$. 
 }
\end{table}

            Finally, we perform a preliminary run time performance study to see the effectiveness of the current approach in comparison to the regular PIC.
            To that extent, we note that we did not perform any optimization to both the regular PIC as well as the adaptive $\tau$ PIC routines. Optimization of different components involved in the algorithm as well as a thorough parallel performance study is left for future work. In Table \ref{tab:diocotron_runtime_adaptive} the total run time in seconds is shown for the adaptive $\tau$
        PIC on 64 cores for the mesh sizes, $P_c$ and sampling techniques considered before. All the timings reported are the average of three runs performed. 
%The sparse routines by themselves for all the cases took approximately less than or equal to 15 percent of the total run time. Among the sparse routines the tauEstimator accounted for approximately $6.1-9.2$ percent and the transferToSparse routine $3.7-5.5$ percent of the total run time. 
In Tables \ref{tab:diocotron_runtime_gauss} and \ref{tab:diocotron_runtime_uniform}, we compare the adaptive $\tau$ PIC timings with the
        timings for the regular PIC with the $P_c$ value required to reach a comparable accuracy in charge density as that of the adaptive $\tau$ results at final time $T=17.5$. The approximate $P_c$ values within parentheses are obtained from Figures \figref{diocotron_pc5_gauss}-\figref{diocotron_pc20_gauss} and Figure \figref{diocotron_density_uniform} based on visual examination. %For regular PIC with $P_c=660$, we extrapolated the timings from $P_c=330$ by multiplying with a factor
%        $1.87$ which is the ratio between timings for $P_c=330$ to $P_c=165$. 
Even in this preliminary performance study, we can see that the adaptive $\tau$ strategy can provide significant speedups up to an order of magnitude compared to the regular PIC for similar accuracy in charge density. In terms of memory storage, the benefits are even
        more pronounced. Using the number of particles $N_p$ as a measure of the dominant memory cost (for PIC methods this is usually the case) we see $\approx2-16$ times memory reduction (in case of Gaussian sampling) and $\approx4-33$ times memory reduction (in case of uniform sampling) with adaptive $\tau$ PIC compared to regular PIC. %to reach a similar accuracy level in $\rho$.

\subsection{3D Penning trap}

    In this section we will consider a 3D Penning trap problem as the test case. Penning traps are storage devices for charged particles, which uses a quadrupole electric field to confine the particles axially and a
    homogeneous axial magnetic field to confine the particles in the
    radial direction \cite{brown1986geonium,blaum2010penning}. The evolution
    of the density in this problem (see Figure \figref{penning_density}) is very similar to that observed in cyclotrons \cite{adam1995space,yang2010beam}. Thus this test case is very relevant to our ultimate goal of high precision large-scale simulation of cyclotrons. The fine scale structures developed in this problem pose challenges for the sparse
    grids similar to the diocotron case in the previous section.

    The parameters for this test case are as follows. The length of the periodic box is $L=20$. The external magnetic field is given by $\B_{ext}=\LRc{0,0,5}$ and the quadrupole external electric field by 
    \[
        \Eb_{ext} = \LRc{-\frac{15}{L}\LRp{x-\frac{L}{2}},-\frac{15}{L}\LRp{y-\frac{L}{2}},\frac{30}{L}\LRp{z-\frac{L}{2}}}.
    \]

    For the initial conditions, we sample the phase space using a Gaussian
    distribution in all the variables. The mean and standard deviation for
    all the velocity variables is $0$ and $1$ respectively. While the mean
    for all the configuration space variables is $L/2$ the standard 
    deviations are $0.15L$, $0.05L$ and $0.2L$ for $x$, $y$ and $z$ respectively. The total electron charge is $Q_e=-1562.5$, and the charge of each 
    particle is $q_e=\frac{Q_e}{N_p}$. 

    The denoising parameters are taken as $(P_c)_{ref}=1$ and $\alpha=0.005$ for this problem with the above mentioned sampling. % This means with $5$ particles per cell, charge densities with Fourier amplitude less than 1 percent of the maximum amplitude will be set to 0. 
     The time step is chosen as $\Delta t=0.05$ and the simulations are run till final time $T=15$.

\begin{figure}[h!t!b!]
\subfigure[time=0]{
\includegraphics[trim=0cm 7cm 15cm 0cm,clip=true,width=0.3\columnwidth]{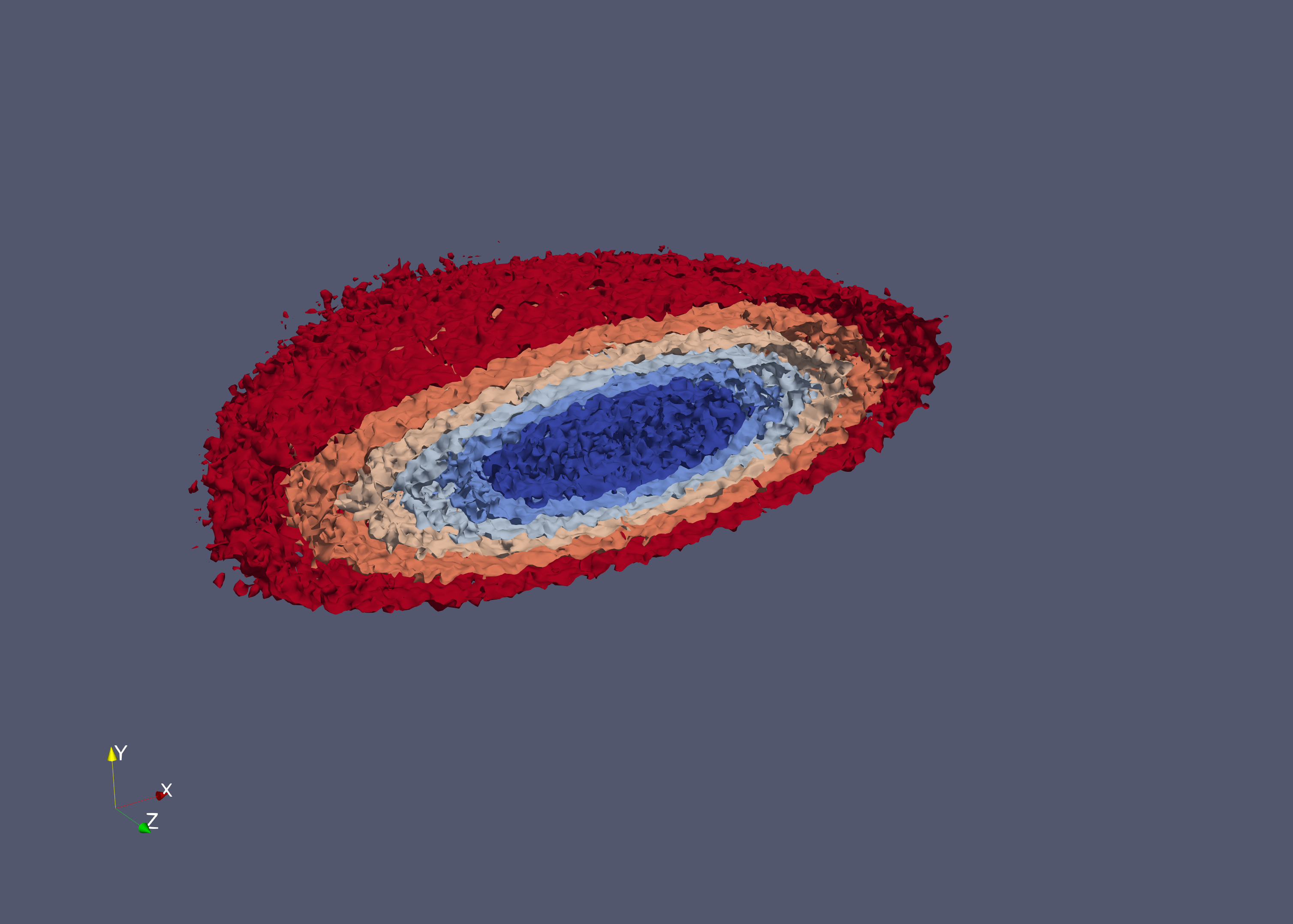}
}
\subfigure[time=10]{
\includegraphics[trim=0cm 7cm 15cm 0cm,clip=true,width=0.3\columnwidth]{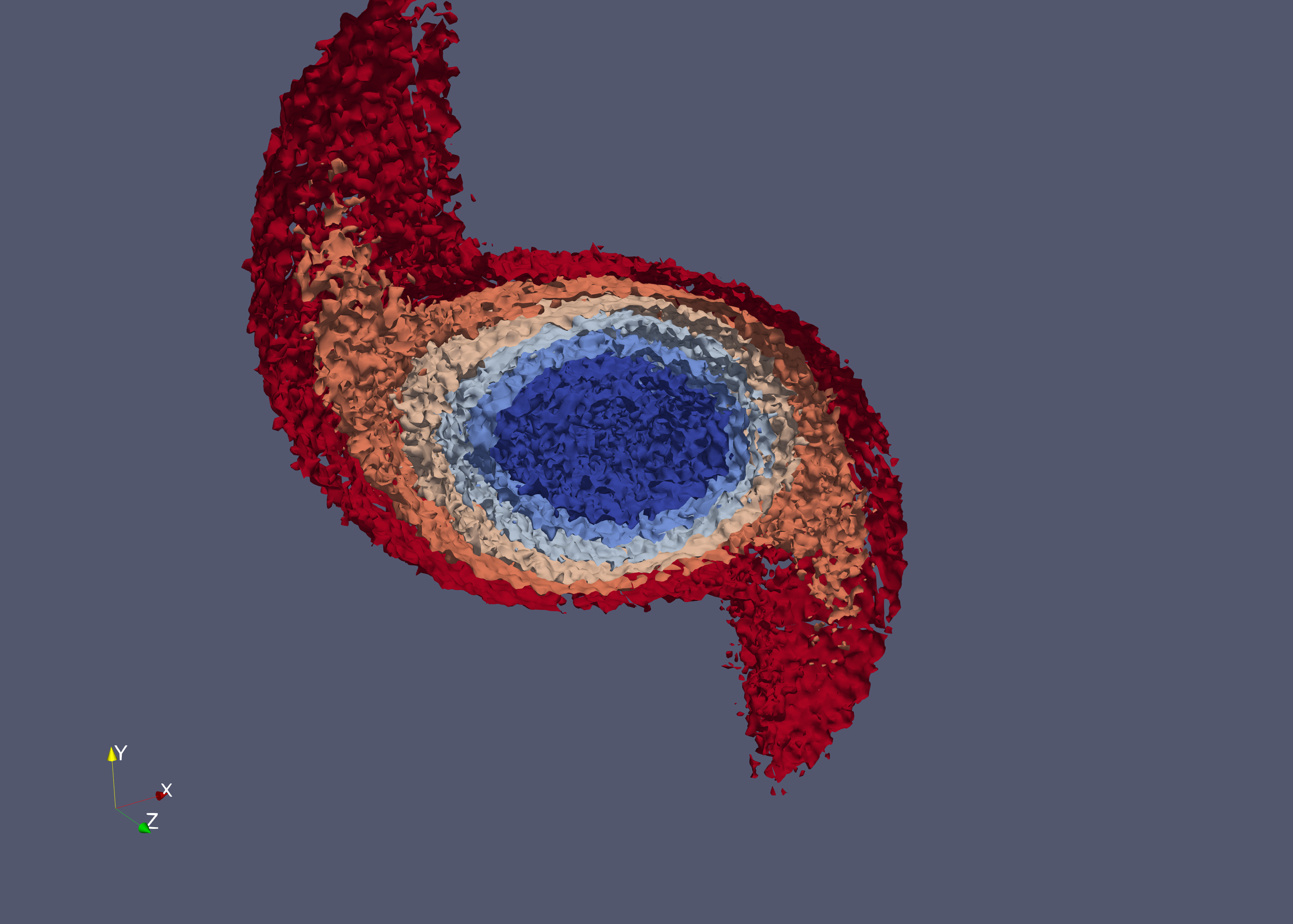}
}
\subfigure[time=15]{
\includegraphics[trim=0cm 7cm 15cm 0cm,clip=true,width=0.3\columnwidth]{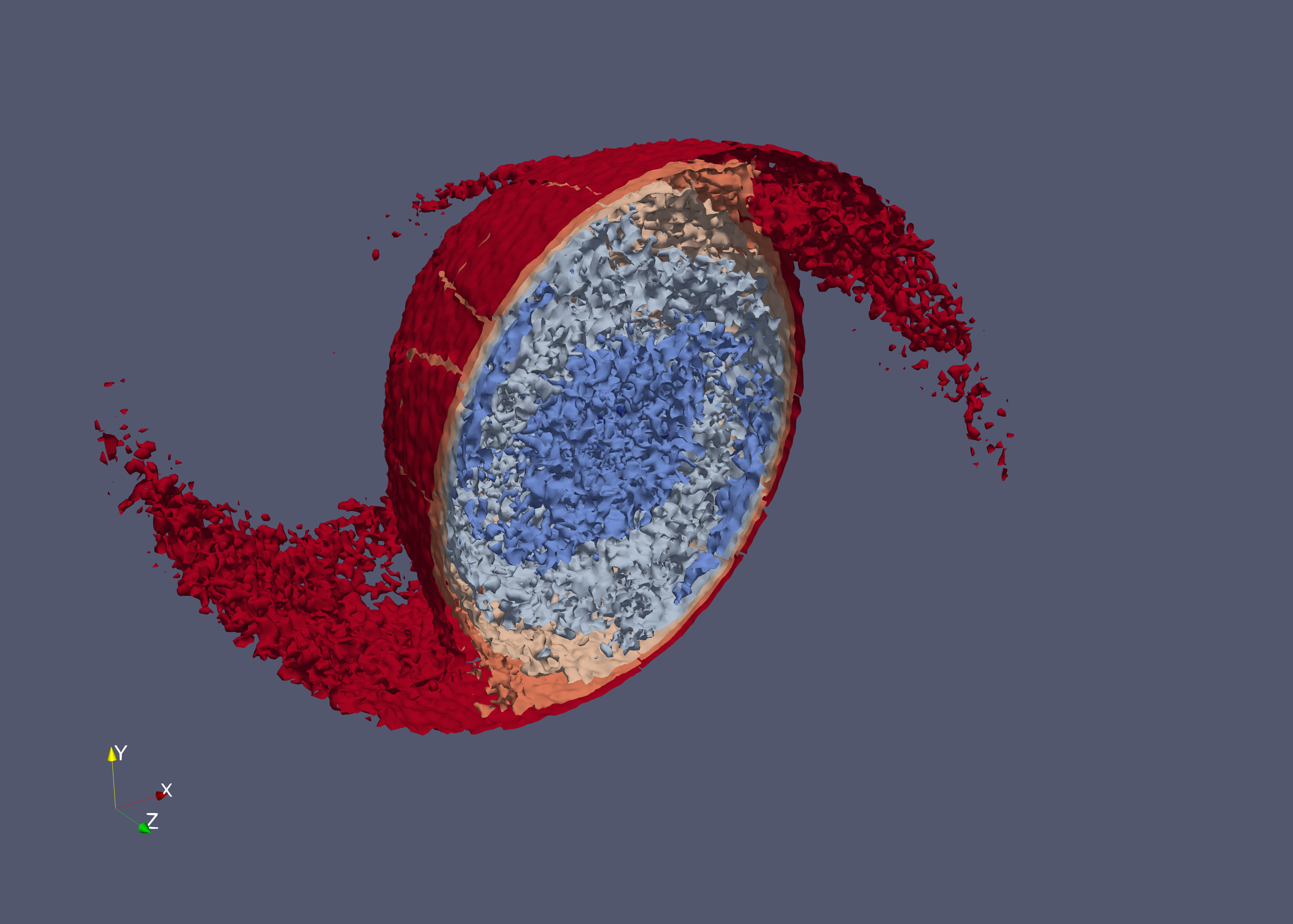}
}
\subfigure[time=0]{
\includegraphics[trim=0cm 7cm 15cm 0cm,clip=true,width=0.3\columnwidth]{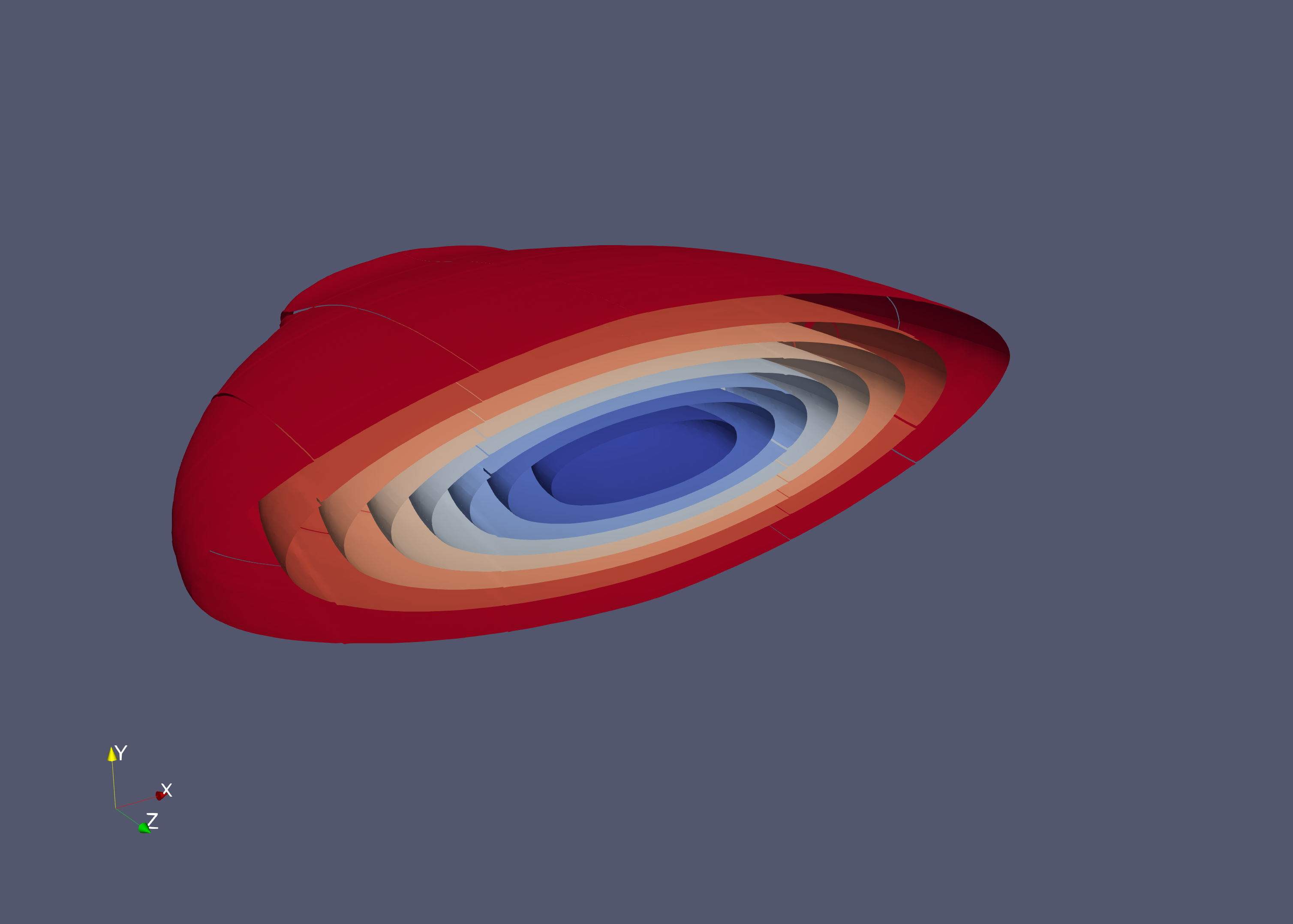}
\figlab{time0_penning_tau1_density}
}
\subfigure[time=10]{
\includegraphics[trim=0cm 7cm 15cm 0cm,clip=true,width=0.3\columnwidth]{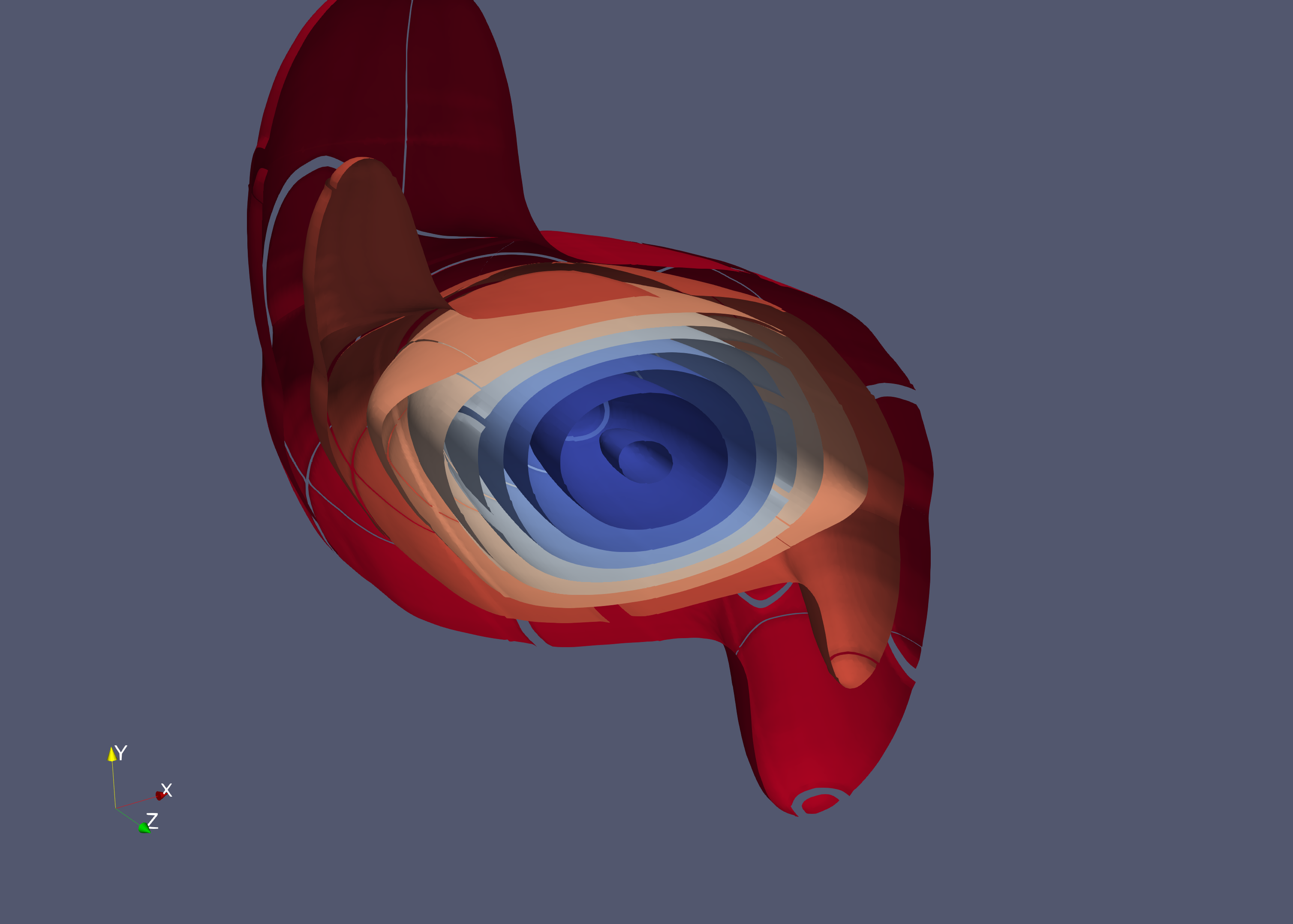}
\figlab{time200_penning_tau1_density}
}
\subfigure[time=15]{
\includegraphics[trim=0cm 7cm 15cm 0cm,clip=true,width=0.3\columnwidth]{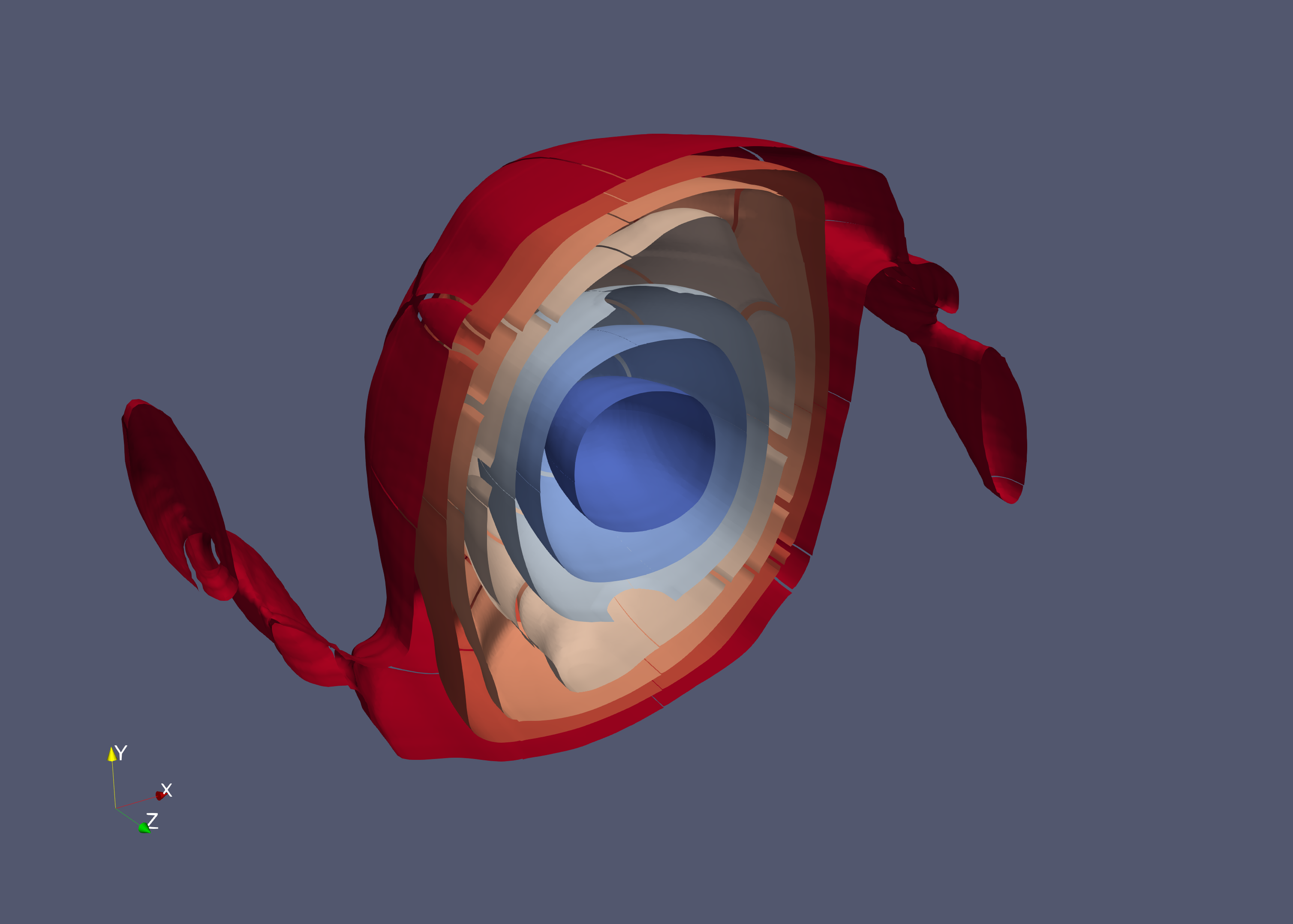}
\figlab{time300_penning_tau1_density}
}
\subfigure[time=0]{
\includegraphics[trim=0cm 7cm 15cm 0cm,clip=true,width=0.3\columnwidth]{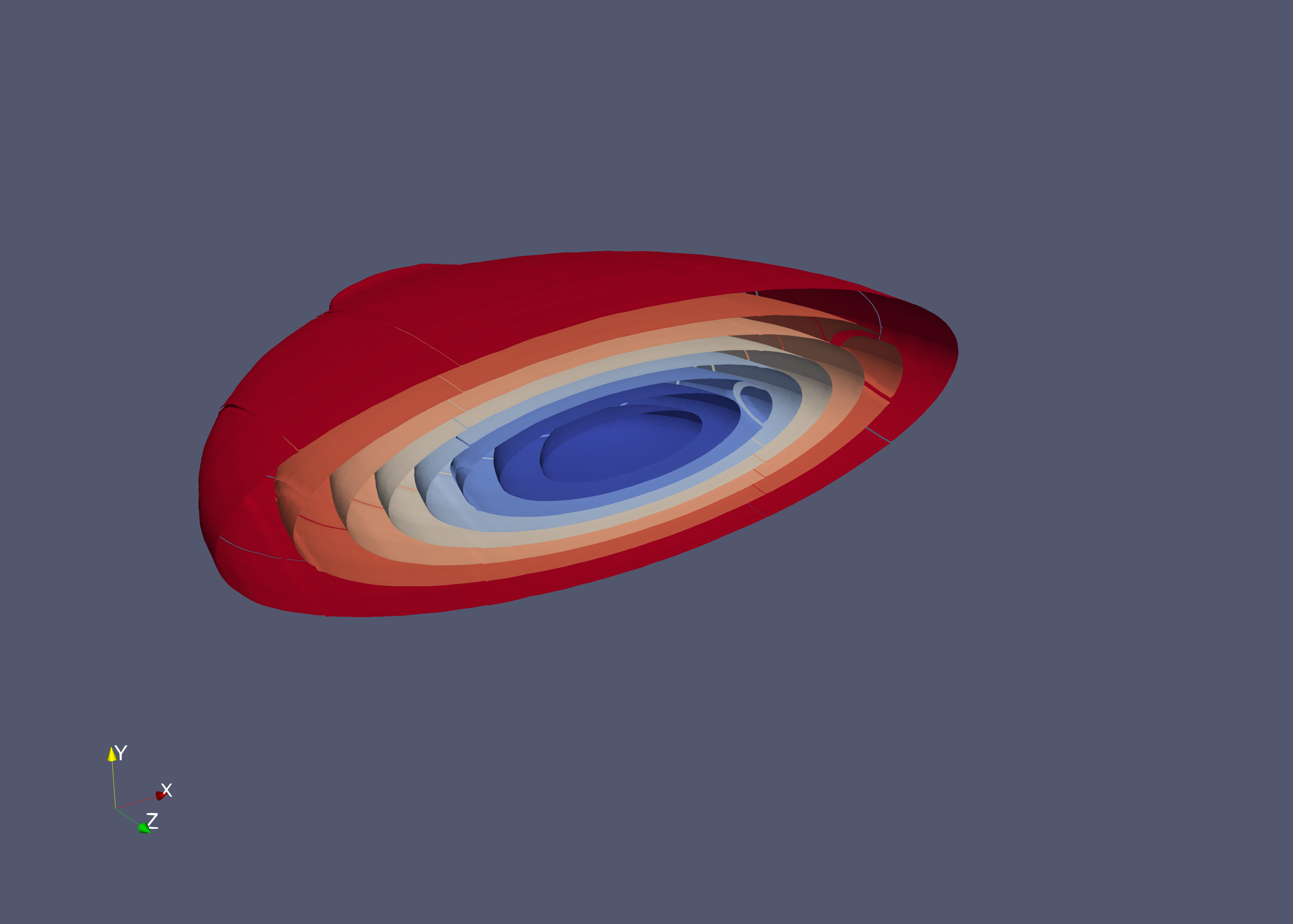}
}
\subfigure[time=10]{
\includegraphics[trim=0cm 7cm 15cm 0cm,clip=true,width=0.3\columnwidth]{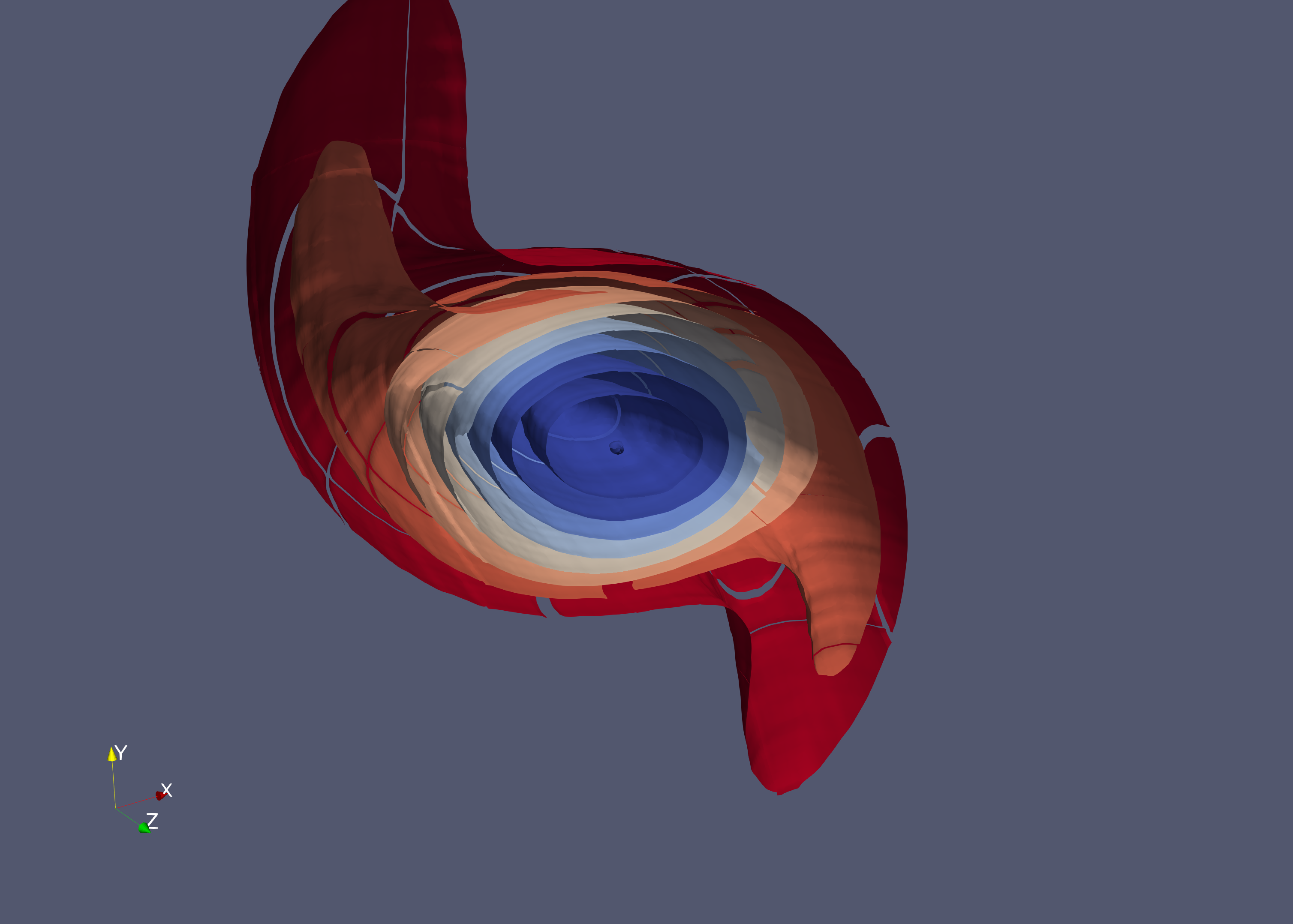}
}
\subfigure[time=15]{
\includegraphics[trim=0cm 7cm 15cm 0cm,clip=true,width=0.3\columnwidth]{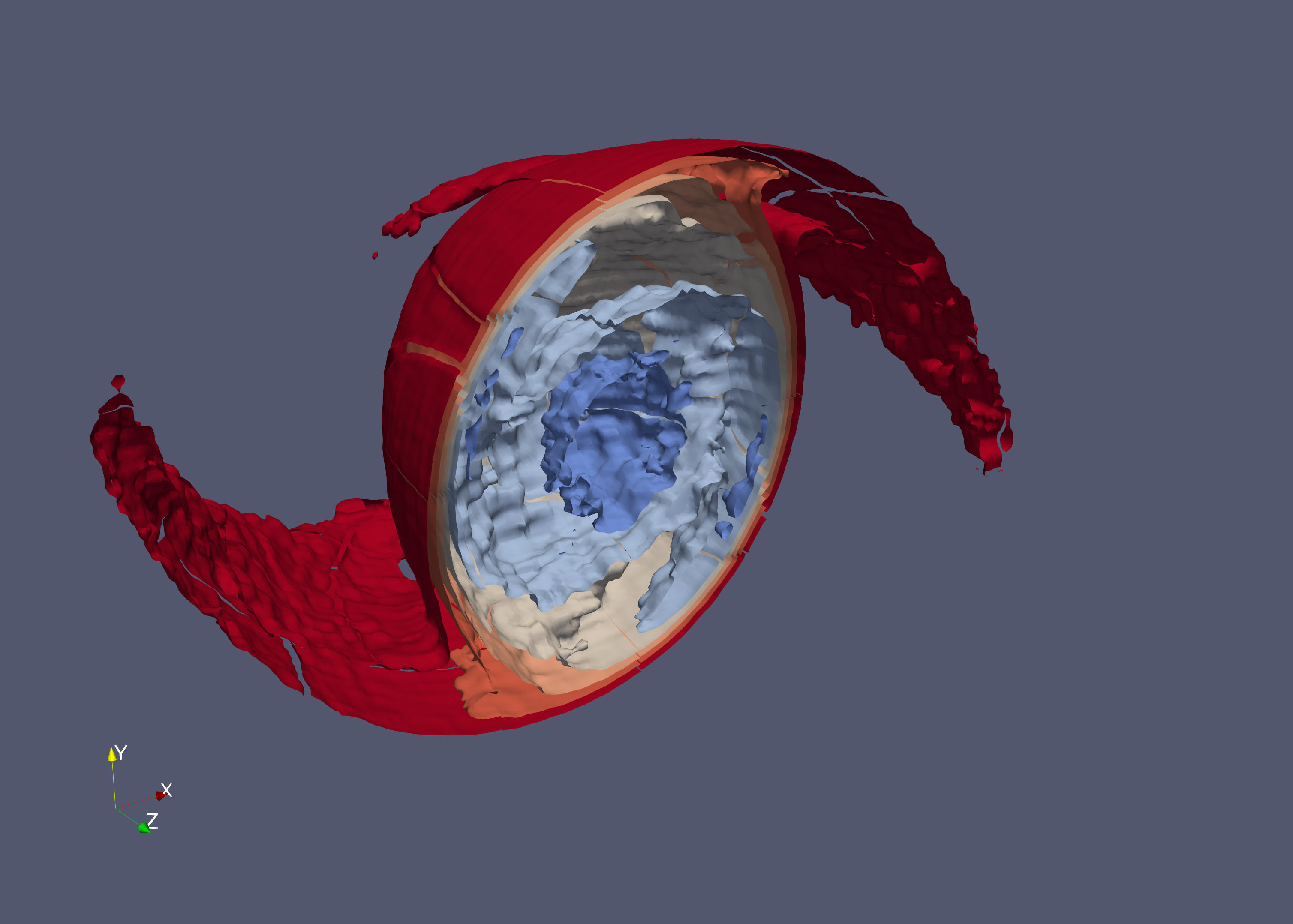}
}
\subfigure[time=0]{
\includegraphics[trim=0cm 7cm 15cm 0cm,clip=true,width=0.3\columnwidth]{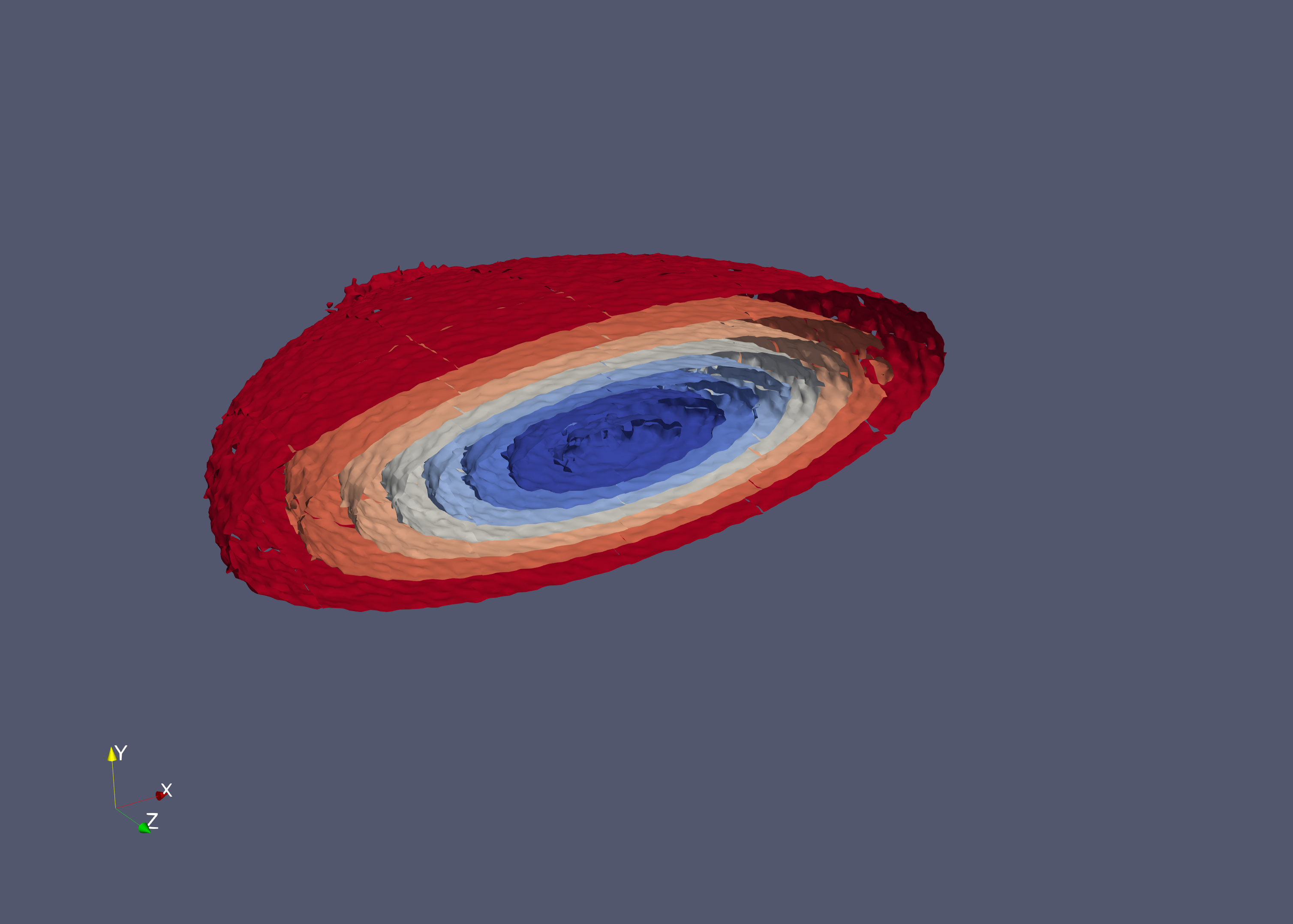}
}
\subfigure[time=10]{
\includegraphics[trim=0cm 7cm 15cm 0cm,clip=true,width=0.3\columnwidth]{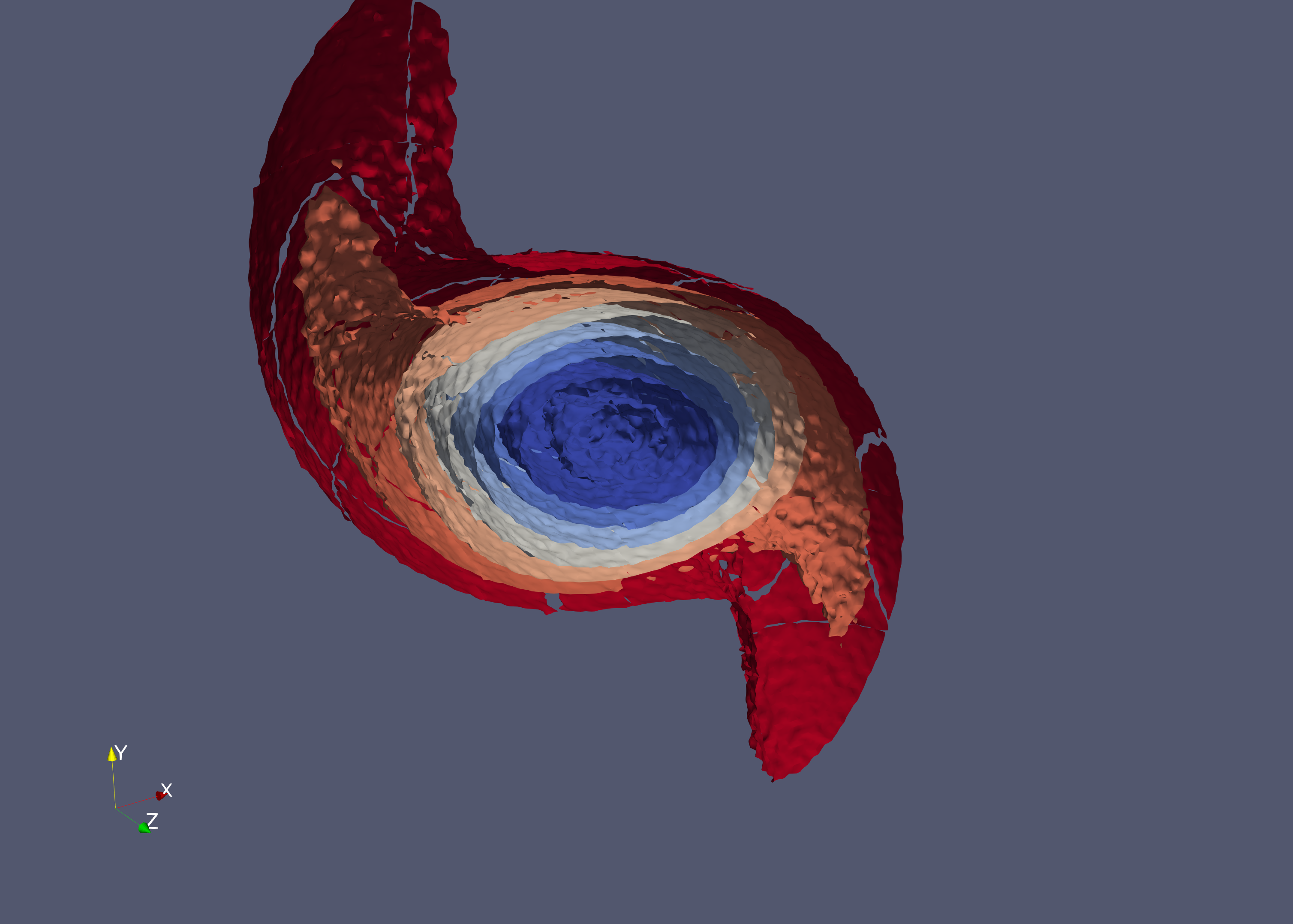}
}
\subfigure[time=15]{
\includegraphics[trim=0cm 7cm 15cm 0cm,clip=true,width=0.3\columnwidth]{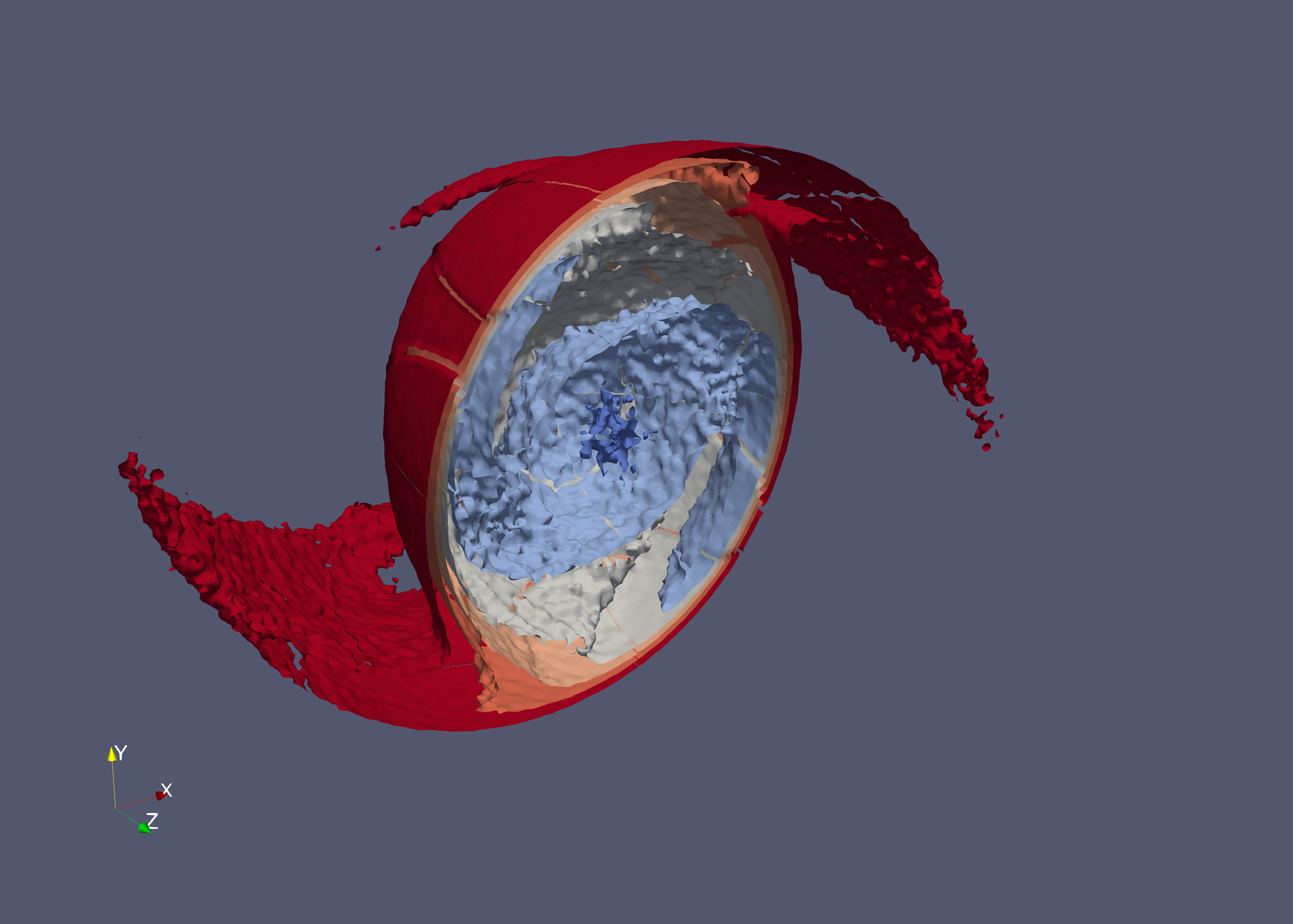}
}
\caption{3D Penning trap: Evolution of the electron charge density with time for regular PIC, $P_c=1$ (first row); $\tau=1$, $P_c=1$ (second row); adaptive $\tau$, $P_c=1$ (third row); and regular PIC, $P_c=20$ (fourth row). The mesh considered here is $256^3$.}
\figlab{penning_density}
\end{figure}

Figure \figref{penning_density} shows the evolution of the electron charge density with time
for regular, $\tau=1$ and adaptive $\tau$ PIC. The mesh used is $256^3$ and
$P_c=1$ for the first three rows and $20$ for the last row. As we had seen in Figure \figref{diocotron_density} for the diocotron test case, the adaptive $\tau$ results,
in the third row are better than both the regular PIC and $\tau=1$ results and are comparable to the results of the regular PIC with higher $P_c$ in the last row.

\begin{figure}[h!t!b!]
\subfigure[$64^3$, $P_c=1$]{
\includegraphics[width=0.5\columnwidth]{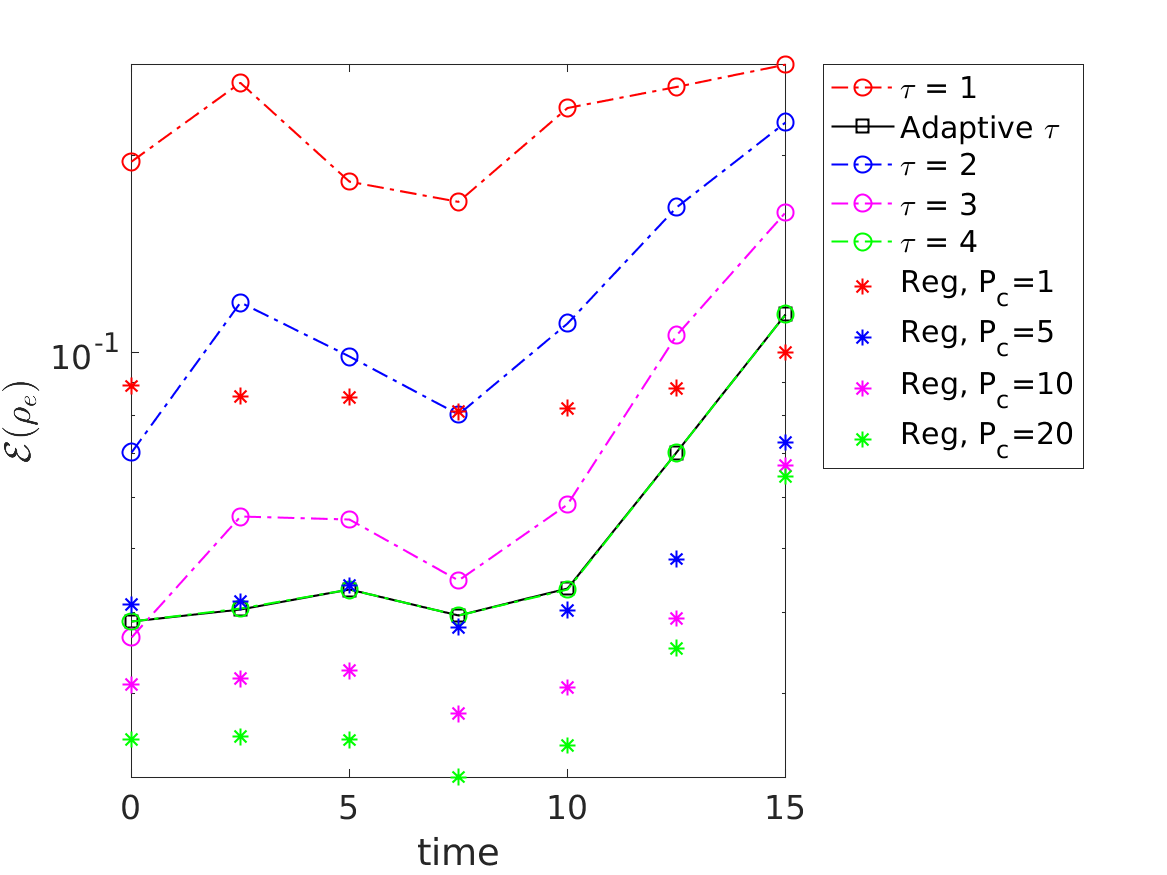}
}
\subfigure[$64^3$, $P_c=1$]{
\includegraphics[width=0.5\columnwidth]{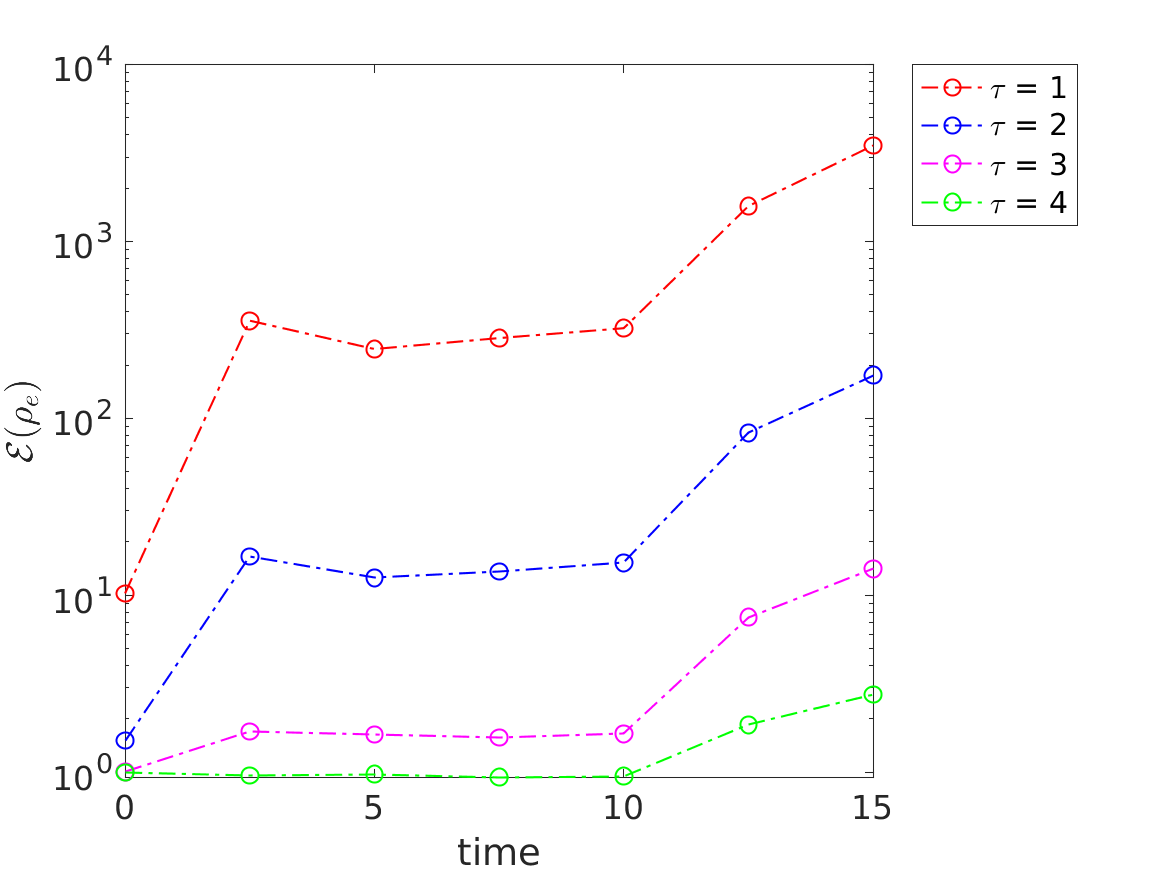}
}
\subfigure[$128^3$, $P_c=1$]{
\includegraphics[width=0.5\columnwidth]{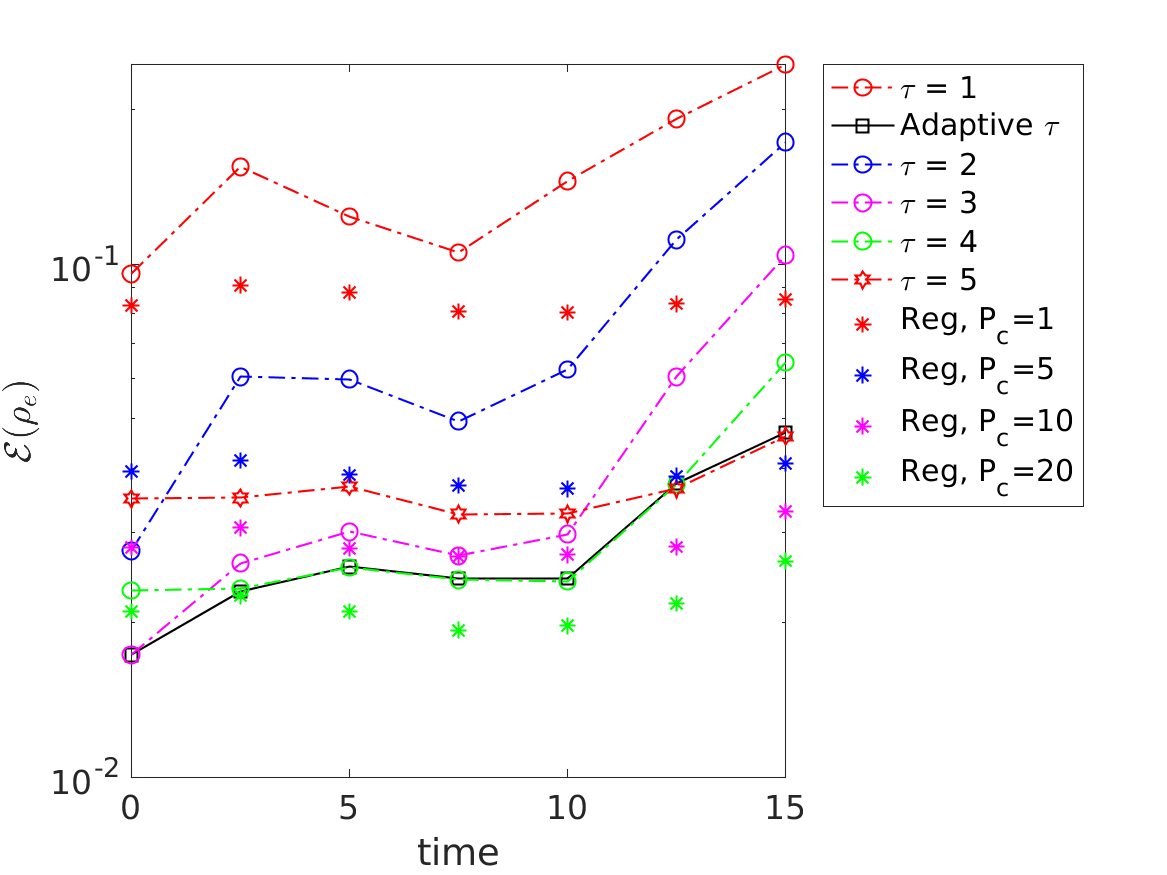}
}
\subfigure[$128^3$, $P_c=1$]{
\includegraphics[width=0.5\columnwidth]{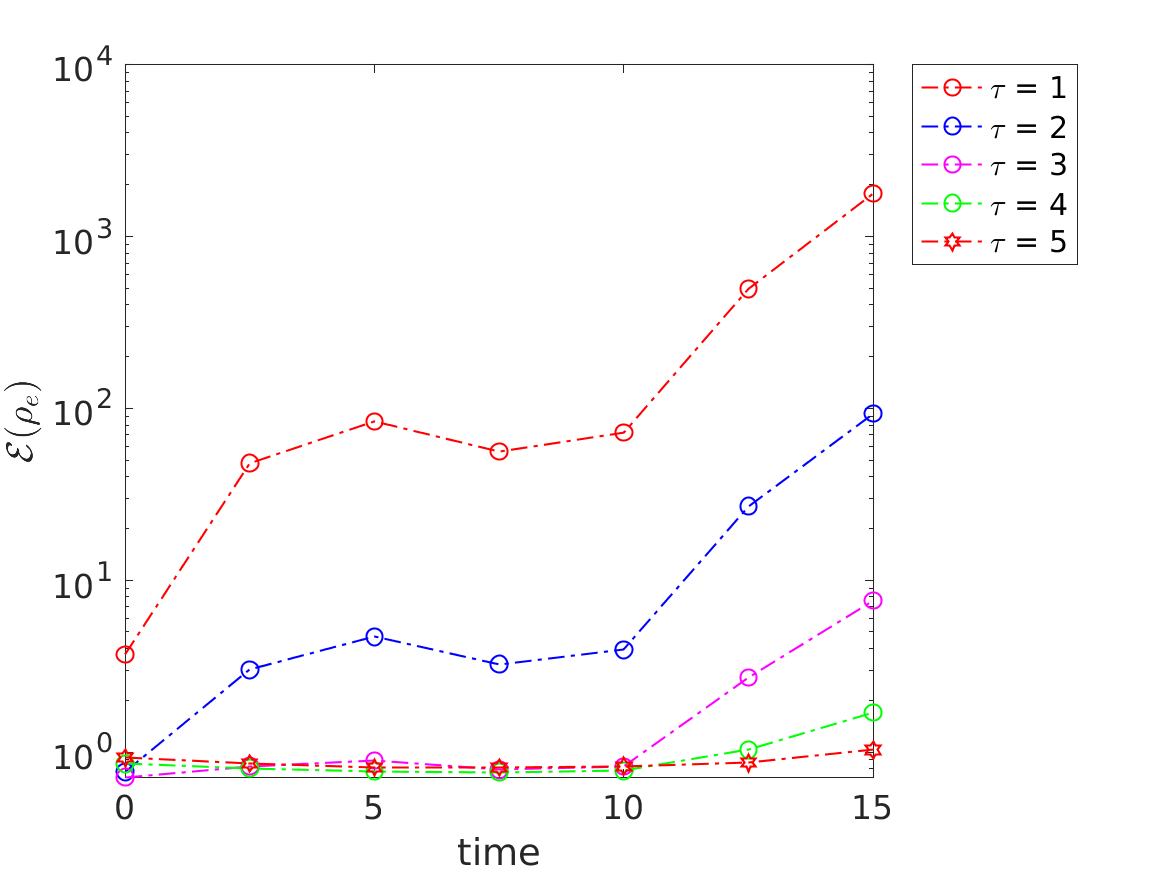}
}
\subfigure[$256^3$, $P_c=1$]{
\includegraphics[width=0.5\columnwidth]{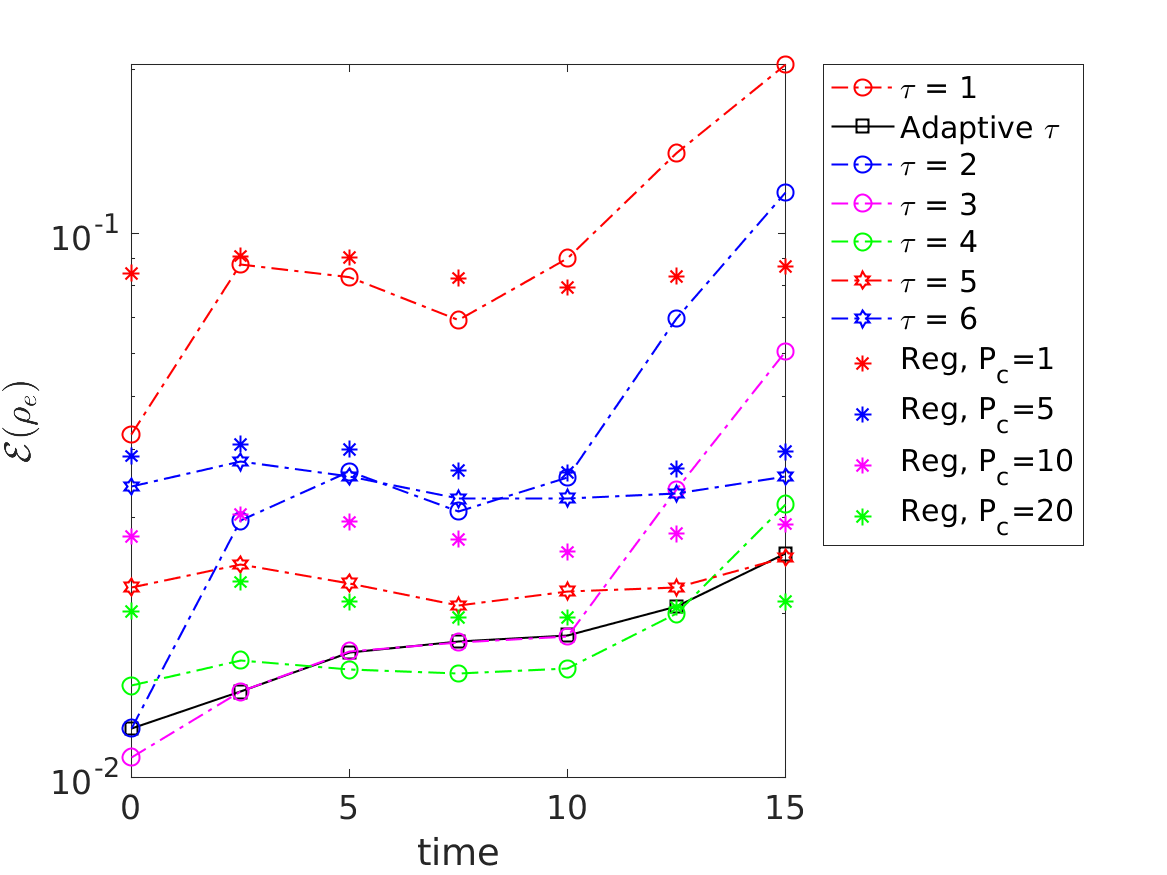}
}
\subfigure[$256^3$, $P_c=1$]{
\includegraphics[width=0.5\columnwidth]{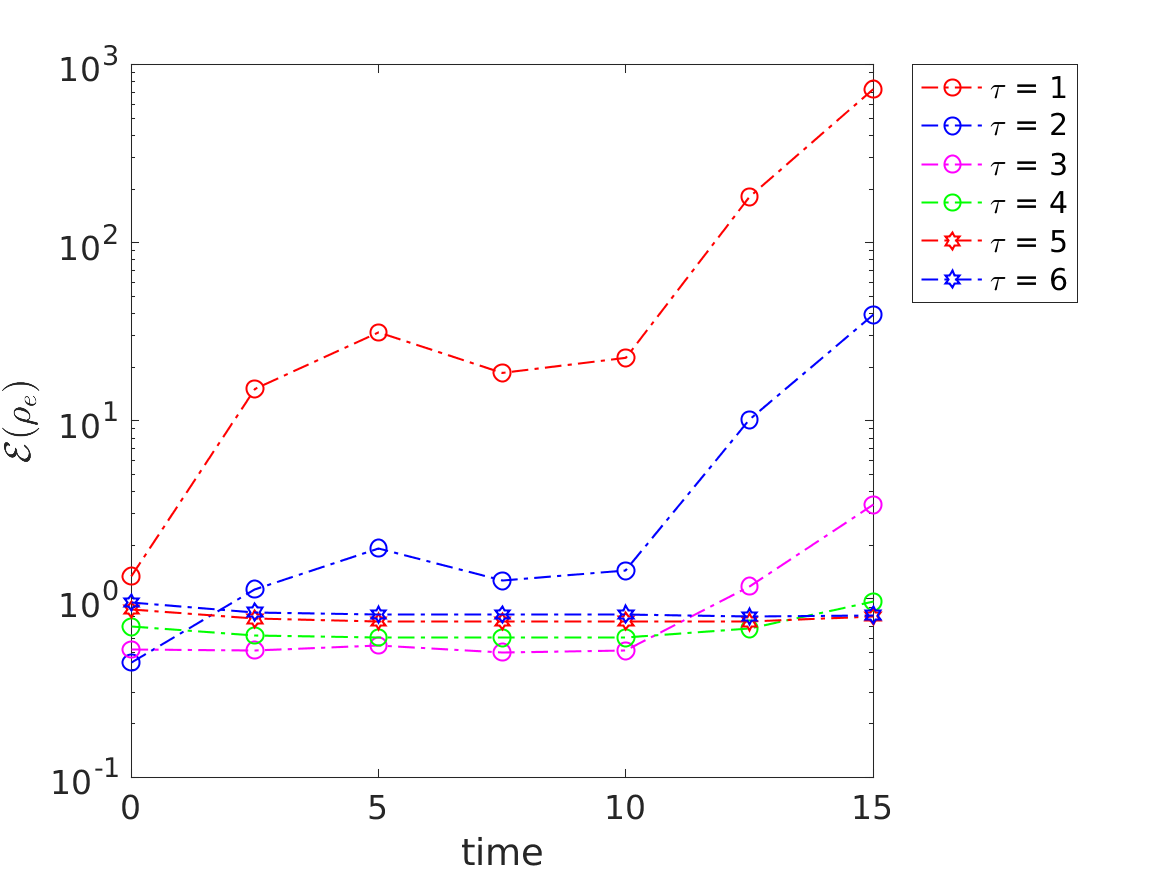}
}
    \caption{3D Penning trap: Electron charge density error comparison between regular (Reg), fixed $\tau$ and adaptive $\tau$ PIC. The left column is the actual error calculated using equation \eqnref{error_def} and the right column is the estimations from the $\tau$ estimator based on which the optimal $\tau$ is selected. The fixed as well as adaptive $\tau$ has the number of particles per cell $P_c=1$.}
\figlab{penning_pc1}
\end{figure}

\begin{figure}[h!t!b!]
\subfigure[$64^3$, $P_c=5$]{
\includegraphics[width=0.5\columnwidth]{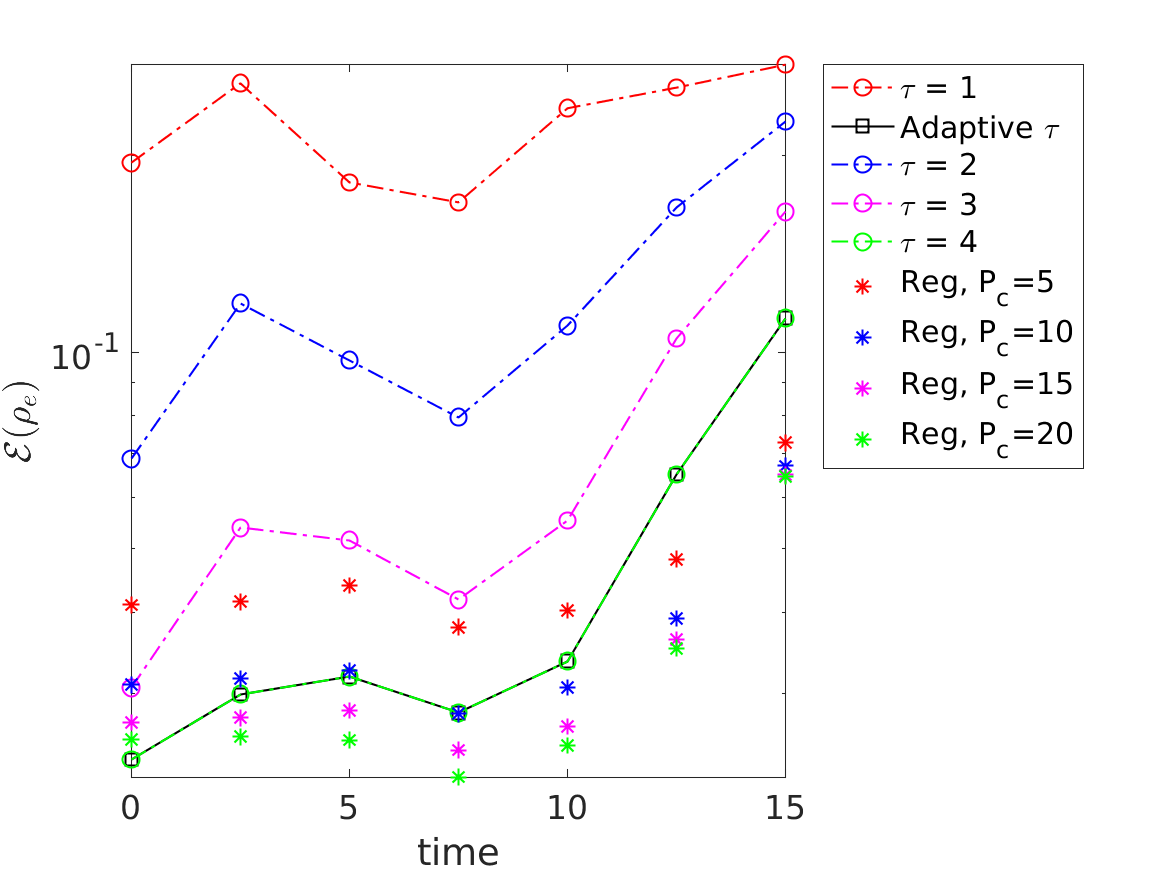}
}
\subfigure[$64^3$, $P_c=5$]{
\includegraphics[width=0.5\columnwidth]{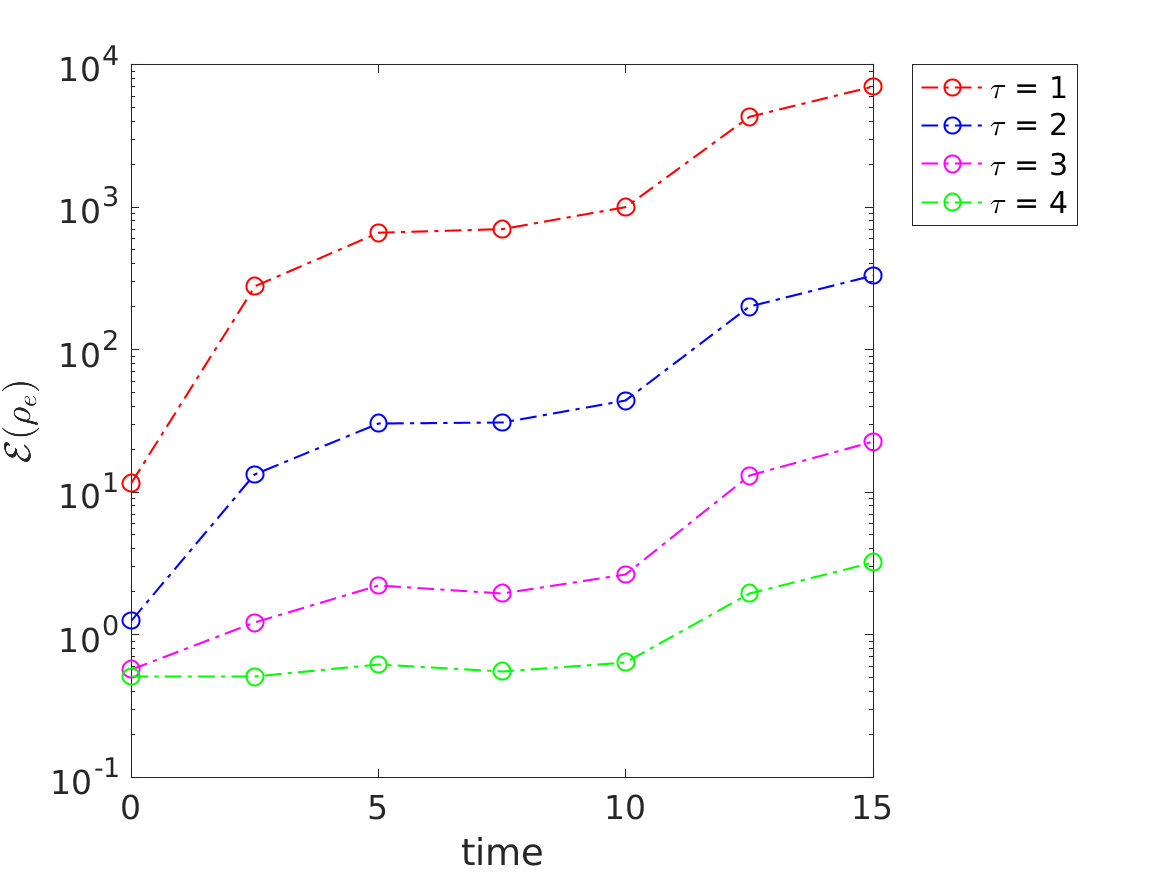}
}
\subfigure[$128^3$, $P_c=5$]{
\includegraphics[width=0.5\columnwidth]{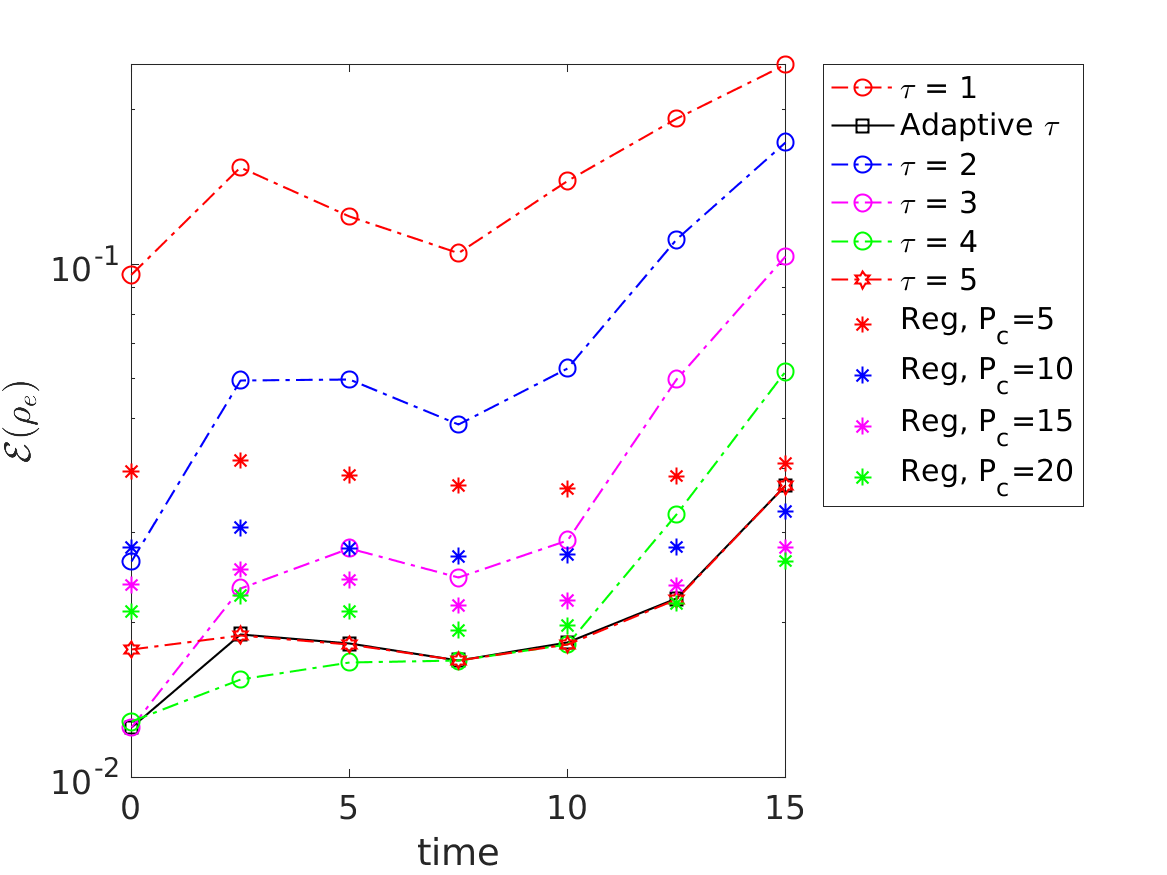}
}
\subfigure[$128^3$, $P_c=5$]{
\includegraphics[width=0.5\columnwidth]{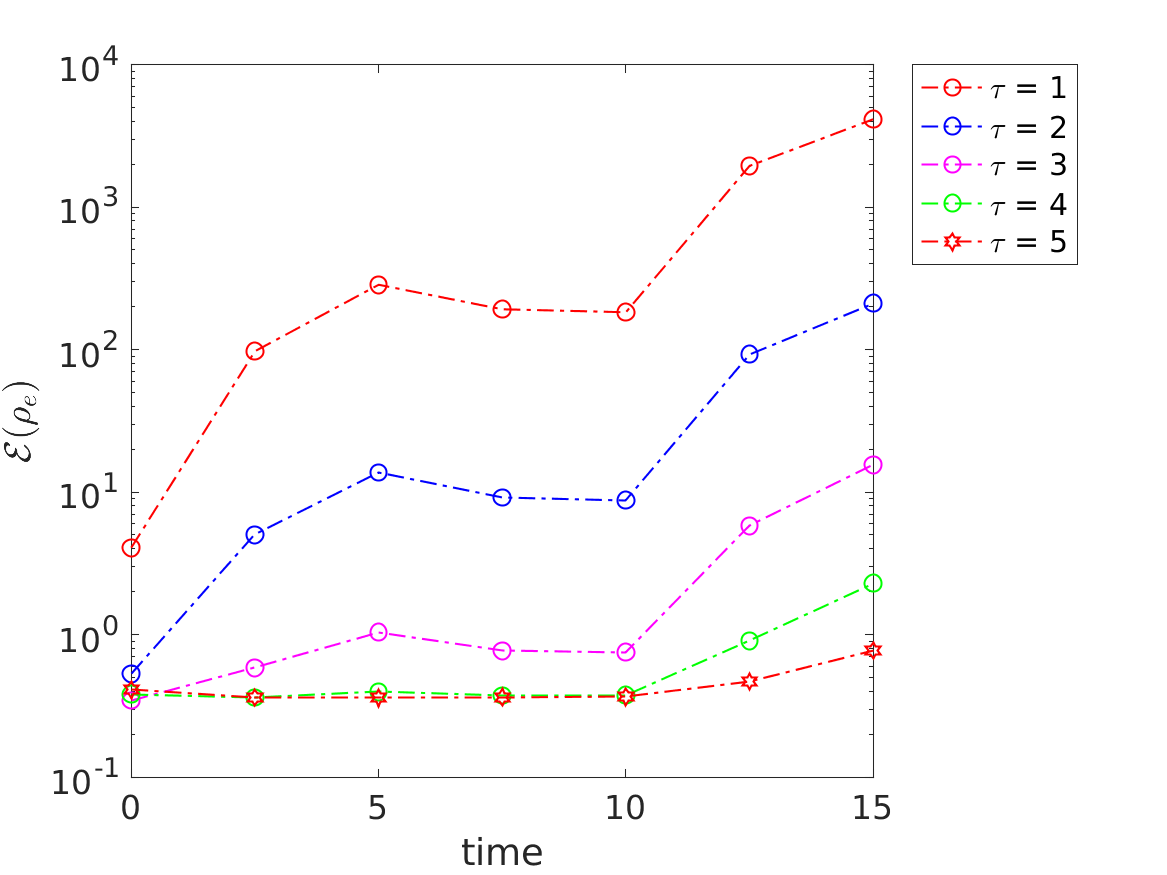}
}
\subfigure[$256^3$, $P_c=5$]{
\includegraphics[width=0.5\columnwidth]{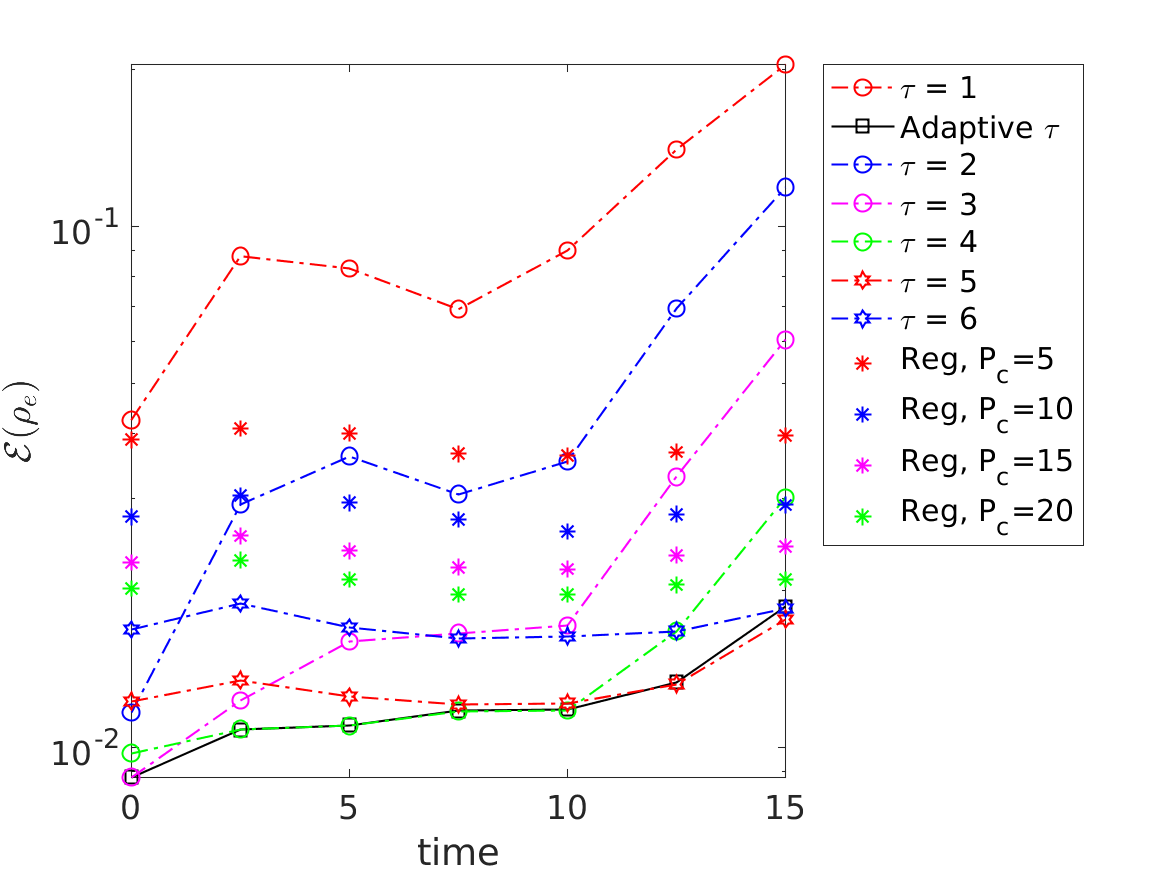}
}
\subfigure[$256^3$, $P_c=5$]{
\includegraphics[width=0.5\columnwidth]{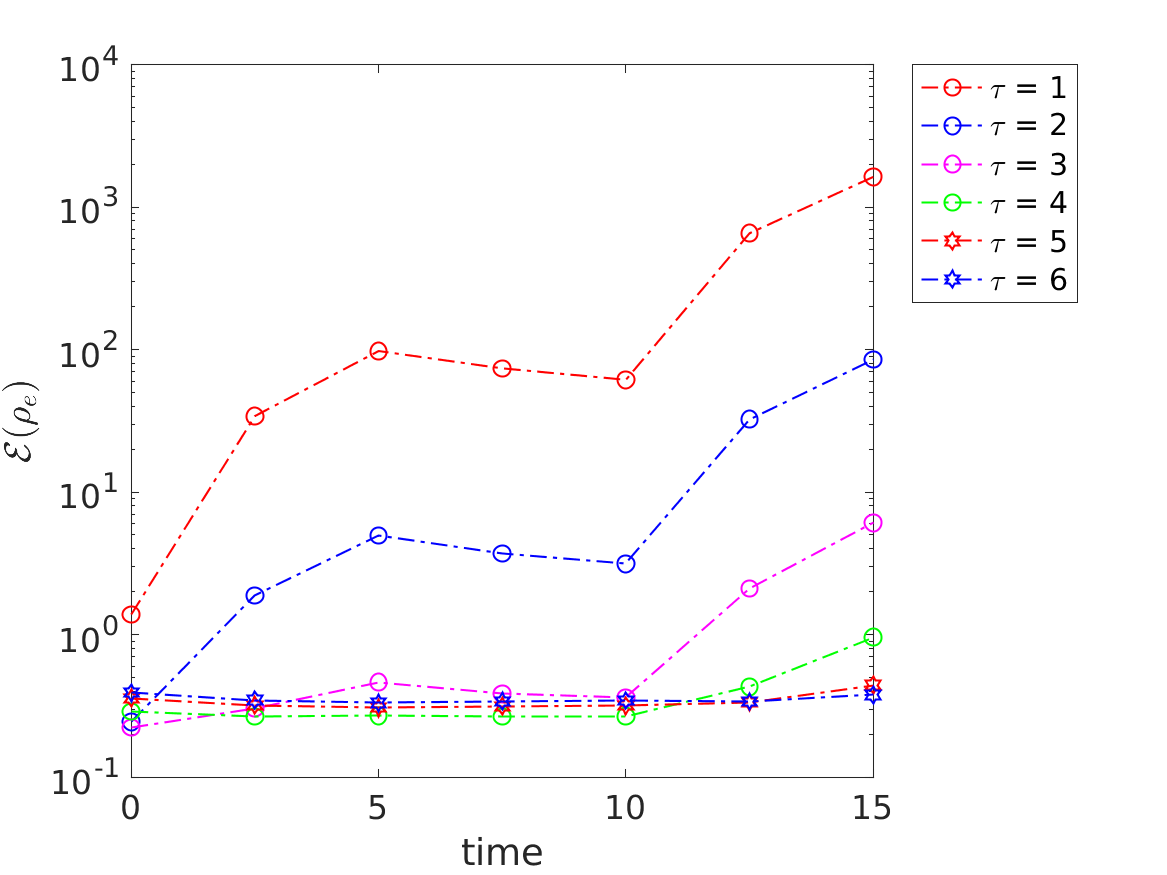}
}
    \caption{3D Penning trap: Electron charge density error comparison between regular (Reg), fixed $\tau$ and adaptive $\tau$ PIC. The left column is the actual error calculated using equation \eqnref{error_def} and the right column is the estimations from the $\tau$ estimator based on which the optimal $\tau$ is selected. The fixed as well as adaptive $\tau$ has the number of particles per cell $P_c=5$.}
\figlab{penning_pc5}
\end{figure}

    In a way analogous to Figures \figref{diocotron_pc5_gauss}-\figref{diocotron_pc20_gauss} for the diocotron instability, in Figures \figref{penning_pc1}-\figref{penning_pc5}
    we show the errors calculated using equation \eqnref{error_def} and the estimations from the $\tau$ estimator for meshes $64^3,128^3,256^3$ and 
    $P_c=1,5$. The reference in equation \eqnref{error_def} is the average of 5 independent computations of regular PIC with $256^3$ mesh and $P_c=40$. 
    For the $N_{points}$ in equation \eqnref{error_def}, we select approximately $4096$ random points throughout the domain and interpolate both the 
    reference density as well as the density under consideration at these points to measure the error. The errors are measured at 7 time instants in the simulation. 

    In general, as before the adaptive $\tau$ predictions are close to optimal and most of the conclusions from the diocotron test case are applicable in this case too. Figure \figref{penning_tau_history} shows the time history of $\tau$ for the meshes and $P_c$ considered and the 
    high values of $\tau$ indicate that the total error is dominated by the grid-based error in these cases. 

%\begin{figure}[h!t!b!]
%\subfigure[$64^3$, $P_c=1$]{
%\includegraphics[width=0.3\columnwidth]{Penning_tau_time_history_ppc_1_h_6.png}
%}
%\subfigure[$128^3$, $P_c=1$]{
%\includegraphics[width=0.3\columnwidth]{Penning_tau_time_history_ppc_1_h_7.png}
%}
%\subfigure[$256^3$, $P_c=1$]{
%\includegraphics[width=0.3\columnwidth]{Penning_tau_time_history_ppc_1_h_8.png}
%}
%\subfigure[$64^3$, $P_c=5$]{
%\includegraphics[width=0.3\columnwidth]{Penning_tau_time_history_ppc_5_h_6.png}
%}
%\subfigure[$128^3$, $P_c=5$]{
%\includegraphics[width=0.3\columnwidth]{Penning_tau_time_history_ppc_5_h_7.png}
%}
%\subfigure[$256^3$, $P_c=5$]{
%\includegraphics[width=0.3\columnwidth]{Penning_tau_time_history_ppc_5_h_8.png}
%}
%    \caption{3D Penning trap: Time history of $\tau$ for different mesh sizes and 
%    particles per cell $P_c$.}
%\figlab{penning_tau_history}
%\end{figure}

\begin{figure}[h!t!b!]
\subfigure[$64^3$, $P_c=1$]{
\includegraphics[width=0.3\columnwidth]{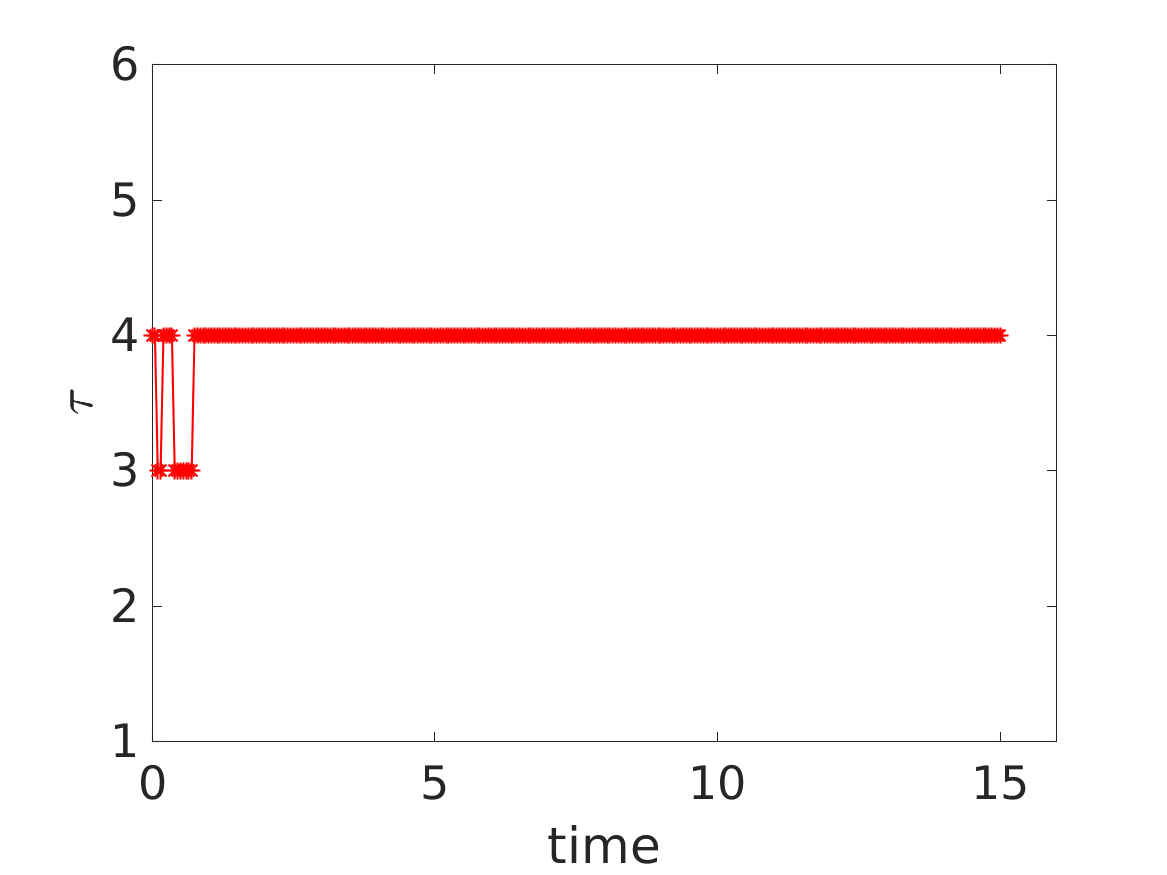}
}
\subfigure[$128^3$, $P_c=1$]{
\includegraphics[width=0.3\columnwidth]{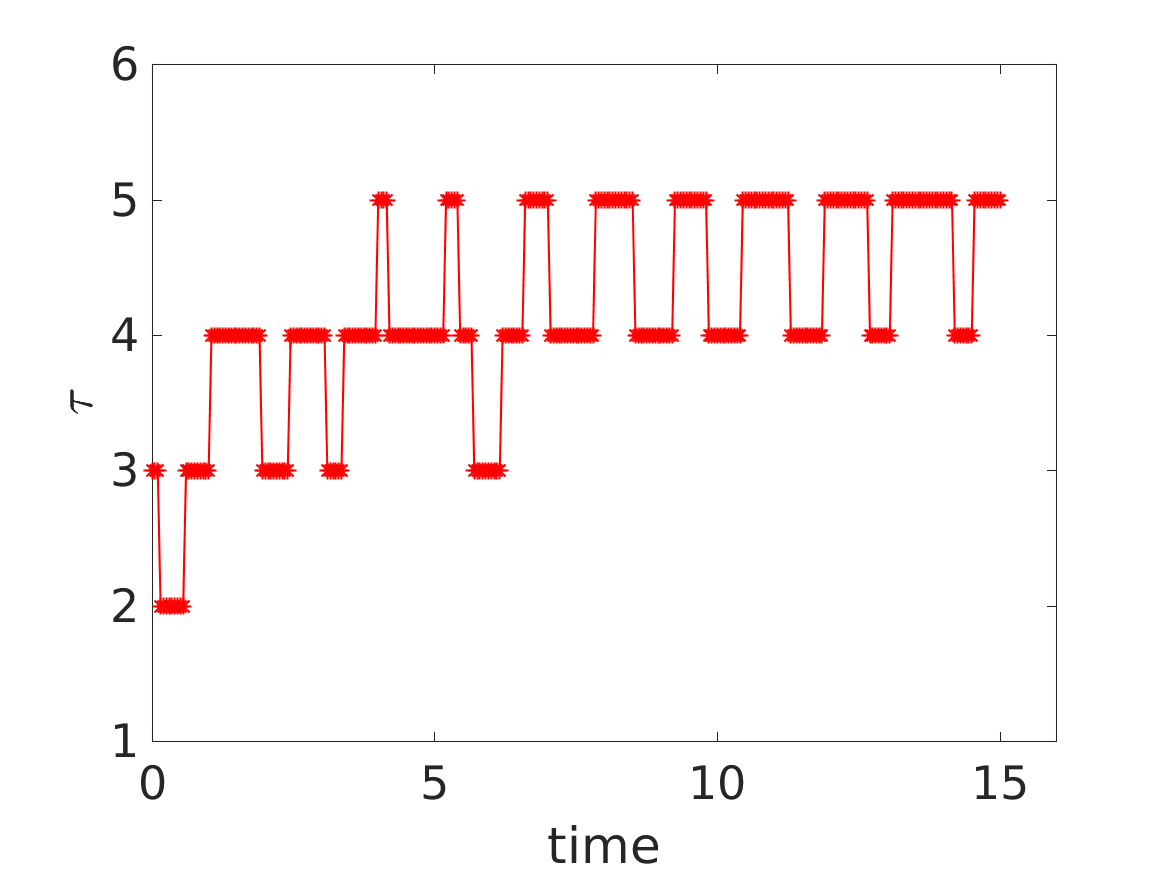}
}
\subfigure[$256^3$, $P_c=1$]{
\includegraphics[width=0.3\columnwidth]{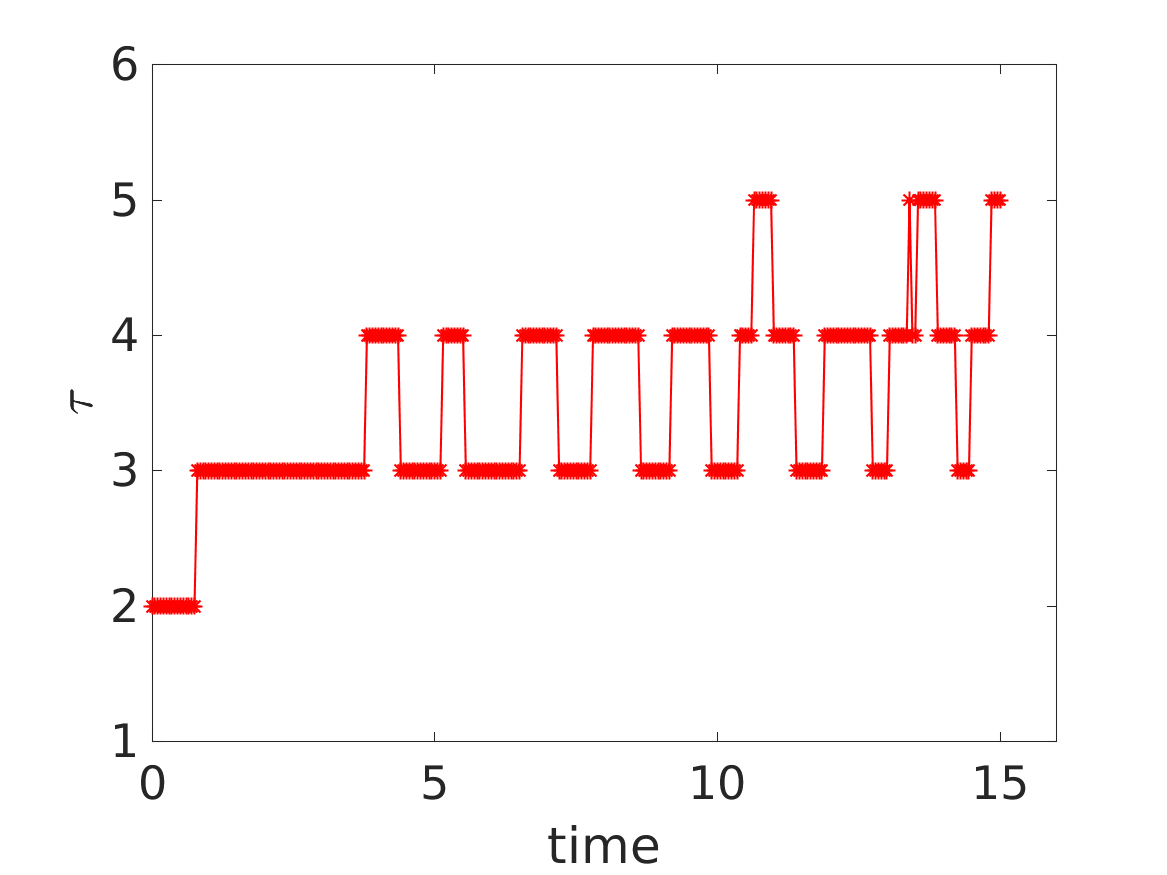}
}
\subfigure[$64^3$, $P_c=5$]{
\includegraphics[width=0.3\columnwidth]{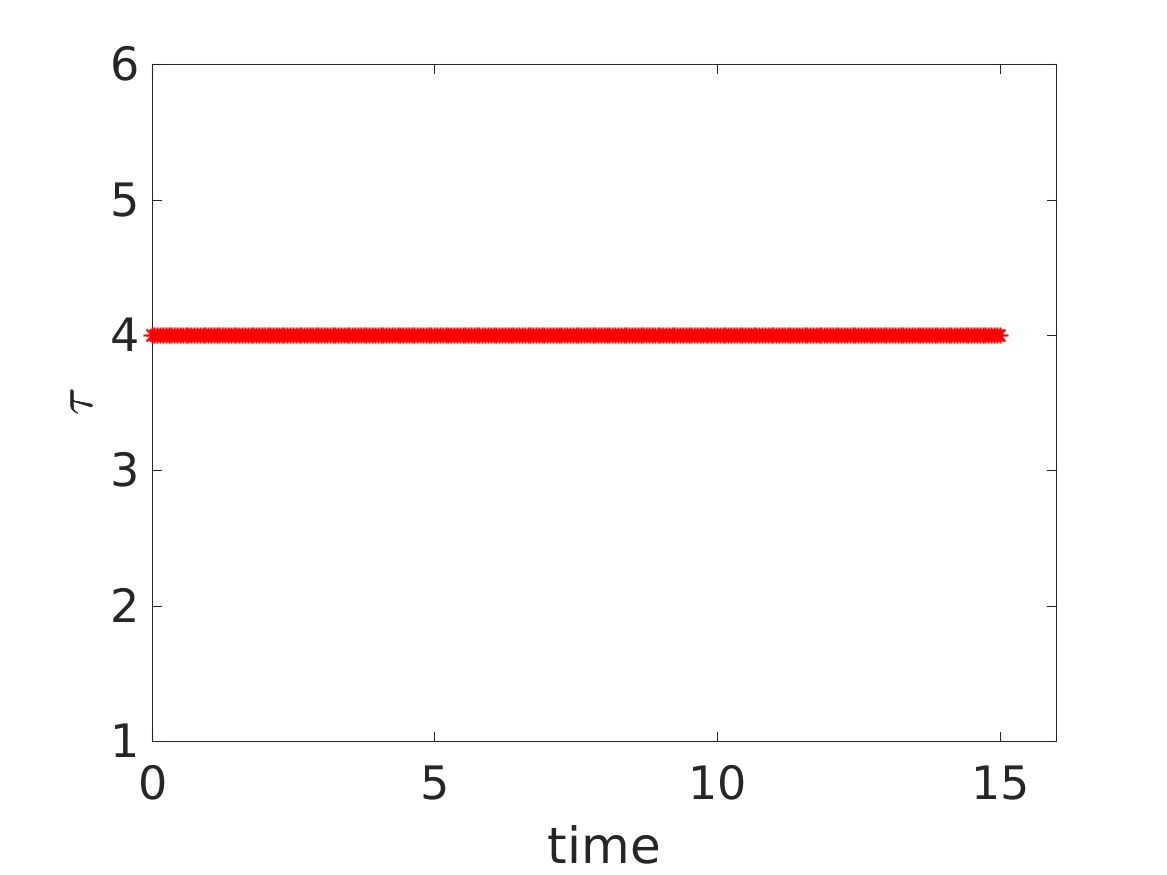}
}
\subfigure[$128^3$, $P_c=5$]{
\includegraphics[width=0.3\columnwidth]{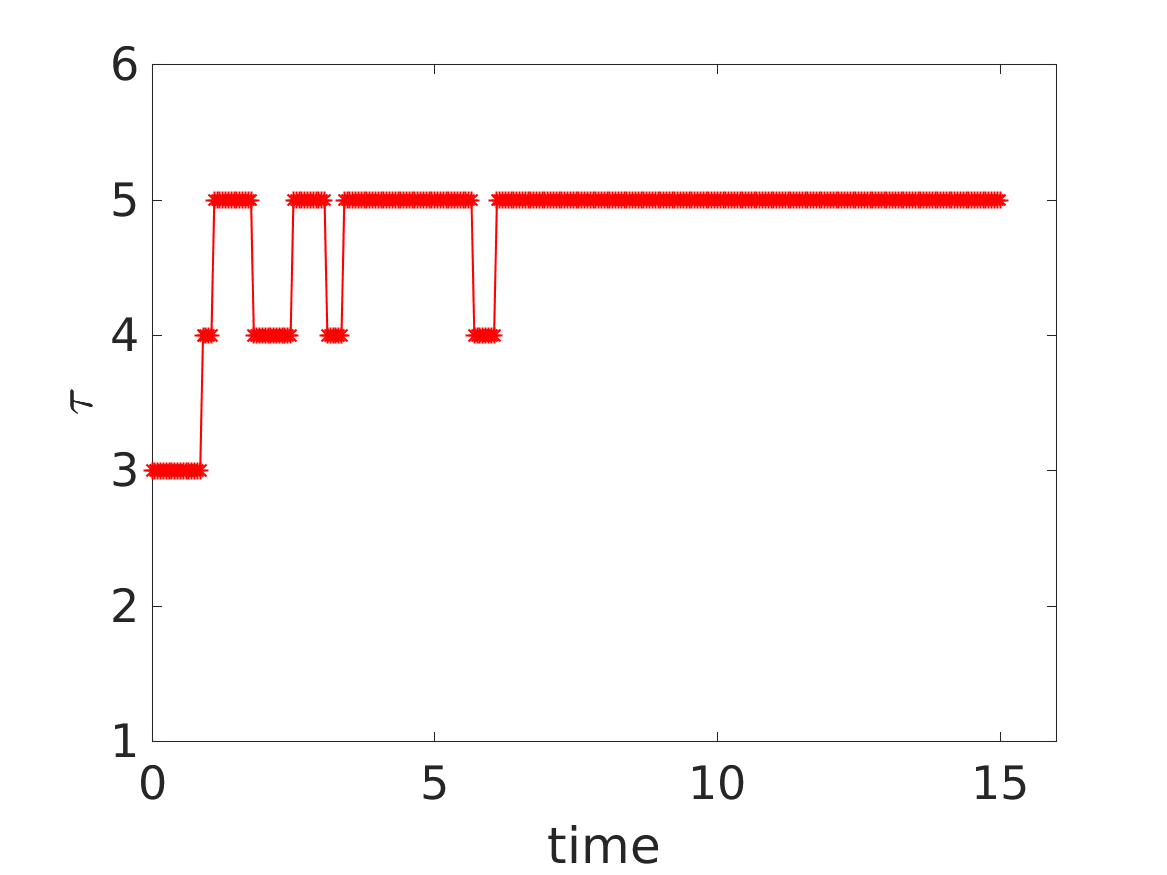}
}
\subfigure[$256^3$, $P_c=5$]{
\includegraphics[width=0.3\columnwidth]{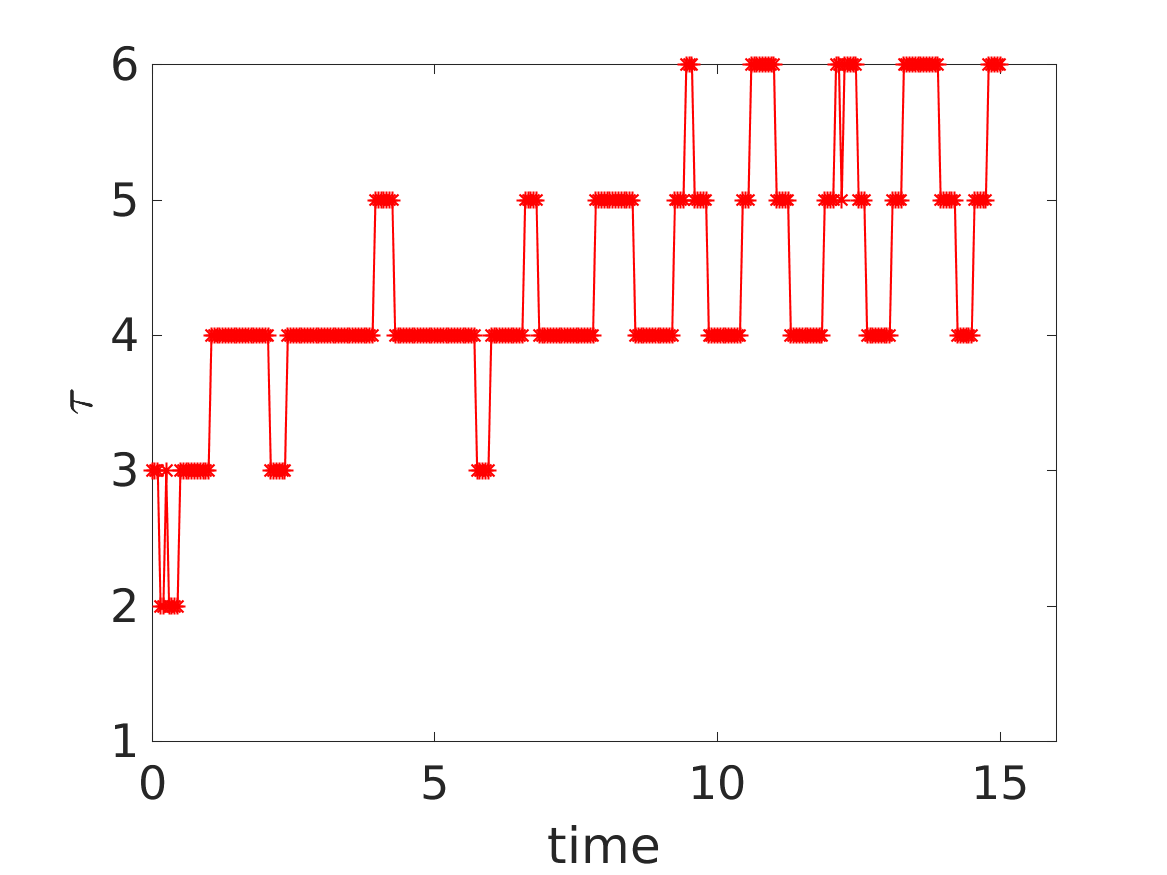}
}
    \caption{3D Penning trap: Time history of $\tau$ for different mesh sizes and 
    number of particles per cell $P_c$.}
\figlab{penning_tau_history}
\end{figure}

    In terms of run time performance comparisons, we ran the $64^3,128^3$ mesh cases on 64 cores and the $256^3$ test cases on 512 cores for both the regular 
    and adaptive $\tau$ PIC. %Similar to the 2D diocotron test case, in this 3D problem also the sparse routines together accounted for approximately less than $15$ percent of the total run time. 
    For $64^3$ mesh, at the last time instant we can see that the regular PIC is more accurate than the adaptive $\tau$ or any other fixed $\tau$ PIC. 

\begin{table}[h!b!t!]%[0.5\textwidth]
\centering
\begin{tabular}{|r|c|c||c|c||c|c|}
\hline
    & \multicolumn{2}{c||}{Adaptive $\tau$} & \multicolumn{2}{c||}{Regular} & \multicolumn{2}{c|}{Reg/adaptive $\tau$}\\
\hline
\cline{2-7}
    $128^3$ & 361 ($1$) & 476 ($5$) & 273 ($5$) & 444 ($10$) & 0.8 & 0.9 \\
    $256^3$ & 829 ($1$) & 1201 ($5$) & 2360 ($15$) & 3081 ($20$) & 2.8 & 2.6 \\
\hline
\end{tabular}
    \caption{\label{tab:penning_runtime}3D Penning Trap: Total run time in seconds on 64 cores for $128^3$ mesh and 512 cores for $256^3$ mesh in case of the regular and adaptive $\tau$ PIC. The values within the parentheses represent the different number of particles per cell required to reach a comparable accuracy (based on visual norm from the left columns of Figures \figref{penning_pc1}-\figref{penning_pc5}) in the charge density for both the schemes at final time $T=15$. Columns $6-7$ are the ratio of time taken
    by the regular PIC to that for adaptive $\tau$ PIC.}
\end{table}

    For $128^3$ and $256^3$ meshes, from Table \ref{tab:penning_runtime} we can see a maximum speedup of $2.8$ with adaptive $\tau$ PIC over the regular PIC for the finest mesh size. Again considering the number of particles as a measure for the memory cost adaptive $\tau$ PIC is $2-15$ times cheaper than the regular PIC. In order to see
    more computational benefits with the adaptive $\tau$ PIC for this problem we need to perform runs with finer meshes and more particles per cell. These 3D large-scale simulations are part of our future work and the results will be reported elsewhere.

 %   We also ran both the diocotron and the Penning trap test cases with different initial seeds for all the meshes and $P_c$ and noticed the results and conclusions are more or less the same. This shows the robustness of the adaptive $\tau$ strategy.

\section{Conclusions}
\seclab{conclusions}

    We have proposed a sparse grids based adaptive noise reduction strategy for particle-in-cell (PIC) simulations. Unlike the typical physical or Fourier
    domain filters used in PIC methods, the strategy adapts to mesh size, number of particles per cell, smoothness of the charge density and the initial sampling 
    technique. In order to construct the strategy we use the key idea of increased particles per cell in sparse grids compared to the regular grid for the same
    total number of particles as proposed in \cite{ricketson2016sparse}. The current work extends that concept in several directions. Specifically, 
    we present a filtering perspective for the sparse grids based noise reduction which helps to incorporate it with ease in existing high performance 
    large-scale PIC code bases and also opens the door for sparse grids based filtering approaches. We tackle the problem of large grid-based error of sparse grids for non-aligned and non-smooth functions by means of the truncated combination technique \cite{leentvaar2008pricing,benk2012hybrid,benk2012variants}.
    We show in the context of PIC simulations that the truncated combination technique provides a natural framework to minimize the sum of grid-based error and
    particle noise. This allows us to propose a heuristic based on formal error analysis to select the optimal truncation parameter on the fly that 
    minimizes the total error.

    We show the performance and applicability of our strategy on two benchmark problems; namely the 2D diocotron instability and electron dynamics in a 3D Penning trap. In both test cases the adaptive noise reduction strategy picks a truncation parameter which is close to optimal for all times. To achieve comparable accuracy for the charge density deposition we obtain significant speedups and memory savings up to an order of magnitude with the noise reduction technique compared to regular PIC 
    in the 2D diocotron test case. For the 3D Penning trap test case a maximum speedup of $2.8$ and $15$ times memory reduction is obtained for the 
    finest mesh size tested. Further speedups and memory reduction in the 3D test case require us to test the strategy for even finer resolutions and
    that is part of future work.

    Our strategy can be very easily integrated into existing high performance large-scale PIC code bases and ongoing work is to integrate it into the
    open source particle accelerator library OPAL \cite{adelmann2019opal}. In terms of future work, we plan on investigating the applicability and performance
    of the noise reduction strategy on large-scale high intensity particle accelerator simulations such as the IsoDAR project \cite{bungau2012proposal,yang2013beam} with a particular focus on understanding the dynamics of halo particles and efficient collimation strategies. Filtering strategies have much more impact on the electromagnetic PIC simulations as reported in \cite{vay2011numerical}. Hence we would like to extend the current approach for Vlasov-Maxwell equations and investigate the
    performance in that context. Use of machine learning approaches to tune denoising threshold in our strategy is also of interest. Finally, we also intend to compare the current strategy with other filtering approaches and denoising techniques. 

\section*{Acknowledgments}
    This project has received funding from the European Union's Horizon 2020 research and innovation program under the Marie Sk{\l}odowska-Curie grant agreement No. 701647 and from the United States National Science Foundation under Grant No. PHY-1820852. L.F.\ Ricketson's work was performed under the auspices of the U.S. Department of Energy by Lawrence Livermore National Laboratory under Contract DE-AC52-07NA27344. Lawrence Livermore National Security, LLC.  We are grateful for the support. The first author would like to thank Dr. Weiqun Zhang for help with the AMReX related queries.

\section*{Appendix A. Proof of Proposition \proporef{MC_filtering_equivalence} relating the direct charge density deposition onto the component grids and the two-step approach}
\seclab{equivalence_proof}
\begin{proof}
Even though sparse grids make sense only for dimensions 2 and higher we can
still understand the essence of the proof in 1D. Also, since the shape functions and 
transfer operators in 2D and 3D are obtained by the tensor product of 1D
linear interpolation functions the proof extends naturally to those cases.

\begin{figure}[h!b!t!]
%\vspace{-35mm}
\subfigure[Node-centered grid]{
\includegraphics[width=0.49\columnwidth]{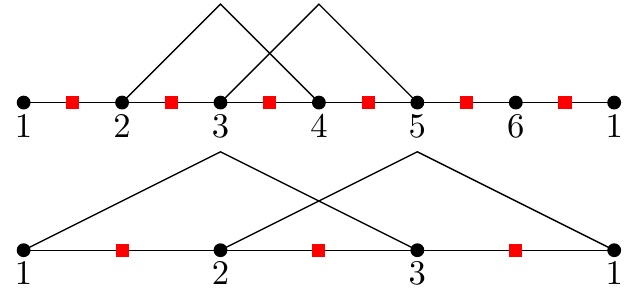}
        \figlab{node_center}
}
\subfigure[Cell-centered grid]{
\includegraphics[width=0.49\columnwidth]{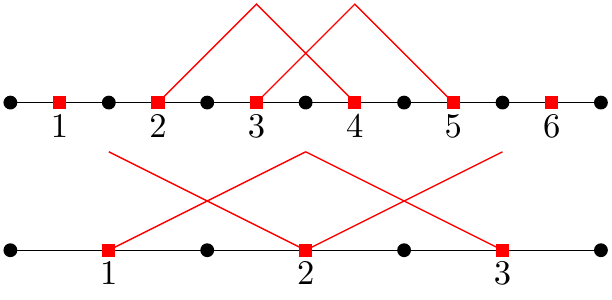}
        \figlab{cell_center}
}
        \caption{Schematic showing the node-centered and cell-centered grids and the corresponding shape functions. The nodes are marked with black circles and the cell-centers with red squares. The domain is periodic. The shape functions $W_c$ corresponding to the coarse grid are linear between the nodes in the fine grid
        in case of node-centered grids. For cell-centered grids $W_c$ has discontinuity in derivative between some of the cell-centers in the fine grid 
        whereas between nodes of the fine grid it is always linear.}
\figlab{node_cell_center}
%    \vspace{-5mm}
\end{figure}

    Consider a periodic 1D domain $[0,L]$ and two grids with mesh sizes $h_f$ and $h_c$. The grid with mesh size $h_c$ is coarser than the one with $h_f$ and assume $h_c$ is an integer multiple of $h_f$. Let us first consider the node-centered grids where all the coarse grid points are also grid points in the fine grid as shown in Figure \figref{node_center}.

    The particles deposit onto the fine grid with mesh size $h_f$ and the charge density $\rhoa_e$ is given by 
\begin{equation}
    \eqnlab{rho_hf}
        \rhoa_e(\tilde x_j) = \frac{Q_e}{N_p h_f} \sum_{p=1}^{N_p} W_f(\tilde x_j - x_p), 
\end{equation}
    where $W_f(\zeta) = \max\LRc{0,1-\frac{|\zeta|}{h_f}}$ is the cloud-in-cell shape function and $x_p$ and $\tilde x_j$ are the locations of the particles and the grid points in the fine grid respectively. 
    Now, we transfer the density $\rhoa_e$ to the coarse grid by means of the transfer operator $R$ in equation \eqnref{restriction} which gives 
\begin{equation}
    \eqnlab{rho_hc}
        \varrho_e(x_k) = \frac{h_f}{h_c} \sum_{j=1}^{N_c} \rhoa_e(\tilde x_j) W_c(x_k - \tilde x_j), 
\end{equation}
where $W_c(\zeta) = \max\LRc{0,1-\frac{|\zeta|}{h_c}}$, $x_k$ are the locations of the grid points in the coarse grid and $N_c$ is the total number of cells in the fine grid. Substituting for $\rhoa_e$ from equation \eqnref{rho_hf} and switching the order of sums we get
\begin{equation}
    \eqnlab{rho_hc2}
        \varrho_e(x_k) = \frac{Q_e}{N_ph_c} \sum_{p=1}^{N_p}\sum_{j=1}^{N_c} W_c(x_k - \tilde x_j)W_f(\tilde x_j - x_p). 
\end{equation}
Now, for a given particle, $W_f(\tilde x_j - x_p)$ is non-zero for exactly two values of $j$: the floor of $x_p/h_f$ and the ceiling of that same quantity. Let us call these values $J$ and $J + 1$ and assume the grid points are ordered such that $x_J$ is to the left of $x_{J+1}$. We have
\begin{align}
    \sum_{j=1}^{N_c} W_c(x_k - \tilde x_j)W_f(\tilde x_j - x_p) &= W_c(x_k - \tilde x_J)W_f(\tilde x_J - x_p) \nonumber \\&+ W_c(x_k - \tilde x_{J+1})W_f(\tilde x_{J+1} - x_p). \eqnlab{sum_wc_wf}
\end{align}

Now we note that because of the way the two grids are related (mesh sizes are integer
multiples, coincident grid points), we are guaranteed that $W_c(x_k-\tilde x)$ is linear on the interval $\tilde x\in[\tilde x_J,\tilde x_{J+1}]$. This is because the places where $W_c$ has a discontinuity in its derivative are
    guaranteed to be fine grid points as shown in Figure \figref{node_center}. So, linear interpolation is exact for $W_c$ on the interval $[\tilde x_J,\tilde x_{J+1}]$.
Since $x_p$ is in this interval, we have
\[
    W_c(x_k - x_p) = W_c(x_k - \tilde x_J) \LRs{\frac{\tilde x_{J+1}-x_p}{\tilde x_{J+1}-\tilde x_{J}}} + W_c(x_k - \tilde x_{J+1}) \LRs{1-\frac{\tilde x_{J+1}-x_p}{\tilde x_{J+1}-\tilde x_{J}}}.
\]
Now we notice that
\[
    \LRs{\frac{\tilde x_{J+1}-x_p}{\tilde x_{J+1}-\tilde x_{J}}} = \LRs{\frac{\tilde x_{J} + h_f -x_p}{h_f}} = 1+\frac{\tilde x_J - x_p}{h_f} = 1-\frac{|\tilde x_J - x_p|}{h_f} = W_f(\tilde x_J - x_p),
\]
and a nearly identical reasoning gives 
\[
    \LRs{1-\frac{\tilde x_{J+1}-x_p}{\tilde x_{J+1}-\tilde x_{J}}} = W_f(\tilde x_{J+1} - x_p).
\]
Combining these with equation \eqnref{sum_wc_wf} we get 
\begin{equation}
\eqnlab{sum_wc_wf2}
    \sum_{j=1}^{N_c} W_c(x_k - \tilde x_j)W_f(\tilde x_j - x_p) = W_c(x_k - x_p).
\end{equation}
Substituting this into equation \eqnref{rho_hc2} we get the density on the coarse grid as 
\begin{equation}
    \eqnlab{rho_hc_equivalent}
        \varrho_e(x_k) = \frac{Q_e}{N_p h_c} \sum_{p=1}^{N_p} W_c(x_k - x_p). 
\end{equation}
Comparing equation \eqnref{rho_hc_equivalent} with equation \eqnref{rho_hf} we see this is exactly the expression we would obtain if the particles were to deposit directly onto the coarse grid with mesh size $h_c$.

Now we will consider the cell-centered grids. In this case the coarse grid points are also not the grid points in the fine grid and $W_c$ will have a discontinuity in the derivative for some of the intervals $[\tilde x_j,\tilde x_{j+1}]$ as shown in Figure \figref{cell_center} depending on the ratio $h_c/h_f$. Hence an exact equivalence between
    the two approaches does not hold. However, we will now show that $$ \mc{W}_c(x_k-x) = \sum_{j=1}^{N_c} W_c(x_k - \tilde x_j)W_f(\tilde x_j - x)$$ can be considered as a 
    shape function by itself. To that end, we will show that it satisfies the three conditions for any shape function as given in \cite{filbet2003numerical}. These are listed as follows
    \begin{enumerate}
        \item $\mc{W}_c(\zeta) = \mc{W}_c(-\zeta)$,
        \item $\frac{1}{h_c}\int \mc{W}_c(\zeta) d\zeta = 1$,
        \vspace{0.1cm}
        \item $\sum\limits_{k} \mc{W}_c(x_k - x) = 1$.
    \end{enumerate}

The first condition is manifestly true as $W_f$ which is the standard hat function is even. For the second condition we observe that 
\[
    \frac{1}{h_c}\int \mc{W}_c(\zeta) d\zeta = \frac{h_f}{h_c} \sum_{j=1}^{N_c}W_c(x_k - \tilde x_j),
\]
as $W_f$ is a shape function and by definition integrates to $h_f$. Now, $h_f\sum_{j=1}^{N_c}W_c(x_k - \tilde x_j)$ is the midpoint rule applied for the integration $\int W_c(x_k - \tilde x)$ over the fine grid. From Figure \figref{cell_center} it is clear that $W_c$ is linear on each integration cell and the midpoint rule is exact. Thus,
\[
    \frac{h_f}{h_c} \sum_{j=1}^{N_c}W_c(x_k - \tilde x_j) = \frac{1}{h_c} \int W_c(x_k - \tilde x) d \tilde x = 1,
\]
where the last step comes from the fact that $W_c$ which is also a standard hat function integrates to $h_c$ by definition. Finally, the third condition
is related to global charge conservation and we note that since $W_c$ and $W_f$ are standard hat functions they satisfy the partition of unity and hence
$\mc{W}_c$ also satisfies it when we carry out the summation.

Now, using conditions 1 and 2 and noting that $\mc{W}_c$ is bounded in $[0,L]$ we can carry out the same set of steps shown in appendix B for a standard hat function. We can then see the grid-based error for $\mc{W}_c$ is of $\mc{O}(\abs{\ppx\rhoe}h_c^2)$ and the particle noise is $\mc{O}(\sqrt{\abs{Q_e\rhoe}/N_p h_c})$ as in equations \eqnref{1D_grid} and \eqnref{1D_noise} but with the constants depending on the ratio of $h_c$ to $h_f$. 
\end{proof}

\section*{Appendix B. Grid-based and particle errors in the charge density deposition for regular PIC schemes}

In this section, we follow the analysis in \cite{ricketson2016sparse} and derive in details the grid-based error and noise estimates for the  
charge density deposition in regular PIC schemes explicitly revealing the constants.  
For simplicity, let us consider a 1D PIC scheme and extensions to 2D and 3D 
are relatively straightforward. In the following, we consider a particular time instant and hence suppress the dependence of the different 
quantities with respect to time.

Let $\f(x,\v)$ be the electron phase-space distribution under consideration and let us define $\fh$ as
\[
    \fh = \frac{f}{\int\int\f dx d\v}.
\]
Since, $\fh$ is non-negative and its phase-space integral is unity it can be interpreted as probability density.
The exact charge density $\rhoe(x)$ is given by
\begin{align}
    \rhoe(x) &= q_e\int\int \f({\xi},\v) \delta(x-\xi) d\xi d\v, \\
             &= q_e\LRp{\int\int\f dx d\v} \int\int \fh(\xi,\v) \delta(x-\xi) d\xi d\v, \\
             &= Q_e\int\int \fh(\xi,\v) \delta(x-\xi) d \xi d\v,
\end{align}
where $Q_e = q_e\int\int\f dx d\v$ is the total electron charge in the system and $\delta(x-\xi)$ is the Dirac delta function.

In PIC, we approximate $\delta(x-\xi)$ with the shape function $S(x-\xi)$
which for our discussion here consider it to be the cloud-in-cell or linear interpolation function. The approximate charge density $\bar\rho_e$ with the shape function $S(x-\xi)$ is given by

\begin{align}
    \eqnlab{rhoa1}
    \bar\rho_e(x) &= Q_e\int\int \fh(\xi,\v) S(x-\xi) d\xi d\v,\\
    \eqnlab{rhoa2}
             &= Q_e\Ebb_{\fh(\xi,\v)}\LRs{S(x-\xi)},
\end{align}
where $\Ebb$ is the expected value over the probability density $\fh$.

\subsection*{B.1 Grid-based error estimate}

This is the error due to approximating $\delta(x-\xi)$ with the shape
function $S(x-\xi)$ 

\begin{equation}
    \eqnlab{eg}
    e_g = \abs{\rhoe - \bar\rho_e}.
\end{equation}
Towards estimating this error, let us expand $\fh(\xi,\v)$ in equation \eqnref{rhoa1} in terms of Taylor's series about $x$,

\begin{align}
    \bar\rho_e &= Q_e\int\int\LRp{\fh(x,\v) + (\xi-x)\px\fh(x,\v) \right.\nonumber \\&+ \left. \frac{(\xi-x)^2}{2}\ppx\fh(x,\v) + \cdots}S(x-\xi) d\xi dv, \\
            &= \underbrace{Q_e\int\fh dv}_{\rhoe}\underbrace{\int S(x-\xi) d\xi}_{1} + Q_e\int \px\fh dv \int (\xi-x)S(x-\xi)d\xi  \nonumber \\
            &+ Q_e\int \ppx\fh dv \int \frac{(\xi-x)^2}{2}S(x-\xi)d\xi + \cdots, 
\end{align}
where we have used the fact that the integral of the shape function $S(x-\xi)$ is unity. In the above equations we have used the short hand notations $\px=\frac{\partial(.)}{\partial x}$ and $\ppx=\frac{\partial^2(.)}{\partial x^2}$. Taking outside the partial derivatives with respect to $x$ in the
$\int dv$ integrals we get

\begin{equation}
    \eqnlab{rhoa}
    \bar\rho_e = \rho_e + \px\rhoe \int (\xi-x) S(x-\xi) d\xi
              + \ppx \rhoe \int \frac{(\xi-x)^2}{2} S(x-\xi)d\xi
              + \cdots.
\end{equation}
The cloud-in-cell shape function is given by 
\begin{equation}
    S(\zeta) = \frac{1}{h_x}\max\LRc{0, 1-\frac{\abs{\zeta}}{h_x}}.
\end{equation}
Performing a change of variables with $\zeta=\xi-x$ in equation \eqnref{rhoa} and noting that $S(\zeta)$ has a compact support and is
zero outside $|\zeta| \leq h_x$ all the integrals has to be carried 
only in $-h_x \leq \zeta \leq h_x$.

Also, $S(\zeta)$ is an even function and hence $\int_{-h_x}^{h_x} \zeta S(\zeta) d\zeta$ which is the second term in equation \eqnref{rhoa} is 0. However, the integrand in the third term of the equation \eqnref{rhoa} is an even function and it evaluates to

\[
    \int \frac{(\xi-x)^2}{2} S(x-\xi) d\xi = \frac{1}{h_x} \int_0^{h_x} \zeta^2 \LRp{1-\frac{\zeta}{h_x}} d\zeta = \frac{h_x^2}{12}.
\]
Thus equation \eqnref{eg} becomes 

\begin{align}
    \nonumber
    e_g(x) &\leq \frac{h_x^2}{12}\abs{\ppx\rhoe(x)}  + \cdots, \\ \eqnlab{1D_grid}
    e_g    &= \mc{O}\LRp{\frac{h_x^2}{12}\abs{\ppx\rhoe(x)}}.
\end{align}

Since, the cloud-in-cell shape functions in 2D and 3D are obtained by the 
tensor product of 1D shape functions the analysis extends easily to these
cases. Carrying out similar steps we obtain the grid-based error for 2D
and 3D as
\begin{align}
    e_g &=
        \mc{O}\LRp{\frac{1}{12}\LRc{\abs{\frac{\partial^2 \rhoe}{\partial x^2}} h_x^2 + \abs{\frac{\partial^2 \rhoe}{\partial y^2}} h_y^2} + \frac{1}{144} \abs{\frac{\partial^4 \rhoe}{\partial x^2 \partial y^2}} h_x^2 h_y^2} \quad \text{in 2D,} \eqnlab{2D_grid} \\
        e_g &= \mc{O}\LRp{\frac{1}{12}\LRc{\abs{\frac{\partial^2 \rhoe}{\partial x^2}} h_x^2 + \abs{\frac{\partial^2 \rhoe}{\partial y^2}} h_y^2 + \abs{\frac{\partial^2 \rhoe}{\partial z^2}} h_z^2} \right. \nonumber \\&+ \left. \frac{1}{144} \LRc{\abs{\frac{\partial^4 \rhoe}{\partial x^2 \partial y^2}} h_x^2 h_y^2 + \abs{\frac{\partial^4 \rhoe}{\partial y^2 \partial z^2}} h_y^2 h_z^2 + \abs{\frac{\partial^4 \rhoe}{\partial z^2 \partial x^2}} h_z^2 h_x^2} \right. \nonumber \\&+ \left.
        \frac{1}{1728} \abs{\frac{\partial^6
             \rhoe}{\partial x^2 \partial y^2 \partial z^2}} h_x^2 h_y^2 h_z^2} \quad \text{in 3D}. \eqnlab{3D_grid}
\end{align}
Note in the above equations the reason for including the only higher order
terms proportional to the mixed derivatives is because these terms will contribute to the dominant error for the sparse grid combination. Hence, the
constants in front of these terms are of interest for estimating the 
coefficients of the grid-based error in section \secref{heuristic}.
\subsection*{B.2 Noise estimate}

This is the error which occurs when we approximate the expected value of
the shape function by means of an arithmetic mean over the number of discrete particles.
Thus equation \eqnref{rhoa2} becomes

\begin{equation}
    \eqnlab{rhoaa1}
    \bar\rho_e(x) \approx \rhoa_e(x) = \frac{Q_e}{N_p} \sum_p S(x-x_p).
\end{equation}
The error incurred by this approximation $\eta(x)$ is a random variable with
mean 0 and variance given by 
\begin{align}
    Var_{\fh}\LRs{\eta(x)} &= \Ebb_{\fh}\LRs{\LRp{\bar\rho_e - \rhoa_e}^2},\\
                       &= \bar\rho_e^2 - 2\bar\rho_e\Ebb_{\fh}[\rhoa_e] + \Ebb_{\fh}[\rhoa_e^2],\\
    &= \Ebb_{\fh}[\rhoa_e^2] - \bar\rho_e^2. \eqnlab{var_eta}
\end{align}
Here, in equation \eqnref{var_eta} we used the fact that $\Ebb_{\fh}[\rhoa_e]=\Ebb_{\fh}[\bar\rho_e]=\bar\rho_e$.
Let us compute $\Ebb_{\fh}\LRs{\rhoa_e^2}$
\begin{equation}
    \Ebb_{\fh}\LRs{\rhoa_e^2} = \Ebb_{\fh}\LRs{\frac{Q_e^2}{{N_p}^2}\LRp{\sum_p S(x-x_p)}^2}. 
\end{equation}
Similar to \cite{ricketson2016sparse} we assume that the initial particle states have been chosen by independent sampling from $\fh(t=0)$
and also they remain approximately independent for $N_p \gg 1$. Then $\Ebb_{\fh}\LRs{S(x-x_p)S(x-x_q)} = 0$ if $p \neq q$ and all the cross terms vanish giving  
\begin{align}
    \Ebb_{\fh}\LRs{\rhoa_e^2} &= \frac{Q_e^2}{{N_p}^2}\sum_p \Ebb_{\fh}\LRs{\LRp{S(x-x_p)}^2},\\
                        &= \frac{Q_e^2}{{N_p}}\Ebb_{\fh}\LRs{\LRp{S(x-x_p)}^2},
\eqnlab{E_rho_tilde}
\end{align}
where, we have used the fact that each particle has same $\Ebb_{\fh}\LRs{\LRp{S(x-x_p)}^2}$. Now, 
\begin{align}
    \frac{Q_e^2}{{N_p}}\Ebb_{\fh}\LRs{\LRp{S(x-x_p)}^2} &= \frac{Q_e^2}{{N_p}}
    \int\int\fh(x_p,v) \LRp{S(x-x_p)}^2 dx_pdv,\\
                                                &= \frac{Q_e^2}{{N_p}}
    \int\int\LRp{\fh(x,v) + (x_p-x)\px\fh(x,v) \right. \nonumber \\&+ \left. \frac{(x_p-x)^2}{2}\ppx\fh(x,v) + \cdots}\LRp{S(x-x_p)}^2 dx_pdv,
\end{align}
and similar to the previous exercise for grid-based error the term associated with $(x_p-x)\px\fh(x,v)$ vanishes and the third term evaluates to $\mc{O}(h_x)$. Hence evaluating the leading order term gives
\begin{align}
\frac{Q_e^2}{{N_p}}\int\int\fh(x,v) \LRp{S(x-x_p)}^2 dx_pdv &= \frac{Q_e}{N_p}\underbrace{\int Q_e \fh dv}_{\rhoe}\int\LRp{S(x-x_p)}^2 dx_p, \\
    &= \frac{Q_e\rhoe}{N_p} \frac{2}{h_x^2}\int_0^{h_x} \LRp{1-\frac{\zeta}{h_x}}^2 d\zeta,\\
    &= \frac{2}{3} \frac{Q_e\rhoe}{N_p h_x}.
\end{align}
Plugging the above term in equation \eqnref{E_rho_tilde} gives 
\begin{equation}
    \eqnlab{rho_tilde_e_square}
    \Ebb_{\fh}\LRs{\rhoa_e^2} = \frac{2}{3} \frac{Q_e\rhoe}{N_p h_x} + \mc{O}(h_x) +
    \cdots.
\end{equation}
Omitting the $\bar\rho_e^2$ term in equation \eqnref{var_eta} as it is small compared to equation \eqnref{rho_tilde_e_square} and substituting the above expression gives 
\[
Var_{\fh}\LRs{\eta(x)} \approx \frac{2}{3} \frac{Q_e\rhoe}{N_p h_x}. 
\]
Defining the particle noise error $e_n$ as the standard deviation of the random variable $\eta$ we get 
\begin{equation}
    \eqnlab{1D_noise}
    e_n(x) = \mc{O}\LRp{\sqrt{\frac{2}{3} \frac{\abs{Q_e\rhoe(x)}}{N_p h_x}}}.
\end{equation}

Similarly, carrying out the same set of steps in 2D and 3D we get the 
estimates for the particle noise as
\begin{align}
    e_n &= \mc{O}\LRp{\sqrt{\frac{4}{9} \frac{\abs{Q_e\rhoe}}{N_p h_x h_y}}} \quad \text{in 2D}, \eqnlab{2D_noise} \\
    e_n &= \mc{O}\LRp{\sqrt{\frac{8}{27} \frac{\abs{Q_e\rhoe}}{N_p h_x h_y h_z}}} \quad \text{in 3D}. \eqnlab{3D_noise}
\end{align}

\section*{References}

\bibliographystyle{elsarticle-num}

\bibliography{references}

\end{document}